\documentclass{article}
\usepackage[hidelinks]{hyperref}
\usepackage{amssymb,ComplexSystems}
\usepackage{graphicx,booktabs,mathtools,longtable}

\newtheorem{algo}{Algorithm}

\newtheorem{remark}{Remark}

\newcommand{\FF}{\mathbb{F}_2}
\newcommand{\ZZ}{\mathbb{Z}}
\newcommand{\NN}{\mathbb{N}}

\newcommand{\Inc}{\operatorname{Inc}}
\newcommand{\conv}{\ast}
\newcommand{\symd}{\bigtriangleup}
\newcommand{\Xor}{\bigoplus}
\newcommand{\xor}{\oplus}
\newcommand{\pc}{\operatorname{pc}}
\newcommand{\Bl}[2]{B(#1,#2)}
\newcommand{\bin}[2]{\binom{#1}{#2}}
\newcommand{\Zet}{\zeta}

\providecommand{\mathscr}[1]{\mathcal{#1}}

\newcommand{\src}[1]{}
\newcommand{\roadmap}[1]{}
\newcommand{\paperhead}[2]{}

\makeatletter
\renewcommand{\verbatim@font}{\normalfont\ttfamily\scriptsize}
\providecommand*{\toclevel@title}{0}
\makeatother

\renewcommand{\mycopyright}{\hbox{\textcopyright}}

\newcommand{\Up}[1]{U(#1)}
\newcommand{\keepstatement}{\par\begingroup
\dimen0=\pagegoal\advance\dimen0 by -\pagetotal
\ifdim\dimen0<9\baselineskip\newpage\fi\endgroup}
\newcommand{\keepsection}{\par\begingroup
\dimen0=\pagegoal\advance\dimen0 by -\pagetotal
\ifdim\dimen0<14\baselineskip\newpage\fi\endgroup}

\newcommand{\EmbeddedFigure}[2]{%
  \resizebox{#1}{!}{%
    \hbox to \csname EFWidth#2\endcsname bp{%
      \vrule width0pt height\csname EFHeight#2\endcsname bp depth0pt
      \pdfliteral{q /EF#2 Do Q}\hss}}}
\begingroup
\catcode`\#=12 \catcode`\%=12 \catcode`\_=12
\catcode`\^=12 \catcode`\~=12 \catcode`\&=12 \catcode`\$=12
\pdfobj reserveobjnum
\expandafter\xdef\csname EF1O1\endcsname{\the\pdflastobj}
\pdfobj reserveobjnum
\expandafter\xdef\csname EF1O2\endcsname{\the\pdflastobj}
\pdfobj reserveobjnum
\expandafter\xdef\csname EF1O3\endcsname{\the\pdflastobj}
\pdfobj reserveobjnum
\expandafter\xdef\csname EF1O4\endcsname{\the\pdflastobj}
\pdfobj reserveobjnum
\expandafter\xdef\csname EF1O5\endcsname{\the\pdflastobj}
\pdfobj reserveobjnum
\expandafter\xdef\csname EF1O6\endcsname{\the\pdflastobj}
\pdfobj reserveobjnum
\expandafter\xdef\csname EF1O7\endcsname{\the\pdflastobj}
\pdfobj reserveobjnum
\expandafter\xdef\csname EF1O8\endcsname{\the\pdflastobj}
\pdfobj reserveobjnum
\expandafter\xdef\csname EF1O9\endcsname{\the\pdflastobj}
\pdfobj reserveobjnum
\expandafter\xdef\csname EF1O10\endcsname{\the\pdflastobj}
\pdfobj reserveobjnum
\expandafter\xdef\csname EF1O11\endcsname{\the\pdflastobj}
\pdfobj reserveobjnum
\expandafter\xdef\csname EF1O12\endcsname{\the\pdflastobj}
\pdfobj reserveobjnum
\expandafter\xdef\csname EF1O13\endcsname{\the\pdflastobj}
\pdfobj reserveobjnum
\expandafter\xdef\csname EF1O14\endcsname{\the\pdflastobj}
\pdfobj reserveobjnum
\expandafter\xdef\csname EF1O15\endcsname{\the\pdflastobj}
\pdfobj reserveobjnum
\expandafter\xdef\csname EF1O16\endcsname{\the\pdflastobj}
\pdfobj reserveobjnum
\expandafter\xdef\csname EF1O17\endcsname{\the\pdflastobj}
\pdfobj reserveobjnum
\expandafter\xdef\csname EF1O18\endcsname{\the\pdflastobj}
\pdfobj reserveobjnum
\expandafter\xdef\csname EF1O19\endcsname{\the\pdflastobj}
\pdfobj reserveobjnum
\expandafter\xdef\csname EF1O20\endcsname{\the\pdflastobj}
\pdfobj reserveobjnum
\expandafter\xdef\csname EF1O21\endcsname{\the\pdflastobj}
\pdfobj reserveobjnum
\expandafter\xdef\csname EF1O22\endcsname{\the\pdflastobj}
\pdfobj reserveobjnum
\expandafter\xdef\csname EF1O23\endcsname{\the\pdflastobj}
\immediate\pdfobj useobjnum \csname EF1O1\endcsname {<< /F2 \csname EF1O2\endcsname\space 0 R /F1 \csname EF1O9\endcsname\space 0 R >>}
\immediate\pdfobj useobjnum \csname EF1O2\endcsname {<< /Type /Font /Subtype /Type0 /BaseFont /GCWXDV+DejaVuSans-Oblique /Encoding /Identity-H /DescendantFonts [ \csname EF1O3\endcsname\space 0 R ] /ToUnicode \csname EF1O8\endcsname\space 0 R >>}
\immediate\pdfobj useobjnum \csname EF1O3\endcsname {<< /Type /Font /Subtype /CIDFontType2 /BaseFont /GCWXDV+DejaVuSans-Oblique /CIDSystemInfo << /Registry <41646f6265> /Ordering <4964656e74697479> /Supplement 0 >> /FontDescriptor \csname EF1O4\endcsname\space 0 R /W \csname EF1O6\endcsname\space 0 R /CIDToGIDMap \csname EF1O7\endcsname\space 0 R >>}
\immediate\pdfobj useobjnum \csname EF1O4\endcsname {<< /Type /FontDescriptor /FontName /GCWXDV+DejaVuSans-Oblique /Flags 96 /FontBBox [ -1016 -351 1660 1068 ] /Ascent 929 /Descent -236 /CapHeight 0 /XHeight 0 /ItalicAngle 0 /StemV 0 /FontFile2 \csname EF1O5\endcsname\space 0 R /MaxWidth 974 >>}
\immediate\pdfobj useobjnum \csname EF1O5\endcsname stream attr{/Length1 3636 /Filter [/ASCIIHexDecode /FlateDecode]}{
789cb5567b6c53d719ffcefdeeb1f3b0af1f384f27c1e018488309754a3ad2084cc8ab059a40020b1497dcc44e42881d6387479286d2258c52d61128
720be391b68c51c668c6104ac7d43dda0aaa0a755bcbb66a9ada6e95b64954eaa43eb6481cefbbd7694b51abad7ff41c9f73bfdf77bed7f9cecbc000
c0469d0cd6fa9ada3a980106005640dcecfaa6c6e69aed2bb711be8bb0abbe796df5dc534b5f22dc4238d1d85ce6eba98ebc48f806e1751d61350a4e
c90120d5106eedd8deef3a7e72620ee16324d3db19ed0affbde75f2f93b3341adfdfa5c6a360a40af261c2a6aede81ce83f16b7f253c4e3aa7ba436a
d0b8ee377600a366afa29b18a6e7f82b841f225cdc1deedf99fd367c4c98e421b7b7af43656be075c2938495b0ba338a5b0ce984af6af147d4706841
c5f218e17f907d6bb42fde9ffc2d9c0348ff378d2f88c642d14307a5d90019f328061368b93141aa488410b8ced35a06144011b09aba952d9046d9a3
924c4ecba6c60fc16170e8e3836a4c6d875135168ec0687b4cdd0ca31d6a244e7d772846fd40ac1746bb427d4477c5425b60b45b8d904c77a89d385b
d4880aa3bd6a9f4bebfbc94258edef86d1c8168dd3d7a5866134b62d4292fd9d912eeabb35fbb745341d97acb0319a03f0727e042a5891f64dd6e21f
a153a22c4b990644943325797a0e9f95a6ceda2071c210313884831d3586d9dfbe2083d3ad409f39401521a663195ca0ed2e59c7113da27032997c72
5a3b95df825b3217d6e4741d606e560693f01ab55fc359b8225d8171b89bea1a7891ed93bc34721af6c25b5c828b708c2d660eb69846df30380c037c
0f3f43e31b49b781acbc054f904dcdd2243c2c0d494dd00957f835384ab54fe77f003f6723701d9e82d7a406da5523d80207a81e85b80cfc3acb0021
95c229cd135580766a85e8e5d7f5fa013c0c43d002a70c9306077b438ffa347b894ec63f29e63770236e258d637086ed91ddf219b9010ea4e2c53638
208db0a3729b5e87681d76c031b98d9d3538e0152d56e23451a49df00b6a3be01abb87edc17d14d9901601bf0ed78cf7c965a9a88cc3b888e603d4ce
c305f06282f4f5b9183ae198d449be3ea648ae610d9490b538001dd704645d3270192506f35dd609c9736f70c2bfbad57575fd2ceffcdba0cb6a744d
40d38479c035994c36b5ca4ebe7e82174ca0276d42f6b8dffdaac177bdf35734b5ba267e565b336db5b6ad8678cdad446a88d8c4afadf1a64e5b1a48
b457b45dd440b35a4a79296559feb86386dd66e559166e329b146e369bd28c14394f47ed644a9c31c8e3f9cefcec1c292b973bf3671649ae42fa1670
a733bf32d76242b9d00dc8b2d257ce710d3a77e599595e7e41aedda270a71925b7014a98c95de834cf75b3b9065ee2cc2f9e6ffdf0c2c97c1678f962
d4d9e694026ffacdc99271aa63257d258d2586c08d1b6f5e7cdedde69602376c8ba9dab51fb51c0277425955d57b1f56ddf0d97216e72c5e48d3355a
f9fb46fe3ee1e95ec3d3ed1648e4fad917167a190b5cf27bdbbc512f06fc3392deb7bdbff48e531df3f6799bbc8ddef40043f78c256cd15d73dcb30d
c6dbe872e31256eecbced1fb2c87c1583632f9485a5a70c3777f3f73687284a8d64735eacf93ee15afc6078f5b5bfe3234763c4d3a7573a374a27449
f686e0ab4fdf3c209df02ccd7b20a49172dbf9f6aeddb1ed0f9d7e6a5623a45647de4bab9307e7fdab528bf1f95a98b962513232a57413b72833ec92
c3465f2bb758944a13da80a5af740c5a769999d9a2408ed5644089db2c6c6e8ee50ea5385fcbbbc202efbd797153ced21c3dbb39a9fc6a597def43c2
54fe574ab5345e5ae85cea1c73220b5c6873b200bb7052ebbf3c55942479dbc80b23e9193d9b0eff61d6d0e45b939e557fda71f8d9743d29c7cb1ab2
5a1e7ce5b99b4fc86d136acff8c1e235a95b4b5a7f24def94ccd264bd54730334dbfcc7eb74b79fbd3ef27efdc2c37d7a5edd777367c7edb19c3a210
c0bce79377fef382b90e82f4fede5a0cb276ee35f367a9ed27f9d7a081dbb5ac7f565ea7372f158301ef8752e8a63b55022bf8b5971515c93cfd3e19
6183769e647a0fd942fdbed66806d98452b4040aab9ba611e6d02b9fa2e55b680eb96c709a3640313b04cbe9ae89c200c460337491f77ebaf1e74107
dd2f2ef0c142aae544b593840baa49a69f6e9d7e920e814a77fd7ce2de4bf77d072c206a19f45275d1fdfea9adb88e42f40d91ce76ea832499f17f78
adf8cc6b0b79da4ebe7a482742d25a1c2ae97c3d8f3544f590de3ad846121d24abead642ba86aacfc8455622d44749a69dec6e263917e9f79177551f
bbdd4eb36e250e8dd3f25b891bfa0a19d76d52ebf408e384fb74af3e8ab31c167d41fb535def6dba12fd3bf888da2eda155f560cfa9e9268a59b6035
c0a4b4db9ffca9c0090f3eefc3f309fc8982e7e20a3fe7c31f0b3cebc1e7143ce3c11f25f0f414fe700a4f097cb6129f11f8b40fc74f36f3f1049e5c
b58c9f6cc6133e3ceec06309fc41061e1578c48e4f0de393973121f030491c1ec627041e3a58cf0f0de3c17a1c3be0e463020f38f1fb021f17f83d81
fb053eb6af883f26705f113eeac3bd0247b37144e077043e2270b7c08705ee1238bcc2c38783f890c0211b0e0e5ce68302077606f8c0651cd82defdc
e1e13b03b8d32feff0e07681db12d81fc4b882b1ad1e1e0be2d6a89d6ff560d48e7d1456df1446fc49816181bd02b76463cfe64ade13c4cde4637325
76df9fc9bb73b1ab53e15d3eec543014c420a90513d821b05d35f17681aa09db36e5f1b6206e7ad0ca37e5e183560c64e0c607cc7ca3c007ccb88134
3624707dabc2d7cfc35605bf3d85ebd65ee6eb04ae6d09f0b59771ed6eb9a5d9c35b02d8e2979b3db846e0eaa6057cb5c0a605d84841343af0fe4c5c
4551ad5a862be9b352e08afb6c7c8507efb3e1bd021bea6dbc4160bd0deb04d60aac11b8bc7a982f17583d8ccb04faa770e9142e99c2aa8a6a5e25f0
9eab58495465332e16fe287e6b18ef2658217b7945352e127897c0f24af44de142139609f40a9c2fb094864befc43bac5882565ee2c6794538778ec2
e706718e821e96c13d3e2c36e5f2e26174f34aee16389bd0eccb388be46739d1353393bb2c48ff227ee53f2acfccc4a2742cf2cb85562c20f182043a
13989fe7e1f941cccbb5f33c0fe6da3127dbc3739661b607b3043a04ce9842bb2d8fdb05dac8aa2d0fad022d0215b2a024d04c0ecdc368ca3471532e
669a3043601a0da525d040e206819c66c12b5126247b11adf4ff28834bb9c83290f965284036c9827b1e67a5df6c816fd8fed72d85ff0562c30a81
>}
\immediate\pdfobj useobjnum \csname EF1O6\endcsname {[ 109 [ 974 634 ] ]}
\immediate\pdfobj useobjnum \csname EF1O7\endcsname stream attr{ /Filter [/ASCIIHexDecode /FlateDecode]}{
789c636018b6808581150000ef000a
>}
\immediate\pdfobj useobjnum \csname EF1O8\endcsname stream attr{ /Filter [/ASCIIHexDecode /FlateDecode]}{
789c5d503d6fc32010ddf91537a6434492a1938554a58b877ea84ea6aa0386c3428a019df1e07fdf83a4ae9293e0f478ef9d1e278fed6b1b7c06f949
d17498c1f96009a7389341e871f041ec0f60bdc937546f33ea24249bbb65ca38b6c145d13420bf989c322db079b1b1c7270100f2832c920f036ccec7
eefad4cd295d70c490612794028b8ec7bde9f4ae470459cddbd632eff3b265dbbfe2b4248443c5fb6b24132d4e491b241d0614cd8e4b41e3b894c060
1ff89bab77abfcd9b29c1b2af8be833fc5fe272c93cab7d7986626e2847537355a09e503aeeb4b31155739bf88527588
>}
\immediate\pdfobj useobjnum \csname EF1O9\endcsname {<< /Type /Font /Subtype /Type0 /BaseFont /BMQQDV+DejaVuSans /Encoding /Identity-H /DescendantFonts [ \csname EF1O10\endcsname\space 0 R ] /ToUnicode \csname EF1O15\endcsname\space 0 R >>}
\immediate\pdfobj useobjnum \csname EF1O10\endcsname {<< /Type /Font /Subtype /CIDFontType2 /BaseFont /BMQQDV+DejaVuSans /CIDSystemInfo << /Registry <41646f6265> /Ordering <4964656e74697479> /Supplement 0 >> /FontDescriptor \csname EF1O11\endcsname\space 0 R /W \csname EF1O13\endcsname\space 0 R /CIDToGIDMap \csname EF1O14\endcsname\space 0 R >>}
\immediate\pdfobj useobjnum \csname EF1O11\endcsname {<< /Type /FontDescriptor /FontName /BMQQDV+DejaVuSans /Flags 32 /FontBBox [ -1021 -463 1794 1233 ] /Ascent 929 /Descent -236 /CapHeight 0 /XHeight 0 /ItalicAngle 0 /StemV 0 /FontFile2 \csname EF1O12\endcsname\space 0 R /MaxWidth 974 >>}
\immediate\pdfobj useobjnum \csname EF1O12\endcsname stream attr{/Length1 9168 /Filter [/ASCIIHexDecode /FlateDecode]}{
789cd579795c1457b6f0b975aaaaf7a6bbe96691a6bb591ac40d02a2a2281de26e1654c213131d5a017105c5258a0645057119b7808931da316a1c34
0e318e01354623891ae24ce6a97993f7cccb989898cc10e3cc337106f1f63b55a0d179cb2fdf1fdfeffb7d557daaeeb9ebd9efb9d5c000c0460f113c
23860e1b0ea1100dc07a516dd8889ca7c67bea7bb4133e9460d588f14f6727eec93a03209ca3f6fa271ecd1da9cb4a5b05800584fff8d4f8e4d4e97f
2f5d41934d203c6fea6c7f19e68558086fa4f1cf4e5d38df03d3a333006423e1bcb86cdaecb97d17ce00d0120e07a6f9cbcb404337689b08374e9bb5
b87888437382f0162272644991bf503be9bd6a8028a57fbf12aa30edd2ac259ce883f892d9f39fab596ade4e7809e1f367954ef50f48c8fc2be15708
cf9ced7fae4c5c2bcf05704610ee99e39f5d94f8dd609adf3980e8b954565a3e5fe885d4e652e67fb66c5e51d920cd5fa8e8a2f9a4125064a5b42897
40188253ad53400f7d201384a1c31fcf05f32cfffc39104132a52b1804b85f527bcf2c9a3707b45de318b509ea5b4bef2fd49eb758184921047ec615
9c19bc48cf3df7f13d0f6157f81b005ca93bade07c0f7f839fff39f3fe5fba18f16522ce2c64671110a9ca2f04ccf414d476bcdf5304497dcbd466a2
fe8aac046a57ea659a450b3a7aeac1401a31a933283ad8022f805dd5c112ff3cff1458e59f377b0eac9a32cf3f1d564df5cf29a76749d13c7a2e9e37
0b564d2b2aa5f2b47945336155897f0ef529299a423533fd73fcb06a96bfd4a33c4997ab66fbe797c0aa3933959ad269fed9b06ade8239d4737ef19c
69f42c51e67f40df0f72acd0f547f8067220d3089acf944aa18a182c209bbadcc56cc1c362ba87b392aef2d33f43b4669a77e1cfd0c1c4cebecafb7e
f981391eaa9ff8003eefa7b240fe7267239533ef0fed4e5ceaef734c4fd1cc36293a94d2a497087575bef15fa058b0d10c0619512b0a82787f44d795
533cac107ce08185b29ddbd976cd6cf6e5437db00b9c00ea8a7b08632a2ec21a7abbc85a902cc203e9900143e1091807b9e087229806d3a114ca61a1
aa230fa440ff07da0bd5f659d4be20180c7e19bc14bc183c1f3c1d7c277822d8183c147c237830d8107c3db8ef617aff9bab33463475616675ad4e50
e84e214887ce98d09f20a38b8fa15d602278a20b94b1e3ba408907b95d40319528ee042b41214111014996b8e8043bc1f42e7010cc222825082328ef
826e040b0814bb8962e944732bdda7a10176b07d841553fd5caa0908876135f56c8233ac95d50abda96e1fdc848bd4b3065ab14104361ad2a816e053
49a008960b47688e0c6667191a5904f149f188384e6c12af8b17a0bf582e5e100bc4729686bba53c691f4106be4f76711edcd0c43e27da8ee1b79886
27c4a1a2193ec70bd8005fd12a8acc5a6123e9bc8268b1b352a8142a8471547356ba00dbe92ea5f60b6c27bb48d41d632be132bc88a2301276b2cbc4
572bfc082b3157a82455a409c544ff599aeb028ddf0ee5145a2e333d70a127d511f5b4d614f5198dbda5cbea7d132a69e55cd82337c9764d1cada248
6c1f3bc3dae4ad10808bf82ccec57f63abc53871bf381236764a000b6023cdbd5d192317b3c5c4bb725728b30b8bc402d600df8a059a2934f7fb0a47
b4e611611c71540c270816c916e269105b8db544a9d21a0d1734a3c5641a4f336896a97a2dc5749841a50a380487a137d6c3469a49e557ee2ffd4823
77885789e78d6c83f0235cc0a19004c5e20dc50ec854ea01ded6c8928802835e1e4ba3e01d55d8e81b3bc1732e3fa677af7f403d168da711721a4d8b
3d4dc160ce04314aca6f949c8de8d5368adeb8abff53e3d5debdc6e44cf034de1d36b46bd6610543a96efc042a2a185553fdb0a16a9bb268a3e4a5df
a88246cfd412cf5acbdab8816b2d45037b2b7e25283b26793252a926f8a5b891b463807088f385ca011b048c9b6deb2374ce10173a1d5111968eb65b
6d8f80e5daad36cb8d14162b582db6b4549bd52224a682d50271b1ca5358b7e39557e8f7ca2b77988edfbe7387df663a29875fe01f115c606974f765
69015eceab790d2f671bd862b6846d50e2d35572e98914d1f5e0f339b231200a01698506023aad5b7622b899c172694c6348ee8466eaec1b90dfd6d2
410425b7a5de6abbd4964232c88f65474230441426f58fb14ae9de346b8c2386b3d1fc2556f4211bddb1a7412c1fd934b2fd72034d40fa124713c74e
d8e94b8cec1685114eab24825592c46ccbabd6174c01fb6691fc162c7a81e99de11694a32d1d631a1db9631ac3729f19d368cf7d8628c1e0a901f92d
97da4e9db2da32baa8b9a552a3b148df69a4ef58a3d3620dcf20da7ca94f8b79529e6689b8445a185513a921af8e14bb917a9df329582fe8561e35df
5905d59155ddaaa2aa9cfb617f9475124cf21213e9fda0ff1096de37212e56d6a40f6169a9a2c32e6b686faf154e773c4e624cf33ff17af52f2e3eb7
e4d2846f987dd83391fc564343c322b679e0ec6da316d5673ff6d123a9dfbcf7ecdeb268fe67e27e07e9bb9cb8ef0e65be3ee008d557ebdcd59ed080
c314d06d959d01cfd6b8cdf27ac76b4961ce50407ba433c16371a2ddad9393142184e5dee35fa7f24f02b8d5465c92042c6dd76e5d6bb37c7dc3a2de
249514e6d315bafc6ebfa730468449ccc51c763126362131dd458cf423ae7ab2f4cec243ec61d6e6d7f8c7fc9bc96767e49e9b7df26cf3de4347eb76
bef6e2f893f3cacfe77fcd8cbf44afbb65d3677ff57acf3c925abf7155ddbe4565e515f109473c9edf1f5e7a40b1708aebe21eb22981768515be6866
4213209ab2010d9a80c470858e19f5e094b5a2d16cb932a6d1408c9954c68c0a6397325bda52ad8a5eaf5dca6c4b255e54c58ae749b9e71595f63040
0f1809f9b4492c82b5a009633d2181f5c47eec49f694f129531e2b660bd8125ccd4ca44a1d8bc1346b9a23ce1a678d4947990b8ca7f3cb97cfdf9d2c
793bbec40b1d69fb7980159c210ded240d1512e5d130d9172776d358ab2dd1dd021a7bc0526b1202b0c2b45eb3c715ee647a7482de22bb2c1dec41bd
5814f2bbbcc5a2780ba9c8d272437160c583493dbca5533ba1645f5645e6e0b0c3436a51b4f11946de0df49ad0ab9dc5f34bfcfbc9674a269e9af9c6
871fbe31f6d55ce97203df1212c26ffce92ffc078fa7f59194a33b761c8d4f20eaf793dc4b887a19a6fa2224ab80025a45f22c8928470999c840d658
3a3e6ab12a3693fc80c7109060155626bc4309aa4f4d7935c48db5ff807c9f6d82c064ec26654823a569d8088db286e44a2cb03816b31f4fddfde222
e377d3a4cb79ed2ba49e4abeb08ee4b84e8d6c71900c8ff9bc11460824ca0157ef806db36b7de26b2911c6f81e4e47bc334447718e825d484c548aa5
a3a5ed564b9b6ac1f7ac5ac532c89c6365873dac5344de3ee495f169a9618a3baa861d171b9fdeb75fe8bd0e244361dda6bd7b376ddab797efadda0c
c17fff9c6f5eb1e5357efbf66d7e7bcfc8cd2babb66ead5ab959787f7b4dcdf697ab6bb6e7790e2f7febe38fdf5a7ed813fbc1c64fbff9e6d38d1f30
fffcaaaaf90464c92b88a31ae22802e261822f4ee38e64d51019d0ef1503501be60e583687adf76a9cce985017c4c63a4d515e8add44febde8fd35ff
a1cb317d612d91ef753b1575ca792afa3d578b5bd3603b61fbd68693d8a4feaa15d842cd8c027b7a5f48f328e19dcc83dd638b6470f5f11d63ce5d0c
197878d61ff91d66f98221b3f237f9578fef6043d6ecdab586c0dd149fc04ccc96f72c0bf9f3d72c4c0dfcbbf8332e61dbb1ddaf1e3ffeeaee638a9f
d2994a6a247bd15236d6c7e7803add0a5667d10a163d4891a65470ea449bba039136548b5182ece18250c626316b176dde18f59dc4d8d65b2c9db9f9
55decab3d92e7698d5f3129ec3fd52f29d452c82f561bd58f83ebe8d2fe7cff37a32305a5d76d3ea4678d7d717ad1aad46b03241abbc50d0e975ccaa
d7ebb2f51a01b508bfd61a249d96b679492f3bc5217a4a9b4d445a8712ff883035628467fc64cfda2e50ecfa70999928f6e5a126441ba213f40ec1ae
09d52708091a8f2641efd1f7d5a4eba70b4b850acd62fd72a14a53a5df248489cc80a12c0ae3582f4cd476d7f5659998a7cdd7156967e8166a1793a4
36601d7b19edaa27508ca1e842ee608d3bcf7ab365ac92f57e9f57b6f2ca16e9728716ffd6de53727750f2dc7eb54beee4c7749a74c29bbe749d5683
7ad98a224a5651c46c4a3c1d283aea74f63ad30a8328c968d58133cc2ce92323456b965def348ab42bb6a592bb10d3d64edd642a1bb22d43b91f746b
75173cec732912c85e12ca24909824c8a8111de06076210cc3452f78995748c044394193a04dd0795cfd583f61381b2e94480bc405d2a2d035f21acd
8bf28b1af724752b090f8dc33eac275322aa47f14232874e0bc50d8f560cb9f0e9bba3d73d77e543768e41c7cabbb57c4b5ddd16e144d8a6e77909ab
ac9f72b756bafcc91f361c139eba7ba366e5cad5948704db29767d4b32d1c0689f5916ea6085c87c14657d92d672e95ac735355ea5a6508cd22b314a
abc6282d68efc5a850d0b9c1c22c825b63d1f97465ba5d3add2454a22cd1288bdfdfbdd17af706c5cef6cb4a8412a082fcb937e5b27af0c209ca45dc
86709d197e152e379bad9e6af73167735c93757db811c231c2a4d31adca8b50f4b20a97f74a92d35b53380b65cbbd541aefd811aafac8ae07d7352a2
535c29ee144f4a4c4a6c56a22fdae7f2b97d1e5f8c2f36273ac795e3cef1e4c4e4c4e6249625ae8eae71d5b86b3c3531ab63372506126f26baee0dbd
37e8de80025781bbc0531053e62a739779ca6296bb96bb977b96c74450ccb81f1807b3fed6b874257424502c4c8b7970670f134e7e7e7045e94bcd4d
4d5927d61c6cbd7b8709af6f2b389a5b7472e27fdc14d28a2ba6947f7a24e9f1bb2b1a8afda777bf73ca56b9ae4f9f86c4c40e455673495613653b6d
694e18e08becd60c667bb3a45d6f6e62dbc87c402b8cb0da0cc3a2d558919aaa048a6b4a1c6fb99172b4c0b5dc1570a11a31bac8a324168824f640b4
c6dd4d4d03df5cda1a8460ebd237ef9e7d7dcb96fdfbb76c791d8f0a93ffdeb6bfd0cf86322ddd43fddcd17afd7a2b41175d95a4433b445146154ff6
acabd6ae911cbf6252b3911d8f68b63519d73ba31c82d6a18531822d64985325b145cd1a95b07c4dcd976e75c6e5a4ace8b2e840f4c7d137a3a52cc8
62594296232b4aeaa549d626eb7ae94ba194950aa58ed228dda4b9c48f2346ddb2fb3b88273506d266aee9c314998b951d878d17de9e7176cad48f67
f25bfc2c4beaf882699a84bd6bb6379b85c9134f9eeddbf7508f5e6c00d3b350f618ffac65db91433b5559f33c7122f164a02c71b42f2ed218adb355
878635876073425c53e2095d73c83bdda21322416b1c21db6c9e6149eaaed929f6966b9d82e797158e3248fa3d96f708f450a47f6ff7208ac32dc24f
19c760d6a5123a578485a7d3d1766fdd0b7bf7be50b7b789f376ffc1b163778efbcd918cc34b7fdbd1f1dba587339a84c1e7ae5c3977f6ca953ff32f
f8b7d1aeb77af578e7dd67a64e61039992610c9c32b541e1e33439f562b219248fee491e7d527c134e0812d38a305c6be9a0bc4ec9ef3ada527c06c5
5f737405e4b3129b144a56a2246ba79be8120bee0464fbb7ca39e5a7f962df866d02d3c270d1d2791449f1992c924fca910aa432e9a624774e4213c8
f6bfb7296369b3d344934c6361a22f41b6e92242408ed6388c35d11e6c8a3a1169d1803544ab9573acda901c6784b6dbf03835bc76d0f6a29e2f3233
afdd52b7181bed31bed094f89cf8b2f84df101badf8dff3c3e18af2309ab7ee620ea3b45fd6021cda1368a49c34e55fdfa64f3bc051bf735cf5bb461
5f737356e3e22507b076e9c21fbeb8fbacb0f3d51d27f7dcad1176ee7ef9ddd7eed6880587a64d59dac58158481c84423f5f24ea00cd4cae315b9b8c
27f4b465c2934a841a6e5788568e8bc9998a15a8c41e2970fcce2128def7df9053d8b47469ddc1e6e6ecb7169cfe40d8a310b06ba742002d5c54f87d
97872d50ad319cac31546eb641b3b14939a1da42c6a2cd31ec1f4ea8beb8acc80aa8902b3595da4a5da5bed25061ac34559a2b432a2d95d60a5b20f2
66a4f5819845deffd041b6fc858307eab61e3cb8f526b3f11b37ffc2bf6756fcfcfaf9f3d7bf3977f6db1dfc1c6fe3df913b6590d7d8d900924c69f0
4b3c4b1426c2755fa6c928980de3dd2ead4ed0e8c7bbddae6cbdc1e5a69dae9ad58af66a476d44b3556cf65260efeed21bdc511a1817a5356bb4f6d8
61dd1537bad4764df19b8caee8605172b61f6ed8ee2516666537d598bbf6544854f6d4d94ebdd3e034f6a100d1cbd0cb384837483fc830c868f08087
c50bddf5dd0d3d4293edc98e1e61dd5ddddd499ea498f8c46a7db5a1da586d52be7b314190f5b2018d68423386a00523b11b46a1538cd62526276525
fd22a9326979d2a6a440d2cda408da7ee732871a6448806ef55427c73d787c48664a72dc8fd48beb9edc3fb1b676ca0b592d7b6fff61e29959c51ff8
abd6171df01d78f18fbf2d3e22661deade3d37d7372ac6dce3a5da1d47e3e24ea6a7e78f1d93e30d89afabda79d0a578f121d2ff04d5eeec30c8e7fc
c9f2d6ebd9097b9391ecce6e78922c70b84331848c4eafbe967adffc4a1da714f30ba5e0dfa9f0fbbb40023ba498df1b4d4d8fbdb9e0f439f63b764c
d877d7bf6bd7c93d42c59dc0c1e2a93771bf62fb83c9f62bc5023ad5dcf125aa071a254f9494170a3250ee418964b680f0ae244b4a86288246f966a1
574e61a09cc294e399f20d4139a2817a740eeffc6af03f248cccf7cb91c20ca142a814aa85e5c266618fa05516d2a14ecd9aba613731814e9c499824
7ab4e990ce06e24031453b1c86b351384a1c2e8d947dda3cc863f9982fe6688ba1984dc7e9e234a9442ed02e80f9ac022b28af5a22af86d5ac166bc5
5aa95aae877ab64dd88e2f8a2f4adbe4fdd2eb72a3f694f6736d503b44c931d3742c8dc50d3ec326b3c967f8b3ed6241472e1ebc13e8fc8f21ffa555
7f5ab4fe1721993f805bab7ec0fdfdf3e66bf7deb73fe978dc9caf53fe4fd2deffde4be334b379348099dffea47dac39ffbf7c218e162fa8df524150
42fa3aca5ac3a186e02a413dc10e8242829d04fb09d611ac9093e0bc7c1dce4b1be1bcf865b05dbc0e15921de68a6d3057790b19705a014d2b1c936c
704cfc8ada7a4229950f8927146ddfbfb26131fc8dcd62e785eec26ce113f461257688d5e249ca9e03e43573e47af988c6a119abd9ad1dabddaa8b57
3988c65ce8092574ae1028637949e158740861f456bec36a60a2f2e55dd4910052d46fdb4a994118619d6501b46c7857191fa8171f284b10c19eec2a
cb6067c5f018944219d13b0fa6c3345a7d3e7868079f0a49f44e8514bad3a834857a7888afe9d45e4e300f8ac00fb3a117d58e8239d4bf0f951e8559
747b60dcfdb9ca55ac88de453466213d0ba9a7fe67acdaeffeaacabf050b692de56beb1ceaadd0e1a731ff672b0ea5d20c1a97070ba8c754eaeb5767
2b5247f8558e3c34cb1c7a96519f2934ef74eae7a1f1a5b4ba5f6dfbc779c6abb39413459465c14caa55562d57fff998a3f2d287e497fed0a87b633a
ff8383e0f3eaff3cfff58a566d5df9174e89620e08a3bd2c02ba51cee8a45d23856655fe4f19012361343c0e4f410e8c25cef36002e43709cb7dc13b
1cdbedf8772ffe2d156fd7e38f66fc81e32d8effe1c5bf9af12ff578d38bdfaf7d54fa9ee38d7afcae1edbdaf1cfedf8278edf0ec46fb2f13ac7af53
f1ab6be3a5afeaf11a75bc361ebffc2259fab21dbf48c6ab1cffc8f1f354fc773b7e568f5738fe9b0dff75197e7a1cffc0f113eafec932bc7c698474
79195e1a8117ff394abac8f19fa3f0f71c3fe6f83b8ebfe578a11e3f6a75491f716c75e187a9789ee307abadd2074e7c3f0c5b389ee1f81ec7d31c4f
717c97e3498eef703cc1f138c763566caef64acd1c9bde3e2e35717cfbe824e9ede3f8f672f1e86fbcd2d149be201ef589bff1e2118e6fd5e3618e6f
726ce4f86b8e870af10d331e3ce0950e16e281069b74c08b0d36fc1511fdab76dccff1758efb38eeb5e11e8eafed364bafa5e26e33be5a8801ea12a8
c75d1c77be629476727cc5883b5e8e947614e2cbdb2dd2cb91b8dd822fe9f1458edbea4dd2368ef526aca34175f5f8c256b3f44277dc6ac62dedb879
d3716933c74d1b27499b8ee3a6e5e2c65f7aa58d9370a34ffca51737705cbfae8fb49ee3ba3eb896d85cfb28d6ae3148b5765c63c01aaaa829c46a92
54b517575b7115c79555566925c72a2baee0b89c6325475ff0f965cba4e7392e5b864b0bb122d72155787109c7c51c9f33e322232ed4e3028ef3dbb1
bc1de7b5e3dc762ce358ca710ec759313893e30c6bb634633c4ee758b20ca71152ccb1886321c7a91ca770f40fc482769c6cc4491c9fe1389163fe04
bd94df8e13f4f84f6191d23fa5621ec7a769e5a7b331d781e399451a1f81e3ec387674a83496638e019fe2f8e41316e9498e4f58f0718e63a8650cc7
d1a32cd2e8501c156d92465970a40947701c5e8fc3ea7128c7c784ded263ed987d1c1f1d833e8e591c870cb64943ec383833441a6cc3cc412629d317
0cc141261cc83183e380fe7669403bf6ef6791fadbb15fba41ea67c17403f675619a09531f3148a91c1f31604ab2414a3161b201fbf4d6497d2cd85b
87bd52b1670fafd4b3107b24d9a41e5e4cb261f744afd4fd514cf46282d7202584a0d780f11ce338c686600cf11963434f21badbd1452cb80a31da84
4e92a09363543b76cbc64842223946146238492a9c63180d0a8b4407473bc7508e36ea60e368255eadd9685986218568e668328649268e46ea6d0c43
0347bd05751cb5d44dcb516347b910456a14c9021c48b5c829b9b048426f6416048eac8915aedec07afeff70c1ff6b02fed72bfa3f013502fd15
>}
\immediate\pdfobj useobjnum \csname EF1O13\endcsname {[ 32 [ 318 ] 48 [ 636 636 636 636 636 636 ] 55 [ 636 636 636 ] 68 [ 770 ] 77 [ 863 ] 82 [ 695 ] 84 [ 611 ] 97 [ 613 ] 100 [ 635 615 ] 103 [ 635 ] 105 [ 278 ] 108 [ 278 974 634 612 ] 115 [ 521 ] 117 [ 634 592 ] ]}
\immediate\pdfobj useobjnum \csname EF1O14\endcsname stream attr{ /Filter [/ASCIIHexDecode /FlateDecode]}{
789c6360a0103013906761606560636067e060e004f2b818b81978b0aae3c510e183b3f8815800a70d826052884118488a00b128982fc620ce20c120
09552305c4d20c32004e4b0194
>}
\immediate\pdfobj useobjnum \csname EF1O15\endcsname stream attr{ /Filter [/ASCIIHexDecode /FlateDecode]}{
789c5d923f6f833010c5773ec58de9101148208d8490aa7461e81f95764219887d4448c558860c7cfbda3c43a45a82a7fbdd3d73c6179e8bd742b523
859fa617258fd4b44a1a1efabb114c57beb52a886292ad187d34bf4557eb20b4e6721a46ee0ad5f4419651f86593c36826dabcc8feca4f0111851f46
b269d58d363fe712a8bc6bfdcb1dab9176419e93e4c66ef756ebf7ba630a67f3b69036df8ed3d6da1e15df93668ae738424ba2973ce85ab0a9d58d83
6c67574e5963571eb092fff2d101b66bb3d6c7ae1e52412f0eef81f709f012469018b2871c9652388f084fdee9c3e785ce4507789c545060092c3d96
c009bee7a482026393c46f92f84d5234e9a48202a33af5277a84c8a24d271514f804ec0f93faf65301dc78ec43749ff2929c6b8ff84b4e2a2870029c
7abc8617776fcb05b92b74f3b6ce87b81b6347631eca7926dc34b48ad7b9d5bd762ef7fc01001bc43b
>}
\immediate\pdfobj useobjnum \csname EF1O16\endcsname {<< /I1 \csname EF1O17\endcsname\space 0 R /I2 \csname EF1O18\endcsname\space 0 R /I3 \csname EF1O19\endcsname\space 0 R >>}
\immediate\pdfobj useobjnum \csname EF1O17\endcsname stream attr{/Type /XObject /Subtype /Image /Width 97 /Height 165 /ColorSpace [ /Indexed /DeviceRGB 1 <ffffff23313d> ] /BitsPerComponent 1 /Filter [/ASCIIHexDecode /FlateDecode] /DecodeParms [null << /Predictor 10 /Colors 1 /Columns 97 /BitsPerComponent 1 >>]}{
789c6d95594c5b571ac70f4bbc04f0d207965163e399d6260f051382970076a8c3e2504c3591268d34429091d28c941613857dbb2e98d88d0d417d68
a21143a2d14cfa508d58321013c6b509d866b3611e5a3b9db1713332e4a1c606e20583ef7c372b8de6befda5738eefeffcbfdf35c27feee456e6dc17
09ad45410c69076eb8a67ea8722c88de33231452eed6cbbfb916a4397e9460483d52a49a3c7f9db593ff2e0fa1bd8bb3dda2cbd18de2dc751142fade
a49bc2418188cfc7b2110aee27362c046dd95213eec190eaab03cd646e8b670eff3513a1f8952f5b2af9726e07f33f7e0c69ce35dd6a3135bb74ac7a
2a42e16fcfd6c9bf5373948cda1086fa0afe813557eeb931f65fb2108acd9eea514f15f27a7cc5e918d2753d1a6a2f5f749e71599f22145074b2f486
b4ec0e4ef2f718c294151b8df6a04727891d476f257cc4d0af1a55715316de1300ed4f82b6eed371d76ee8472bd096d8149a9e024e51d6bb1184d4f6
96d1f0d0b264d6bbeec0d05e5a5f9ed8afe4252548f381769bbc66d662a2c4f85c2e8682c3d1bf7737a3ec24c6c355a0f50a3e2ebb8d7bf6d91fb281
56626b98ae3fc1e433ffe903dadfd0358f542b7e879f4306da7f53c389dd8d542ad50dbfd75747cb49baa90885427fc8c0500ca308f7d9135922eb17
9b408be8960ae036450ee8408ba806c375e07670a0310cd1148256030dcf77e7021fbe332aaa290d68cd997c84b443c579a646635e73c28609f83eb7
aea1d2b23575fc8214dad4277f8e1ba75924dedf6cc07735a63d21976d4444ce30f00d9fa4ad4cce38071dec4c68d3bb1418b356646f53b2d61152a9
524fdb49064f49d807a7c4bb033d147229d792a3865334d5d2a170a4db75e4877678cff0f81c4bfeb5c6bb972983f7ecfb80bc3179399c50b8de0993
15fb57b45f745e1c5f94fe09264bf799a0cd7459c04bb3fd17262bf05b9ba267d0263a4a2f72427f9fd147cf50a61ccf82b3702fbf4cf8c0f6e98eb0
3cff546e17dc9236d0fb9dae7cd2fc68b55c08b4799ddc144365bc73eb410eb4b976c3b5cb996254504c42a065a5d617b9c55b8630cac190fe49a076
9621c82fcd44b02ff845cdbda42d5bae711dc13ed5c1f82789f9f4d57984605f9c537d27c9bccd26e1b805218d64ccb3df5ded8bd86e8b31145e5130
f9f231b25078c90e6d5e9df03b262511b1450d3dc4863ffa78bc623ec32c6e871e74ef8c36543f98de34da87e10603c76a348d0e19bd8cef851bc4de
196f518dcd50da4d2a7817fce7ea5bf1d28ab05eda6d21dc1c5bd7741b3283b63231e1a65dda72d3bcae124edb093729dc3e76028a5bdee703dfc5b0
abc317c70bc48f4dc0d72bbeabbb50c05ce61c93126e0a6a53ee2eb39412af0df8beb27dba2b51326a162474e0bb32d55fb580b1c769f341e03b278f
de6f46cceac0344c79f8db664ed56ddc3fc696ad126eaaddf7d3cfacda7d335b849b24c6d2d3877914570505f8ba22ecd4df9dcc0f730c61e0539c66
06feb86416bbb939849b3dfe9a816b7133cf29446f257c44ab1d5756318cce7b04ed4fb4e69291fbecb26c17415bb273db72a9c837ed6111b4f68fd2
5b2fcc8a65dc27046dda68d6e0dd24fb8cb8388f7053eabd5abb1fb0daad6b849b73b5bdf7f82603df0076a8bce47b9d22a9b4d4645e25dc8c3a2b1c
735c2332b2093705198631b2ab0c2ff3116e8a32cde628679a3b0dddf6d509d78d090277a14b06ddc63031361fb76d2ca6cf048016097a483c3aa9e9
69059c1240a2ca234e6ab45f60805330249cdacb08715a8f9792615af1058b6c53ed1efcde1801be211a79e67c3b633bab2c8370732762fd46b655e2
9dde8469d517d7240b67d62c92f7e9849bd6b19885ef6fe53d0e126e26e77e487650079d7f65019f37b6da15190f51b31f6f106e9e6cf8327b252be4
394622dc5cd23cf31cf58a7cde28b4599d4a3fc57c9660724938d0e6788052e83f1befe1ccbb615a3f908617f94d8c33927606e1e65c7993e9ecd6c3
85e12dc24db2a11f35516669bfca27dc8cda8fe0fde1a4c0313361a380bac73d929998e795a0b7123e60a3c95c7beb496b9205c2cda940677d214a64
cd0b0837e52337548bf8c106c90a6dae4d5e4a4d4863e6932211c24dd1854042d06f8ff2e05ef44f4cf5d2b8823f91eeda24dcec51cdc9274ccb592c
b05175a08d774c16e329de2760639cd3ccd31db73277ddbf8f429b12b5332068f357f1ee4063e195f60ca975807f4de4e1136eea374f902879d71d59
26c2cdc6c1956878ad2d1f0a033755c1a31c71ee8019fed4c0cdeedc676ef3ea8e0416829b1a7e3fc3c82e866f14e1668b29c22ef3599d44d20ef4f5
087dd3e2e475228594b1a185eb327b8c08e0e64956736b6740f73ced5d5cda507f7d83ad7c9ef4bda99fec5dde6562cf53703fe58e7eb0c8ff3c809b
bbd2c6e02cff458a5fb96e53e476995e24cdb9b6a989d5f21701dc2cb514373c7899fa0a8c64abe6e54270b32002dfb9978fae6bb9a6f4d6ab145028
1b8d75af12a6acf97301f6ff133e32fee9f2ab006e96f42b5fa75089a5edf54270b355f13a809b83a36f927e9b6a7a9382c3a13701dc141d4a71c9a1
85e0e6a1006e1e4e7d7587530c3b9c748703b879f8f9c5c2ff018f70b04c
>}
\immediate\pdfobj useobjnum \csname EF1O18\endcsname stream attr{/Type /XObject /Subtype /Image /Width 97 /Height 165 /ColorSpace [ /Indexed /DeviceRGB 2 <ffffff23313d388ead> ] /BitsPerComponent 2 /Filter [/ASCIIHexDecode /FlateDecode] /DecodeParms [null << /Predictor 10 /Colors 1 /Columns 97 /BitsPerComponent 2 >>]}{
789c8d97c16eabc816450b049661442283628f48144798af282363c51e610b1066e4442102be02ac10258c48042866841158c057beaadcd7fd6e77df
db7a4cb7393ea7d8b56b1580de862c798798a4f288bb1bb67cbe3bdc4faa1602275c9b80930d76a68289b2059c52f1cbb2d6013032d5d49a3d6417b2
5ee675069e2ce882ae8300f645db7053c53eaabaa63044e69d6601a808001ca13cf98e9ff255d135b946188fd461bca565080cda3907c6be33224ba6
cd46ee239056753f0700525c0be65b829e2fe7454b0541a12be43b0d5135381e68c4bd6facbf72f1ba9b4d875d9a271440bdad1ead66fd19c0fd681e
bd18a3784729d21409003e463b6e78eb658e4ec063ef89d52e0d643d43d5a239d52cdf0f72dc4187250f6bb29a4df75bdc5b9181f88b4f830f3bf32d
8b4c727a14bfd6019ab4e440b2ec675342ed83e5295d4b9a207c0d840038638701eed0e0b2b90b5c4267970de95eb773548d8d35b0dbf50e7517660f
e74e95a8f45a3bc9018043a199561519eb17b07fd9f8c57207ee9a337ae7f78ab36ca9988c24b67bce82e33928bf2afd8ddaa06a86120cf85b59b1d9
b791c30e41b41b76d7de77b5a729e943d12df78c2e5b7a56545be2856f73d4f5fbbc16d371e8bc268aca74844dd6f9650f5ab4a209a1473a3b922f56
299d5486bdd1e7b5a0eba8eb66ddc99da5a8e7a4834d3de2b56d3e68bb0ead283764d5c79354d05b62ca757adf6c4c33a808d41b60ac227a5644bfce
e19812588a20dbf1d047df07b0279b93d2487827337667dafad9641eb701aab6b208d679d07d9797b80f8ab73bb27d8c360a5ad105a3b1f18b10ec7c
c5213ac3664db0f1cf29aae6d98d38bc74c1424865adeacbfd26d3c456c0be56fbf5ee69b63eba820a6a136ccf72e39c5ad45b36b712def393c5a82d
b64cebd49b208e291dfb3adb01ee26688e9e20d61530decf82f0d17568dd84bed29a278f5bf0e47a40cf7abe9db777049ac7a0f7f584ba3c3847df4c
be3499f581cf28868baa5163c0cf9edcf8e6d04eae19777fd0a5c7bc0fb1af79c6e0bd508dee67f65ad3a7b7032532c911ea6de56b90bb19d1f38743
f9c1743ea4e773c5f4b0af0fcdda79eae0f38b3b252a825f19f0396f0f68d2e89e62e3f711308f61acd5b99f68de9bc448a85af1b953d545b76e6f46
fba62383667b188d2a05552bafaaa2388e58e65e897c2275388bf4bc618a7dfdb0b1451eb0f667ea1f8c2e767673fe713b4393b22f04bbf2f544a895
20ef09953320774ba27d0a8797869888cad47c4f1dba2464c7501c6f33c3eeada968fb9ac62ddfc54623ab9ce6c6210bd13755de75bf96043ee00ebc
d1cf65671baaeb0dca11f874d305ef4b17a620ed7b3257398b903f5888dd7bcf4e78259c0e52851d5a3c1c9f72551b6668d284b1d43e1fc56697269f
a73e63a9b9dc0c64f02f0a6c164c319414de0cf5f5ad4dcad636577dcb445d7303ad7cd8a57d3be992f79254997a2e1f4e24f635bd051adff1814068
8f03b3b0995cbd2752f4b501456a639fe8c7aedf4464cb5b465b7c1ad83b2bbd9e3c1e8c807fe0638e04dc88d2d5cdb7771603c68eeefb89ffd2f380
d41bbdeb8af3b7773c5353e70fae2a1c793aad057a40f8a2f4ed9d5069900766857be31b1de065d370b74199a315cdee29e4015f0c6ea59ed57b9536
bc7ada20c7c3fe73873c104ca79ec78b5bb2300c7f10b300fbba6693dd3d90e7f6c15fd5a94a51922b6cf08ad29e35a93ef5e0b94cc5a403bede79c1
058bd78de277283784c91be81c50e96247c9e3e70ddec190c3b9d1aacbd4f0b7b41e3deedcc5058b76a3b1720cfbdc09c5d7ac0fee3585bb65c3d9f3
06f7f618f7a5c4bafcc926279fcdbcf12c2234cfb8b7482da3601822a749eac680713830327388f682511403ce598e6c8d0ae5739f25cf562fb7ba8c
f35ab0d6f19752363057dde168fd76722da673b1afdd9338bc96a61c688b80414967878c4d08d8d7c179bd5b2b70bc0525603bb52b73bb245a9cca63
32e1d5f9d4aeb3282d0f053b90ca290bb0af2b736acc65a87a04ec1889176901c89b0cfb3a32e1fc793e2dc23c0b2bc55f19e4f670c6e9ff24b7d9fa
2d872561ca137ace2fa8f4ea8ae6b1afd593a12ea5ec3533d59222faa3ae3c2c6dec9084b2fbe24b875f725b808111545d7af755228fc2069aac3812
b2933563b7563f2109fded7a8b56d4e1c02679ed4c995842b666787673eeae350bef6000eca9c4b6b6462bf6401b2593a1bc6e763e762f30e1f53e64
9bb954ba0d682c3dd8cb153a019dd5b6cd26afa3c4582b604764f10964cebef6827f518c0509fa3d0d9afe43722aa3e75f9a3e7e650e38af374db02d
74e3705662b227fd4b9614be0c9c3be19a9dd4c3aebfa2e76a4e9ac19569bad7848e7dadd2ece093d82f8bbc90ac168c6833b83310a1200eb16db3ce
23496cf91123689d618e951e110ae210b36c07121d4cc3deb3dd8637ccc79c24705e3b260868c148b99c0d698febfb36324de39b4336592abbfdcc69
f7447f700412146ddb43cc216779a67a3cef032d2bd16f2c6fc80c701eac1e87f6423ef8fe419fc883d93ed0e59d6db9df1c32288f6a2a1ef2ce5649
a8397ab0189ef0c914cddb295bcc9c5bfa501666a6c9ddecb8635d9cd7d929b6443986943b1db6a3895ac1aada0838afb933bf5bb96c3a0be3ddc0b3
0b5a21238d0798435ab8b809138832beb20ea560a779d1e033988d8173bc273470db454f8c3435cb8eb6a81ce7b5a0cb379f86965e1250af14a80c88
e2d491dfb4d1ed37f77d33a8f34cb752f0d4fae21d814e5a441b55443cb8b2f94ecadd49cfbc413075edd9376dd4be7617064a42aa8f9420872d9003
13653ca20d5d6c5c629d375211019ecf4e9990d20ace6b024cb9eb3c516245a41aa2972978012912a7f25a8f9d6769927f48da91c85915781a187ca7
f2b013e337bdbc582b0d6bb4fb42979b95d97ed30615a9cb6ef2fc21c9560f34550fe264e37dd3c64e2ebee4fddb1de21d7ad614dd4c1c9ff1a9b9b2
2ad5ba0eb40713f10e75a02d164eed21fea60b66589cb44cbb6b05a30169712a3dbf64305779f656bc6bfac665cc9e6a3a96d81e0eaf9a8fddab6ea2
37caa44737ade811bead59e4d505137cd306cb3d78adefddcf9cf01ca8404f4f6b2dc3a99cd140bb0176f870882743e06f3bf022ee08f82f0ae2105b
1bdf7b25a1a57b75a08bb5319bad2a03f3f5de6ceccf03205640a3daceb9ece14d528fb0afc7ad5c5e915ab6c8b4e380e06a33bb9a7688ec10879c82
d7533a918f04a82cdfb96c4727a3ea305ffb9423117a69dd18d3ab13ef3fe9dda4a7d1ba210ed9c5cb4cd746573d1cd97ea00f8852a470c246f7158f
765fa333fcda2bcd49671aafafbb10f3f5274ddf0a84dfd97e723368d5c7b6bf70abd1378768f4bbeb07a1ea4eee375e119ddc7550e31d3c7e606072
cd67d95cd87f9e0f22658722e846df1c5265cdb32fcb19196d4c7ea20f89d7d400df1c42f7d405cf0bbdc94d5ac8eba7fc026a48411c828e44a2f7dd
61bbde9f66866e4beb14df171087d062a70576c008c9f40c2134950ff87d9378bab1d76c93d963db9dcccd6c9dd2735c0b600e29f77b592e17c3b0cc
4c398114f12d60da60a2b1ca47b39d0cfad69ca61dfca120daa8b8aaf0b97011800398c733ea878069830691154c26c7d9fa6a97f3fc7f054c1b4456
9c94bd6ac3fdb29220f7878268e32c8b77e9b428bd482295ec0f01d3c6066d4c0196539f0bc814fea920da384fe63b337b8d83c978a3ff2960da18aa
7755cbb982b05f10ff13306d6c0b250a9be0a28d66da4f0aa28d5abd9547464a78f2e22701d306534051e8679ac1c39f15441b15ba7d98ec02fc6f94
1fd5686fa841deb48fdbec2f02a68d4193a2dbc04d0dffaa20daa0e399de973f8ff2a3daca21d410061101ff6f05718801b3d54cd6fe26600e31d672
22ab7f173087f4c9a1f9eb903faa950239bde546ff10308758c665f34f017388ced5dc2f14c4219df6f7217f545ba25bf73fff1e57532256fb958039
64f34b0173c8af5fc11cf26b0173c86f14c421bf798cdf09bf1e053fff01402b986e
>}
\immediate\pdfobj useobjnum \csname EF1O19\endcsname stream attr{/Type /XObject /Subtype /Image /Width 97 /Height 165 /ColorSpace [ /Indexed /DeviceRGB 4 <ffffff943d6523313dddaa55388ead> ] /BitsPerComponent 4 /Filter [/ASCIIHexDecode /FlateDecode] /DecodeParms [null << /Predictor 10 /Colors 1 /Columns 97 /BitsPerComponent 4 >>]}{
789c9d99099235a7b185991600140b605a00902c8021f7bfa677eadefe9f255bb6860a2b426e355d90c3774e1622e52c7dc83339a187d7fd3a6ea636
95b83ed4a29a648d8bbc9fa7971c4f5aa2383e39f173c9ccb5936d39f9f2f89ae838e7977c9c7ac4ac5b4a7b322d234475b4e412aebb12651ba5b8b5
9d34cbba49cee87149671e4fe313a368be3491337e2cb356e44ef2616f31932d732b8ac7151a269b40ebc41c4e15c99c26a2c06e864b95fabc83cfce
cfedfb69db9de3b60b35ab266ef2e7994294b9a248393da3b0799c7d7a18a9517676dcde72d159571eaef63cac91f22994d5e7e4ebe4211ee1d4a84f
58c9313637efe6d58847dc621517748f566ec797164e3e843aa2e184ba4b3d0dc2d6b2aa87598b27609794734af63695f3435891b310aa50730e3f2e
341f72dd1f61244d97cd0e94a3127269e76b265675e1e45861325951ac32a52102d8744e91abe5a1bb63be2e22bfe1e29de9693879d524d929b7557c
9847d6a57153928e27bc16ef2aacbb4dba61ef033b2aee3a17f30c42d67578075ece37a190c1f6bed6ba875db24a33874c23459cbc0b791074a62851
0dda8d1d7c3d5eac946fd3856775eb36fc314d346ac4c9f5b694bba74cfd39298940b3f52033d57cb7dbf41c37cc4dd49373e10a9c5c5b758534b9e4
1ca3ee59dd51695a2e0f175514a3c05677abeec5e7e96fcee35c9a6d26e4c7f864b9b9b0dd638493c47706bcd98d9c3d9f25b6b4eb1fac28ee3433b3
3f74a31927a92c8531a74c27a4231c5748895cb4988ba16aee1bab8906e3165c3fb5854579fa32735cdc238965a68b3637f456d8253fd158c44a258b
1a505471e26744d1d104d43cd5c87d60a597ec356d3d1c2bf9bcb1521a3b50c1dd5c1ea164b6b9945cf1ff8d28c913377666b2adc7d21af93d39ca8f
c4d6e8b52d73e9ac471ee14ecaccd36c47bee5e721cf179db5240a018df6b83c1dd8817a1243b8c05237e447d21375459eb393ce086d28fbf8f6472b
63f3a372abcb4419b2356b8c320598b37cba7c251bea4002fec8d2efc9259a8067603d77f111bb3a7187fde48da6cad911a9ccaa389f99c49d2f1964
43e01f89262ecedc8b16160c3e31da6d054631d58c4c52606cde3d56bc4cb4d2c988628ea389830e100b47726e90d9209f5d6f261d7a5a6ca1e38789
d1a2e0f9595aa34639faec915baaed4dfe455250f0f8dd315dbbf3cb445d481551418fa5b43b8133ceeceecea9f44e36d04936a580ec81721f2636f4
9c214908783b6e8ea89207571bf56dad03e1f6d21da4cd288b2f13b7cf47b12964549a74ba2fd674613d5057b688d735d4f4dc222ba0f7c3c416c5ae
bcb72ba38e1c807f953577c371bbc67e7cfe7599ce19dbf8301104e221d070e8cae5c0923252229bddd9d6d38a799736a3dab4ea6a1f26bef18c6ed7
db75c301b6bfcb60ffc24d9c87fc720d221405e3d4b57d9888e71ab64b6dfc86b023a1ca91b24872c92c680b0f11dabe71e76d3f4cc47354ae7c6b0a
d129c43240279e0d16da9d71943b722cc04dc5820f13a7cfa0d7337958dd7268eafd778084fcf653b0763d3002adb8bec97e99f8687a561ee9091adb
44722ecd8ae274066592951328e7a212c4047bfe3051b338cb85538eb99e7b2ea2c99b93ce6d99e52a3823fd0deba101f5fd30d191081d02e0e6d3d1
aa61a33891f0c49e66a4d9d945719e8156616e1f268a8618e4a7ee474848178e1b25128efe461e186d397050152b3b9a5fc26d4f8af29516083e144e
4467915349e45162863fe866eba56706e8db8770aabb3a85be3d356c4d6309359dfba25186e72090f3092e2a3189da8770f7525cc4d807149210c33a
77623b190d2d099457663108002135f143381886b37242ad409b1ce187aa58c286ec88d645c538ddd10840e0bd3e84038843e7233d3673890aa50a38
a1927c80f2089551d6d359841350f8106e04aec47907e7537700d87abacf6d3a68adcc1590404aa9414ff34707612504380143f3a86e2d2ab172998d
4c75057f9e77cc17e4aef02260cb9770ede6dbdd7843d0152ab1eefd78070ab71765e96448b46796020d1dffc18a97897eb86e05cd73733e0eedafad
343ca7f3f30c2151be0f05842db3f4ebc344ad73aea2a293f826b134091f1f020a14df45f7881c1183e37124b13e4c740938c61fb0169a4a3e9a6e10
e489304aa41a81e5fe744d0122f451037eede2d1c8d1c4de2eb77950c3ce3d17157a3a2366d49dad171df75503d845cdf9e27f23d7d7e43c139a2a85
76d9370db083a250545587d8af72be4cec09247528822b51b44942e7657428ddee13c00ee51955569619ef5e1f265ec5bc584040e1081cc44ea3be04
23d5a362c521964fad70b9b079ebc3c45ef9692f55fb89083ae04b00391d42c4454e939f7ae51d228770e387890919dc62df79ce1b74e7bb03c87580
895b5358d8ba3b6dd39bc5bbab9789b23342394b0bcdc0340df830905f471a60a140a06d936aa9ec72fce8a010d2827e705e5c3b40d5c3ba16e941ab
04b0309b9257b231854be994af4f9490029d6f540a4a4d5dcfb2eda6142912fc4d2d6e5805a77ec8be5ef465e2e8141d4f4d78a743aeb0473f994af6
cea1a6387715a11f93e3295f9f18ac838590c925f0c3299595a699e143a490507f18de03e0200245b7f561a2e17cbc8bd2c2331bbe017e15a7ba9df0
fa20606e624456bdc9f1359d2f1361e25517d72af46c11d061984d492826bc5e08f43df8d5244c7fb560c9cb44f9e46e5f2f879ead829741416f1c48
1c302fa3ad514f1b71c166f69789173bc51f804caa741d41f6355800a9d8f0ef395b0103d0e67a0d1adb0f1345a3f1b02002c035b465c04f39a77cb6
11dac8d1617d33c0b44f8ff8301153801a02a1442ad0d273f5076f5030354d93a3c20bba620ec86de597894f7e19190c9b89283f0f34ccc0e0c3d454
b4cd70d44ef4290c46377f5d1f461368b88685e841e3ade11a0c65e818d858ca05db09076314cabb90f830b1e977843ba7a2e3aebbd893f488c810bd
c16e37a8c6483abce50d31fd30118db70ed2a45c3d97d89c8140f33b57f2cddc05c9b01a22421756fc43b87e0843caf2d52ed91dfe7b78c612d00d7e
00749b5cecde9fa885bb9fa948623620c4e042f623e6c2bc1f29031c3ff82401748189456117984050135fd7870e989a616e21fa15052ca3d54a406c
40089b3d62d63dc0f5d03b327d08f7301d38fb2565538a15f50386882f21289f2bd40e30a2685a54ce8770e5d5e5bc47459da97cc6ebc761d608f938
c3cd82d7ea538672f235242fe194151db23222ec4a7118564d467d6583b122e0fc03d311021b509ef9c7f591cf1aa2b23d5cc516b00690ff2e7679a7
3b8d72005d5e49543007d77e08076f464cf8f5326f156724c8bfcd30be98d7013614994281f78ad669f11fac78994830ce4849998ad7030b8609244e
1142deaf5138287cf82d6717d55fb373815c628cd4bde2283831ef0997c415e183518803b302b542dba17fc4878970c85e6352d03385ea3024f5a781
d0c3351503aca18bcef93b9368df58819744012e1f1d76c703b63aab5f420791ea5af0d6e2769ce41d549b8c1f266a4111f60bffb8803459340e8a2a
6936cf90a2011ae82fb3134d78fc0f13db65cdb44f09669103cc8caf81141c0a5c87c3e405ad850363913f1f1b00258f5108584868c3785ef7f42073
181d608d34863234b2d0205ce3a0f29789ea007c0c4703e7f8fe1ea88dee647ed13468c07482d8805b8623f8323134e49043e901eacba58c00bef033
e1430ec4a77300b17d81e96cd27e98a88d0868bd6d8d8f170c0c023a8209487815393f16225475c2b856e0703e4ce49d15c2b07d3237a39999dfb1b3
ba2c54a2b770a1366e99416e3f5f26ca460bd62085a36a1650d5270ac8fde618d63bd6e7ad5b312be28fba2f130fd1e6dc606cc2855bce491e7803ba
d983c3c235b27539105e60f4fe3211076b5041189b852494b010334848675977cfda6038e704c267d8f3f8616217ce1b6168c891c0eb6794661481a2
13a2ca13b338a07012f6fe8e5c2f137bce30b40563705893d548685e789ece4f991de9e3882cbc9f8430bd7e9968372f34504078ef41972e695ecf93
08139f818abb430d6f86691bb97d988833be4329f4042e356bd898671e1625535420d3bb53d441815c8577320213f1fe84700d5862d119a227dc033d
b902fa911d4c9a68664ac4ac3e243e4cc4147659480c20db83a8e1c05c6906fce39d10680cfee9c1e48031943e6a906500189e0ce4efa64b8e3da010
30fb6b64c09060c97da981a9919c73df49b8c20f89d727828602d30a2c3aa66c3066fa925f87bc470a2a1d422f7e5d9fe878594e514293327555e954
383878870e67d79095b01675c8299af7c3c489e85d0c32119a244885cbe12ac75760985b984fe94e84454bcc35e3cbc473b06381134393bcab0fb2d6
f16e86780b14174e89b730580ff8890fe16616af37807d01d5b93a0b3a2c2e924c87dd86ba28bc9ef0feb8edd7f5814d0a7d01fb724aa263d5b40a96
c9ba524f07db5dc5eb975ed4acfa12ce04c2decb25543e3cd785f7bbb0acb0eed09a8c450b9306c67de697552fe1cc761adb4633f7c7bd1fc49eb767
78d3537aa36031f21e0802545acfef246ccc6b77cd4990b1c2f7e9f24153a686c1db7acceb0e9e15710584d1fb1fc2f522002378521810e2c2077626
e12d18f213fe98136d6ee0bd73ce3f93b03d80940b5163b3cd6d1a23604a813d223d9a1b3bfba76d4d9d40a37fb2e265a24032bce807520fbdc22416
8412ef40da4e96c00144c74f783a37de4f522f13dfb13da3a55442368544ec7269906691e61af40a3421d952517aa7d49789b0da4809fb9a16a68403
dfc92bc12ac16f8c0295730f86776923dbd73cbe4c44589316cffbdd0510c4cc66c07c912d9b2581b903e53732cac35e8fefd741f4c331682f09905b
8c05187e2b46558c373e16315fe54fe8d59901f3efecbc403b8ca9fc825c672e63c2aa0814dddafeba67a1c33b76f54019edf7eb207e8c203d73953b
c6ebf3a29524d21109ee4e688f8a3a98cb66eeef9789978998820d66d9287747606bc6d01c0f76e10035a864ca43c5e89ee71dafbf3e910ea418a851
f9f521224afcec9d1c3145bcba0ac5c91555155f687e99e8a630689185ba3233cb9378b2777da2f03b7ebe49f23e3d30bdaeffc3c48958b04349b32f
6f4c6e7c5d6289acdd3bb934013737311fbbcf37a29789c8d8e18c4c801d0793563f4b47b70fabb93195167a3fd2a0fabf1f955e268a21622664226b
2887371cb0d23574266a4bd4927b478ccdcf47a59f1b13a8d9a429df6f4b6644f8c70dd904f7de5d7cb2e0c6cf825f37261bea0dd3219f2df0eeb1b4
cd982ba13cfded73e17f3e73fdf28998c3e0234cd48d97cc083a76a5de0f6dda2cf8c6233caaffff9f9f1b93560966d73743193413987c73532e15cc
32aa5088ff5af0ebc6c4b338dbaa85e105f38f878384340b0b730398d26f17fcba3169707bb044981304864ea95ba2ad1970aeecb2f8ddf373638251
0a405ff06a7015af45c2e47498d03aa47ebfe0d78d499d02cc85201927555d7c2b266834df34fff606f1ebc604c89fa5471fa1c5584918699375c043
fcf705bf6e4c2accc5a8911a5a1e43ac8320236b20da7f2cf87563b231ef23abcd3cb05d1ebbc3d8fbbb2cfcf6e49f1b936c08be04f3b9c84b423921
6914fee3083f27ffdc9874c88a350f26793234cd9b05fd073bfa9efc7363f2ca8a2f1cf578d5e79385ffb2e0d78d0941cdc2a6328a234d273d7f7884
9f937fee3f56ccea1dfe9461852e36e7bf2ef875ffb1fc7bf0859116c645fe41167e7bf2cffd47c4a0aa40fd81d2fdc32cfcf6e49ffb8f33c40c8c5f
8d2eaaffb9e01facf8b931a141b13fcb6772ffe3d09fe7e7c6446b314f311716ec4f16fcba31892ee3d418e3fef78edee7e7c68413ab939f3f39c2e7
f9b931a989207a7f7684cff3736372908767fc9505bf6e4ce69fe5ed5fcfcf8d498fbfbed5ffe9f37363d2ffd2113ecfcf8dc95f3bc2e7f9b931f9eb
0b7edd98fc8de7e7c6e46f3c3f37267fe3c97fe7089f27fedd05ff07fa5daa7a
>}
\immediate\pdfobj useobjnum \csname EF1O20\endcsname {<< /A1 << /Type /ExtGState /CA 0 /ca 1 >> /A2 << /Type /ExtGState /CA 1 /ca 1 >> >>}
\immediate\pdfobj useobjnum \csname EF1O21\endcsname {<<  >>}
\immediate\pdfobj useobjnum \csname EF1O22\endcsname {<<  >>}
\immediate\pdfobj useobjnum \csname EF1O23\endcsname stream attr{/Type /XObject /Subtype /Form /FormType 1 /BBox [0 0 306 194.4] /Resources << /Font \csname EF1O1\endcsname\space 0 R /XObject \csname EF1O16\endcsname\space 0 R /ExtGState \csname EF1O20\endcsname\space 0 R /Pattern \csname EF1O21\endcsname\space 0 R /Shading \csname EF1O22\endcsname\space 0 R /ProcSet [ /PDF /Text /ImageB /ImageC /ImageI ] >> /Filter [/ASCIIHexDecode /FlateDecode]}{
78daed9a4d73243512867dae5f51473850a3d4b78eb02cc412b1076022f640702060f80a0fc440007f7f9fb7aadd2dc91e7f0d47c6e3b0fbedae9452
ca4ce553e537abad8e2f5b3f70faffedebe5c5c7affefce9db575f7cfad1faaf2ffb57dffebed8fa33df3f70c1cf7cffc5659ff2fdc32213af97e032
3faff79fd6e216f9dd9d7ffb7159be5f5e7cc8c77fe7539f2ea16da5ae296fa1642e36d7b6e0534ce946bbee352b716bd28ecbce2f77ab6fd6c1586e
1bd71c5f5c69754b35aebfbd5affb7feb2bef8d06b0678c2f75f9ac97af171c1c737ba9ecfefcbc2b57534feedebf5c57f6cfdf8d7f5f3e5f3f5cd8d
3dc7aac8a6e3f387559425bacd7b67d92e7e7692dfeaeee6f2118bfad7f266d15e7ca0cd087573c58a0fd9c5b2c6b8b91c9236e7a397cb8b4f6ccd5b
595f7ebf2ffbcbef96afd6f7aeecfdf5ebf5e567cbbf5f2e9f2ffb24169628255f7cb90cde496f1fbc30c3946af2fb071f3378bc72b78737c7da61a5
857e932fdadb27608e69d61aadc6e81e378372d5fa1974c65238b9bc7b13fc96cd8572630d135baebcc48e2c97cdcbf27b571fbfbfbefc7949ccc3e7
68990d3b067befeaa7fd1dbf398bbe7aeff3f99d6ff67788f6e84b6cce8778f3ce0fa777b2f32eb952cddfbcf3ebcd3521e55c5b4ce1e69d5f4eefa4
165b71c5057b789cebb7ce6d3dbd135c4c3945dfcaf1ce8b4ffcc5e9d7fb87ba3d3c477e8939b14ae478de7cafdcb57fbe6eb106f364325bd71c297c
efe6dd113ca72cb7bcb950fb913be9cea1e366357b657dd842f50f0d5def195cc9c28ef86ef08b7467dc52376a537059b1ad8560e581a4bd4a73dcee
a558164d31b17be2c266b93e14b02f6f82c25772a7e6e21f0ed863bf0be9e54329ed1c14ef5dbdbab146018acd5bad4f0aa45f2e81a41af9acea7cec
c557ab57a1372dffd78cf41d7a7178e803658417c1b5ea0253c4abcd7c8b66255334bfe8ebbaceaeb1189f033ada66ce11af6b2a846d29fbe6a262b6
72900d2a352b9b29012ab66adda33006ac594ca348b27be33aa9e4af3b0c505028012d8e6a56a846a32ab3506676586035b32fa18d6add2c72c0124e
6e8b690f70d4b65106024e4c2ac9c76aad19275b88bb85c402d6e87062506dcb992224112def06a8d1294685d2a006822c369cc854d392ca6120b215
29b00d839ac8c79471227395e576582087736a7262541b43c809b238557758281b2538cb895ead4401d755a9b5357f5820b4aab6aa17b5229a0f4ea4
2dc7bdf8a3da86e7554ef4aad70e2a96b2f6c4d26181047435c989516d4a349c28049ecf870546cecde444af2642a42539813b99ddded54c8cb4cab1
34a8654be68c8f56ce0b6aff2ed62d347a013faa6df31c0d7282462966db0d707e5b225ef2a81a7593d0a93adf1dab544e6a53c2d8a8fa8d6ae07002
b5fa7a14bc12d8571f716250cf2d05b6526e47652e89ea458968a3caf2590878c108c1dc9155a5280c335e0caa5aa0a87397715db223ad505bc2721d
550a2f8eca0b72b18523ad2a816eb1c98b5e653e2d0579c1cc433cd2aafa8d3cc8f2a257394392cedf157f589a23ad6adc68ccbcbce85522dd086abc
20ba2d1f6985da1af19047356f2a29f2824a94ca9157b5e8f7944791ed0ea5ca090a51ab475a5502ddaa97139ddaa8868dee0a27a844a11d69d5a89d
189013bdea0992c63925b5553bd2aab1a8017fe3ac36e31533a31271821c168874e74c4ef42ad9545dcc6b55da8523ad5a2606687bc2a83275aa354e
b0279ee3e43040a47b4eaa32aa44bae3f86d52e9928eb4a27013b49e56f996dc38d270837dcd351f894552d1d7fb821fa34cb4fb4053b15655b6d24e
46087717429ed4a8de864b57c280ca65271b690b5411f932c8b40d31063903af30d2c9884a43ccf266900979a536deb08ac9a2bf914f3d732f2f5ff6
f8e1d7cf3af4180fbe09b42680c2ce23488c4f390d71a29c839c7e18faecac1d5284ab1658d51674cd4f9d5a9fff5efd7af5ddd51f57d7fc5caf7cdf
049d18b1e74c77e2ccc54222ed19a1830a82448d75073ad7bd76f6b4bb76c447ade3f94ddf1b7f22494eb47c9b24ef1e65474aff48a494096224d2f4
5f76acd3eec1aa40398b9c8b24a07f2e549a5cea89eef4fa9e6139fc52d871e0d92ca9ed340a22adebb0ed37da3da353396ae60432f5d0ef869246f3
85afff30e49318724cd973df6d81837896efdc41cfe94d256ee5f8f07391b24fff9e2bbb893c009706e0717a1d3746c23b21e6548bce9cd94fe601d8
341a8e9c8f7a1b9e4d9cef5efbfe764e1b2b5c17309a575183d4339905424f58502719500a350e50669c748e9ed1875b3265d11d0877c6320318a1e2
90cb24834ac04b1dc18c6c61bd6888fc24ab7ba08f88039a59142e01417992c54bae9511ce2c0a982cdc125945fa9c81ce2c0a988c2a30c92226f1c8
c067f4c72a4fc4d92827f5699ea239101a58c10912542c4659d844a79c0746b3246e0ab5d45b72a30f1b20cd92b8897598658153ac72a6c3344b22a7
64726690412797621d418d917991aa9c19647a4a98218ca8469d86b0994a9e64f889dfe54c076b929bfacf51a4a3ac25fa91d68c35cb58886992c9f3
589d9ce9784dc7260538c89941c60540cd466233cdb55217c324d352e61604151db349e644c8a54ef27e96f7c466451045b9f5932c8a72398ecc6645
18e5303cc9e228f3b8d2539b0a9c0faca04db2488a301fb94df21e166992c5523ec9958edcac08a6283ca32898c29791dcac88a6429233bd5c8553a1
ca998edd4c5852a2c999411650c524673a7cb32aa28218c22db9914d69043803070a7c2b6706392935eb007006449271266706b9e8be618c23c2191c
c952170ef3516e1494e2e44c077106493a0255ce4c320150c2c871066112a70e6746992cc0d132921c1b4be8d459a41329b4217e04396832246a589c
e42c8739fc46900328899b56eb240bab687d6c24b926aca21b0c93dcd46db434929c0729e9e17c29934c1ee848eb39ce0394d1e1be9fe4b0a9a19233
1dc7756decf3f0ee2e24ba0b75b0f8587a7a98f6741210a671bfd9fd34d60b8f643d907ce3d84c31759355375f3bcc3b5e9f7decaeb98d7897377ba3
7f37e2dd3dca8e78e191882713641f0d7be779a7bd9d7668693818f96a8eb4792ee3798e1ed5ee214c3aed9ef1393ca92d37bdfaf3798f7d8da155dd
fceef7fe46bb6706940c5208c6ad2d8577043e72f570fa688cff01bf2781df98c0e73e1ec3ba5b3bca776e25b8c56ae7e28e0f3f17fcfaa2d0835f37
91879e2a1ae772743a9a8f4f3f1ffca60a7506bf7e32f7839f370e5912c1bbd3a79f0f7eef5611ff76f01bebde256082802b39acf584c7fad3daf165
932c32a15318198f38dfa80dada549a65c941cac0e8ce7614d474ec55baa1e54e411f13c045a697d059b834c2d6ef41223e2d1b86c697fd83aca512f
e8586c803c0f6c0605669864888b9eae8d90e7814d472368e596dc14bc03e3a9b8553a31bc19652117943c329e8f422eea4f9e642197abf2a6633c1f
855c66f26690855c106a1828cf2721176d789a642197371b294f32752bc6491573f9226f3ac8f349cc0530c549167385e846c8032ee80443f17592c5
5cd1a511f3800e988b16284cb2982be636629e642ab49337834cffc02cc388793a488b22d04f32fda6362d0e98c7894fdf9df57757a3ac1bbed81831
8f338904607e36c9e713be473dc9b4eb097746997a954bb311f55831a287dc6f93acbea426dce951cfabd4b85ae54e2f332d8e6feff2807a9ef1356d
5f2719f28295c6077a5e0603ab1d6ec974e451de74acc751884157e5cd200bbc28f723eb1165a097c56c932cf40232e3c07abe08bda8da7992855e24
f2c87abe0abdb47fb3aa879665843d5f455ec44e9a649157285606d8f355e445c2b449167bc18c23ecb15bb057ccf26690c55ea4fcf884cf57b117cb
566fc994dc3c3ee2a3f86e8aa8e227997108919606daf350670e395b9964a2913284373ded79a893f693449964e08bb5ac23ed711a519dd8f93cc949
ed8097373ded2137607c424368b4009d65c23d6854f9a1879083cc714a32c99b9e022fcdedf328f04e5aba8386b0f818b07a9800752094940a8eed98
f934084c6f7bc4e28b6afafe774f5a8738b435a75eb3b3fac5fbab6eb0b2b8fbbf8af4eaeaf7ab9ff6a15e31d09f57df30e41ffa83a6cb8067dcdcf2
cd83c5ac6dd97b6b3daddab7f6b5fe56a0490c17f1ba134d7f5074522fd7f7ea8feadfce9077c2bbce5dd2430563ff53b99b119ed0dd92e39c641c86
eaac38f613f10ec73b9582e355587feb5ded2e59eebb64d99f680302594171591035837152af7bb577be33f19435d99ff7d0fce92c78d4a24c040b36
e86f0da2904777c8e45f4c7a9177550b30afcae59ae5be6b96fde67f2cac5318964577af27f5ba578765b99878d2b264ed11973d3656fc041ff9f44f
9b9dfb17e1fc6f5e96cb35cb7dd72cfbedd6c64112c668a1448749bdeed561592e269eb42c545f226cbf11f5a86509e3b2a4834de86bc61cd0bddf53
48c479592ed72cf75d03d0c08d34ef61ac2ad44b5a55d4215a3ab5f7bf33f194651105594ce1d195250e6cb82c9fff1f826e9b58
>}
\expandafter\gdef\csname EFWidth1\endcsname{306}
\expandafter\gdef\csname EFHeight1\endcsname{194.4}
\pdfobj reserveobjnum
\expandafter\xdef\csname EF2O1\endcsname{\the\pdflastobj}
\pdfobj reserveobjnum
\expandafter\xdef\csname EF2O2\endcsname{\the\pdflastobj}
\pdfobj reserveobjnum
\expandafter\xdef\csname EF2O3\endcsname{\the\pdflastobj}
\pdfobj reserveobjnum
\expandafter\xdef\csname EF2O4\endcsname{\the\pdflastobj}
\pdfobj reserveobjnum
\expandafter\xdef\csname EF2O5\endcsname{\the\pdflastobj}
\pdfobj reserveobjnum
\expandafter\xdef\csname EF2O6\endcsname{\the\pdflastobj}
\pdfobj reserveobjnum
\expandafter\xdef\csname EF2O7\endcsname{\the\pdflastobj}
\pdfobj reserveobjnum
\expandafter\xdef\csname EF2O8\endcsname{\the\pdflastobj}
\pdfobj reserveobjnum
\expandafter\xdef\csname EF2O9\endcsname{\the\pdflastobj}
\pdfobj reserveobjnum
\expandafter\xdef\csname EF2O10\endcsname{\the\pdflastobj}
\pdfobj reserveobjnum
\expandafter\xdef\csname EF2O11\endcsname{\the\pdflastobj}
\pdfobj reserveobjnum
\expandafter\xdef\csname EF2O12\endcsname{\the\pdflastobj}
\pdfobj reserveobjnum
\expandafter\xdef\csname EF2O13\endcsname{\the\pdflastobj}
\pdfobj reserveobjnum
\expandafter\xdef\csname EF2O14\endcsname{\the\pdflastobj}
\pdfobj reserveobjnum
\expandafter\xdef\csname EF2O15\endcsname{\the\pdflastobj}
\pdfobj reserveobjnum
\expandafter\xdef\csname EF2O16\endcsname{\the\pdflastobj}
\pdfobj reserveobjnum
\expandafter\xdef\csname EF2O17\endcsname{\the\pdflastobj}
\pdfobj reserveobjnum
\expandafter\xdef\csname EF2O18\endcsname{\the\pdflastobj}
\pdfobj reserveobjnum
\expandafter\xdef\csname EF2O19\endcsname{\the\pdflastobj}
\pdfobj reserveobjnum
\expandafter\xdef\csname EF2O20\endcsname{\the\pdflastobj}
\immediate\pdfobj useobjnum \csname EF2O1\endcsname {<< /F2 \csname EF2O2\endcsname\space 0 R /F1 \csname EF2O9\endcsname\space 0 R >>}
\immediate\pdfobj useobjnum \csname EF2O2\endcsname {<< /Type /Font /Subtype /Type0 /BaseFont /GCWXDV+DejaVuSans-Oblique /Encoding /Identity-H /DescendantFonts [ \csname EF2O3\endcsname\space 0 R ] /ToUnicode \csname EF2O8\endcsname\space 0 R >>}
\immediate\pdfobj useobjnum \csname EF2O3\endcsname {<< /Type /Font /Subtype /CIDFontType2 /BaseFont /GCWXDV+DejaVuSans-Oblique /CIDSystemInfo << /Registry <41646f6265> /Ordering <4964656e74697479> /Supplement 0 >> /FontDescriptor \csname EF2O4\endcsname\space 0 R /W \csname EF2O6\endcsname\space 0 R /CIDToGIDMap \csname EF2O7\endcsname\space 0 R >>}
\immediate\pdfobj useobjnum \csname EF2O4\endcsname {<< /Type /FontDescriptor /FontName /GCWXDV+DejaVuSans-Oblique /Flags 96 /FontBBox [ -1016 -351 1660 1068 ] /Ascent 929 /Descent -236 /CapHeight 0 /XHeight 0 /ItalicAngle 0 /StemV 0 /FontFile2 \csname EF2O5\endcsname\space 0 R /MaxWidth 974 >>}
\immediate\pdfobj useobjnum \csname EF2O5\endcsname stream attr{/Length1 3904 /Filter [/ASCIIHexDecode /FlateDecode]}{
789cb5567d5054d7153ff79d777717d85d175810585816d605292c10502884d155f94afc0003124037b2b0cb87c2b22e440582928235c4a61ae36c22
f52b0db5c6a609b5694a6b93a64d5293b14ed3a8ed64fa61d238d3e90cc92433249932e3dd9ef7406b9c66a6f923f7be73eff9dd7bbeeebde7ddf780
014034353298aaca2b2a2116b4002c9946e3ab6a6bea964c650f125e46385055b76975e6e4cad7084f120ed5d4e5156cb3f85f0690ca0837b4f57802
3080ef123e4278a46d67bfedf8c9a90c0034934e777ba0a3e79fdb3e799d9c8dd0fc810e4f5f80bc913fae27acefe81e68dff3c87815e104d259d6e9
f378b50dbf8d01d0ada3f9a24e1ad03fcb7f47788cf092ce9efedd9621f884f019c2c9ddbd6d1eb8c8fe4a987c80a9c7b33b80dd9a08c25708dbfc9e
1e5f6ed19a20e1598ad114e8edeb0fbf045701a2649acf0f047d81c38f4be98469bdb21e94bd5122538a440841a38e291409b950065279e5ba7a3076
7bfafd90008a150887016e71aaf4765fd00fba053d467392daeba89fd7b15195d45145ebe73427d138a75e21b6e0ef301c01b3ea6fd013f4b4c29827
d8e387b1d6a0a70bc6da3cfe3e6a3b7d416a0782dd30d6e1eb25be23e8db0e639d1e3fc974fa5a6964bbc7ef81b16e4faf4d6929eeb11e4f7f278cf9
b72b23bd1d9e1e180b3ee827c9fe767f07b59d8afddbd6f6dfc264233b4431022fe447a18859953e5c817f8676894e4d8ad220a21c25c90b6bbf556a
db2bbc3452093d1ab330b3096d0ffbe00b32b840c9eaca01d611622a96219f7a9dba7b4016eaa0478d6a810bbf123e173e7a9ba5f9f35b769bddca05
52a2aa5b205a03f42c908602b1b33c98868b44bf81b37041ba00a7a098ea7df00a1b979c34731af6c3bb5c8217e1182b61665642b3973566cd00dfc7
cfd0fc16d2ad262befc2136453b1340d7ba521a916dae102bf0413547bd5f18fe1976c94b2f029b82855c367308af57090ea04f451125c659120a46c
98543c510568254a4127bfaad68f612f0c413d4c6aa6356676598dfa347b8dcdc0bf28e6cbb8057790c63138c3f6c976f98c5c0d07e7e3c51638288d
b209b945ad43745ebbe098dcc2ce6accf086122b8dd452a4edf02ba25d7089ddcdf6e1384536a444c0afc225edbd72de7c54da615c4eeb01a2e7e11c
383144faea5a34ed704c6a275f9f512497b01cb2c85a1f005d0921887b49c3659418e4d84c5392e31eef946b63a3edcda63467ce1dd066d2daa6a076
ca30609b0e876b1b650b6f9ae2c953e8d04dc90efbfb5f36f9be33676d6da36deaa715e50b562b5aca69acae915805d1308d57943be7df72ca2bca27
25db8a2945ced12a23a0dee5d4f10888d07288908990d225a254d32c4bcd5029f7e974115a7519caed204558224db3d77f1fbdb824a6e42ec89b999d
317d944feeb426fea176fe21ae299da569d350213b6bc6c6b41b3b1aa4dab41b5307dee257cf8aeab3378a95f7411ffe40be4c679c09475cc57a6688
e2d6546b04d345f2d4546b696494355566713e73534233dde14d72b3a3217ae3526b646a9445cbc0926cd426a72f355d99b932f3eaabd1312554ee02
d3f5d999d9d74d1f7df61185381f98f1436217baa6f4739608066e972d0a34fa24166bce822ccaee1256185b682e8c2b8cd79b23c6e3c6e3d1cddc4c
9bcbece99a38737c2ab3b23833a4a56764165b596141d1f26519792c972d5f066905f1f2e54cdf05f7cffe54333970e994382efebeebe381acde6bbd
2f5fdf1c1af8c393ec7e16bbfd1d3ef94651f189d1eaf539d6e5cf3cf10f713d2fffed8aca33fb9bdd19a9458fedbdc652326847aa29d75652b666b3
38579f393626dac4e31671bd416fe406835ea71c04570e884e823306893cc99214bf588a4be096a454ab644ba13e995b2c49a5098bf428a7d801595c
c4ba0cdba0654fa28125262527c42c32728b0125bb8656aeb7a7580c997696a9e15996a42539a6d973279398fbf5170396168be4bee23284b34e513d
94d59b5593a571cfcc5c79f1057b8b5d72cf4497508d511ea2c504281dcacaaecf96cd14d046dfdc7c352594fd9f6f4df3f9a1d06d70fe60f29d8cb9
5f72395b9c0127ba5db161e735e7af9da7a81e72f63a6b9d35ce0837437bec0adaf60c3a17ed1d7ca176051d4dfc62b58d336bb479a3d30feb74dee6
6fbf933a343d4a5ce3230af79769fbdab7fa068f9beaff3674e8b84e9abcb1453a91bd22bed9fbd6d3370e4a271c2b1337fb14566e79beb56324b8f3
a1d34fa5d5ccdfdf52d3d1b7db366fd8baa8ec5348d5a997ee1ff718afddec3f7fef46a1a15277407dd36e7d47e84c7b440a8061dfe7effdfb17864a
f0d27fc8ed452b2bf79062fe2cd1017a3327a0989b41afbda864c4adb29abeff136a1c5adc00d9d049f7bf042670297f1968940ce455b9f7b5d0acbc
e332fd1bb07cf5bba0f00ce209cdf3121859e5028f90c1ea1778f9369e43021b5ce035b0841d863574ffd13f1004a10b3ac87b3f7dd997421bdd7936
28a02f573e1412d74a12368ab58be6fb8882e0030f7d7b7268f41ef0937c2e71aba09baa8dbe39376df5a9c847bd8f747652eb25c9c8ffc36bd12daf
f5e46927f9da463a7e9256e2f090ce57f3584edc36d26b800749a28d643daa359faae1515764232b7e6a0324d34a76bb48ce46fabde4dda3cedd69a7
4eb5d207350bf23b68d4f72532b63ba41ad408fb08f7aa5e0b28ce4258fe05ed9bbace3b74e94f22fc29d11eca8aff55b46a8e4a74d259b493b5d3d2
882bfc1381530e7ca1009f0fe18f8df85c9f913f57803f1278d681cf1af18c037f18c2d373f883399c14f84c297e5fe0d30578ea641d3f15c293eb57
f1937578a2008f9bf15808bf178913028fc6e053c3f8e4790c093c42124786f10981871fafe28787f1f12a3c74d0c20f093c68c1ef0a7c4ce077041e
10f8e8b8953f2a70dc8a8f14e07e8163f1382af05b021f16382270afc03d0287d73af8b0171f1238148d8303e7f9a0c081dd6e3e701e0746e4ddbb1c
7cb71b77bbe45d0edc29f0c110f67bb1cf88c11d0e1ef4e28e400cdfe1c0400cf65258bd73e8778505f608ec16b83d1eb77595f26d5eec221f5da5d8
b9218a77266047bb91771460bb117d5ef4929a37846d025b3d7ade2ad0a3c796ad89bcc58b5b1f30f1ad89f88009dd91b865b3816f11b8d980cda4d1
1cc2a646236f5a8a8d46bc7f0e1b369de70d0237d5bbf9a6f3b86944aeaf73f07a37d6bbe43a07de2770636d2edf28b036176b28881a336e88c2f514
d5fa55b88eba7502d7de1bcdd73af0de68bc4760755534af1658158d95022b04960b5cb37a98af11b87a18570974cde1ca395c31876545ab7999c0bb
dfc452e24aebb044b802f8cd612c2658243b79d16a5c2e7099c0c2522c98c37c3de609740acc11984dd3d977e1374c9885269e65c7a556cccc30f24c
2f6618d1c122b9a30097e813f89261b4f3526e17984e28fd3ca6917c9a056da951dcb608e9cfe655d7849c1a85d608b4bae4141326937872082d214c
4a74f0242f2626c4f0440726c4e0e278075fbc0ae31d1827d02c30760e63a213798cc068b21a9d8826818b041ac982318406726818467d949eeb1330
4a8f91027534a50ba186c4350239ad8297a24c4876229ae89f2d924b09c82291b964484636cdbcfb1e63d95f6f81afd9fe572d29ff0192b9464e
>}
\immediate\pdfobj useobjnum \csname EF2O6\endcsname {[ 70 [ 575 ] 83 [ 635 ] 109 [ 974 ] ]}
\immediate\pdfobj useobjnum \csname EF2O7\endcsname stream attr{ /Filter [/ASCIIHexDecode /FlateDecode]}{
789c636018448005a70c2b19a6b1010003270010
>}
\immediate\pdfobj useobjnum \csname EF2O8\endcsname stream attr{ /Filter [/ASCIIHexDecode /FlateDecode]}{
789c5d503d6fc32010ddf91537a64344be9ac9b254a58b873651dd4e56060c87855403c278f0bfcf01ad2bf52478ba77f7ee8b5f9ad7c69a08fc169c
6c31823656059cdc1c24428f83b16c7f006564fcf1f22f47e1192771bb4c11c7c66ac7aa0af80705a71816d8bc28d7e31303007e0d0a83b1036cbe2e
6da1dad9fb6f1cd146d8b1ba06859acabd09ff2e46049ec5db4651dcc4654bb2bf8ccfc5231cb2bf2f2349a770f2426210764056edc86aa83459cdd0
aa7ff16351f57a4d3f9d29bd4057f09ee8e763a6137405337d56994ed015bca72ebff552c3749d751b3987408be413e60dd2ecc6e27a65ef7c52a5f7
0036657d8a
>}
\immediate\pdfobj useobjnum \csname EF2O9\endcsname {<< /Type /Font /Subtype /Type0 /BaseFont /BMQQDV+DejaVuSans /Encoding /Identity-H /DescendantFonts [ \csname EF2O10\endcsname\space 0 R ] /ToUnicode \csname EF2O15\endcsname\space 0 R >>}
\immediate\pdfobj useobjnum \csname EF2O10\endcsname {<< /Type /Font /Subtype /CIDFontType2 /BaseFont /BMQQDV+DejaVuSans /CIDSystemInfo << /Registry <41646f6265> /Ordering <4964656e74697479> /Supplement 0 >> /FontDescriptor \csname EF2O11\endcsname\space 0 R /W \csname EF2O13\endcsname\space 0 R /CIDToGIDMap \csname EF2O14\endcsname\space 0 R >>}
\immediate\pdfobj useobjnum \csname EF2O11\endcsname {<< /Type /FontDescriptor /FontName /BMQQDV+DejaVuSans /Flags 32 /FontBBox [ -1021 -463 1794 1233 ] /Ascent 929 /Descent -236 /CapHeight 0 /XHeight 0 /ItalicAngle 0 /StemV 0 /FontFile2 \csname EF2O12\endcsname\space 0 R /MaxWidth 974 >>}
\immediate\pdfobj useobjnum \csname EF2O12\endcsname stream attr{/Length1 10184 /Filter [/ASCIIHexDecode /FlateDecode]}{
789cd57909781455d6e8b975aaba7a4977ba3bdd593b49272189ecb19b1002282dfb36182422e0806948c24e42c21e30809080c2b007458406012120
0686c10e9b20516000c719606678e28f228a4b40c6c12d24b7dfa9ea0e44e77ffffcef7bdffbdef7aaead4ddcebdf7ecf7de2a600060a59708ce7ebd
fbf48508c80060eda836b25fce93c3521f693f91cabd095eec37ece99ee93b7a9c06103ea4f6aadf3c91dbbfeb4fb3ef01e0046af73f39aca36b52ed
da4800c949edc3c74df516a73e73c24de539d4fef2b859339c30313e1b4036539917168f9f3abdd3ac49003a2ac3def1ded26290e906dd312a878d9f
32b770fe9ef9349eee2c802d674281375f3bfadd0a8084686aef3c812a8c5be5f554cea172ab095367ccd937d841f325cca0f2fe2945e3bccbb216ff
0520b11595c74ef5ce29165fd34ca7f2322a3ba779a716a4df7e8ce64adc4df45c2e2e2a9d717ffff16d0049f3a9fd467149417137f91f944dce231e
268022ab30085e0295103aaa750ae8a1037407a177dfc1b9609ae29d310da249a6740502000f722af6e4829269a00df563d426a8a996d21b2ae63d16
4b520887ff4b57e0daff41df530f47095c0a8ed42277ee3fc57c90232e65b08015ec10070e488024aab182599504a87268be509584a4e635d46ea19b
1186522f518d8646d2d25b47920730905e8c6082665dac85f560537531cf5be21d0b4bbc2553a7c192b125de89b0649c775a29bd271494d07b6ec914
5832bea088f2e34b0a26c39209de698433a1602cd54cf64ef3c29229de22a7f2269d2e99ea9d3101964c9bacd4148df74e85252533a711e68cc269e3
e93d4119bf85dea105e70a5d9fc0979003ddc340fe5865783199461eb1792588a4e45b5ecd653621947ffadf6b8891148459ff1e0f46057195f441be
c518bfa81fd5a25cf2302f7401b8bf8af2dd1f747d046c90acea2d0d86c02cd2c341e25d52b5960cad2055958418c246c2a73eba99f44ea7fc23ba95
d49ea6d4eaaa7e81d7f6015e9b0778adff054fa039084f712eca51ab36e861cdaded1f8cd24eedbbf217ad1dd5d6427a77505b9fa756269ad86a851f
c92dbd4283250453fc1b140a56ea64d0206a4541504651383bd92c889cc23ef9e001678a4163e336b6499eca3e0be134f3140407806ac3c7a8c4d4b2
083e4a5b819372149789cf76900959900dbdc10b0530118aa10466c21c989f6c4d31a896e6248adb12070a562fc2ca87f13085b06684b0f48140e033
f2d48f0257037f0ffc2df041e060a026f056e0cdc0bec0dec09ec0aeffb07c1cdb82baffea0ac6c0e6286251670f02aa1600447150276d55492b1e0c
aa743341550951190425c665874059037a11f40e8de90d4104413e4101818d603c01ad4940eb027108240b8018829210241028f15fd17222c19c1024
11cc0f015949b2350890420fc93fc5a0c89c65821fced37d0aaa6133db4525c51ea6538d4f38084b69543f9c66e7d972a13dd5ed82bb7089302be13c
568bc006829b6a01ae4a0245f15c38446364331bcb9635a4ca21e221f129d12fde122f4296582a5e14f3c452e6c6edd27069174136be4736758ea8f6
b3eb500a47f02b74e331b1b76882eb7811abe1739a4591eb7958053ba08c68b1b1222817ca84a7a8e68c741136d15d44ed17d9167689a83bc25e802b
f0328a427fd8c2ae105fe7e107780173857252975b2824facfd05817a9ff26282537bdc2f4c085b65447d4d35c63d5773cb697aea8f75d28a7997361
87c6afb1c929348b22b15dec34abd7ac23cbbd84bfc5e9f8115b2aa688bbc5feb02a2801cc835534f626a58fa690cd25de95bb4c195d982de6b16af8
4acc93c7d2d8ef291cd19c8784a788a342f28c4298ad31134fddd8525c4e942aadf170511e2876a4fe3482bc80b80628c24c9844b932d84f71a73d56
c12a1a49e5579325fd403d378b9f12cfabd84ae107b888bdc9520bc53b246bc5b42892bc2d6b24110506ed9ce61a2175407e8d67e808e7d99149eddb
fdaae834cbce1ac8a931ce75fa03819c11629c34b24672d460aab6464c4df9f47fd5f869fb76837246386b9afaf40e8dda27af37d50d1b4159a544d5
54dfa7b7daa64c5a23a5d23320afc6396e82f345f38b295d5f341774554219c52cda3590ff2b91a3867f279469acb40e6679c2352fc346935106b46a
20426f325f1b5413913ba216f481935d460eaa0957f3e0e932f2a6abde929dfd28981beb339846b0dbac5129694266276b965056b1f885a5beaa0deb
376aac5ff0c76fdde2dd3eff86bdffc97556574ff3eda0f98ad4f9123de1b2329fccc0601523b440f375bff770dc0877a4d56e13e494ced6cc4ec20e
1a7243956fe90b2f68acf5bcfbf54f78d76f3e67efddbac5dea551af52a4fd481269454ff15840c736cb28097611a2f41abb36cc7cadb17b63f7fa47
a1e3654aeb329825c99e6449b12465265970bfd0bee9d2eea64b427b496cba54ad64aac94b196c0958d969e014c5633c61b8055ed0a0c862205a43c4
5dbe10a431cb6dc79488bb97762c7a8aefe32799478985819ec2415a9b11da7b6c10cb0426c422604f612b2c126903881deb5416efd56790fa64b374
5b8191c9cccd52046d75d34fd5d2959fa72a5aaa0c7c26ae22cf314014f115a1f159c117b6c6ba225ae7084f40873d2e9a68b8477c996fdeab37dfc9
60c982c56c75bbac16b390ee028b19529295b7f0d2e6d75ea3e7b5d7ee331dfff1fe7dfe23d34939fc22bf4070912676b34eccede3a5bc8257f252b6
92cd65f3d84a65ddf99442f228e286d4eeb1f7449f28f8a44532f874da448d03219119cc974396c114cba8af6b54045defba577f596590183b148ee1
a2303a2bc92265a6ba15e1733690bfc20afec80636eea8164bfbfbfb375ca9a601c897c481c4b103b678d26362e330da6121b55a2449ec69de66596f
f4d9d6d03a2780592f30bd23ca8c9a7873e3a01a7beea09ac8dc6707d5d8729f254a50b1d7bacbf5274f5aacd9216a5a885bbacd6a1c664b5436d1e6
713d2d0e9786cbf3c479d2acb8ca1899226e8c184baee79801b33433634be36638164345cce2d8c5718b1dbb61779c65348c4e2526323b43d6e32cb3
535a4ab246ce7c9cb95da2dda69169dfb95c38d53898c4e8f6fee68d8ae72ecd997779c497ccd6e7d9187eafbaba7a365bd375eac601b3ab7af6baf0
a8ebcb777fbbb3389e7f43dc6f267d9712f78f40b1a703d823f415bac40a6784cf6ef4e9d6691c3ee7ba94359a15f6d75b473a22006d318e34a7d981
b6449da6b52284c8dc66fe752aff240072a728d5d8ea6fdebb596ffee28e59bd492a19cca3cb4ff0267a9df949228c6609cc6e139392d3d233138891
cec4555b9619ccfc823decb1e675fe21ff72cc9949b967a79e3853bb73ffe10d5b5e7f79d88992d27323bf6061bfc3d4c4bad51f7f979a7afa5157d5
aa251b76cd2e2e2d6b9576c8e9fcf3c1f97b15bba67559dc413625501458e48967463402a2b127a041f6490c17e958981e1c1aad18a6c62003316654
190b5318bbdcbdaede6551f47af372f77a17f1a22a563c47ca3da7a8b48d81b67bfd61242df8b3e1459023595b48636db1331bc29e0c7bd2389c15b2
996c1e2e654652a58e25a1dbe2b6a7a81101355c603c935fb972ae698c94daf8195e6c74efe63e96775a8d099f89f944793c8cf1a488b1b2a5c21c1f
eb936d3ef372a3e08345c615f28e842807d3a303f4664d82b991b5d48bb9451c352bde422a32d7dd511c58f160520faf0b6a274289508accc16e835f
a845d1c6c718d3e46b37a25d036bc52ff36fc79c9e30eae4e437fff8c737876ecb95ae54f3b5e1e1fcced7ffe0df3b9de71fcd38bc79f3e15669443d
9d57c50944bd06c679a2258b80025a44f22c8928470999c840239b1b2fd4a931b8e3bf0428859511c7e908e5518f64327163c9ea32d2631d21300dc6
4ad9527f693cd6408d4626b9120b2c8525edc6934d372e31dee496ae0c6f5824b5552c6011c9b1922c3d9af6ab233c2972620cab80189f7ea7e883e5
91893ef39ac815a9b2c391149100c9c90e635c2a453a125173acfb827f1f32634f645dccbbb127e34e3a4ec6bf9b509728575b8f59bfb2e268363a4b
959935c2c4280c667602b7530986244cd62c48bb2df2d3c19b079dbd14def5e0944ff87d66bec19059f801fef9e0cdecf1655bb72e2348f4b74a6346
661dfe5b16fecd172c520d935bf9b309c2c623dbb71d3dba6dfb11255296063e93d289a718e8ec89356e33edd76fb0b06db05fdc10b5c6b222568e31
4286cd1c4b8cb8426c2841fb873b1987c2e312e304369ad93bb094208daece5976d38342a4945e786b7100f85d6666b0f856e1a4db4bf89b7c1eab60
c32a6e4b63af3c37869fe17fe757f99931cf5deadf9f6d65e3d904b6b51f499aced6520de95c4bbbe30e1e3b6cd02d621bcc5ac1ac0729c6e802874e
b4aaab08e95cd5ba12280fe64530a2c71292586a929ab6666cdd3d96c912f9a7fc3cef49b31c64557c02cfe15ea9e3fdd92c9a7560ed58d42ebe912f
e4cff32a9abd8cf4dc9e766b7a3ac11da3889e6888d299604f94a6d6647156241e71d4a6f82d2ba2c2200aa38d3aad2111b5b63e6944cf85cbf52e57
d00ceb6ede6b2449bdaf283cdb92ade87c5a467c464246628633232923b947ba27de93e049f4383d499ee49cf89c849cc41c674e524e724e7a71fad2
f8ca84cac44a6765d2d2e4d5e9bef4bbe909cd5d9b3b3577c84bc84bcc73e62515271427163b8b9316262c4c5ce85c98144db6c45463210f7c8c6559
523215934acbecd4d99dd4323e460a27aeef5b54f44aaddfdfe3d8b27de79bee33e18d8d7987730b4e8cfae75dc15d5836b6f4eaa1d6839b1655177a
4f6d3f7ed25afe52870ed5e9e98d8a4f4c27598dd2d8283038a08b2726b6164cb65a49bbc2e4671b314a04add0cf6235f48957b5e57229aaba79af4e
891f1987f3121626f81250d559883cda0a0091c454aa8354e276bfbfeb81f9e70310383fff40d39937d6aeddbd7bedda37f0b030e6e7faddf95ed69b
69e9eeede5f6f3b76e9d2708d1554e3ab4411cad4badc0ce7415da65927d0f936ac3d8d1e85aab3f6c8523ce2e68ed5a182458c3fb385412ebd4b557
b1f39beaaa732fe8afad7bc417c7fbe23f8cbf1b2ff5801eac87d0c3de234e6a2777d476d4b5d31741112b128aec4571bad1d3159f485203df4377a0
9028939f90ccc5f2c6836117df9e7466ecb80f27f37bfc0c6bdd7883c97e61e7b24db52661cca813673a75dadfa61debc2f42c82f5e21fd76d3cb47f
8b2a6b3e5c1c453c1968ad1de84989098bd7592b22226bc3b1362dc59f7e4c571b7e3c363e2d06b461fd3456abb30f2db075cd62afbb19143cbfa270
944dd26fb3b08daf8d22fde6a842144799858771fb31165209edce22a332e9f0b673c3fa9d3bd76fd8e9e7bcc1bb6fe8d02d4ffde150f6c1f91f3436
7e30ff60b65f78ececb56b67cf5cbbf60dbfc1bf8a4ff87dbb36c7df7976dc58d6952971baebd871d54acca1c023e693cd4440274f0cea004d4c5369
b2f8c336ea99a085218a4ff5b5292147d92676ec5e4f745bac51d9e4e176d5c3532c419229e356d79b4831df3f7ffe867db5b53d7f3ff3d4fbc28ea6
df0a5bb66e39b1a3a952636bda5290ffad22bf5334f95c9a17691968eb31694e8807e0982031ad087db566da6b2b72ba491b648fc1acf3e8727479ba
629dc44647b89589ec29a7fc748979f77d1adb570a1f0fc74b7e1b360a4c0b7d4573702399e1319a258f9423e549c5d25d49131c8406d0d87eae0fc9
408e275d26c3284f9ac6aa8b0e074dbc6c0fab8c77a23fee588c59064bb856abc9b168c3731cd1dad8be298a401a1b1beb83bbc3eedd6fde53b7148a
603c1119ad725a15b75addca47f73badaeb70ab4d291a454d9d85bcaeba1e0ec41c1b5ee7372f15b276a4b66aeda555b327be5aedada1e3573e7edc5
e5f3677d7f4311e3b6cd8a18852ddb5f7de7f5a64a316ffff8b1f31f68913888a095e3175a3cf69f6bf166b3160fe5d9ff64177ead47fbbfd1234dac
a831e8d933552f88222f88d0d45aa136ccaf9c2face143d16aeff3abf38527a5474c199469cae5726db9ae5c5f6e280b2b37969bcac3cbcde59632ab
2fe66e8ca545aca4a8f38b6348e9fa7d7b37acdbb76fdd5d66e577eefe837fcb2c78fdd6b973b7be3c7be6abcdfc2cafe7b7c98db3c95b6dac0b5178
84fc740751a8c4c4c73d71cd31d16f5ac18ee3b1788a87fdd4c8d857898a2e5790d69bcd61d1a30bc6c54f1244363af58168881681c2774b6765a5b5
b55d0f945d8040e042d901a10b45c63714d8ddb45fa3afcef7f263fc27ba8f79d937cd8131a8371c48d45920c363d318c8ce0c5869f2eb8ec97a8d16
b47dad4ad0503d81a2e1e50b4af83b9413b13542d1981ac25aa82b0a07260e68b7f90da2e3c8d2880e0e3c64b59c3fd174909455384e9268b62c72bc
efa42d6425b43a98246d38ee010b3ba6add41bb43a81a460b69a142be95e478f2bb49607b7cae4ef6f05fd5dd956da22bb31bbb27ed14ae1b6b0d9ac
8c2f1d547afcf895ed9595d216feeeaa26dff2219bb6fe45c85bc51e57bc7d3fd9c908d53e6dd0cde37868a12bf4ec98cd1f46f669330c214bed6b57
0c263bc8f34dd703332db29f54cc3482624cd0301eac52696cbf62a66ffafdbd0ecc3c7596fd891d11763579b76e3db14328bbefdb5738ee2e2abf59
a01bf9c85c310f0cac97a7af64a16396684151561249a433365a04811994affd7a8b4ecf94c4a097b5b2cea2d5ca3df5b2c8442dbc2309a19ca0d584
2907577daeb21da783a3597959d45dba46d9992b5b76251f3c49d6d1314a8d9faea8ec879b626d08a4dbbf2e2ba79097f5a2a88f15edfa34fd63e2a3
faa7c567e411fa42fd2c364f9c25cfd0af1417eb5f11b78a1be5b5fad5fa5d6c8ff896b8537e5defd33bf4284a924e6f8845bb64d7c51a5a639a94aa
6b63701abbb26ccc923ac99d75d9860ce300ec2bf5d10d34788c2361381b298cc467a4e19a91f270ed70dd48438eb1c83887951b5f65ebe5bd6c875c
63fc93f1ba3160eca89c7984141da3c7ad63623e9fccaaaff223fcc855f67b5e7295b566adc5bca6eb4da7989ff717060a917c3a5b453a60228f67fb
35f114a5ad6f333820822488e6cb8acf355ec8506210db5fad89ff29f497eaaad05efa48f81b615b0ec366810e17845c47c897ebd5e3cd554968e0c2
dfd4754cf90a3582e2bf04a01889db5ee3bfcbbf93937f9eba8a5acfa8df8cd456964426eb168af888bb7efedd2acd3ab5775323e6c86b298e014bc9
34475823dc808285959ef61dd9e5e38da3ca9b1abfc68dec8690c1b0e99fbca2e94ee3ed603f562a2f57be3f2b6b4acde9d3f2f21f4a83f4d868c4e5
ea887665c0143a2c59325969f9baf3be0f57fbe4e55f37eee4cf701b2f607dd81d210d5d5fff9a924c330acc6d8db05aecc2caf251bcd1b7eb884f21
2442b0b2b9828937357dc8131a277e0d4c58a9f65bae7cfd56e49881392a29ea9f8696b4b8ed3428a628830a2b4ffb567fe83bbfae5c21e662d347dc
ca8fb257593ddb83cfc27f8b3b459e6a4460a5ece3a63c79f98f434e2bbf399ae999a7f465d497e63a2dcffb61f183bef3d4be6ee50b1ab535e5b18f
a979d369f54f8030f29539b5f191cf8577ff1e12d57f2af0e7e74d379bd31fffda38d83452772df44d3f78513f792a8f0730f11fffda30d434f25ffe
29a48b17d56fea2090c1082f51e44d841a821d9a4364692fc116e1209d74a2a092e053822a82cd04f9045b0876132c92aaa054d31ace89b7a04cb2c1
74b11ea62ba964852342369c52403e0f4794b2f8b9da76040712a44016d5ed178f4137e12013953943f39f2195d5046f61658bb4e5d505c6c14ad803
1fd0f1ad1f7b5fe82a6c147ec06eb813ef8a1dc5f7248b54265dd7c4697a69bc9a77355c7e52de225fd15ab5b9ba485d7b5d2fdd53bab1bae7752fe9
ded09dd6fd59775bd7a037e8a3f58f04e50619980b6d61028451a436c32b8a5445bb1049a9f2cd5f8651ca775e5147c819eabf1625cf20924ac1bc00
5ad63794c716f5628bbc04d16c4828af011b2b845e5004c530174a60228ca7d967a8ff92c6416b4a5d9041b79b726309c3093d096706941294400178
612ab4a3da01308df03b50ee099842b7139e7a3056a95a2aa0b480facca2773e61eaff1bb3767e306b2ecd348be652beec4f236c850e2ff5f9df9bb1
37e52651bfe1309330c611ae571dad40ede1553972d228d3e85d4c386369dc8984e7a4fe4534bb576dfbf538c3d4514a89223aefc064aa55662d25dc
22752417cded86cc5ff46aee13fadf1d785efd9ffaaf57866a17ca3f6fe5ffb609c2c1aaaedc7688a47d5e34c4d25e2a0dd269ec7ed01f06c220180c
4f420e0c25fe87c1d3f00c8c8067e139f8565955e9ac21310d93e954a8a35394818531a33c73da44973be38950da33986637a7bd4269ef50da2794f6
0da64f64845257287587d2cc50da19c02f2cf404ee736cb0e1cfa9f8930b7facc21f4cf83dc77b1cff998adf99f01f55783715bf7df109e95b8e77aa
f07615d637e0370df835c7afbae2973df116c72f5cf8f9cd61d2e7557893106f0ec3cf6e74943e6bc01b1df1538e9f70bceec2ffb0e1c755788de347
56fc1f0bf0ea51fc3bc7bf12fa5f17e095cbfda42b0bf0723fbcf49738e912c7bfc4e19f397ec8f14f1c3fe078b10a2f9c4f902e703c9f807f74e139
8eef2fb548ef3bf0bd48ace3789ae3bb1c4f713cc9f11d8e27381ee7788ce3518e472c585b912ad572f4bf7d54f2737cfbf068e9eda3f8f642f1f01f
52a5c3a33d013cec11ff908a8738febe0a0f723cc0b186e35b1cf7e7e39b26dcb73755da978f7babadd2de54acb6e21e227a4f03eee6f806c75d1c77
5a7107c7d7b79ba4d75db8dd84dbf2d14728be2adcca71cb6b61b4d7c3d7c270f3ab31d2e67c7c7593597a35063799f1153dbecc71639551dac8b1ca
881ba8d3862a5cbfce24ad7f04d799706d03ae597d545ac371f5aad1d2eaa3b87aa1b8ea77a9d2aad1b8ca23fe2e1557725cf152076905c7973ae08b
c4e68b4fe0f2650669b90d97d116992a2af3b1822455918a4b2db884e30b8b2dd20b1c175b7011c7851ccb397a02cf2f58203dcf71c1029c9f8f65b9
76a92c15e7719ccb718e096787e12c3dcee438a3014b1bb0a401a7376031c7228ed3384e49c2c91c27597a4a9386e1448e1316e0782a14722ce098cf
711cc7b11cbd5d31af01c784e1688ecf721cc571e408bd34b20147e8f199c818e919170ee7f834cdfc744fccb5e33066968645e353361c3a30421aca
31c7804f721cf21bb33484e36fcc3898e3206a19c471e000b334300207c41ba50166ec6fc47e1cfb56619f2aeccdb1176d8d7a3560cfa3f8c420f470
ecc1f1f1c7acd2e3367cac7bb8f49815bb77334add3d8170ec66c4ae1cb33976c9b2495d1a30abb359cab261e74c83d4d98c9906ec94806e23ba1e35
482e8e8f1a30a3a341ca306247037668af933a98b1bd0edbb9b06d9b54a96d3eb6696d95daa4626b2b3e929e2a3df204a6a7625aaa414a0bc75403b6
e298c231391c9388cf242b3af331b101138885847c8c37a28324e8e018d780b13d31860a311ca3f3318a2415c531923a45c6a09da38d6304472b2158
39edcbdb4b969e685e80e1f968e2680c8b948c1cc3083b2c120d1cf566d471d4129a96a36c434d3e8ad4289205d8916a91d33eca2c09ed9199113832
3fcb5fba92b5fdffe182ffd704fc9757fcff04ad082f53
>}
\immediate\pdfobj useobjnum \csname EF2O13\endcsname {[ 32 [ 318 ] 40 [ 390 390 ] 43 [ 838 318 ] 47 [ 337 636 636 636 636 636 636 ] 55 [ 636 ] 57 [ 636 ] 67 [ 698 770 ] 97 [ 613 ] 100 [ 635 615 ] 103 [ 635 634 278 ] 108 [ 278 974 634 612 635 ] 114 [ 411 ] 116 [ 392 634 ] 120 [ 592 ] 124 [ 337 ] 8722 [ 838 ] 8970 [ 390 390 ] ]}
\immediate\pdfobj useobjnum \csname EF2O14\endcsname stream attr{ /Filter [/ASCIIHexDecode /FlateDecode]}{
789cedced96a02301005d00b2e75abfb5e5bf7ffff4543f049faa4a02fe7c00cb9814c267951e32137d32abd9daf9a3ae9a6977e06f9ceb0e451a9f1
3f5326993ebdc1acf67916a52fb3cabae64db6f9c92ebfe5fc576a9f43bd3fde5f9d9efe0f00000000000000000000000000000000deebfce90578ab
4bae37f46402bd
>}
\immediate\pdfobj useobjnum \csname EF2O15\endcsname stream attr{ /Filter [/ASCIIHexDecode /FlateDecode]}{
789c5d53bb6e843010ecf90a974971e219b8931052746928f250482a7405d8cb09291864b882bf8fed315c144b309a995def62d6feb97c2965bf30ff
438dbca28575bd148ae6f1a638b196aebdf4c288899e2f8ed9371f9ac9f37572b5ce0b0da5ec462fcf99ffa9cd79512b7b7816634b8f1e63cc7f5782
542fafece1fb5c41aa6ed3f44303c985055e5130419ddeeeb599de9a81986f930fa5d07ebfac079d768ff85a276291e5215ae2a3a0796a38a9465ec9
cb03bd0a96777a151e49f1cf0f33a4b5dd1e1f9978400dbc58f908f9e4e49dc26d41b973770ab7b3347e72aea328148780081003922dc36e1067a019
36d008f904d9f514bb6612ec912490efd4ba29ea19a881905131753dde295c3490ba4a1b3d6e2a82b8a5993b394753012040b7c5d8940c5f6da00642
46f5cc3573a77051d6400d84eccabb7f90d9c38fa2d09400d44023c74163e43868adfc875ecc986cf36026c68cf73e8efca6949e447b07ec089ae1eb
25edd7641a2793659e5f81f4dc93
>}
\immediate\pdfobj useobjnum \csname EF2O16\endcsname {<<  >>}
\immediate\pdfobj useobjnum \csname EF2O17\endcsname {<< /A1 << /Type /ExtGState /CA 0 /ca 1 >> /A2 << /Type /ExtGState /CA 0.18 /ca 1 >> /A3 << /Type /ExtGState /CA 1 /ca 1 >> /A4 << /Type /ExtGState /CA 0.8 /ca 0.8 >> >>}
\immediate\pdfobj useobjnum \csname EF2O18\endcsname {<<  >>}
\immediate\pdfobj useobjnum \csname EF2O19\endcsname {<<  >>}
\immediate\pdfobj useobjnum \csname EF2O20\endcsname stream attr{/Type /XObject /Subtype /Form /FormType 1 /BBox [0 0 306 201.6] /Resources << /Font \csname EF2O1\endcsname\space 0 R /XObject \csname EF2O16\endcsname\space 0 R /ExtGState \csname EF2O17\endcsname\space 0 R /Pattern \csname EF2O18\endcsname\space 0 R /Shading \csname EF2O19\endcsname\space 0 R /ProcSet [ /PDF /Text /ImageB /ImageC /ImageI ] >> /Filter [/ASCIIHexDecode /FlateDecode]}{
78dadd9acf931db511c7f73cb9e52f980a175309b3ea6efd681d3104aab819bb2a07cc21b5368ba9b58981040ef9e3f9b666de937aece7ddc53fb612
e3c56f7b66a4d657dd1fb534efe54c73c07f347f12ecefc5f3e9fcf3a7ff7976f1f4eb2fefcf9f3d1c7fbbf879a2f907fc5ce2811ff0f32b1efb123f
979335f17c9290f1ef55fb97032d199fc3f1d3f7d3f4dd74fe296eff19777d394d224b4981e3cc752939694ead0d5e6aad250fd6abc14a5597d43e5d
f5064663ebe6e5fc9ac639d785634d3c53d64545b4a6f9a7a7f33fe617f3f9a76c7ef1fc15c683b12d3aff3a8525d7403587a21103f552c4b070cdd1
39fa7cb08e3e4d0fa707f34b7422ebe02f7b27268559ae6d8ff3c20739a6fb98855fa797934dde27367ba20b450da2a485662a4b4c22b817f379ffd1
74fe05cd19bd3dfaaecdd4a327d337f3bdb3f0f1fcedfce8abe9ef8fa607cdbff7ae592d0b73e0eca7bc5b6fabd975edbd59b38a98621b111704eccd
34a33b508d52c4e4c690c40d7330df56b7eb5b7cb37214eba2b150299ae24dc38def403a265a42c61572031dccb795eefa16df2c1d87bc94126ace51
ca4da34eee42ba5c805cca3e44baf5d6c25dd3de35b26581b3b1248c996e1a71d1cb3606b0c4454a9ad3128289b8b6824741041034590bd6665992b5
79efecf38fe7473f4c09539728e6a831aedddc3b7bd6aef012544a0c2cfdca3fdb950895123cc7e0f2e1cae576a540fab876b85df9f1f08c6a660931
d7c39517876792442a35ab5edfcfd549dfe6ed8a680c221c840e579e9cf4ede9762553e45494f878e55f279ff965eba746648728ebb19fef4f8ee794
6fe75f709f90e7eda60f970ec70ec048a921819163b1d2ad37cb8613cd2dbc33be9ebe876769212a44a992a621864b941ec25b5a60e9da04e71a534a
e0cf712ac27645325600209d506d2cd9e6a43763933bb67307caabc277921a9df2dd7a4be57d739bf2aeb5d3ca2bc081a04d5923d0d595cf6b5abc2f
e9e5aea427e6852ab3f8127d30df52fc5d839bfabebdd3f213e982982fb594022d3f54e4a73b933fe525540970c4c9dfcdb795df377890dfb5f706f9
132f0aa7454be00fa67eb933f5b54270a954bcfadd7c5bf57d8307f55d7b6f505f51417288d089ab7e38f8d4aeffcba99d1858437581785688c1ab84
d283c7926c5bad8792ecb3b31fcffe7df6e2ec97b3199f7ec2ff9fe1b727565b505ef2e14fc1adbfbd75c92be38c539b6f6204394a2bab5096d8bb43
f006d45c200fd7f9eb7d284c312d5555539a0595670c24b65e63fdcd901c9fbbf56a4a05a52524656745fd1a32ea6528ca4b8a5c48cdaa18128a3367
2c11c52d0aa93aa3db0457a5c2aa703e0ab17a6bc2a85114a385b264d4c4c13aab6477c05f6f3dee8bb17d63b57986952ca00956d99975c90172f39c
105a12886d4d24b2404231493b3362576b51ab494a8ee61a6211f240ef3a1a85a00e3c2a3320521075da7a835421676c04bc19b2945419939ac4ea4e
8470332b14a25ca237f7edeb5840c1016884f1b13767935115bb8984d51c757fb1c4433d0595b0c0d49d39432541ec98b908630b676665530965f7ce
6c65840a63387911a68aa483b90a323c20d276e60aa1248041302b86a5361c0ec9f62856603873df69b6488394b199b395df365da3150b6a4c8cb833
b32650439ad98e8b629bafd18c0c43304ac2686c00b92d042c15425584b4b3220ad9844a662d549aa820196442cbe2ad48ed00aac0b98a9235a5d89c
cb8ce62a27da990f5b4258d11289850217cc343604d0d79b2b468b9548cdac89edec11662bd0d02132db992bf23c2386d39c914e98b6cc1dde6fcd97
994098583508c2261954428aeb02a3ef192ace7a84ca008a3f06959eb81d2963e876a4b864e948c968001bcf441e29cedc916201c1ba56df1d2998aa
840e2b7ba8b89b3b56d034f25443f458c9d9f6a0c095c70ae6172e15f2542988e22c2dc006a814733ad09ab31d2a05c510340a3ba8205e298458c943
057747a4d08aa00e15d776870aee0646b2eca0525a8aafc13b4045214260d4cf1e2a0a4f206f610f15c5ea1c0314f450412f5453688b40870a3a21cc
28250f159805d553510f1598238aa4143d54ecee2452aaa70acc98dfdc941ab0526d3a50d224cf15334760a37ab050001f62a4409e2c14700ba44b1e
2d168b580962f168c118978a845fe1d9d942844168a81a3d5c6c2382e9138a9e2e16bc369b9cdf162f87d2e59b198fcff626e55b54534fc0162abcd1
c56a3f46f4c33f84bdbb7217a5cc507274ea749274e68c4bfdc09c6145ead419d3ba5307f35cd155da1532986691505ba08dd4c1ec4b2d2b0306ea38
f3913aa375a0ce681ea8a3607db2be1d750c024991199e3ac87b82da75471de477d8765f23769cb9630766e1d0e27dc4ce78f7801d6d0c90bcc3ced8
c8801d98b1669178eca865bb48d51d76b020a4523673c70e06863059f9d2b1531133150c208f1d8340d6b4c78e8d8634c80e3b98d25c112e3bec9819
0fd30e3b662ea905c9881d98edf04c3c760c24a5d8c2e2b9b3dad1a2eec183cc05255876e0c1d814b569140f1e9004bed6a07bf22082aaae4c77e401
4994cc034f1e863ba15a487af23002074d3624bf23f29071271dc1a3adbcc12e4a8c35503312422ab55fdade0af1fccec0e310d3c13330e6089e8e98
013b43553e8067a8203b78b2ad9335e90e3ca8bdb18cb7cce9e04185205883843d788a4d2e122479f060bd4e62fb140f1eacd7a186f6c6a373472386
16522e9e3bb62c63e243f2dc819f189f6af1dca908e752360a74ee20e4952dc53d786c0048ecbcda87d78ac85af36f2d3f3a7a2840bf883a237af6b4
d5b3927d74f021e473928895dfd387181572c100b3c78fc5762678cc9e3fc465417697b64d18004462dbf748ab0a9d40f69687b1c96f7e0e0832a12c
0557220cef022374c55aaed541086a20c3a5793340a80988946eb33a5088108788a2b29539470c617fb32414a061c721133644d37f0722b004b9bac6
c708224424827823c20022ab6511e9adfc1b4184a064a455f21cc23480899b970387d4be718028293b0e01feb823afcd0c1c421929903bef2ba06aef
d25086f02ba763b663ea2763270edddc77425ef73d0fb4798b2f90e0eef34fe37a3667fbb5cbb65f332fd62fb04c8296a9f561a99d6a3b19450505d5
bdd5220c150ad48644dd6cae8133dbad583a43156744a292dd78310dd66ad17168b55b23eaa20ceacd6357c869dabe0273f4ea68bb180670345e4d08
f41c107c325861a3b4eb67b01d3cba98baf1e8fcd5683c0e73e8e5a8c76bf4bcb0efe6dc9f76fbe63f7a268704d2552e0042b7098bd952c019b19860
37f28a1101619110da9799d6f3d9e1b035e941611437b134effb39abbee6947b7d1f09ce616b47ed4de7fef52c21c5410ea9c7f3d7dfda15c358cc0c
1e61eddfae3cbeb7bd0345b915c1163b9add2e9d3ecefddbd65c04102af2bd1edef67ef4f8c589e6ec95eab54370e33c3fbc5266b00f80a1e30b6a3e
e5d8477f3a3998c71ff743e535206e7180d22300f56ac0cad75e26f510e8d6310646eb4d83a06019e6f6cc1005edac7b2fcf7fb7a1a20872f238a11f
1e945244fbfaeafb13b2ef6ff5f7273c1cbc3f3fbc2947a18547500aed6e7ea5ffa6e7bbda31b63d2eed33cd4a13a17d5691f042f4da5cbb9cbf3938
f2aad4b622a7f47f9370efd4b90f153a632ebea39a7f089d62dfd87925748e56173a83f556a183340dffe3a1e326fb8bed26ac8176c864375d3fd956
8da358ce3974c88faff1febaf9a4d8ea14d476b6d07a6f13b52f62a1dc3c1d2d7ff9f3a6612d01dea1ba90eb073d46d883697af03b38c003cd
>}
\expandafter\gdef\csname EFWidth2\endcsname{306}
\expandafter\gdef\csname EFHeight2\endcsname{201.6}
\pdfobj reserveobjnum
\expandafter\xdef\csname EF3O1\endcsname{\the\pdflastobj}
\pdfobj reserveobjnum
\expandafter\xdef\csname EF3O2\endcsname{\the\pdflastobj}
\pdfobj reserveobjnum
\expandafter\xdef\csname EF3O3\endcsname{\the\pdflastobj}
\pdfobj reserveobjnum
\expandafter\xdef\csname EF3O4\endcsname{\the\pdflastobj}
\pdfobj reserveobjnum
\expandafter\xdef\csname EF3O5\endcsname{\the\pdflastobj}
\pdfobj reserveobjnum
\expandafter\xdef\csname EF3O6\endcsname{\the\pdflastobj}
\pdfobj reserveobjnum
\expandafter\xdef\csname EF3O7\endcsname{\the\pdflastobj}
\pdfobj reserveobjnum
\expandafter\xdef\csname EF3O8\endcsname{\the\pdflastobj}
\pdfobj reserveobjnum
\expandafter\xdef\csname EF3O9\endcsname{\the\pdflastobj}
\pdfobj reserveobjnum
\expandafter\xdef\csname EF3O10\endcsname{\the\pdflastobj}
\pdfobj reserveobjnum
\expandafter\xdef\csname EF3O11\endcsname{\the\pdflastobj}
\pdfobj reserveobjnum
\expandafter\xdef\csname EF3O12\endcsname{\the\pdflastobj}
\pdfobj reserveobjnum
\expandafter\xdef\csname EF3O13\endcsname{\the\pdflastobj}
\pdfobj reserveobjnum
\expandafter\xdef\csname EF3O14\endcsname{\the\pdflastobj}
\pdfobj reserveobjnum
\expandafter\xdef\csname EF3O15\endcsname{\the\pdflastobj}
\pdfobj reserveobjnum
\expandafter\xdef\csname EF3O16\endcsname{\the\pdflastobj}
\pdfobj reserveobjnum
\expandafter\xdef\csname EF3O17\endcsname{\the\pdflastobj}
\pdfobj reserveobjnum
\expandafter\xdef\csname EF3O18\endcsname{\the\pdflastobj}
\pdfobj reserveobjnum
\expandafter\xdef\csname EF3O19\endcsname{\the\pdflastobj}
\pdfobj reserveobjnum
\expandafter\xdef\csname EF3O20\endcsname{\the\pdflastobj}
\immediate\pdfobj useobjnum \csname EF3O1\endcsname {<< /F1 \csname EF3O2\endcsname\space 0 R /F2 \csname EF3O9\endcsname\space 0 R >>}
\immediate\pdfobj useobjnum \csname EF3O2\endcsname {<< /Type /Font /Subtype /Type0 /BaseFont /GCWXDV+DejaVuSans-Oblique /Encoding /Identity-H /DescendantFonts [ \csname EF3O3\endcsname\space 0 R ] /ToUnicode \csname EF3O8\endcsname\space 0 R >>}
\immediate\pdfobj useobjnum \csname EF3O3\endcsname {<< /Type /Font /Subtype /CIDFontType2 /BaseFont /GCWXDV+DejaVuSans-Oblique /CIDSystemInfo << /Registry <41646f6265> /Ordering <4964656e74697479> /Supplement 0 >> /FontDescriptor \csname EF3O4\endcsname\space 0 R /W \csname EF3O6\endcsname\space 0 R /CIDToGIDMap \csname EF3O7\endcsname\space 0 R >>}
\immediate\pdfobj useobjnum \csname EF3O4\endcsname {<< /Type /FontDescriptor /FontName /GCWXDV+DejaVuSans-Oblique /Flags 96 /FontBBox [ -1016 -351 1660 1068 ] /Ascent 929 /Descent -236 /CapHeight 0 /XHeight 0 /ItalicAngle 0 /StemV 0 /FontFile2 \csname EF3O5\endcsname\space 0 R /MaxWidth 686 >>}
\immediate\pdfobj useobjnum \csname EF3O5\endcsname stream attr{/Length1 3500 /Filter [/ASCIIHexDecode /FlateDecode]}{
789cb5557b7054d519ffcefdddb3bbd924379bb8840d097143b22886002e124b64ec26e405a844133089447293ddbc4836cb26420244b0842aa2e565
b742a3604b2da2b551fba0c5e9b45aab1dcb8c16331dffe828d3769ccec48e767cd4cc70b6dfbd09a04c9da97f78ce9e73bfdff7feced9730e0922ca
e449274f4d6555355d450e2291c7dcec9abab5f545cf15bfc1f806c64335f5eb2aae39f1cd97193fc338b1b67e71b0273b7a8848ab64bcbebdcf8c51
ae38c0788c7163fb9641ff63c7c6e733fe886d7a3b629d7deff57cf87b0ed6caf27d9de6408c9cdc499f649cd6d93bdc31ef93b5ab19ff87f55fea8a
9861e7fa97b2881c2fb2bcb48b19694fc95718b33f2aeaea1b1cf27e464b889c3ec6bedefe7693ce539c31f3c8e8338762d8e448611c62ec8f9a7d91
45a52b2d39c7d73cb1fe81c1e409da45941263f9a2583c123b74509bc7f838e79046d6daa4d174d31881a4cdb3869b16d10a2ebcfa9606327acdc128
f9780db925934497285b7b53241e25d78c9d6099667f5dfc9db6f1d3b54c6b74b109a6753b9635846d33fdb5e21ea247c86bc7dd66c6cd361a35e37d
511a6d8b9bdd34da6e460778ee8ac4791e8ef7d26867a49fe9ce7864138d769951d6e98ab4316793193569b4d7ecf75b33e73fda670e76d1687493c5
e9ef34fb68347e4f9435073ba29d3c7759fe3f57e3e5267483775d12c9a5f208958a7ceb9bacc25fa843e3ddd3521d00f4544d9f59834badaea32acc
9c72aa707895571c75f689bf7d41073323cfae9c78c5614bf3380fbfbdceba8d2bec8cca93c9e4cb33d6d3fb96776945596ae94def8d28148be934bd
cee377748a5ed55ea5e37423f73be83762af56c29227e97e7a5b6af4331a13cb85572c67e93987d7312cf7c8932cdfc0b6b5ece56d3acc3e2d4fa769
a7b65daba30e7a559ea5a3dcfb6dfe07f46bb19b26e8517a5daba54f68371a683ff7a334c09b3c21dca4b4623a614522ebdcb4f1988b123961f70f68
276da7063ae138edf08a7376d64f8a97c524fd93733e870dd8cc16637452ecd10bf5937a2ded9fce17adb45fdb2d8eeaad76dfcefbb095c6f45671ca
e1a557ac5c9953c79976d08b3cb6d2597193d883bd9cd9762b033941679dabf5c5d3593947b08ceb211ecfd2f3548204dbdbb5383a684cebe0589f70
266751490bd8db0091976f089af50b87d4a1095ae8f78c6b8155e1f1d0ed8dfed79a0a4a165e01fd1ea77f9ceac6d387fda793c9ba463d57368dcbbc
71045ce37aa0f0fc9709cf972c5c53d7e81f7fa1aa72c66b556b25f3ea1b99b410b3995f5559327d8af9d4f17f053c1fb6ee1daed2c9ffaca7432197
4c71094db8a526741e909a2652245f4f2ee9e6c3eb7448724a0b97b9486f4ea16a87e64e7159e50957aae64cf5bcf5a7ccd9cbafa7c52bfefed664d6
f2259c81d323df77f2cfe37adf2667e034a769def34e97102da185197a86cc48c970bb5c7b5d0e17b95c2ef71c9aedb88e6a6995a8d56b1c77ca2eea
945be503ae8c166a114b45210a5098220ab47fff43cc3a77f385d6731338a95ab4223931552c76aafba47121c73a3937f299e4ff1819349bcf4b6f68
458e24c3c727d5f048838c0c8b2a736634e7f89aa9dad39c33909ee39b9d9921bd1e4338d3e7f82579f3dc73f3f8ca32720b3c1f71656f4d729559cb
b3ec5227273f9af4fceb72ad978b6c9a270a82d9b3bc0e67be98e5a50542a02058baec86f985a9628610cd9afee7375e38b863877a41bd73e143d1d0
23aa7ef987dd07cf368a92277ebeeddb7262ecb1036f66e73d3e7a5e4d697fbcf0d79347b6f46a1deae68776984d174fb5d674e4bbabdecbdf98b1e2
63bada651ff637ef35deb9f8fdf4dd0b4bd3ab5dfbec9dbf7cbf92b34fcd254adff3e9bb9ffd2abd9ac2fcee7dbe3974eb5c58ee4ff1d8c7ab389c3c
2c47acf5bcd4e2fcd64ce7e0c06d544c5d7ce768e4a190f5a2c1d0d239a275e339a9d9bab5747e87c412fb3eb36841d98ca6698d0c513d4383e68b86
195aff1c2dc927b6cdd00e2a128768259fc5180d731eddd4c9d107ed77a49dcf9f9f82fc2a2ea1a54cb5b1869fefbd6e960ff08853844ceaa385cc5d
4551d65fc45439f572f7f3fd77d1d7808d22fc8db0cd169ec3ace9fe3fa2965e8adac091b670ac1eb689b2b69587c9365f2d6225533d6cb79eee618d
76d6356d6f11dbc2b42bf2b39728cf31d66963bfddace767fb7e8e6edab22bfdd4db5e0668ed8cfe66e646be44c77f85d67a3bc301c6fd76d420e7b9
94967dc1faa26dc915b6fcca273fe6712fff2bfe5773ccbcffa0220a109dd6768592cf298c07f0d3209e4de027069e1930e433413cad702a80a70c9c
0ce0c7093c39851f4de184c20fcbf003852782387eac5e1e4fe0d8ade5f2583d1e0fe2312fc612f8be1b47158e64e1d1117cef0c120a8fb0c6232338
ac70e8608d3c3482833538b03f571e50d89f8bef283cacf090c23e8507f7e6cb0715f6e6e38120ee5718cdc66e856f29dca7b04b61a7c2bd0a236b02
72248c1d0adb33b16df88cdca6303cd42287cf6078973eb43520875a3014d2b706b045e19e0406c3183010df1c90f13036c7b2e4e6006259e8e7b4fa
a7100d2515fa147a153665a3a7bb4cf684d1cd31bacbd0755baaecf2a1b3c3909d41741888841166b37002ed0a6d669a6c5330d3d0ba3147b686b1f1
6e8fdc9883bb3d687163c35de97283c25de968668be6049a1a0dd9742d1a0ddc3985f5ebcec8f50aeb1a5ae4ba3358b74b6fa80fc886163484f4fa00
ee50b8bd6e91bc5da16e11d672126bbdb82d15b77256b796e316fedca2b06675a65c13c0ea4cac52a8adc994b50a3599a856a852a85458593122572a
548ca05c2134856f4ee1e629ac28ad902b146e7a0d654c95d563b90ac5f08d11dcc8b0542f91a51558a67083c2d23204a7b0240d8b154a14162a14b3
b8f87a5ce7c10278e482425c9b8f6be61bf29a30e61b08f0e31608a228cd278b465028cb64a1c23c46f3cef07be29105b9f05f9d2afd19e057f6b7a1
a3fad5a9c84f417e489feb411eabe725909bc09c9c809c13468e2f4be604e0cbc2ecec809c5d8eec0066297815ae9a4256668ecc52c864af9939f028
642818ecc148209d03a68f202d354da6f9909a06b7828b45ae041cacee50905c852c83ce482f013cb05f651f841b22a4531ec46911def3b028fe7a1b
7dcdfebf6a9bfb5f9be9dd93
>}
\immediate\pdfobj useobjnum \csname EF3O6\endcsname {[ 65 [ 684 686 ] ]}
\immediate\pdfobj useobjnum \csname EF3O7\endcsname stream attr{ /Filter [/ASCIIHexDecode /FlateDecode]}{
789c63601870c0c2c00a000097000a
>}
\immediate\pdfobj useobjnum \csname EF3O8\endcsname stream attr{ /Filter [/ASCIIHexDecode /FlateDecode]}{
789c5d50bb6ec42010ecf98a2d2fc589b395d2428a2e8d8b3c1427559402c362219d01ad71e1bfcfc25d1c252bc16a9899d5b0f2dc3ff6c16790af14
cd80199c0f9670892b198411271f44d382f526df50bdcdac93906c1eb625e3dc071745d7817c6372c9b4c1e1c1c611ef0400c817b2483e4c70f8380f
d7a7614de98233860c27a11458743cee49a7673d23c86a3ef696799fb723db7e15ef5b42682b6eae914cb4b8246d9074985074272e059de3520283fd
c7df5ca3dbe5f70dcbb9b50a3effc0af62ff119649e5db7b4cb31271c2ba9b1aad84f201f7f5a5988aab9c6f403e74b4
>}
\immediate\pdfobj useobjnum \csname EF3O9\endcsname {<< /Type /Font /Subtype /Type0 /BaseFont /BMQQDV+DejaVuSans /Encoding /Identity-H /DescendantFonts [ \csname EF3O10\endcsname\space 0 R ] /ToUnicode \csname EF3O15\endcsname\space 0 R >>}
\immediate\pdfobj useobjnum \csname EF3O10\endcsname {<< /Type /Font /Subtype /CIDFontType2 /BaseFont /BMQQDV+DejaVuSans /CIDSystemInfo << /Registry <41646f6265> /Ordering <4964656e74697479> /Supplement 0 >> /FontDescriptor \csname EF3O11\endcsname\space 0 R /W \csname EF3O13\endcsname\space 0 R /CIDToGIDMap \csname EF3O14\endcsname\space 0 R >>}
\immediate\pdfobj useobjnum \csname EF3O11\endcsname {<< /Type /FontDescriptor /FontName /BMQQDV+DejaVuSans /Flags 32 /FontBBox [ -1021 -463 1794 1233 ] /Ascent 929 /Descent -236 /CapHeight 0 /XHeight 0 /ItalicAngle 0 /StemV 0 /FontFile2 \csname EF3O12\endcsname\space 0 R /MaxWidth 974 >>}
\immediate\pdfobj useobjnum \csname EF3O12\endcsname stream attr{/Length1 11756 /Filter [/ASCIIHexDecode /FlateDecode]}{
789cd57a797c1455d6e8b975aaaaf74e77a73bfbd249a71342586242088120cd1276310a22e044134802c8924040060313d684c812100882802d9b80
8c13918124208312410ccc380338e3133f54a2a046649c281a939bef547587c59979f3bdf7cffbbdaa3e75b773ef3ddb3de7deaa060600367a88e01c
36246b2804426f00d68d6a8386653f3cd6bdb3fb03541e42b073d8d8c70625ec19701a4068a2f6ea87068e1bfe50f7ed2f03e0226a6f79786ccf94e9
3f15ad02908aa97dfc945979c5586fb94ae577a97dc59467e639617a6406806634957961f1d459737a3df334808ecaf0ead4bc9262d0d00dba1fa96c
9c3a736161c9c5ac0200bd0810da36ad202f5f9bf37639401c3543ef6954617a49b383ca7da81c376dd6bc5fff35d8f509657f45e5f3338ba6e4d5ff
e5c8a700ee2e549e3d2befd7c5e26e790e955fa0b27376deac82846ffa9fa0722dd173a9b8a864de1f337f0c00485843ed3f14cf2d28eea7f93b65bb
ac201ea681222b23f82e814a08696a9d027ae80199200c193a7a1c9867e6cd9b0d212453ba3a3a00eee454ec1905736783d6df8f519ba0a65a4a6f29
984c60f124050bfcc7abe3ca7fc6f99f8ed271ec4ef9d8bfc7fdbf1999d28b1d5714f0d575e688670dd99b03c220129c10032eaa7140802a17df25a8
4f549f2248f494a9d54a58ccdf8aaa6c25aad7900475a068420b06554f2630d3689dfa791e36815dd5cfb37973f326c38abcb9b366c38ac973f3a6c3
8a2979b34be839ad602e3d17ce9d092ba61614517eeadc8219b0625ade6cc2995630996a66e4cdce831533f38a9cca93f4bc6256debc69b062f60ca5
a6686ade2c583177fe6cc29c57387b2a3da729e3df630b772f1f5d9fc00dc8864c23683e56595e46ace6124b977d484afedeabb3cca6f9f38ffd673d
30338dfbcc7fc683493e5c25bd93bf678cfbea27dd539e7b372fd05afcb98af29977ba76013be957247ee3610c3c0387e1a66af78a4613a895d16df3
4908148d26297d74f3e9d9955abbe8d6527b22e540570d700f5e8f3b78ddefe075fb273c091e50f1141a5354bc17082f59cd3d775f5d4fb5eff2fbfa
a6fd8bbebdfe45dfd47fea2b104fd45752acd4a6b44a8acf60a299ad572c594a95b692e4a27c29fe150a05c2170c32a25614046594bbeb80aeecc2ac
7cf0901c17cb766e67db34b3d8b5fb70d00f11a05815c0112a31b52c82e2ef62551d08f4ec4a52ea037da13f0c8111f03014c054980ec5300f4a61b1
6aa34ea2b71bc9e82e4e9e8a3303e6c2b30a4ec7355ac71f75fcade3af1d973bfed4d1d871a6e3f71d473adee838dcf17ac7a18e57efa7fddf5c3e5f
fabebf6451e7f501aa56a3e8df27cf6e7ea0d840540128fe5ef1a07dfd6025e8ef0745ea43fc104830c20f768287fde020c823a0d802410453fd104e
30dd0f8a1c671028b12c9240b18179042e8267094a09e20916fba10b4b835a68a4fb2d3808dbd93e2a1552fd1caaf10a876125cca79ad3ac91550add
a96e1f79fb8b8459018d789016c748b2a146c2ff90eca4858d23fd6d6719ccce3234322d9531e211f151b156bc2e5e8074b144bc20e68a252c157749
e3a57d0419f80ed9cf3988865a76154aa01ebfc4543c210e11cd70152fe041f89c665164d90855b087a8af053b2b8232a15478946ace4a17601bdd45
d47e81ed641789ba7ab61c2ec30b280ac36127bb4c7c35c20fb01cc70965a4a254a190e83f4b635da0fedba08496f365a6072e24511d514f734d569f
91d85dbaacdeb7a08c661e077be45ad9ae71d12c8ac4f6b1d3ac59de085eb888bfc239f8115b29bac4fde270a8f2490073a18ac6dea6f4910bd942e2
5db94b95d18505622e3b085f8ab99ac934f63b0a478af50b8f124785708260816c219efab19558499496aafabca01929f6a4fe34826631710d508469
f034e54ae135f250ddb11aaa6824955f395dfa817a6e173f259eabd85ae107b88043c827158a3749d68a6991c739a691251105f2414e4b8de01e915f
e3796482f3dd8931ddbbfda2e8b4689c35905d635ae8acede8c89e20864b136ba4881a746b6b44b7ebd37fd7f869f76ea3b227386bdab386f847cdca
1d427563275056295135d5670d51db94496b2437fd46e4d638a74c733e6779ced5f7394b41dfeeca7a13941d07ad76c54bd4f0ef8452d9461133dd13
20bf005bcc260da04d8640bdd97265544de0b80975a0ef38d567e2a89a00350f9e3e139b529aad19190f80a5ad3999c982c36e0b76c50b69bd6ce942
69f9b2e52bbdd59b376d916d5ff007af5fe7fd3eff9a9df9e42a6b68a6f9f6d07c45ea7cd19e008d329f8681c126066a81e6cb6cb93b6e606a90cd61
1734aedeb6b45ec21e1a7273b577e5f2e5b2ad99675efd84f7fdfa73f6cef5ebec6d1ab5bf30124f938d59618ca79bc56400d1a8d38882a8c5ad92d5
b855cfcc1b6c5aa31e75b215190a600f100d3abb49b6592e6536a55c2108ceb0da8295995b2e6536a7343453059593998969dc5697554ae8ced2256b
aad58da7f95636b52faf99cb6bfab2a97c6b5f963d97658b57df3e3db99157b0858d934fbf3da5912de4158d44d9ce0e1b3b0d9cfc7ea8c7883b61b9
8c220b85109958bc74dec7697aaa035d81b72eee59fa283fc44f310f69a8a2e39a58451c1920185c9e40d96b03af71836d4d882e22200a231ce121d4
b3a59906686a69b6dc4c66b182d5624b4db1592d42420a582de08a559ec2eaed3b76d06fc78e9f998edffef9677e9be9a46c7e819f27b8c052e9eec5
52bdbc8497f30a5ec2d6b285ec59b656892fb47f1627d16e8454ee710c42af2878a5a51af0eab4d17204423433582ef9ad822956d1dcd04604f56c4e
6969bed49c4cb63931961d09c00051c8498fb14a69ee546b8c2386b39124c182f7d8c8b63d07c592e1b5c35b2f1fa401681d892389e308d8e949080d
0bc79008ab24825592c4419697ad9b4c5efb0691fc2958f402d347045b508eb4b48daa718c1b551334ee895135f6714f1025a8d86ac3a5e653a7acb6
0c3f352d2a351a8bf48d46fa86d54458acc119449b27e53171bc345ef3acf8acf44c7845a886bc6da81846cb2e621e3c23cf0f2b099f17b10cca4397
852d0b5f16b11ff6875b7320c74d4ca4f586f407595aaf7857acac497b90a5a6880ebbaca13d6aa5f056db6812636ade43af943f75f1d7cf5e9a7083
d9b39e08e52d070f1e5cc036f49db565c482ea4183cf3f9072e3ed5fed2d8ee45f13f7db49df25c47d1728f6f40047a0be5c175dee0cf43a4c5edd46
39c2ebdce8da20af71ec4e0c8a0804b48746c43b2d11688fd6c9898a1082c675f2af53f92701d0520a56d75273534b53b3e58b9b16f526a924338f2e
3f2a2f3acf991f23420e8b620ebb18131b9f9016458cf426ae92589a2f731f7b3860c36efe3ebff1e4d9a7c7bd3bebe4d9babdaf1dddbc73f70b634f
ce2d3937f10b665c87eee886f51f7fe7769f7e20a5ba6ac5e67d0b8a4b4ae3e28f389d7f3ebce855c5f3e49396f7904d09e401967a2299094d80681a
0468d07825864b75cca88708592b1a55ff6320c64c2a634685b14b99b434ad8a5e9b94754abca88a15cf9172cf292aed6aa0adc3709848817c013c07
9a209604f12c097bb331ec61e3c3a6f1ac90cd67cfe24a662255ea580ca65a531d2e5ae23169287381f1347ef9f2b9f6272577db35bcd096ba9f7b59
ee6975255f13f389f24878d2e312c334d6724b64985763f75a2a4d8217969ad668f6440547303d4680de224759dad8bd7ab1dce3432dca6a2115591a
6e2a0b5859c1a41edee0d34e20d99755913938ec709f5a146d7c8ca1edde6e13bab5b2387e897ffbe4e969934ecdf8ed7beffdf69197c749970ff2e7
0302f8cdaffecebf773a1b1f483eba7dfbd1b878927615515fadfa933898e0890b94c1546e046f90ec8d08da6bf11a2b633744ac711b637511a15181
1118131dee26074346d4a4ba98a6b6a6bbe6e3b1d35e875d102ee005b1516a9489efc351420ecb61b1b2c31ee4a395397a3057ac809d8cb89c8a3b8a
490912f6ac7ae9a555044c37fac5d1ef5e0ce87778c6a74ce2b73ee3edfc26cb66e1a35fc47ef5bb5e3e7efce55df5c2c2dab878fe1dfff6f11cfeed
d75ff0af54073599ed8d523cd47eb2a669a41319a6784224ab80025a45f21712e903256422035963693bdfa046959ef7f801023217454113dea4cda5
473d926a4847d6f43e133db6090293314cca90864b53b1066a640d590b2986b958cc7e3cd5fed945c6db53a5cbe35b974a49caae7535c977b52a5f17
9d26067bdc2124dd04d91bd5dd6bdb10b5266177728831ae6b84232e224047de9b5c78404c78b2a5ada1b9a5a159156ce75a554b19b448ef11a6bb07
f99ab8d49420c5c9a8cbd5151b97d6ab776027025986b07afddebdebd7efdbcbf72edb001dff75956f58fafc6e7efbf66d7e7bcff00dcb976ddcb86c
f906e19d6d1515db5e2cafd836de7978c91befbfffc692c3ced833551fdeb8f161d51996376fd9b27904ea1b0960a74517f1a681088f595e2eee83e5
022d511142b496365a7d8a776d51c3974b5943b72ed2c5b9e8e2d4fb43962b7d84bbfcbd4db05d90b78b5a8989e0523a375067cb25254c28d6aedcb4
f96de50ae0aefdb70ed208e7e8a45443bad55254efe171c066dd52b6d9a2152c7a90424d2910a1136d6a0c24c9a9da5568399c1bc8c80eada93e6b73
c7a86922631b5b581a8be69ff2463e88bdc40eb36a3e8d67f33ca9e7cf0b5808ebc1bab1e07d7c0b5fc27fc3abc9186876b69166a7b3e631f89dc0e8
2c252a8ed6a2ac5f50a39d3225b95253b2e491b2a55ca94a7a4992730253adae73efbd275d6e4d52fc5d49c735e9c37be2f8661b6cf6c5f1d080540c
75589438de766f1c57084e552378828f09f5890957783bc32b5718e31d57585ff66bbe8a9fe1ef283b0e6934afe59ff32f782d1bcec258381bbe873f
c177f297f8136c0fad155a2d444929596877da39ebc10d2728c2461b82756638102cd799adcef2e8fa883a57ad754db0118231c4a4d31aa2516bcf8a
27eace5f6a4e49f12da086a6963622f38c6aafd60cc517cc4e8e4c8e4a8e4e7626c724c70e48f0447aa23cd11ea727c6139b1d991d951d9dedcc8ec9
8ecd4e284e5819591155115de1ac885919bb3ec19b702b21aab36b67a7ce0eb951b9d1b9cedc98e2a8e2e8626771cc92a825d14b9c4b6242eef532fd
59bad595662637134f6b2135e6de7815249cbc7a6869d1d6badada0127561d6a6cff9909af6cc93d3aaee0e4a47fdc12520b4b27977c78247174fbd2
8385796fed7af394ad6c758f1e071312da14add593acf6c876d25a04f4f184629d31405717e25813501bbe25146cb6612146591b3694f61f249a16d5
129a94757ce666f2d1dca82551de28243a55727ca432d50dd2de8c684d50960b7efecaf3cfbfa240fbbabeaf979e878e8ef3a5aff7adab137a365ebf
de48203c9a9fc74ff01fe93e9197bf9fa86130a7e31a5e271d86c2004f3894b355a2b9dcb44a5f6715eb824979611a9b0986dbb3c22c6db443f7db14
6fb969f9fe66b2c710106e095f12be3edc1b2e1171aa7bf65397ee5084e8f7cf787dcc8eec37ce9c79237bc79887f6e6b4f30f5877263fb64b4c3b94
9474edc2856b494907e3e2882133b3b1be2e921651254e22fa2c3e6985d581d95e2769d7986bd9160c16412b0cb3da0c5991ea6a4d49b923ad86fba4
459b6b9f3205750104b17b7c1beeaaadedfbfaa2c60ee8685cf47afb5992dbfefd243b3c2a3cf953f3fefc3c368469e91e92c71d7ef1f9e92a2369d9
e9745fec890307d3956b57498e034caa33b2e32175b65ae39a887087a075686194600bc88a50496c50778e8af07c41afc517f512074416477a23df8f
bc15290d80016c8030c031205ceaa6e9a9eda9eba62f8222562414398ac27539731401c7a8615b95add3bfa835aad0356259db61e385634f9f9d3ce5
fd19bc859f65896d9f314dadb077d5b63ab3f0e4a493677bf57aad6b37d687e959201bcc3f6ed872e4b59d8a65cee1e3c549c49381768a233dae5063
a4ce561e1854178075f1aeda8413baba8037c322e343416b1c26db6cceac4435c6f8c4ded0e4133cbfac709441d2efbaa4abb7eb2f6c35d822dcdd75
f4677e95d0d92228382d1577edddbc69efde4d9bf7d672de9a77e89147763efafb23198717fdb1aded8f8b0e67d40afddfbd72e5ddb357ae7ccd3fe3
5f4646bdd1adeb9b7f7862ca6472594a3cee3b79ca41858fb7c8bf2e249b51e24312459793e2eb7082a28b5684a177a24b131d8e3c068bcea3cbd6e5
ea8a7564b78a7b5582cd5bb57489b93f7b65fb97ca4ea09e7602db693c3df9b644ad5592258d559625753b20a86708ad24209e802d3a99a2904cd602
430dcafabdd4a0ee27339b9b823b7793df746e1194fcc4d8c316230515cf53a2a0d706090942a294a41d2f140a53b525c20269995029add36e14aaa5
2ddadd824d27e964c1807a4d174c10bb48497292c6639c86b9c64a5c29564a6be52acd36dca23988af484735ef683ed0dcc65b785bbc2586e5cc0185
3d96aaa35d86d5555f27b8bf6e7f4d9871abfd6c9d6c6f9bceaeb5b7b41f125ced1f13bf77e5177b0cb6080a37149f7c31c963b2f84352b1744b927d
422381c9f69f9afdb2d244920dc5c2244fbc6cd38504801ca971182b229d581b7e22d4a2016b80562b675bb501d91121e4e85c8aa0dadada9a7d67aa
cccca6167523ae1c913d81c971d971c571ebe3bc74ff21ee6a5c479c8e2c4af5c20ed296cfb4eecda43ad4463131ebd4b2df9dac9b3bbf6a5fdddc05
6bf7d5d50da859f8ecab58b9e899ef3f6bff95b0f3e5ed27f7b457083b77bdf887dded1562ee6b53272fead4763e711008bdc93beb00cd4cae305b6b
8d27f44cd0c218257e0db52b44abce3953b17a95d823b98e3f3904c5dbfc0b72f26b172dda7ca8ae6ed01bf3df3a23ec51087869a742004d5c90ffad
dfa3cc57575f30adbe40b9ce0675c65a259adb021e419b23eb17a7728f6b40682994ca659a326d99ae4c5f6628359699cacc65016596326ba9cd1b7a
2bd47affbef9bec37bc9a643af6ede78e8d0c65bccc66fdefa3bff9659f1eaf573e7aedf78f7ec97dbf9bbbc997f43ee2383bc849df5512217f9873d
44a1e28b1ff48477fae25af31af6269e88243f3c4cf5c8f7c42e3a8074ba638fcee78f3f8912598efb8e68fc81ebbe8056525777376e097d3aa3d9fe
f6d764fdc17b2217fbbad321fbf48623893a2bd0214236909d19b0c25cab3ba1d1cb5ad00eb529ce4a5df9e4852f9d57dcee91ecc09702158df9e2d5
5d7505e3c8e811ddb6bf4274d4af0cec1181476cd6c693ed87495985532489662ba2687996664b80eb9e4c9351301bc64647697582463f363a3a6a90
de10152d3a288a568af6724765881245dd1445bb44e90dd1e11a78345c6bd668edb1595d14aa2e3537293e3323a333ac7eaf84555ba7cf307f43a76d
8dfa248701098ac39815a18f3044187b5070e866e866eca7eba7ef67e8673438c1c9e2842efa2e86ae813ded3d1d5d83ba4475894e7426c6c42594eb
cb0de5c67293f22e9b0982ac970d6844139a31002d188a61188e1162a42ea167e280c4a712cb129724ae4ff426de4a0ca103ca9cbb513d5a3dd5cbae
7b8f8f3d99728ce84db2c3d563f64faaac9cbc6940c3dedb7f9b747a66e199bc656b0a5ef5bcfac2277f2c3c220e78ad4b9771e33c2362cc5db7566e
3fea729d4c4b9bf8c8a86c7740dce6653b0fa967af7472e3df493b690d52cc374bda003c005676425ba137908cc9c62c36b3b206331b94edbe7f87ee
3bbe67241ffe9d43dda72b475d7b503fe650f67014ffc9fd2d60a57ce5a89237dfbcbcaba242dac9dfae6af7568ed9f6d25f84dc2af6a032f304b2a2
8fc45cd0c10e4f98d677ead3683583b407e0041e90b4c84064b2bef3f581513d7fcbbef3b7fa3e5356cfdf29fe5361f33f1d0b3d8324214888178609
233492411b6008c1706d92d669e88d19da64838779842cf4881e69b0f6719ca87dca90cb728542cc1573a5c9da32c312c3ef0ce1fef3a2f2ae88c5cc
c1a7db470b47da160b47da0bc4dcfd6d1f6ddc8f6ee285013f2c06931f8f83059e5ea1010e83ecd685591c9106c919836038a08303ec94ce7120f075
b751a797e2824221522f050a7670860ed607486ee5159c51614a5298ea143699aa62ad0aa8ef181b147744c64ad5771a93d9618867399222fbfeacf3
a50fe57aabfa9023997ad88c54366262f0833f1d98b9ae7fffaa19077e7a70c8daf14fcc2e9a347eedc9f59b3efe76cbbcaa92cdb73ede583561ed8f
3bd68586afdbfee3da09f7f356e1e91ded0e0f320668c30d8ed000517222841e088703ae53e10107acafbb2342c31c017438d486395c3611c2a21d83
2d601055eef40a77a24f7d6635efe3348558f533f92f79552bfdfc76f2ea6730bd07bbc37190c2e05d8e9388eea16b1e9b545434e9b1354307fcf4ca
cc750f3eb86ee62b3f0d3839beeac7edebc243d7edf871dde3551b3fbeb5b9a46ade966f3fdea4d8644d7b1b666b9ea7c800cc956609b405a6020a56
5672da5bbfcfcbdb2695b5b77d855bd8674232c3f67ff0f2f69b6ddff8fab1124da5f2dd4b89d235a74f6b2a7f285147e4761ab1521dd1a10ce82243
b7a6b192b28d8ddef7d77b35955fb5ede58f733b2f6059eca6108f295ffd9292340b0a2cd51668b33a84b56593789b775fbd57212450b0b1858299b7
b7bfcfa3daa67f45be66addaaf52fd1f0279d964cc564951bfadde4b4baa8306459732a8b0f6b477fdfbdec68d650a3117da3fe2367e9cbdc89ad901
7c42f15fbbdbbf922668668149918ab5573a1dd21c31b11ad98139f5450beb5920bf5e3024a3fd2bcdfcab1b76d6143ebe5579a74374e468ca2184fa
0405078552474b42bc3b212dd599de3b94a55a2d1a59d8d4b76b41fdec51fce7fefd193e34bbbea08bf096b3ff463eb8fd866064c79e1a28b2e35953
d81b82b1fd061ffe42a6420bb7132de52a2d567b70904a4a02791dccc91852c042f937b5a5459af21f376d1b5f58b3b3ea534dbe8f8316a2668da21f
a61cdb72eaebdb5bd4efbefbdbbfa2161f6fd65e09f1ae98b4989460daf108bb693c7e9d05d62f2c22d6166f7d9cc6db70f5c7d5bef1da542a14eee2
13e2557e828362521d1a5123abbca6f7c699995d88bb8718f6efcf7f1e45dc7515de8ac97c81d50a21ed2dfca12959c79938f0293ea2bd4508616f6e
ecaf50431acaf173976ab56b64959cf4de69c2eefaa2d25afe0d0b2549137bdb3fad22498fdfe6fb1fc5c4ad232e9ffaeea980ccef215aab7e8cfdf3
6fcc4d9de9ed0fda469b27eaaef8bfb9fa2eeaa799c52301ccfcf607ad8f9827fed317df64f182fafd130465ab4f6c6334d410ec21d3ea4fe59d5230
54107c4a504db09d209f40a9af22d84fb05a380cb7e423f0a19c08e7845fc13979349488d7a154b243bd580873289d2336fb522103de127741bd926a
1aa15eb211cee76a5b3d8ea47c1214a10bd2c55498200533508068a9b91f84b54a2a83b01bf60bbb7d4fe50de39dab0f4c811a384f27a22cd620f411
360b2df814ee114d6237f121718d58237e2fad96fe2ac7cab3e5bf693235f99a2b9a56ed13da2dda1bba74dd6c7da47e92fe94a187618fe194e182e1
4ba3640c34f6348e30ce37ae37ee30d61a9b8cdc275fe88de32009a68191767416d8aa485f740841948aeab7f049cab73b5147c8c9ea377325cf2088
4abebc005a36d49fc77beac57bf21284b031febc0c76560883a1088a6121cc85e93095669fa7fe1b600a24529a02c974a7526e3261386110e1cc8312
82b9500079300bba51ed08984df83d28371066d2ed8447ef8c55a2960a282da03ecfd0339f30f5ff83597bdf99751ccdf40ccda57cad9d4dd80a1d79
d4e7ff6cc621947b9afa8d87f984318570f3d4d10ad41e792a474e1a65363d8b0967328d3b9df09cd4bf8866cf53db7e39ce58759412a2a888ee1954
abcc5a42b845ea482934772aa4ddd7abb38fefdf4ed0f11bf5bf34ff7cf556ed42206d29ff6f324100edbaece080203a878440288441b8fa4faa04d2
420acd90054361180c8791301ac6c0c3900d8f9024c6c26334ebe390034f422e13d4f3b8c464a6615aa6637a666046666266cdfcd9d353523306f9d3
c1fe74883fcdf2a7437de9c0647f9aee4ffbf8d30c7f3ad09ffac71be81f6fe010805a6189a7e3678ead76fcc98d3fa6e0ed6afcc18cdf736ce1f80f
377e67c6bf57e32d377efbdc40e95b8e37abf19b6a6c6ec5af5bf12b8e5ff6c51b83f03ac72f52f0f3a6b1d2e7d5d844884d63f1da673da56badf859
4ffc94e3271cafa6e07fd9f1e36abcc2f1231bfeafc5f8e171fc1bc70f08fd83c578f9d230e9f262bc340c2ffe255cbac8f12fe1f8678eef73fc13c7
3f72bc508de71ba3a4f31c1ba3f0bd143cc7f1cc4aab742602df09c2068ea739becdf12d8ea738fe81e3498e6f723cc1f138c77a2bd695bba53a8eb5
c78e4bb51c8f1dcd918e1dc7634bc4a3bf774b47733c1d78d423fede8d4738be518d8739beceb186e3ef38be968fbf35e3a157ddd2a17c7cf5a04d7a
d58d076d7880883ed08afb39bec2711fc7bd36dcc371f72eb3b43b057799f1e57cf4128ab71a5fe2b8738791f6d0b8c388db5f0c95b6e7e38bdb2cd2
8ba1b8cd825bf5f802c72dd526690bc76a136ea64e9bab71d346b3b4a90b6e34e3f3adb861fd716903c7f55539d2fae3b87e8958b5ce2d55e5609547
5ce7c6b51cd7acee21ade1b8ba073e476c3e37102b5719a44a3baea2831d5554e4633949aadc8d2badb882e3f265566939c765565cca7109c7328e9e
8edf2c5e2cfd86e3e2c5b8281f4bc739a452373ecb7121c75f9b7181119fd1e37c8ef35ab1a415e7b6e29c562ce658c47136c799313883e3d3d641d2
d363713ac7698b712a150a391670cce73885e3648e797d31b7159f34620ec727384ee23871825e9ad88a13f4f87850a8f4780a8ee7f818cdfcd8201c
e7c0b1cc228d0dc147edf8c8c840e9118ed9067c98e398872cd2188e0f597034c751d4328ae3c811166964208e883449232c38dc84c3380eadc6ac6a
1cc271b0d05d1adc8a838ee3c051e8e13880e383fd6dd28376ec9f1920f5b761663f9394e9e908c07e26eccb3183639f74bbd4a715d37b5ba4743bf6
4e3348bd2d9866c05e51986ac294070c520ac7070c98dcd320259bb0a7017b74d7493d2cd85d87dd5230a9ab5b4acac7ae8936a9ab1b136dd825c12d
751988096e8c771ba4f800741b308ea38b636c00c6109f313674e663742b46110b51f91869c208926004c7f0560c1b84a15408e518928fb417938239
0651a7a0507470b4730ce46823041b472bf16a1d8496c518908f668e26639064e268246c63101a38ea2da8e3a825342d478d1de57c14a951240b7020
d522a7bdaa4512ba23b3207064b52c7fe55a96f4ffc305ffaf09f8df5e91ff0d4021f9ef
>}
\immediate\pdfobj useobjnum \csname EF3O13\endcsname {[ 32 [ 318 ] 40 [ 390 390 500 ] 44 [ 318 ] 48 [ 636 636 636 636 636 636 636 636 636 ] 58 [ 337 ] 61 [ 838 ] 68 [ 770 ] 73 [ 295 ] 79 [ 787 ] 97 [ 613 635 550 635 615 ] 103 [ 635 ] 105 [ 278 ] 107 [ 579 278 974 634 612 635 ] 114 [ 411 521 392 ] 122 [ 525 636 ] 125 [ 636 ] ]}
\immediate\pdfobj useobjnum \csname EF3O14\endcsname stream attr{ /Filter [/ASCIIHexDecode /FlateDecode]}{
789ca58ed70a835010440f187b8b35c61e4bf4ffbfd04544ae7912b230bb531858f873b41ffd40c7906b1edac2c6c1c5c32720e4295eb427f1a5979c
2c55dcecc607392f0ade94c22b412d6868e9e8f930881a999895c69745f6ba01bdc10332
>}
\immediate\pdfobj useobjnum \csname EF3O15\endcsname stream attr{ /Filter [/ASCIIHexDecode /FlateDecode]}{
789c5d933f6f833010c5773e85c7748808909854424855ba30f48f4a3b451dc03e4748c5204306be7d6d9e4da522254fefe7bbf3713ae24bf55ce96e
66f1bb19444d33539d9686a6e16e04b1966e9d8e9294c94eccdeadffa26fc628b6c9f532cdd4575a0d5151b0f8c31e4eb359d8ee490e2d3d448cb1f8
cd48329dbeb1ddd7a506aaefe3f8433de9991da2b26492942df7d28caf4d4f2c5e93f795b4e7ddbcec6dda5fc4e732124b579fa02531489ac6469069
f48da2e2609f9215ca3e65445afe3b4f38d25ab5c5a72e1e72857eaff80cdc78eced63a00812b0c20709e00c25b33370b0092485649023e404e1903c
e4a35c03eb5bc9fce59904961e4be0232a3ab94281d1b9932b1458012b8f153047affc041c2c5ae7689d1f430c52d0b4932b1418f7727f2ff7f7f276
c5b99fbab71cf3e478334e101542d7cc1c4de4fefd82cd0245104696b73e68b33845fddc4f2e7793b3bb1296c2ad8ddbf16d27c5dd18bb8eeb87b0ee
a1dbc04ed3f6ad8cc3e8b2dcef172fd5deeb
>}
\immediate\pdfobj useobjnum \csname EF3O16\endcsname {<<  >>}
\immediate\pdfobj useobjnum \csname EF3O17\endcsname {<< /A1 << /Type /ExtGState /CA 0 /ca 1 >> /A2 << /Type /ExtGState /CA 1 /ca 1 >> >>}
\immediate\pdfobj useobjnum \csname EF3O18\endcsname {<<  >>}
\immediate\pdfobj useobjnum \csname EF3O19\endcsname {<<  >>}
\immediate\pdfobj useobjnum \csname EF3O20\endcsname stream attr{/Type /XObject /Subtype /Form /FormType 1 /BBox [0 0 306 154.8] /Resources << /Font \csname EF3O1\endcsname\space 0 R /XObject \csname EF3O16\endcsname\space 0 R /ExtGState \csname EF3O17\endcsname\space 0 R /Pattern \csname EF3O18\endcsname\space 0 R /Shading \csname EF3O19\endcsname\space 0 R /ProcSet [ /PDF /Text /ImageB /ImageC /ImageI ] >> /Filter [/ASCIIHexDecode /FlateDecode]}{
78dacd58df4f1c3710dee7fd2bf69154899919db63bb521fa069a3e6a5a241ea43938794902b14d2025123b5eaffdecfbedd5b9b1c77802ae53856b7
37f68ebff93db357030f840f0fcf28ff9f5cf6fbcf4fff3a3b39fde9c5e1f0edabfad7c94dcfc339ae051e38c7f5098fbdc0b5e8338bcbde92e2fba2
7cb37726e29e5677bff5fd7bfc5ce081735c9fb0f505ae457fd5670ccf32086b8848dcc0361ab6249e6c8674789c4f301a831f8edff5fbdff3108c1d
8edff77bddc193e1f8bc5763835020cbe2c1ac6c9279d3376553302e068ea21c64b969affba7ac38a32c1a28a6e8a615be73e5df3b579e8e6028303b
eb3ca5094c85f870421c552c398df17f422c5b107f77dc1fadd5364563c947f57eadb62b583f144602e3448ed6271ba7233e8c87fbe452289698564e
ca8a35292a450b542b58aff7467651a24b2cbc565ff7b2f05765933764830bde477fb7e69d49c106ef224e5dc3e9f5937197384aaa21699af03ec628
ee914609c1888d21c5dd350914e9f17141541fa6c8bbacf565a4f9126e611fe9160e3025910b5bfde2f9a8630f488e9558a713ce46a5103b89222beb
ed7537a3264308220c21c3b4f276c4eb9c049748ec0aef629284843c85c82be9afc77352443e94986875cee9c44dc9b9241ce3f695cb51cb84ac1902
9410b63f73b7cf7c9c1cc3897ab63adbeceb710592aa68f2f333c3b882a4edd5bb0ac163ecfc74d23507724cecf5412963ff0085f766a8ab318dd5b8
e7608d0b1a513d251848e05229cf30454a4167e2454564cf2694bb8b9a414d2e25fc6a58ad2213af58b124e34592872bb2204c344a18ae4f879f870f
70516c214e0affe80981a2f022544718cad8427560325c2f9a9dc3a69d59ca64ac3ae15a488e3edb87707625654565b10861ebb8d0572c1a7291f368
789ca434c1777d03998c77cc415d401c1449e79dc3a69d903459635d030292a6642412696c249da9291845d4d92579e25053774d4ed4776359253526
158669bca24ba905ada8e8140d5212c4cbf4994743df3959858d24b1561a5925977dc93f6a5967aa45da4eea5c28e4158b9abc73925a9f4fb1a1c946
e2009e22b782ae8801850be5c917eacca026ef9c9c0ea1e79ce3564eef0da3ca586e04ada8d18087d558e49f5954e49d9354117b822ae81a493519f4
19125a49672ac2d1064e54249d59d4e49d93146d3876bad4a45e89160daa621469246da8685e63c9470d8789ba73722636e85bbdd5464e140b4665f4
b69173a65a672c5a487285bc6251935792ee1f486e58a637006462f50e007db3b1a8be1cdb623e531d8685e5f9fd21dae24f7573cc116d43f4d06264
680cd982634a63839c9b62352137c5b4eca37e19f63a7a32bc198e5f2e7ba982051585503ab5cdc72bdaa6f391b260f9105420f53d01b835006ccc33
9fc4b6fc55d48d20ac9ad22cda704f08710d0455782bc18bdbc89ea91b21281bb4169883d0d4de1304e7d7169fc1480e7d60ce84ad3bced48d306242
5fefbd2638c57d51688da2b16d1e7af254018970ebad555ff14265bec5eb7977d6dd746fbb451e8050a2036198cd7f098ba7f85ce2fad07dec066c3b
ebfecef3cb6d05acefe019b934dda2ae8d07c5008da65d00199a02724c8e0f8c860a81c798843643b9853093d762f016a327866160f00822ece33c6f
ea4620bea32d6002a2c15b343d2d9899bc164caeabf0cc9c65730b81e11d01e55193363b056d85c36431a7617e732d9e8abe0d1013da541f22264a40
da8268bb8230690384cb135b8b68a66f45248a1e3a22990776762322f95c47577d799f9c39e36994bcd10910b27277e4146e3fe6906134f0d959ca1f
88ef102eb723294ffd5818ba5fbb8bee8feea4fb1df77f22e8cebaebfc16e34d3fa191e1e572402ef5e6d6ebecf551b67ef8ed5f3d64aecebb71f2b2
d229548d5202ed594cf47930f126c1d6680590d5c8e73862975f7b334c449464e0e40d7986ea608d00ff70798c41cb886496df62e57935d918a0a193
7c18153117e38bf5e6a57a9e94021a0c4a49d138b343b1b4281209e1b82d2b467c02543b6533782ad44f1e287999cd6acb1ff5fdd17fda8cd7a8
>}
\expandafter\gdef\csname EFWidth3\endcsname{306}
\expandafter\gdef\csname EFHeight3\endcsname{154.8}
\pdfobj reserveobjnum
\expandafter\xdef\csname EF4O1\endcsname{\the\pdflastobj}
\pdfobj reserveobjnum
\expandafter\xdef\csname EF4O2\endcsname{\the\pdflastobj}
\pdfobj reserveobjnum
\expandafter\xdef\csname EF4O3\endcsname{\the\pdflastobj}
\pdfobj reserveobjnum
\expandafter\xdef\csname EF4O4\endcsname{\the\pdflastobj}
\pdfobj reserveobjnum
\expandafter\xdef\csname EF4O5\endcsname{\the\pdflastobj}
\pdfobj reserveobjnum
\expandafter\xdef\csname EF4O6\endcsname{\the\pdflastobj}
\pdfobj reserveobjnum
\expandafter\xdef\csname EF4O7\endcsname{\the\pdflastobj}
\pdfobj reserveobjnum
\expandafter\xdef\csname EF4O8\endcsname{\the\pdflastobj}
\pdfobj reserveobjnum
\expandafter\xdef\csname EF4O9\endcsname{\the\pdflastobj}
\pdfobj reserveobjnum
\expandafter\xdef\csname EF4O10\endcsname{\the\pdflastobj}
\pdfobj reserveobjnum
\expandafter\xdef\csname EF4O11\endcsname{\the\pdflastobj}
\pdfobj reserveobjnum
\expandafter\xdef\csname EF4O12\endcsname{\the\pdflastobj}
\pdfobj reserveobjnum
\expandafter\xdef\csname EF4O13\endcsname{\the\pdflastobj}
\pdfobj reserveobjnum
\expandafter\xdef\csname EF4O14\endcsname{\the\pdflastobj}
\pdfobj reserveobjnum
\expandafter\xdef\csname EF4O15\endcsname{\the\pdflastobj}
\pdfobj reserveobjnum
\expandafter\xdef\csname EF4O16\endcsname{\the\pdflastobj}
\pdfobj reserveobjnum
\expandafter\xdef\csname EF4O17\endcsname{\the\pdflastobj}
\pdfobj reserveobjnum
\expandafter\xdef\csname EF4O18\endcsname{\the\pdflastobj}
\pdfobj reserveobjnum
\expandafter\xdef\csname EF4O19\endcsname{\the\pdflastobj}
\pdfobj reserveobjnum
\expandafter\xdef\csname EF4O20\endcsname{\the\pdflastobj}
\pdfobj reserveobjnum
\expandafter\xdef\csname EF4O21\endcsname{\the\pdflastobj}
\pdfobj reserveobjnum
\expandafter\xdef\csname EF4O22\endcsname{\the\pdflastobj}
\immediate\pdfobj useobjnum \csname EF4O1\endcsname {<< /F2 \csname EF4O2\endcsname\space 0 R /F1 \csname EF4O9\endcsname\space 0 R >>}
\immediate\pdfobj useobjnum \csname EF4O2\endcsname {<< /Type /Font /Subtype /Type0 /BaseFont /GCWXDV+DejaVuSans-Oblique /Encoding /Identity-H /DescendantFonts [ \csname EF4O3\endcsname\space 0 R ] /ToUnicode \csname EF4O8\endcsname\space 0 R >>}
\immediate\pdfobj useobjnum \csname EF4O3\endcsname {<< /Type /Font /Subtype /CIDFontType2 /BaseFont /GCWXDV+DejaVuSans-Oblique /CIDSystemInfo << /Registry <41646f6265> /Ordering <4964656e74697479> /Supplement 0 >> /FontDescriptor \csname EF4O4\endcsname\space 0 R /W \csname EF4O6\endcsname\space 0 R /CIDToGIDMap \csname EF4O7\endcsname\space 0 R >>}
\immediate\pdfobj useobjnum \csname EF4O4\endcsname {<< /Type /FontDescriptor /FontName /GCWXDV+DejaVuSans-Oblique /Flags 96 /FontBBox [ -1016 -351 1660 1068 ] /Ascent 929 /Descent -236 /CapHeight 0 /XHeight 0 /ItalicAngle 0 /StemV 0 /FontFile2 \csname EF4O5\endcsname\space 0 R /MaxWidth 635 >>}
\immediate\pdfobj useobjnum \csname EF4O5\endcsname stream attr{/Length1 3880 /Filter [/ASCIIHexDecode /FlateDecode]}{
789cb5567b7094d5153ff73bdfddddec6e36bbcb920779add96c480a49304fb230b2405e90400259c206086493cdfbfd8090c4402a612ca2f29c5522
168a54912a4dd1b6a1b10e3e5a74943f5a651c676cd511c67606ad9d41ad19b9dbf37d59283275a6fed17bf7deeffcce3daf7bf67cf77ec000c04293
0ce692c2a26298035a001647dcc892ca8aaac21de5db09e710b697546d5831fff4b2d7087b08072aaa32b3dadc5d0abe4eb8baa1d3d7038fc90100a9
90f058c38e01fb9327265308bf4e321d4d3dcd9d9fb6fd93687911adef6ff6f5f79037f227bf44d8d8dc31d4f4afe2cf1e21fc0600b6b534fafcdaea
57ad003ac57f5e0b318ccff23f12ee209cdcd239b033ea13e6207c90705c4777838fc54acb093f47d8dce9dbd9831d9a30c28a7d7b97afb331236f65
1fe10f2826734f77ff40f05d25017a25fe453d7d8d3d870f49490006dabf6c0465c908b34d2284a05179aa0ac44102b0c2e2720fe8287bd482c190ec
ecfa61380a36757dd8d7e7ab87715f5f67178cd7f7f95a61bcc1d7d54f734b631fcd437d1d30dedcd84d74735f633b8cb7f8ba48a6a5b19e38edbe2e
1f8c77f8baedca3c40163a7d032d30ded5ae70ba9b7d9d30deb7bd8b24079aba9a696e51ecdf1551282ed9c40e0207e0d9fc18e4b104e5192cc2f7a0
49a22c4b060d22ca06490eede176ab6c2af213a71dfa3536616313da4ef6c977643034e2d49d039413622a9641f9af75d4257ab6432ff4ab5185a8e0
b1e0d1e0913b2ccde63be70ebbeda1a144d51b1ab407e80f0d0d05e2609930056fd17805cec225e9129c847ceaebe165b64f4aa795a7e141789f4bf0
221c6705ccc60a68f51d8d4d33c4f7f233b4be85744bc9cafb70846c2a96a660b734225542135ce29761827ab7caff027ec7f6c015781cde924ae12b
d8831e38407d02fa65e057981e84b4004e2b9ea803d4d388c7747e45ed5fc06e18010f9cd64c696cec1d35eaa7d96bf406fd9d627e07b7602f691c87
336cafec90cfc8a57060365eac8303d21e3621d7a97d84feaf41382ed7b1b31a1bfc4189953895146913bc4463102eb3256c2feea3c8469408f815b8
ac5d2d67ce46a51dc55cda0fd03807e7211d03a4afee45d304c7a526f2f5154572190b218dac519e6df4c6c3dcdf68b88c12838576f3a4e45ce59f74
aff3dadfa8b9277de15dd06ed6da27a172327cc83e150c567ae5585e33c9e326d1a99b949d8e8fbf6ff1e3f48565955efbe40b458521ab457585c4ab
f212a9206213bfa8307df6ada4baa27a52aaad944ae41ae5420f1fbb4b355cabd3722e739d4e8b693297a53426110ee3a0d7d11ae8258e921ef48c4a
49efd2710971176c09d384e9f41a2e3199cacda94f36986f5c7ddb1255602d28b81732975ebf1a55b08842d19af9675ab3ee332dfd4294f93fcf9aa4
f33a99b15af77d5c8a94f2a42ac9c33dba16deaa1b938e487a23ea6497c6a5751bbde8e51bb51b75adbc43d3a61dc221cd9076aff428eed33ca49d40
5b2dd4b2b249abc7fb7b30042fd20e0dd4dd8b2df98b6b5836b2ec302a7b74b0c1a94969f01f37dd92f4a1486aa3aafa763b7befe6df6e4e498e9b7f
a10c6d159be44b542f7320099e7797c7b1f8586eb544c74851f3b8c56a919884408c39dc6ab59889083749c60865c51501f151e5185f6e1cc675d6dd
e65d8e3936b3056c11f3ac4e8bcda44972986fa8f15da077d0bdb8e6fabbe7cd16567bb5a0c0a2f67bc17cf5ead51bd7cd9f5b0a6ea78d92a9a68c2b
9953c06d1ecd94b9ba64cadc8b2793c792a55aa7893992527273acf9f7b1ecacc8a8fc6cdab12d323bcb9a9b93e248d2684d0cab57bf397cf695c1ee
d2e9e6fb4f26db2f882b1784ebfeede7beda3db8f168be6be5199fffda9ba758c6ceb59efa8633df7c8beb8f3ec1d67c7ef0d8bab5eb3fa09c46053f
a1ff7a0452e1b7ee65a67029c2c8ed89ba303d6a0d3cd19ee84a488c311813ed329b5b165911536e8532562197a70c5bd7a52518628c768cd3328833
4568b923cdfc6ed9a4dea3a46473d9a4c9b39912e3085e541273fde2458b35aaa040cdc98d6f2927e2cbcf8961b18612635293609a4d0273a7dac1ce
32e76646a64579c3bc7aafc16bf42678130d6990462757015b0dab590dd4309d5223da0c4a9366ae2d2a329125b0b9369992363f3f815296979b9329
65b0dc9c644a9f0c8e55a7aa0f3df36a8476c9d488b876ade6f51dc98daf6f7af0d9bfce59f6f0032cfcd3cdcfcbcbcea5a6366c68c895256b768fef
fc4749c92fe7e40dd4d42de166e7d0d617fe9c70eb7e936a8e3d955af9f6b688a55f42a24e3dacffb4cbf4e1ade7d71fddcc0e2fd6ed57dfd0dbf70f
9d3b9d221e207cefd71f7d7321bc18fc54957736adac9c5f8af9b334f6d31b7d0a4ab90db6ca0b20ea0eb987592cfbb51a8716d7c20268a17b430233
b895af093449e1a13b590b9b94b341a66f00b648bd4f149a4124a1595a02132b0ed10829f465334bcb77d01ca2d97088d640323b0c2be9dcec8121e8
83566826ef0360a70a6aa0b3d20e5974e32d826ca2ea49c20e2b4866804ed001926e041f74c242e2ae822e92cf206a397450b7d35d75cb56bf8a1ae9
d9483a3b68f693a4fe7ff09a77dbab873ced205f6da4d345d24a1c3ed2f9611e0b896a23bd6ad84e120d24eb53ad35aa1a3e754776b2d245730fc9d4
93dd5692b3937e3779f7a96b77dba952adf4434548be97b88ddf2363bf4baa5a8db09f70b7ea358be2cc86dcef68dfd24dbf4b97be40825fd2d84555
f1df9a56ad5189fee9b594bfea2969cc1dfc95c04927fe320bcf05f079133ed76fe2cf65e12f049e75e2b3263ce3c46702f8f40cfe7c064f0b7cca85
a704fe2c0b4f9ea8e227037862cd727ea20a7f9a854fdaf078009fd0e384c063567c7c141f9bc680c0a3247174148f083c7ca8841f1ec5432578f040
2c3f28f0402c3e2af011810f0bdc2ff0a17d09fc2181fb12f02759f8a0c0f148dc23f001813f16382670b7c05d0247cb9c7cd48ff70b1cb1e0f0d034
1f1638b4b3960f4de3d098bc73d0c977d6e24eb73ce8c41d02b70770c08ffd26eceb75f23e3ff6f65879af137bacd84d6175cf60973b28b0536087c0
f6486c6b75f1363fb6928f5617b6ac35f096686c6e32f1e62c6c3261a31ffda4e60f6083c07a9f91d70bf419b16e5b0caff3e3b6ad66be2d06b79ab1
568f5b3687f32d023787e326d2d814c01aaf89d7a4a2d7841b67b07ac334af16b8c153cb374ce38631d953e5e49e5af4b8e52a27ae17b8ae3283af13
589981151444850dd71a700d45b5663996d3a35c60d96a0b2f73e26a0bae12585a62e1a5024b2c582cb04860a1c0952b46f94a812b4671b940f70c2e
9bc1fb667069de0abe54e09237d04594ab0a0b84bb07178f623ec13c399de7adc05c813902b35d9835838b8c9829305de042810b6879c1bdf82333a6
a199a739303501e7a798f87c3fa698d0c9f4dc9985c9c6689e3c8a0eeee20e81498492a6f11e92bf2716ed89066e8f40fa22bae89e90130d98108609
6e39de8c71241e17c0d800ce8b71f2797e8c89b6f21827465b312ad2c9a39663a413e70ab4099c3383564b0cb70ab490554b0c9a0546083491055300
c3c961f8281a0d466e8c468311f50275b4a40ba086c4350239ed82bb502624a7239ae95b4fcfa568647a646e19e2904d31ffde47d882ff6f83ffb3fd
1fdae2ff0d51184750
>}
\immediate\pdfobj useobjnum \csname EF4O6\endcsname {[ 107 [ 579 ] 113 [ 635 ] 115 [ 521 ] ]}
\immediate\pdfobj useobjnum \csname EF4O7\endcsname stream attr{ /Filter [/ASCIIHexDecode /FlateDecode]}{
789c63601896800589cd0ac46c00014b0010
>}
\immediate\pdfobj useobjnum \csname EF4O8\endcsname stream attr{ /Filter [/ASCIIHexDecode /FlateDecode]}{
789c5d503d6bc33010ddf52b6e4c87a03886643186922e1efa419c4ea1832c9d8ca096842c0ffef73d49a90b3d901ef7eede7df14bf7d25913817f04
277b8ca08d550167b7048930e0682cab8ea08c8c0f2fff72129e7112f7eb1c71eaac76ac69805f2938c7b0c2ee59b9019f1800f0f7a030183bc2eef3
d217aa5fbcffc6096d84036b5b50a8a9dcabf06f6242e059bcef14c54d5cf724fbcbb8ad1ee198fdaa8c249dc2d90b8941d8115973206ba1d1642d43
abfec5eba21af4967e1a28bdc0bde057a2cf55a613dc0b16ba2e74fda06ba2a9cb6fbdd4305d67db462e21d022f984798334bbb1b85dd93b9f54e9fd
004d4a7d84
>}
\immediate\pdfobj useobjnum \csname EF4O9\endcsname {<< /Type /Font /Subtype /Type0 /BaseFont /BMQQDV+DejaVuSans /Encoding /Identity-H /DescendantFonts [ \csname EF4O10\endcsname\space 0 R ] /ToUnicode \csname EF4O15\endcsname\space 0 R >>}
\immediate\pdfobj useobjnum \csname EF4O10\endcsname {<< /Type /Font /Subtype /CIDFontType2 /BaseFont /BMQQDV+DejaVuSans /CIDSystemInfo << /Registry <41646f6265> /Ordering <4964656e74697479> /Supplement 0 >> /FontDescriptor \csname EF4O11\endcsname\space 0 R /W \csname EF4O13\endcsname\space 0 R /CIDToGIDMap \csname EF4O14\endcsname\space 0 R >>}
\immediate\pdfobj useobjnum \csname EF4O11\endcsname {<< /Type /FontDescriptor /FontName /BMQQDV+DejaVuSans /Flags 32 /FontBBox [ -1021 -463 1794 1233 ] /Ascent 929 /Descent -236 /CapHeight 0 /XHeight 0 /ItalicAngle 0 /StemV 0 /FontFile2 \csname EF4O12\endcsname\space 0 R /MaxWidth 989 >>}
\immediate\pdfobj useobjnum \csname EF4O12\endcsname stream attr{/Length1 15368 /Filter [/ASCIIHexDecode /FlateDecode]}{
789cd57b69405447b670d53d77e9bebdd0dd74b336d0d02ce206015151a21da2b8c6e01a3563028a4b340a8a3b1a5c46d0a82346c1c418ed1874dc83
c631a0c468241aa3ce24a3ce1b5fcccba2c62cc43879664328be53b71b4493796fdefbf37ddfbdd4adaa7bab4e9d3afbb97d21941062c38b485cfdfa
f4cd2243c93c426847bcebec97fdf8f0f0e5b19db1df1f07fcb5dff091991579bb9d84c858c89dc71e19d1df5d99339d10c58263663f3e3c29e5995f
f29713a2fb3d3e1f35617a6e01b4332762ff2a3effdd84b9b35de499887442d4eed867930a264f9fd965ee54428cd8277b27e7161610054f62bc8e7d
e3e467174ccad831c18bcddb84c4364e99989ba71bf74e09210fd5e3f3ae53f086699bb297909458ecc74e993e7bbe7b680cc33ee24b1a9ecd9f906b
ce0cce2024b50ff6974ccf9d5f20ee936762bf0efbae19b9d327267cfb702df6713d7aa920bf70f6e03357ff41485a353effb260d6c4829eca3ff852
88833485705a1989ef10f004928ef77a905ed8e6cf54d2996410a14fd6e011c4fc6ceeec192404e98a47733321ad2d3e924e9b386b06d1f196564484
c06b1d11a09a8f844fc576c446c27c8b35bfe95f35bdb94ff37ef2eb23fe37ee91e6ab2dd7d6f9ff4f1ccdb96dda1cbb3e5aab4cbb4ed3ae6ffaebb2
7ba35a67bcc9bef0cd78e0fed507eea5b77dfa2f6276f157772a7d6bb63c6dbe8ab85dd5b06e596f7ff349bc6ef6f7a661ffe2bd996da1b4a9fb3447
e175e53d8cd97e86cf986fb593fcca2ad94586bc6667d97eff3a2b1fc06db1765d790ff3e6956df1fd9f1fdabecafed5d96da9f0bf387e5366ff8bd5
389517ff7abde68bc88f8bf7a49df75b9f9dd5ae65ff6b1cdbae73b58533f7f0697350d4569904123b719060d4fb30124e9c24824412378925ed4907
b40d292415c7059268bc72ebc16d828457196d9e8ee8d17eb41c86563bc30f13316b7580bf6f215684202344fa9ba8725b02adb065bc2a5addb20687
ce6106e0c9610520ee8184db368e7d90867f28217e9bb695ec2376cda62dcc9d953b9e94e5ce9a3e83948d9f95fb0c299b903ba310af5326cec2eb82
59cf92b2c913f3b13d79d6c469a46c4aee0c1c3365e278bc332d77462e297b3637dfc5af681b7f3f3d77f6145236631abf933f39773a299b3567068e
9c3d69c664bc4ee1f0ff89fdd428feec339373efb3a122f1d9504aba6ab5a4edc84962483bb4d1807c68471ed6ea046ea7b1d713af0968c105bc47f1
8cf5eff953f225c9261946a27cac1174192e9f833ee0b26f69de6e7bb4f4e9147f7be46ff2e57e81419e0a73fffb7164ac6f2caf5bdb6d60dc777f6c
9bfeac7b6d01fdebdd75d8ce689dda0e391fa2f9ad489245666812c925866a922cfa47015212c7eae7e0350de7b4d3afc5e75db045f41584b41997de
3aae7bebb86ebf1a27206c1c277109b5f1a712c7808a665ac6a5554a955e42cc237d35fc1b9924e078c12003e84441f0f1b7cd913da96f1ef1101799
2fdb999d6e56a6d36bf78d017fe1f10a97fc4aec51ad2f126e3723910a809ae02289a413cac1c3488bc790f3c3c808f2047986cc21f335d97321ae9d
48323ecf2483c9e3fee7b9e45932afb9b9f91a5a9d2bcdffd67cb9f9c3e6f3cde79adf6f3edb7ca6f978736df3b17fa2a1f70e9f9e57fb7b666d2d5f
018d4b0431f3d1af93bf707d4ec6d28370bd268893af70ca6662c9c2c2b5793096c7b08460799c130b0b8f2786f90ba7c9087f89c0f284bf4462e11e
fa192ce8a5708f84ccf1e382b121998f2585a621cee7f03c49f6902d7427f626e1fd9978c72b1c222b7046353945cfd1554227bcb793dc2617716429
39077b444207a2253c87e3af20ffefd011e430c248a7769aaec8a8c643c4c3e230b15abc295e20ddc442f182982316d254d82e8d927662498777512e
ce2276d5f41352488ec257900ab5621fd14c3e810bb087dcc05538cdce9175c8f322c4c54ef349b150240cc33b67a40b64339ef9f8fc02dd4a2f2276
47e9727299bc08a2d09f6ca597715fe7c88f64398c108a9115a9c224c4ff0cc2ba80f33793425495cb54254ce880f7107b5c6bbc768d804ed265edbc
4d8a71e511a452ae96ed8a1b57e114db494fd17a7903f1928bf03b98091fd115a25bdc25f627eb7c14801cb20e616fe673e4497401ee9d9f451cba30
4fcca17bc857628e321e61bfcb77846b1e1686e18e26915a2cf3640beea9275d01ab10d3228db71794816212ce4708ca62dc3521f99046a662ab881c
20874827a820eb1092b65fb99bf423cedc227e867b5e47d70a3f920bd007a57092780b698d0695a046bfa9c8920802251d5d962a216e405e9567e868
d77b63a23b757ca0ebb228ae2a925d655ae0aa6e6ece1e2d864b63aa246715c4e9aac438f767ffece1679d3a0eca1eedaa6aeadbc70fb56f4e1fbc37
7c3436790f6fe3fdbe7db4677cd12a290eff06e454b9264c713d6f79dedde379cbc41e9d7c7607236bd464aefd55ec7ba148b6a107ece609905f249b
cc2685800d9db76ab65c1d54153862740d519b4f741f33a82a406b134ff731d753eaade9e90f114b637d32950587dd16ec8e17d2bad8ba094525cb96
aff056946fdc24dbbe60bd6ede643d6f7c434f7ffa09adabc7f52a71bd7c6dbd284f80c2d7532831d8c4401dc1f532eedc831b981a6473d805c5ddd5
96d645a84490e515de15cb97cbb67a96f1c9a7acc73737e8bb376fd277106a1efd44281696e38eac47c8164144ff67b97a1ee15c4238d18ee83c21bc
e986b0bc92efbdb4f99ab80ee5d1809edded0994bd36e235aeb7ad09d13b0322c1e9080fc1d5efd4e3e4eb77ea2db792698c60b5d852536c568b9090
42ac16e28ee15761f596575ec1bf575eb94bf5eca7bb77d94f542f65b30bec3c960b3415cf2e34d5cb0a59092b6585742d5d4017d2b5dc227f86466c
2cfa4f24a6c791095e51f04a4b15e2d5eba2642790286ab05cf2d39b727ad7d735224249f52977ea7147c8f53131f470000488c2b86ed156292d2ed5
8abb6474207b894e7c9f0e6cacdc2316f6afeedf70790f0240091507e28e9d64ab2721342c1c429c564924564912332daf5a379abcf6f5225a2a6251
05aa3a832d2047581a075539460caa0a1af1e4a02afb88271113e0525077a9fec409ab2ddd8fcd1d0d1bc5227dab48dfd22aa7c51a9c8eb87952468a
a3a451ca4271a13437bc3454413b162a86a1403b6793b9f29cb0c2f0d9ce65a4247459d8b2f065ce5d6457b8751c1917879b48eb4abaf5a2695de2dd
31b292d68ba6a6880ebbacc8048de7c9c6c148c6d4dcc7fe58f2f4c5f90b2f8dfe92dafb3e19caeeecd9b3671e5ddf63faa601f32a321f3dff50ca97
effc6e474104fb0677bf05f95d88bb6f470a3c9d8923502dd14795b802bd0e9357bf41767a5d1bdcebe5358ed712839c8104eca1ce7897c509f628bd
9cc8891034a265ff7a6dff480014d2604d4aebafdfb95e6ff9e296453b912ac9d4a3cf8bcc8dca75e5458b641c8da40ebb181d139f9016891be98abb
ea40d37c8dfbb607bdd7bfc63e605f3e7566ea88f7a61f3f53b3e3c091f2adafbd38fcf8acc2b363bea0c63f405c545dd9c7dfc7c59d7a28a562ddef
cb77ce2b282c8a8d3fec727d7868d15e2ed779c8e54a942901756ba927829ac044004c99040c8a57a2b0544f8d2a71ca3ad1a869b6013766d23666e4
1bbb9451579f62e57cbd7e29a33e05f7a231563c8bcc3dcb59dade80817b7f32065de13cf23c51826807124f3b40573a843e6e7cdc348a4ea273e842
58414dc84a3d8d86546baac36d755ba3d34066026569ecf2e5b34d4f49718dd7e04263ea2ee6a539a790435b91437988790479cae316c3146b892522
ccabd8bd965526c14b969ad6289591c14eaa8293a81639d2d248dbf2c5d2c63a59b8b6208b2c75b7b802730d46f6b03a1f77b829b0729a13879ddcc7
16ce8d8f21b4c9db7174c7061acb2eb1ef9e3a3565ec8969fbdf7f7fffd057474897f7b0170202d8adafffc17e70b9ce3d947c64cb9623b1f148ed75
887d85664f62c9684f6ca04c4c2546e20d92bdcea01d16af7155cc7ae79a38638cde191a19e884e8a8f03834302844d7351373bdf1fa3df1f1d8318a
a017840b70413c279d9371df87228571741c8d911df6201faed4d199ba630468d988dbc5cd51744a9050b972dbb69558a87ef0cb83dfbb18d0f3d0b4
cfa8c46e7fce9ad82d9a4dc307bf0c3d8f6e7ff5d8b157b71f151654c7c6b3efd9774f8c63df7df305fb5a3350e3e98e486ea176a1344d419ec86482
2744b20a208055447b21213f40028a2656562c8de7eb347b9dd4c60e604171e10c1afd1686641e04a610057964edd67d8cc7365aa0328449e9527f69
3254912a59416941c650378dde05279a3ebf485953aa747954c352a9038ffb56237d576bf4759324f2a8272e04a99b207b233b796deb23d724bc961c
628c6def74c43a03f468bdd184074487275b1aebeaefd4d56b846dd155ad978e4ada8698719dd1d6c4a6a6047123a3a9ab3b2636ad4bd7c096012819
c2eab21d3bcaca76ee603b96ad27cdfff1095bbff485d7d84f3ffdc47eaaecbf7ef9b20d1b962d5f2fbcbbb9b474f3cb25a59b47b90e2d79e3830fde
5872c815737add952fbfbcb2ee34cd9dbd6cd96c2c28314b7147a5b8a3104d62dc4a54282d21a15e7587e825ab82a2bc96f5416be214a7333a3092c4
c4384d9ac020fa2d3ee90bf6438bbc04d585be137622fc84f344c43b917551ca1e5baded2b1ba0c474d364db1668465921695d48aa4f4a62e269cbb6
90069f0dde3208e5a4c7a1673f6577a9e5730ad4ca0eb21b83b7d05e7e598a4229a1266a1bf53b1af0cd1734487367dbd89391c2a61649e2d6e7360a
cd29d18dfc5288d36396978b3bc97201cd8e4842749646b428dc63a09874e33601cbed8b783026ba19cebe4273a48f60bb7fb6095db8bc45d4495424
6e3eb92e4573e57e67ce4f0c951b182fb07dd7ed3d38ef2cca6b3dcaab0e3378174a893b9494abfa72db525aaeee8fb21a744260689444ccce2029d4
d9594f9c36311aa99ad28850b935d71c5a3d1791f4e44301311455ce1aeda3546b232ebaad9844d30db4cf6bafbcf21aaba51d36ae5fbf911904f166
c39245e53bd8edbb4d5f0a679b3e2e5dbd66853089f5ca9f35b360e78983abb6db5de75e7cefdf51250a9baf49092803a1a4ab27ccf4aaf9805a6ea5
af92036279f07aeb9a3025d44492ed96308ea29fed1cb71f6f251f0e088f0a17103d6e05fc9adfb59bc3dcda09921226dd5cd64cd86d6aa164d9cd49
53bffd3ddbcf16d2123abce45b69fce5a79f6267d8dfd91576e6a9a72ff6ef4fb7d1c9740addd60ff9805494aafc54ecec7190723dd2cfa2132c2a91
424d29c4a9176d5a74843ae5231972f4504ea0462ebf84c5456b7522a51beed0341ac53e63e75826ae728856b0292c9be54a4977e7d110da9976a4c1
3bd926b6843dc72a90269c87ab7175035f5b2e178572b254572eee5725aa57d0f68b464e8e4b7575adbc4a3e1465c2b535bbee2f67a1aa294c38d394
2efcdcd88b9bedac3d4dd7f6f8a1bb11ba9e247a6c7ee8e27e14320db4ea03eddb12071c60680bd87d16c6361508d94d55ef7398fdf73475237e9853
11a644e23c460d2075420f22cadcc273575a9fecd1272bd9ca1258228a084f8bd3cebe2ffc5be3d3d2655f748630e4288461246f7bba8055d1298295
0a3a5e81a057f5d4aaaafa4c55114007e4759d41d2eb30d59054d929f652117313aed5c86518d7d37c7870fa3d5bacf3176e930f1598714b9e51a004
e802f482ea10ec4aa01a2fc42b2e255e75a95d9434f519619150a42c509708cb94656a991024520304d27070d38e90a06ba7ef423360946e8c7ea26e
aa7eae6e01cac75a28a72f835db3e2e8f5d1dfa3294782d14e74312da69dde65c5e758719d74b951073f377490a21a31816ff80ca58debc095369178
b98d94fb22f1d0805408755842b4adb589c4b960a56a3178824fd8b42b245c654d14ae5ea594355fa53de87cb6929d66efb252ba401accaad90df605
aba6fd69180da7fd2bd9936c2bb760b412bd1dfa3b9fdc8b6b35b90f243d3c2128f35cf46d165527885cf27b5bb9e8db7d06c3476a2d12f718021c51
8ede8ea71daf3b244d075a2d0586756207a403ddc0d66edebc9675a7efdde5f8dd65ef4b494d7f79a1b4e4859dd73efaf8f3a65d7ebdc31803577792
839e34bd4e0155b6820892551421531689034447b9de5e6e5a6a102519ac68c182cc921a1a2a5a7bdb55a7518cd090abe3d8597dba99c151b4a5f3b3
ad73d622f4439e482e0b990b03a944242a093228a28338a85d088260318ec4d138211e12e478255e17af774576a55d852c9a254c91e68873a479812b
e595ca8bf28b4ad4382dcc0d0e744367caf76b8d767123d94a0a58fb4851af0b57de1eb87afed5f7e97b94342e6f5ac55e282f7f41a80d2a7b8e4da1
c515e39b564997fff6f7b54785c79b6e952e5fbe0275abb90179f215d2442103d19ff894cb83aaea9174964b1835693c4849c64843e591864e8b3474
44d7126904127d14b1508b10a558f41e7d817e9b5e3f0e3465b646cbe2774db7ce35dd42656eb8cce30c4afa10a2ece47a488b3c03a4305942dd13c3
543d84a9065508a382c1803c41f5449e487ef534e2682b21c64c55124046f534e88c0655aff3bd0e3028c464b974dea799f52929bfad98adb5af89bc
219a9e7e2f0bb224a89823aa36b59d148b1ada4be825755193d5c1c26352a6ea51c708538569d26435472dc2f47791542c2d512b84722942217a01ad
85282373894c1511e9a2e8895e5455233187814374e8428d16b34b8c965cb24b71e9dcfa5835cee032bbcc19420f481353a5645d577dbaa1b731d99c
45b2e840c123f6953c52a69ca964ea3c3a8fbe8ffa98d163f698470b680f8cd9e649c264c815c74b39728e92a3cbd3e7a9798679642e2d12e6c33c71
b6b4405ea0ccd315e8e61b8b8dc5e612a114568aaba415fae70debcc9bc46de6d7cd4f721b92aaa7fc8fbaf5d4dde73c2a73fa357eb9c05631d4e877
9874f9ae4dbcc50bda124bc36dd49e228ca73a894548a938528b196e9421586f26bb83e51ab3d5551275d459e3aeb6ae093692600831e9758628d0d9
fbc6a3be9cbf845cf105b075d7ef34a29139adc58b56ae329e19c911c991c951c9aee4e8e498de099e084fa427cae3f2447b62b223b223b3a3b25dd9
d1d931d90905092b224a234ba34a5da5d12b62ca12bc09b713225ba6b64c6a999013991395e3ca892e882c882a7015442f895c12b5c4b5243aa46d94
ff30ed6675a7f1d02d1e63d1d4e8b6f9629070fc937d4bf35faaa9aeee5dbb72dfb9a6bb54f8e3a69c2323261e1ffb9fb785d44945e30baf1c4e1cdc
b474cfa4dc93dbdf3a612b5eddb9f39e8484461ea71d455a55ca76b4b94ed2dd130a35c6007d4d88634d4075f8a65062b3f50b31cabab02ccd9aa4dc
d15e3a5ce771f4e95bc947722297447a2301f16c891f1155aaa521568b80b826f0d00e6efcf18517fec84bd31f7a1c2c3a4f9a9bcf171dec51532324
9dbb79f31c1661585e2eab653fe3599b9bb70bb1a16466f335b8893c0c25bd3de1a484ae14cd25a6956a8d55ac09aee60191cd44fadbfb624074bd25
20b2b03bb72c3fdce24638dc12be24bc2cdc1b2ef903a316ecb4c028c61f18c1cd21af64bf71faf41bd9af0c796cc7b826f637f451f2c8ed62dabe0e
1dae5db870ad43873db1b1b82133b5d11e6ea41662258e45fc2c3e6a85d510b3bd46d2ad3157d34d6826894ee867b519fa466831514a4a2bb5eaeea3
9635d5cf4c41735f41b44d6e01dbabab7b1c5c74ae99349f5b74b0e90cd26dd72ea41d1c119efaa57e575e2eed437578f6c9650e3ff9fc781523b5ec
249c147862d16eeb4b742b25c76e2ad518e9b1901a5bb5718d33dc21e81c3a3248b005f4756a28d6696f6e38f17c49e71d5f1691d83ba220c21bf141
c4ed08a937e94d7b0bbd1dbdc3a58e4a922e49df51cd27f9345fc877e487ebc7cde4048ed6d2e67b41270a80a2115d118b1b0f192fbc39f5ccf8091f
4c6377d8199ad8f83955aa851d2b37d79885a7c61e3fd3a5cb81f61d6977aad240fa28fbb86ed3e1035bb9fd4d42f1fc19691d48c6789c92851a75bb
655a4a3699e55a5508c414522fe94c0186c176fe5646e5c9be8127fb83aacc5a9bbf04c8a8c35ca1cea629f475f4d4965b29369eef1df138b21d5e07
a6450e443282fa9c803b2d952b97f073d584c76812fbb0a6aaeac05bb2fda5ec2913d63526c187eb861cdbcb69cd46896391d606d20e3d913bd418a1
b7950406d504404dbcbb3aa1565f13f05658447c28d119fbc9369bab6fa2967bfac4a1eeba4f20d8654ee974948af64bda7bdb3fa043c116e1dedb88
87a95f546c282ac169a9b07d47f9c61d3b3696efa866ac2177dfd0a15b87fde970faa1457f6e6cfcf3a243e9d5c2c3ef5dbdfade99ab57bf619fb3af
2222dfe8d8feadb79f9c301ecd26cfd37b8c9fb087d317d334314fa36f17d47b3d0133954bcdd66ae32615c34d32845bc62c2dc4d1d43e83a7cf561b
faac43390e2dbe775b7d286323557b8b1224e6552f5a54beafa626f38d39274f0b954dbf13b66edb7abcb2a954b6376d9d98f71db7382771f105b82e
cfed3aa0273f2e1e24b59819ea4492d59a195e6fe4c114f7d3d9fa1cf4d592162b6b89e2c96a3cc49cbb5ed9fe15dfc73d78316f924d02d5912cd1e2
7b3d9aec3159d047654b395281745b927d4010806cffa5de4f032502791943c67ae2659b3e2480c8118ac3581ae182eaf0da508b42ac013a9d9c6dd5
05643b43d010bab5b0aa11a33eed9d6746c6f53b5a90cd09e3094c8ecd8e2d882d8bf5e2f976ec27b1cdb17aa494461b475b7add239cc347b8c4be27
96bd7ebc66d69c753b6b66cd5bbbb3a6a677d582857b61d5a2b93f7ccec9f8ea164e4661ebf697df7eada954cc393079fca2562ee20e02316fbc8f8b
b5bfcdc5eb2d5c3c9ce3f88b4378908f8eff868fb83067a3cfe2ccd1b42018b52050aeb1911a63358fd56d0143c1e6e8fbc05b738fbb77681129928b
95625db1be582d3614198b4dc5e6e280624bb1b5c8e60dbd1d6abdffbdd67d2fd70b37eedb5bbe61dfbe0db7a98dddbafd0ff61db5c22737cf9ebdf9
e57b67bedac2de63f5ec5b342fe96845ecb43bf76ca8a7958821b7d5bd3ce12db6badabc86be05b51168a7fb6916bb8d6fb35cbfde62ae3d7a9fbdfe
3412d3b4b856d2f81ddb7d0eafb0a6e69e5f13bab778bb5d4d0764754f1bcf46bf6931d8f759110dbb16bf5b1db026fcadd0da08cdebf643ffdbc697
b46077fa01ecdaba0fdac6ada0adc07423a9c5830885f7fc4a8feaea56efdb74a08d53c9dbf3cb8f3ea98281889d95247bec32c6ad5603949aabf5b5
8a2a63e89865e3264dd353f42197ce73a771383b705b2097279fb7bd274cc130306a40c72d7f442a1d5d11d8d909876dd673c79b0ea1284d9a20f1df
d4f3d1d79fc1d512c84d4f86c928980dc3a322757a41518747454566aa86c828cc454ae82ad15ee25815c26380388c01da45aa86a870850c0bd79915
9d3da66f3b8ed5a5faebdcb2a6a7b704053ff0a0c0d6126b9b79bea398fd590f49e091f574a7ea34388d9dd1b575347434f6d4f7547b1a7a1a0d2ee2
a2b1423bb59da17d60923dc9d13ea85d64bba8445762746c42895a6228319698f8efca541064553680114c608600b0402884413838c5087d425262ef
c4a7138b13972496257a136f278660503bf35e4c12a5fd2620bbdbbe7c4ea2fc256457a41dac1eb26becaa55e337f6aedbf1d3dfc79e7a76d2e9dc65
6b26eef5ec7df1d33f4f3a2cf63ed0aedd88119e01d1e6f62fadda72c4ed3e9e963666e8a0ecb880d8f2655bf7696f6ebba1b87d2f6d450b81118b59
d205c06e62a5b5ba52d58034460db0d8ccdc4268ce32c5ff16c7f7f21f6dfdeb3e5bcf3da43da827f797f169dc535ae93c5ac4560c2a7cebadcbdb4b
4ba5adec9d754dde5543366ffbab90b38ef6e2327e006dc468cd36d9494f8ff39e755aa3d25a7bb5116d93dd3004ad5496830b79ba4fa2aea7b49aa8
7cc7096ea202d1bff864bc35728aa707b889da5f5dfde8c13927dfa37fa147859d4db9dbb61daf148aee7af74d9a701b76f1dd8f22444e1373302ffb
06b3b891be2c6e24667123791637f25fcae2defe8d2c8e871c83aaacfcd7301bbf04f28b41fbfd81f03004d350a3ef0795415516df6d1e91fc8f933f
ea6996842021488a51d3d401c200290bf3bc278527a5916ab63a4398214d521760aeb70073bd52e125e14569a35a2bd44a7f16cec05fa40849d0832c
1a245567d063657408a110248649e1ba70bddde030f2bcde2d2440b41827c5c8314a9c2e01f3be6883db980e5dc5aeba749eed09fd214bf488999247
f6281e5d1fccf4fa1878a6378a8ca2a3846c71a8344c1ea664eb86eb47a8230d13481e9d284c8589e25469aa3c5599a1cf354c36e69be790397481b0
18e68b8ba58572b1bc502956e6eb16e88bf545ea5cc36263a9b052c2cc8f6ca21b850db0457c59e2ef135ed279922a8cdbcc3bc94e5a2954c25e71af
b45bdeadecd5551a5f37ff4938086f89c7a46afddbe63ae1149c17df971668596338e57fd46da0ee51d55fdcb8f2c58d6af6d1957f7c7f45cc69ac80
a9bcdcf54245e35494919ee84317a08c18e8a39e2cc92a2bb2680551e195245281825540b6f3afca54ab5ea5bc32a828327a2b0a4ca6aa8814f3e9b7
3137f7b5049d6c6c1190003fff35514121907d21aa9d87a8b2eff7d33a6b8b4cfc339178f0c51df5bca88aa21a263ad478f561f12175a4f884325a9d
a4cea50bc5b9ca6c75adb84c7d49dc266e525e50cbd49d74b7f8bab843794df5aa4e15440975c08039bfe4d0871912215e8ad3b737b84c98544337a9
8bc2b3fc64d300c892faea071a3ca631c8e531c21878421a258f5146e946e9c718b24df9a6f9b4d8f432dda8eca5954a95e92fa64f4ccda624fe4b9f
c053752d6b17f3d834bae70a3bca8e5ea16fb0595768224d14739a3e693a49ab597f61a010c466d2755c4fab9a1a215b7901fd21a1ee344ba02d3095
8060a585a7bc47777a59e3d8e2a6c6af6113fd5c48a6d0f49faca4e956e3b7be79b45059c5bf07e2b15dd5a953caaa1f0b3588cc8e105769101d1ca0
5b408069b4b078c339ef07655e65d5d78d3bd813ccce26d2bef496100f295f3f88499a05153ed51668b33a84b5c56359a377e7512f472450b0a1309b
5953d3072cb2f199afd107acd5e6ade25f23f1982a19b23554b46fbddae292ea40a0e0e64085b5a7bc651f78cf6d28e6c85c68fa88d9d831fa32ada7
bbe1497f0e24edc1d83696f4f704c66b298f313ac414a9b31aa32df6c171da6f21dc6e5b3278a653f710f158f526eb6e9b10564a4236c951b65a4340
52c61729292ce3560aa63f29c9f7a53cf7d21ecdb42bfc01f73cd29e961c88295a1a74a06a42423cfde5be7ca825277aa95dbb29137cb95112ffb240
e6df498692611e6b781609d60505d8459d0e82547970d83d7c5906c65b1e9b0e9d91a5d41c723ce8a079939ed44a94637b8b69bf80a6600ad08c497c
1926f3162d95ff35ca883152d42d0ef461baff8d1a8ef92f35353c4368c1f1cdd739d2f4d0577e9afa71ece60908cec2bc4335ea7416d1661e1cc4f1
f3a1c7b1c3ec72b75ec494d3aaaf35091c31a661a5bdd48e7e307144692635eccd7bb9234f32a4a207b2475c5d6ec4d5f91775633d214959c11d74ed
2de10e5d587b3d899275b191fa98f8c19def11aa2e855f1b3572058747b977c75a310bee74bcfd410bd914a4c4d6864644f3342425051d67bda53e05
ff7c5cf673b35bd76eada46ae1799b745742eaf19497b3f70967c290e5c8eac784439c9c7ede035213d9eee3f24867600227660b715bb626b4d056db
5d3c19e2096a9765d2598242ec3a8b9e7f66101dae8f720f4e68b3336d639a1884385dbba3ad42a9317e934389ae0d088bf46de94ec6aff7d335f501
f23f90bbfbe4b42d2ffcfb68ddc3deb67c69e50dff165118f3921c1f61793a20e30712a5d33e3afcf039f3f596faa7bf350e368fd1f32fa875addf28
e23c653a8b20c4cc7efa5bc350f3985f7dd5d84bbca07dff47044cbc85d5185b47912a2c95c23592270593522c9f61a9c0b2050bbfb715cb3a2cbbb0
acc6b25438446ecb87c9156931392b55904239116b33392b6ec632879c956f62194c0ac58b787f1ddebbd6dca07c45fa88374991642747c5496426d6
33c57a3253f89024f1b66423478574729217e51c39cafbe20d6ddc51fe1c0662bf03c90737e986cf0ec8abc928b196f4441357757f11d6f25a5a8e70
555f516a511a92eea3437732812c20c7c8f77426bd21640aaf62c43c148e8ab1e2a398c3d64a8a9425cd91b64927e55ef272f994f298b242375ef7b3
be8f7e9dfe23b5afba5efdc1f090618c619ee1ef46bb71a2f1b0094c3d4c5f99e3cd5302165a322ca32d132d4b2daf58f65ace58a3ac33ad55b628db
731a377ac308d2814c2146d4400b7989734f740841588bda77a263f9b76fa21e9999ac7d5bcadb940461cfd716888e66f9dbd0e6bed8a62d91103ac4
df96899d4e228f927c52803b9e459e219371f5d9da97b21350f75d248524e3998aadf138c2453271cc6c52886516994872c974d211ef0e2033707c67
6c3d429ec5d34586b5c22ad47a13b19e8873e6e2350f47aaffc2aa5d5b571d812bcdc5b5f8d78e337034c72317e7fccf56ec83ada9386f1499832326
e0d85c0dda446d46aeb62317429981d7021c331ee13e83e35c383f1f57cfd59e3d0867b806a51031cac7731adee5ab16e2d87c0d520aae9d4ad2ee9b
d53247f0095cf373da37debf3e7a6b7acbffcb807fed6ff1ff6742cbb7fd2dff9de0227168c912907a8928c9c9da7f28a4613ed503d7ec4bb2483fd2
1f3934900cd2be7ace2643913ac391a623119727c8685cfb491e3d52914a54a60ad59113e424d553559933e39994d4f44c7ffda8bfeee3affbfaeb2c
5ffd4832afb332935bea547fdda55a58e269becba0c10ebfc4c1cf29f05305fc68861f18dc61f09f71f0bd19fe5101b7e3e0bbe71f91be6370ab02be
ad80fa06f8a601be66f0550ff832136e32f822056e5c1f2edda880eb38f0fa70b8f6799274ad013e4f82cf187ccae09314f80f3b7c5c0157197c6483
7f5f0c578ec1df19fc0d87ff6d315cbed44fbabc182ef5838b7f0d972e32f86b387cc8e003067f61f06706172ae0fcb948e93c837391f07e0a9c6570
7a85553aed847783a08ec12906ef3038c9e00483b7191c67f016835a06c7181cb5424d499c54c3a0fa4d8cf319bc79649cf4e63178738978e44f71d2
91719e6638e211ff14078719bc510187181c6450c5e0750607f260bf19f6ed8d93f6e5c1de3d36696f1cecb1c16e447a7703ec62f047063b19ecb041
2583d7b69ba5d75260bb195ecd032f0ef156c036065b5f31626e0baf1861cbcba1d2963c7879b3457a3914365be025155e64b0a9c2246d6250618272
9c545e011b3798a58ded6083195e6880f565c7a4f50ccad68d93ca8e41d91271dd1fe2a475e3609d47fc431cac65b0667567690d83d59de179dce6f3
8fc0aa950669951d561aa0146f94e6410952aa240e5658e1f70c962fb34acb192cb3c252064b181433f0343fb778b1f41c83c58b61511e148d704845
71b090c10206f3cd30cf08735598c36076031436c0ac0698d900050cf219cc60f06c344c6330d59a294d1d0ecf3098b2182663671283890cf2184c60
309e416e0fc86980a78c308ec1930cc63218335a95c634c068159e080a959e4881510c46e2ca233361840386538b343c0486d961e8c0406928836c03
3cce60c86316690883c72c3098c1207c3288c1c001166960200c883049032cd0df04fd18645540df0ae8c3e051a193f46803641e8347068187416f06
bd1eb649bdecf0704680f4b00d327a9aa40c4f7300f434410f06e90cba77b34bdd1ba05b578bd4cd0e5dd30c52570ba419a04b24a49a20e5218394c2
e02103242719a464132419a07327bdd4d9029df4d031053ab48f933ae441fb449bd43e0e126dd02e214e6af70824c4417c9c418a0f803803c4327033
88098068dc67b40d5c7910d50091b885c83c8830811329e86410de006199108a9d50062179108c940a661084938242c1c1c0ce2090810d07d818e698
9d246b26581643401e9819988c4192898111471b83c0c040b5809e810e87e918287690f340c487224a8003f02e30cc552c92d009a80508035a4df356
aca51dfe7f38c8ff6d04fecb23e2ff00352c5a9a
>}
\immediate\pdfobj useobjnum \csname EF4O13\endcsname {[ 32 [ 318 ] 40 [ 390 390 ] 45 [ 361 ] 48 [ 636 636 636 636 636 636 636 636 636 636 337 ] 61 [ 838 ] 66 [ 686 698 770 632 575 ] 76 [ 557 863 ] 79 [ 787 603 ] 82 [ 695 ] 84 [ 611 ] 87 [ 989 ] 97 [ 613 635 550 635 615 352 635 634 278 ] 108 [ 278 974 634 612 635 635 411 521 392 634 ] 119 [ 818 592 ] ]}
\immediate\pdfobj useobjnum \csname EF4O14\endcsname stream attr{ /Filter [/ASCIIHexDecode /FlateDecode]}{
789ca5ce570a03400800d107e9bdf7de7bbfffddb22c218490bf080e3aa2c89f91f8ea9352af2a1d99919593575054525651558b93fa7ba7a1a9a5ad
f371a7ab17d837081c861c453bfef1c1c4d4ccdcc2d2cada26daad9dbd83a393b38bab5bb0778f270a2c052c
>}
\immediate\pdfobj useobjnum \csname EF4O15\endcsname stream attr{ /Filter [/ASCIIHexDecode /FlateDecode]}{
789c5d53cb6e833010bcf3153eb68788f0b213092155e985431f2aed09e540f0122115830c39f0f7c51e0352916034bbde197b59fb97fc3557edc4fc
4fddd7054dac6995d434f60f5d13bbd1bd555e1032d9d69363f65b77d5e0f94b71318f1375b96a7a2f4d99ffb524c749cfece945f6377af61863fe87
96a45b75674f3f9702a1e2310cbfd4919ad8d1cb3226a959e4deaae1bdea88f9b6f890cb25df4ef36129db577ccf03b1d0f2005baa7b49e350d5a42b
75272f3d2e4fc6d26679328f94fc970f6294dd9a6d7d68d6034ae0d5864f089f5d78a3c84a50e9b212e1085a5185f04a0340088800312001708000c0
2c3aaf6a108767e43c23e7194336e608af142e315ce2645d83921ad429ed14d9c6d2c43564a7369b40df4009441846064a20c23890811268c31c1de1
aeb92b8534c7d6391439b6ced1200e3d7e5aeb21870388c4c98172f48b1300c710f825028e028e028e225e65acaa8097384175a7573358eb04991933
17621be0faa1f532bbf6d6d8a135e3da2ada2ed6d00fa6cabc7ff06fe5e6
>}
\immediate\pdfobj useobjnum \csname EF4O16\endcsname {<< /M0 \csname EF4O17\endcsname\space 0 R /M1 \csname EF4O18\endcsname\space 0 R >>}
\immediate\pdfobj useobjnum \csname EF4O17\endcsname stream attr{/Type /XObject /Subtype /Form /BBox [ -6.93649167 -6.93649167 6.93649167 6.93649167 ] /Filter [/ASCIIHexDecode /FlateDecode]}{
789c3354c8e23250f0e2e2d235d4b3343633b1345240b072b9b008e6200491c474b109667071397101009d5e12c9
>}
\immediate\pdfobj useobjnum \csname EF4O18\endcsname stream attr{/Type /XObject /Subtype /Form /BBox [ -7.12132034 -7.12132034 7.12132034 7.12132034 ] /Filter [/ASCIIHexDecode /FlateDecode]}{
789c6d90310ec3300845774ec105be05546e9db563af9125aa94fbaf71ab802de1c5c21f78c057fe92f087fa032b6afa303e494a7d5a6d3a342d2aa6
5b656869dbab49edd2ff5779a7d01075de8860b922bd21624f4ed41bf0233b3da4bb68c1f468a73c11998081479e8db41fe6ed918fc3c2024c0e6161
2056748c430ea2375dab904deb
>}
\immediate\pdfobj useobjnum \csname EF4O19\endcsname {<< /A1 << /Type /ExtGState /CA 0 /ca 1 >> >>}
\immediate\pdfobj useobjnum \csname EF4O20\endcsname {<<  >>}
\immediate\pdfobj useobjnum \csname EF4O21\endcsname {<<  >>}
\immediate\pdfobj useobjnum \csname EF4O22\endcsname stream attr{/Type /XObject /Subtype /Form /FormType 1 /BBox [0 0 306 327.6] /Resources << /Font \csname EF4O1\endcsname\space 0 R /XObject \csname EF4O16\endcsname\space 0 R /ExtGState \csname EF4O19\endcsname\space 0 R /Pattern \csname EF4O20\endcsname\space 0 R /Shading \csname EF4O21\endcsname\space 0 R /ProcSet [ /PDF /Text /ImageB /ImageC /ImageI ] >> /Filter [/ASCIIHexDecode /FlateDecode]}{
78dad55b4b731cb7119ef3fc8a39920741e8c63b55394496c52a575c896856f960e5a090142985a42cd289fcf3fd35667607b35cec82a2ca54482dc9
e9c1f403e86e7cdd187d1a68d0f8a6e199967fa7d7fdf397e7ff7b7f7a7e7cf462f8eea7f2eaf4aea7e1033e1778e0033e9ff1d8113e17bdb0b8ee8d
f6f87d957f1b0ecae36fbdfeebb2efdff59f064a2a26fc348aa31d38b0327620226570797b3efc3cdc0ccfff06a677835649bb18d9384e72618982b7
41fb1e17ce47cfc13a1a6e459fa3fae8616374df93b68a7011e3c0cea9648283fa9cb4721c222ed6d4ab05d5b2d2780ed492c34ccd26be1ebe3d23ad
56d6df3372455d1ac95e397bcfc899fa202379f8010e0267810b7cc6cfacac1860615e485e87680d59dc89b4527f383e1a36dc2e4b12091469328234
1b6bd7c4abbeff296b853b3669caac71c1da24ceace502b24c606793cca7e84559abf989befec4f1513beffe532f71f54c026b521ed364301001f6e2
44e2457999ed93b3fef92b1a82e2e1e45d7fd0bd3c1c4e3ef44eb9c8de92d72411854107dd79be63b162da5a88c5a24c777e9dee78cdda4133e2d59d
dff21d56d1b277647c58dfb99c9e71c9a6a08336b4ba334ccf186d9d77965318ef3c7fc5b39a7779d0f727fdeb85a530d507091bb1168f1b6237593c
5a69e5713d32fc6538e8f870f8d770f2c37d4ece2bf6c93e8497adf14a980098cfedac7c85157392f98f0fe1156bbcbc5181f8010652a7ab268ef111
54a2bd2ef6735ebba0c8471712152ef671e562c6791f93756675e776728a1435748e49af9f39ab3a5fcd910ebaabe90eb22647466a6970f29baac35e
fcf9eeffa9e2feacbc15d79085880dae81254df518b049b1899493722347ee7467ab2e87e95414827f808a0e2a72351aac0ac1641eadfcc0ad25221a
b999fae48d11e15939bacf8737f8bc3a1c02298fbd2e1af902e97d7783cf6fdd793774d7dddbee775c5d77ffedae6b02e7dc5711baa97ca8e63ea312
bb1cbbadbc4c470d6bdec84cd628b4ac79233ff1a19a7aec92d2ec1ec06d74f1b07bdd6d50daee5ff7efb0bc3779916f4b7e5ae9d52e2fc8c49430c5
cd884b20596d64bf1899a14323cf8a53552cfa22a76ae4d5e6548dcc9a9daa915fa3533572dbe254ed38b23ab2ff9af8d1907261bf4bbf189d998232
9ea62f906fba8f485e57f87e8f8476878476ddfd8aabf36fddef2b667f91df37f26af3fb4666cd7edfc8afd1ef1bb9fd1ff83d79947afbfdfea4bb84
839f77cff0fd239cfd164e7ebe7fafae70ffa23ae5b1bc8a3ae591ac1675ca637979af34a0f057b02fafa7465962f6afe73f0f078bb248030aca5704
e92d3216925bf70ecb7bd59d658876b17f892b029bcbc772891fcbab58e247b25a2cf16379ad97f84b184d7dacb22da8a7b6a0734ac7c158a4fd38b5
6dc2eaf26aba242fad278beb71f47cbd6e342dd8b071caf3400c6f64bfd96722458fe833a100716c92057b0734eaad4b0cbdb792d70da7a7d0ce6227
a564a4a1477004ef9cf35490ff7ce588c475e441f08f4a074b3a887a2ea8e034055bd29f563f722a728cde8a7ec12967181b50497f5afd34163a1acf
d2eaa414957556a3282de84fa0dfec66899431d18424214d606b036af582fc14da5965c90768178081122085cfda49930957a9203f817606b031790b
c742a68d98ae64443bb81e52afa592fc04da012263a2b0630e2e77d1bc9b94b3de23e1ccd427d0cd3b3895f1c8b7969571c1eab8546e26cf2700c3c5
74e610a7bd0894227baef7a382c4d39ed4bfc0b6f8b9dc1c3d294e40128961028c4508bac5eeb809f897601a2a94b1bd965dd2eac24923fa49dbe44c
4cae41badd227dcec0b3f482b643ba78460a10af81c35a6c5f76ce47f9ac2d408fe15400819954972e60c44a7bda2131c706e97e9bf5ec4825ac9c1c
49aec517b41df22dc233f94031702ec0f7c9671467be8a475108238d4be92849c9d90d88e50a6699db4be0dbbbeeb78c724ff17be83e4b5983d2de2b
bffa0ad281c7b0a17b73004cfcb1bb004abecd8dcb4b54f8eff1a454fba7a0a2d6efde1c42bd7e394113ea2af10e967d49da363d981d84bbd4d5d2a8
4e3a1acb8c5daa3e3ffafeea4cc20b3833c92e91cc2ed12660328c4ec991f50f8b8b4976095626e10b9cb24b3a230d5188ce23a1d10ee9b12abd8422
93f4050ab92f5d4f750d915114d863aa501fed0c4aae8a2f91c6247e01327688477422259a40463ce4813961145f028951fa0243d485c7a87cc49e85
49c243bb22b2ea73254c18652f10425d76f0ca1a1f23ac66b35376d5e74a1030ca5eecff75d9de024790b7463a283b649bfa92175bfc28badcddeb92
256302bf686da3b3bbe2ac3ee3e5063e8a5eecdd75d926a9a0111621f8b033fd9b65f2cd6fb5083b0014adf3395a64e0062a4e48e58cb13c2375e301
df3fa663f804dc8db9967a783a1f7c379d0f622e4cd4f002bbffce785a6ec02dd8e569e3eabc1351cc2ec442ceeaec3261cf080609e1fe0965b49376
e509e564c07fe613ca46c03623a65ff26f4422633af570d6da12bc87de0a9836efbd336daef427e0269871066d4b5e1bad84cd5e0138ec6e36b42357
5d6caf62e0884d75927743687ef966f10249e58961f1c4f151df3872dd08cd007442a9d278b6c5ae3c6eb6c3f31ff5f0f2e37a348207f802d8852d03
bad82d436c50d83c88bcfcdc3f5c7f952184ec4ad169ecd526ec1ffeccc9d977482945ecebfbb91b154292c8d2d88cf60f8717901cef27d3329ae11e
483f1c59ef98afd7dfe6eb47f2e696d556338bdb20648c8a46e33967ca44ee55d876d2b0ebf4fb9b392e92d70524e1268362dac82e99086be5c27297
f40b903d9aa8eb6f238029e2130548708834849d0f2807382d37c06d4c779cd02c99ca0b782031f6dff0309e0f48eaa38f352ec043f20e9541020483
6ce229b1f6716b902c860bc406a02003adecfee1847235b273c441a706ee2c1bbbd18151943470c7da2623518121fb87638b0a641d6cf566ff684c23
6372937352eb373047dd8bb29fa41fb37f747dc82a1d356e385b37b37ed766d6cebbefe1602a4a1f414abb28a046671c4851401e802f97f4ab051d8e
07609472195af229e997821ff570b1eef62c0e5d00f3223288d4a93a915b445adc88b4bfe7025b47eb4603403a1768b7499482fc1cb5f7c5f294319f
82b4e5b0e3f18c04364515a409286d4217e128b9ad8f72c4207f452ae942052b975ca646f1f141268b470649317c0780f6749c426c10da65bad71171
294459a6f1080153ce88be2531940c56645228f5e0f243c915d44d0580e867050a5d419f1428cc02311a26b9bf9e8299769a27069f8840588c0db306
33dbb055033f6bb0561644cf3a7be26c56412ca760415e4d57c1759ed8b5028b3528749d97ab30ab5cd9fb5e70dae8d77822b1cb227778f631f6f08b
eeb2e6b9dfccb9bf246f147d92f243d958888b33e26cd21b39147e73088386eefbee13008af4b99ee55ed70d0c942ed8dc370ba28d49c8f2264d7db3
bfc85c6cdc02dfdfc167ecb9bd5bbc08789adfa29137c6eeca46da3de5c5cb919263b5d08cab3749df1c4cb5216e23d652f196ecbfa71a948d8f24cd
caf5dbabd2c6ab3cb42a0e3dc314a9b2ecfe3b63b10b949846bad978b37c9b0af53b1fa73bda061b9d71f3ebc57753b10b08c5888c306bf07ed28d31
4b6325bc2c8361a9013ba070ecfe8bd21972021c0ad1386b5db7f474f5daaf411e8f3e50daaff5cdca52c6faa5e09d7f8c06f567ceaa337a51bd53e7
569febba6e6f57dc907b2c6163e5fdab50e776f6055ad7bdaaaec165757d6aba49a3242eff2783747a802e83f46056838a28fd6b1e84bc4c29a660fd
2c3d167d96be7ffd07920a9f1d
>}
\expandafter\gdef\csname EFWidth4\endcsname{306}
\expandafter\gdef\csname EFHeight4\endcsname{327.6}
\pdfobj reserveobjnum
\expandafter\xdef\csname EF5O1\endcsname{\the\pdflastobj}
\pdfobj reserveobjnum
\expandafter\xdef\csname EF5O2\endcsname{\the\pdflastobj}
\pdfobj reserveobjnum
\expandafter\xdef\csname EF5O3\endcsname{\the\pdflastobj}
\pdfobj reserveobjnum
\expandafter\xdef\csname EF5O4\endcsname{\the\pdflastobj}
\pdfobj reserveobjnum
\expandafter\xdef\csname EF5O5\endcsname{\the\pdflastobj}
\pdfobj reserveobjnum
\expandafter\xdef\csname EF5O6\endcsname{\the\pdflastobj}
\pdfobj reserveobjnum
\expandafter\xdef\csname EF5O7\endcsname{\the\pdflastobj}
\pdfobj reserveobjnum
\expandafter\xdef\csname EF5O8\endcsname{\the\pdflastobj}
\pdfobj reserveobjnum
\expandafter\xdef\csname EF5O9\endcsname{\the\pdflastobj}
\pdfobj reserveobjnum
\expandafter\xdef\csname EF5O10\endcsname{\the\pdflastobj}
\pdfobj reserveobjnum
\expandafter\xdef\csname EF5O11\endcsname{\the\pdflastobj}
\pdfobj reserveobjnum
\expandafter\xdef\csname EF5O12\endcsname{\the\pdflastobj}
\pdfobj reserveobjnum
\expandafter\xdef\csname EF5O13\endcsname{\the\pdflastobj}
\pdfobj reserveobjnum
\expandafter\xdef\csname EF5O14\endcsname{\the\pdflastobj}
\pdfobj reserveobjnum
\expandafter\xdef\csname EF5O15\endcsname{\the\pdflastobj}
\pdfobj reserveobjnum
\expandafter\xdef\csname EF5O16\endcsname{\the\pdflastobj}
\pdfobj reserveobjnum
\expandafter\xdef\csname EF5O17\endcsname{\the\pdflastobj}
\pdfobj reserveobjnum
\expandafter\xdef\csname EF5O18\endcsname{\the\pdflastobj}
\pdfobj reserveobjnum
\expandafter\xdef\csname EF5O19\endcsname{\the\pdflastobj}
\pdfobj reserveobjnum
\expandafter\xdef\csname EF5O20\endcsname{\the\pdflastobj}
\immediate\pdfobj useobjnum \csname EF5O1\endcsname {<< /F2 \csname EF5O2\endcsname\space 0 R /F1 \csname EF5O9\endcsname\space 0 R >>}
\immediate\pdfobj useobjnum \csname EF5O2\endcsname {<< /Type /Font /Subtype /Type0 /BaseFont /GCWXDV+DejaVuSans-Oblique /Encoding /Identity-H /DescendantFonts [ \csname EF5O3\endcsname\space 0 R ] /ToUnicode \csname EF5O8\endcsname\space 0 R >>}
\immediate\pdfobj useobjnum \csname EF5O3\endcsname {<< /Type /Font /Subtype /CIDFontType2 /BaseFont /GCWXDV+DejaVuSans-Oblique /CIDSystemInfo << /Registry <41646f6265> /Ordering <4964656e74697479> /Supplement 0 >> /FontDescriptor \csname EF5O4\endcsname\space 0 R /W \csname EF5O6\endcsname\space 0 R /CIDToGIDMap \csname EF5O7\endcsname\space 0 R >>}
\immediate\pdfobj useobjnum \csname EF5O4\endcsname {<< /Type /FontDescriptor /FontName /GCWXDV+DejaVuSans-Oblique /Flags 96 /FontBBox [ -1016 -351 1660 1068 ] /Ascent 929 /Descent -236 /CapHeight 0 /XHeight 0 /ItalicAngle 0 /StemV 0 /FontFile2 \csname EF5O5\endcsname\space 0 R /MaxWidth 634 >>}
\immediate\pdfobj useobjnum \csname EF5O5\endcsname stream attr{/Length1 3908 /Filter [/ASCIIHexDecode /FlateDecode]}{
789cb5567b7094d5153ff79eef7edf97cd66b3bb2c7998809bc7262004700381442a0b84100c8140164884956cb2793fc9833c3088250c05449e8d42
43a18a0ad4da95522714ebd05a076784693b4ad5a98e3ae2b49d42a533a8d394dcedf9360b52a6ced43fbc77efbde7775ef7dc7bcfdefb0103001b75
0a5817e72f2a8071a003b064e2c62d2e595e9abf716917e199849d8b4b572dc83c3eef75c25ec283cb4ba7bb1bceb72b84af125e5dd5ec6fd376e9ff
06e0f9840f566dec741e391acc20fc1ee934d5b4d536ffa5e19fbfa3c90cf9ae5a7f471b685441f90361736d536f4de1489f21ff00002bebaafd016d
f56fed00ba87e43975c4309f141708f7104eaf6beeec49fe3dfb1ee12123bea6d62a3f3fcfae117e95b0bdd9dfd386cd6a1461c3bfb3c5df5c3d2d67
613be1eb1493b5adb5a3b3c6d37711209ac430a7adbdba6dff3e9e4a780ec56006636fcc30563821a4580d9ed14c900c1381e5172cf5d28e29864a28
14d11d93ef8783e008cbfbfcedfe4a18f0b737b7c04065bbbf1e06aafc2d1dd4d755b753dfdbde0403b5d5ad44d7b65737c2409dbf8574eaaa2b89d3
e86ff1c34093bfd569f49de4a1d9df5907032d8d06a7b5d6df0c03ed5d2da4d959d3524b7d9de1ffae8822712916b6170480c8168720874d34c6d022
7c176a38ed328f56115189e686ede3771a96d42c0ad0aa1af18cea900e76586b669f46d6395630d292c32b07f0126261acc09c304fa71de4d0082dd0
8767c271dda64387424f859eb8597a87bf5bfbee892025ac3fd60cba25d2682dd017692a857066ac1939c5d2d8741886b7a8fd064ec1057e018ec16c
aa2be135b6836791e479d80eef0b0e676088e53207cb25e9dbaa43ed15dbc40992af23db42f2f23e1c209f86a761788c6fe225500317c425384cb535
ccbf0ebf625be1323c0d6ff142f812b6a217f6503d0c1d0a88cbcc04924f81e3c64c60e46a25b50998252e87eb75780c36d19e1d578755077b3b1cf5
f3ec75fa57fd8d627e1bd7e106b21882136c9b92a69c500a61cf58bc58017bf8567658a908d74d7486dd30a454b053aa03de3062254e09455a03af52
eb864bec01b60d7750649b8c08c465b8a43da44c1f8b4aebc759b41ea0f6129c862c1c24fbf05ad41a18e23534d79714c925cc87c9e4ad03c041b700
8c7f45150a7206539dd620772d09043d2bca9c6f96a7644dbd0b3aad9a330825c1985ee770285452a62489f2a0480ea24b0f2aaeb44fbe49f849d6d4
a2923267f0178bf2235e1755e413afb48c4803119bf88bf2b3c6fea9946b946146061652ba7c467b61824f3c85aad0744d0845e8ba869315a1f0c98c
138e1260d24906262e909bc0c428ad4c79bae0889b615d941aa59b54c19942a9e732a5475b6f5cb9688bcfb5e7e6de0fd3e75ebd129f3b8342d1ace2
9a66d5af69f48b50d6afc7f2d4d3bac298cff3a0e0713c879772aff0ea75a25edfc20f70931975254fcdd33ce6322c136bb4357abd68521bb45eec55
7bb56dfc49dca1eed40ea3c3073e5614b47bcb7e0dd1a1f3b4c268aa9e39b6d973ca5936b2ec284a7b4c63ddc341defdf9a887f38f646a0365d5cd2e
f6eee85f478779dae887917dd94ef992082f798a758d0e50a0203e178c418cb0c45a4cd13cca2c622de3ecdc61a3d12a62632d7966b4018b5aeae88b
dd1cc362622d106f35abc8852d9665c6c7de6749bfc77ae3f4510bf35d79e7ccfaf879f1dc77d5961b9f6bcbb5d9c39b75e506612ab7374c5ca39d8c
f4e1ed1a6bf1b9e5a9afcc489a97b4370999ef744512f3b1d3478d1ed3c63dc866cdcc484b5535a2b3b50759b63b6ebc4355bab69edd1a656a587ff0
4f299b86df1f7615bfd77df0d9287e7c741d3f32bd70bcf791374e8e1e502a82fe8663fbd257d205111aa15d28532a200a3a3c13f4291c394ea114c9
d35b6133b68a289df21a5413534dd6778a82666f51d0e25d5b14b41a9dcdbbb6ec2ca8a1f373ca8d64a0c55dbd71f5eb3c089f780ffd297c1ec736dc
aeedd4f7e1215df4f06e6d17dfaa293ee61b97a2b1142d65364bf90736a4dcfc251f9e70b323058b461b958a1337ff7ce004ba28978b439f2a5d7443
24c3340a99964d8715171f875c53d3523333c6e5ccce999d31db9ee28e67aa46353e2edb4dac9ccc8ccc0c654dd9e7a3f17608bd59d19933f38513d7
3f6f6cccc93cb960e5eee459ecec870da3af2e2968ca9c34cff3c3034dcdae8cc2e247e7095b41d7ba19aeb6edf2eff2aba1234786d8b88bcf79d297
ef66d6023969c2920d17df6aefe8ec7ced42456db47aebade3e587123e70df581f3bf70bb8570f5fd87fdc6cf9e8d6f8d5c7a3d93105faaef03ff3f6
5b44f74db39c0010b3edab8fff7536a60002f4ed7167d115e3de32dc9fa2b68bceea192814762854b24323ca67507c87ee132c96f9d9cb9178745c06
53a08ede100e567a43e835400b8f89bcd31a3c6cdc0d0a7d17b019e177cba019c4111aa3395858418446c860de08addc410b48607d115a8574b61f16
d2bdd906bdd00ef5504bb3778213264115dd954e70c30caad9445592861316904e27dda09da45d0d7e6886a9c45d426f5a159db313e643135527bd55
b77c758451358dd564b391fa00699afe8f59736ecfeaa59936d25c0d64d342da461c7eb2f97633e613d54076aba18b34aa48d71ff6561db6f08757e4
242f2dd4b7914e25f9ad273d27d9b7d2ecfeb0ec6e3fa5612f1db03ca2bf81b8d5dfa0e3bc4b6b7538c20ec2ade159dd146736ccfa2feb5bb65977d9
72fa4afa82da66ca8aff55f4704e713ae965b0822e5f2e02d54d9dfe61bec5137a5962d0853f77e34b83f8330bbed861112fbaf1a7124fb9f0a4054f
b8f085417c7e049f1bc1e3129fcdc36724fec48dc78e968a638378b478be385a8a3f76e311070e0de28f4c7858e2213b3edd8f4f9dc341890749e360
3f1e90b87fdf62b1bf1ff72dc6bd7b92c45e897b92f04989bb253e217197c49d3b268a9d12774cc41fb871bbc48138dc2af1fb121f97b845e2631237
4bec2f7289fe003e2a71930dfb7acf893e89bd3d3ed17b0e7bb7283ddd2ed1e3c31e8fd2edc28d12bb06b133801d166cdfe012ed01dcd066171b5cd8
66c7560aab75045b3c2189cd129b2436c661437d9e6808603dcd519f8775cba2455d02d6d65844ad1b6b2c581dc000990506b14a62a5df2c2a25facd
58b13e51540470fd2356b13e111fb1a2cf84ebd6c6887512d7c6e0c364f1f02096975944f9242cb3e09a115cbdea9c582d7195d727569dc3555b146f
a94b787de8f528a52e5c297145c934b14262c9345c4e412c77e0b2682ca6a88ae7e3521a964a2c7ac8268a5cf8900d97482c5c6c13851217dbb040e2
2289f912172ee8170b252ee8c7f9123d23386f041f1cc1b9390bc45c890fbc897944e59562aef4b4e19c7e9c4d3047c912390b7096c49912b3f3d03d
8233cc385d6296c4a912a79078cafd789f1527a3554c4ec349133133c3223203986141173309971bd3cd0922bd1fd3449e4893984a28f51ca6907e4a
123aef8d16ce58a42fa3f39ec3cabdd138310a277a9409564c26f5e4414c1ac47b125de29e002626d845a20b13ec181fe712f1f331ce85e3253a248e
1b41bb2d51d825dac8ab2d11ad1263255ac88365106368c2987e34479b853901a3cd6892a893481f4495d45589825621f25021a464215ae99bcf2478
023213328f02c9c8865960db6e36e5bb2df01dfbffb665c27f0073675542
>}
\immediate\pdfobj useobjnum \csname EF5O6\endcsname {[ 107 [ 579 ] 110 [ 634 ] 122 [ 525 ] 948 [ 612 ] ]}
\immediate\pdfobj useobjnum \csname EF5O7\endcsname stream attr{ /Filter [/ASCIIHexDecode /FlateDecode]}{
789c6360189680054cb2e29065a39f4346c1281805a360148c70c00e00693c0017
>}
\immediate\pdfobj useobjnum \csname EF5O8\endcsname stream attr{ /Filter [/ASCIIHexDecode /FlateDecode]}{
789c5d50bb6ac43010ecf5155b5e8a43779723698c215c1a1779102795b94296c646104b42960bff7df4b838108134ecec8e7667f9a5796e8c0ec4df
bd952d020dda288fd92e5e827a8cdab0e3899496e116e5574ec2311ec5ed3a074c8d192cab2ae21f313907bfd2ee49d91e778c88f89b57f0da8cb4fb
bab4856a17e7be31c1043ab0ba2685217ef722dcab98403c8bf78d8a791dd67d94fd557cae0e74caf1b18c24adc2ec8484176604ab0ef1d4540df1d4
0c46fdcb9f8baa1fb6f2873e9617e80a5e338d42e346a3d08f22d309ba8289beefcf89ced015bca6e6bf6dd21c69699b49b9781ffde5cd6663c99236
d896efac4baa747f00837d8403
>}
\immediate\pdfobj useobjnum \csname EF5O9\endcsname {<< /Type /Font /Subtype /Type0 /BaseFont /BMQQDV+DejaVuSans /Encoding /Identity-H /DescendantFonts [ \csname EF5O10\endcsname\space 0 R ] /ToUnicode \csname EF5O15\endcsname\space 0 R >>}
\immediate\pdfobj useobjnum \csname EF5O10\endcsname {<< /Type /Font /Subtype /CIDFontType2 /BaseFont /BMQQDV+DejaVuSans /CIDSystemInfo << /Registry <41646f6265> /Ordering <4964656e74697479> /Supplement 0 >> /FontDescriptor \csname EF5O11\endcsname\space 0 R /W \csname EF5O13\endcsname\space 0 R /CIDToGIDMap \csname EF5O14\endcsname\space 0 R >>}
\immediate\pdfobj useobjnum \csname EF5O11\endcsname {<< /Type /FontDescriptor /FontName /BMQQDV+DejaVuSans /Flags 32 /FontBBox [ -1021 -463 1794 1233 ] /Ascent 929 /Descent -236 /CapHeight 0 /XHeight 0 /ItalicAngle 0 /StemV 0 /FontFile2 \csname EF5O12\endcsname\space 0 R /MaxWidth 974 >>}
\immediate\pdfobj useobjnum \csname EF5O12\endcsname stream attr{/Length1 13300 /Filter [/ASCIIHexDecode /FlateDecode]}{
789cd57a797c1455bae877eaabaadebbab3bdd593b49279d85b025260408a0b4c80e62908080a209844d8144022a06278043420486201004115a0444
402622420291018928226e80338e3a2aa0b844646650c7909cdcef54272c3af7dd797fdcdffbbdae9c3a4b9de5db974a01030017dd64f00dea3f6020
0c81a900ac338d7a07e5de312a36d8e17dea0f0690fe3668d4e87ed5c9db7fa4c91de8f9e5db6fcd1bec5f947f0f80e2a53973ee18959e39fd97a2df
03a855f47ccca49905c552b5ad90fa5f8b3d263d34c707d36373008c19d4e7538aa7ce7cb0db43f70398a90f3ba716941483812e307f427debd419f3
a6fc75df343ff5697d7cc6b4c90585c609af9503748ca4e7dda7d1806d93e139eae7523f69dacc398f5c7eddbb9cfa73a8ffd71945930adefae0834c
804ee279f1cc82478ae52dea83d43f4c7ddfac82999353bfbfb99efa9f133ca78b8b4ae60c0fffa4174097bdf4fc2fc5b32717f736fc9d9a5d97118e
d340d0ca0aa19f4417c24d349609dda92d9e99a12bf401a9ffc0e179609f5130671644125de9d7da0a70b52566b20726cf9e0546d1d28b4c3b88da08
92641433a524e91e3acb43a339faba022a9fe8d7fe3608a8df7e6f5dd23a979ae9a2cfbfbafee98d6d7d6e7f2a8fb5f5b7e8777dc7d653fa09fbc519
d7ad4e81fff0272014ebf88b9c76e5a19d8fe8106de1a7f88b541fa772e43fddef3f3af3936bf7ffcd5feb29a2caa96be789fefff2918cb8ef823070
930c84430c78c10f49c48d8ed0899e854114dd8504ca575728a0deb08341972f535bcf0c165ae182687d5d68a542e3aa3ecf44cfc50c21df66b0819d
5a0ed0c0d90603b4c9f746d8056e5dbe1f2d985d3011aa0a66cf9c05551367174c87aa4905b34ae83e6df26cbacf9b3d03aaa64e2ea2f6d4d9931f80
aa6905b368ceb4c91369e4818259055035a3a0c827eea427bf9f5930671a54cd7a408c144d2d980955b3e7cea29973a6cc9a4af76962ffff4697746a
cd983eb5e0067d9221a44f8cb453d40a61e6262a2642071a416a2540b65ec74337a2470264d13d9eb459a2314697b70de7cfe16bc8853e56307c2a8e
9216d1f1f944ee33a1a345fbfa5f7b9f4d6b6b8ffe0f984df4961efa9fe7c1f8d05c515f6d5fb7c70de3e3afebcfbed6967a025c5941ed3e57977620
4e47ea362c0e06c22c5d428474082a38af4a189275a3b9266169bad09a0ea6e5f4bc33b5c0540d70ddbc9baececbb83a2ffd37f324da9be629c2f239
c5534540c0643bab1292a96429eb08f2b8508d7f8629928b16595444a32c4921fe5ef7cb9d32a01002e08379aa9bbbd97ac34c76ee8639d856c85711
570136508fe97d1916e8b5462356da218d742c076e865b6138713e0fee87627808e6e992e723483b115eed4f474101cc80d9e269eb39b20c1fb7feb9
f574ebfbad6fb436b4bed6baaff595d6bdad7b6e84f437bf904fd9d5d6b3e9a7840aeadc01822844b74e6d45d5691bf20de42f099a501114bdb5ad10
b508c2500913046a2b429b4751c9a322fc8bb0dbe4832182ca0c2ac55412a808991132293cc0bcb6d28965432d9ca0eb08ec800d6c1bf5a6d0f88334
1294f6c062984b2347d909562975a1b16d70094ed1cc0a38813b64604349cb4ed0fc8f88d797591eeca53d72989be5185452d911f25ef94eb956be20
9f841e72897c52ce974b58166e56c628dba8e4e0eb2403c7494b6bd967500207f01bccc27ab9bf6c87cff024ee802fe91441a713b002b64029c1e266
455026954a77d2c81bca49584f57113d3fc936b25304dd01f6389c81a7509606c3467686f03a013fc1e398279511f9b3a42904ff1bb4d7495abf1e4a
482dce303370a9138d11f474d644fd1e8b5d9433fa7509cae8e43cd8a2d6aa6e839f4e1114dbc68eb246751504e114de830fe2c76cb1ec97b7cb8361
458802980f2b68eff5628d3a85cd23dcc5552a76971e96f3d90ef846ce374ca4bd5f1718d1997ba53b09a329504fe56155239c7ab3c55849908aa7b1
70d230544ea7f5b483e131c21aa008b389d745f47c37ec812e580d2b68271d5fb587f213addc207f4138af60cba59fe024f627c99b225f245a0bb121
eddd6f50151925d27b9f5623250f29ac098c1ceb7b735c4297cebfeafa3483af06726b6cf37cb5adadb963e518655c8de2adc164638d9cecffe2bf7b
f84597cec372c7fa6a5a06f46fdb75407e7f1a1b35969aa247c3343ea0bffe4c1c5aa324d3df90fc1adfa469be27b427fcbd9ed026f7ea12b2311451
91de0a4dafe1ff904a5517e9578f80437d0ad6da6d0640970a6166bbf6c9b09ab0bcb175606e3ddc73dcb01a87de8640cf71e7331b9d39393781d6dc
98c154c9e37645f853a4ec6eae1e5269f9a2c71707abd7ac5eabbabee2b75cb8c07b7ff91d3bf6f967aca191cedb42e715e9e7c5071c06719e8181c5
25871981ceeb73f9dabe6159e12e8f5b32f8bbbbb2bb495b68cb35d5c1c58f3faeba1a799fcf3ee7bdbefb92bd7ee1027b8d76fd882ce3c78a4cbeda
1f7082896d30a02279648830ab1ea355fba4b94f739fc69b20fd34d50d19cc99e04970fa9d09d9094edc2d756939b5bde594d445915b4eed108d1da4
a51254b49e935790dc5ac806f803616ad00541eb4ad7b24893d711875e4f4c2441799976d5ce5f6ed42e66b044c9a9b9b2325d4e4d4acd04a706fe44
7197966e78e619fa7be6992bccc47fbe7285ffcc4c4a2e3fc9dfa6729265d1d58d650579092fe715bc842d67f3d8a36cb9b0d25f90811b4f3e95881e
f0f4c3a02c05958506089a8cf1aa17219e59b4d36d7c61822f8d0dcd02cdc6cccb8da71b33483ac625b2bd0e74c8d2841e094e253b394ba0ced950be
8e4d7e8b0d6ddeb2432e195c3bb8e9cc0eda8024591e4a187b616320352a3a0623bd4e22aa5351e47edab3ced5b6a07ba54c160d34b3c4ccde080dd5
58ad79588d276f584d78deddc36adc7977132428a4a5e174e3e1c34e574e1b349775680c9af2bd41f99ed5783567440ec116c81c2d8f51c6181e951f
551e8aa9883290bd8b92a349f0bd73e021756e7449cc1cef22288f5a14bd286691773b6c8f714e8009c984447677e8710bcbee96e24f540dd9b7b0ac
4cd9e3560d2a90913dd23c9cc8985570fbf3e5f79d7ae4d1d363bf66ee017747f1cb3b76ec7898adec3573ed9087abfbddf6f64d995fbf76cfd6e258
fe1d61bf81f85d42d87780e24057f08499cb4df1e5beb0a0c71634ad52bd41df2aff4a7599e7b9b4706f18a03bca9be2d3bce88e37a9698208e179ed
f89b74fc890024cc11ba34379ebf7cbe51fbeaa2a65f44950c163015c615c417f80a136498c0e298c72d2724a6a466c71122dd092bf229a1c60de861
df95cff1f7f8d7f7be717fde9b330fbd51b775f7be351b9f7b6ad4a1d925c7c77dc5ac7fc0e4f886aa4fff919c7cf4a6ccea15bf5fb3ede1e292d2a4
94bd3edffb7be6ef14ba4f99a7bc85644a221d5c18886536b401a2ad1fa0c51054182e3431ab19bcaa51b6ea16c04288d974c4ac02b1d37d1a1a339d
82afe74ff769cc245c74c6cac789b9c7054b3b5a281c1f0ce3603a3c0c4f80219c758214d609bbb311ec0eeb1db6316c0a9bcb1ec5c5cc46ac34b104
cc726679fcba3ea2ca25c6b3f99933c75bee55929bcfe1c9e6aced3cc8f28f12873612870a09f258b837e097a30dce722d363a687007b54a9b148485
b665862d71115e66462f9835354e6b66d7f345bbce8a69425b88455ac345a1c04283893dbc21c49d30611f04cdc1e3861bd822b8f12946b5043b8fed
dcc492f869fec3bd47a78d3ffcc08b6fbdf5e2c867f394333bf8930e07bff8eddff98f3edf899b32f66dd8b02f2985a8bd82a0afd6ed49128c0d2485
a9602bb742305c0d7ac3b76a416b65e24aefb2646ba2c91b1517e6c584f8986432302444e7751373bef9fc35f109b829da6027a59378523ea19c5009
ef3d71d204368125aa1e77780856e6e9cafc8912b623e2f709739490192e6d59b269d3122acc34fce9e16f9e72f4def3c0174ce197cef2167e91e5b2
98e14f63ef039b9f3d78f0d9cd07a479b54929fc1ffc87bb26f01fbefb8a7fab1ba8896c6b9cb050db499aa6114f54981488549c124ae894c95e28c4
0f5490c90c5483d6fc76836ed7d3afb30354485c0483c6be4ae15a803633808178e4ecd1735cc03556622a462b39ca60652ad6408d6a206921c6303f
4bd88e875bce9e62bc254b3933a669a1d249c4844b89be4b75fafa29aebe2d901c49d44d5583715d82ae9571cb529fcb88b42675f47a92bc0e13596f
32e18e84980cadb9a1f17243a34ed8765dd57b39a4a4d71133b92bd99aa4accc7061647475f527266577eb1ed63e8124435a5ab5756b55d5b6ad7ceb
a295d0fab7cff8ca854f3ec77ffef967fef396c12b1f5fb46ad5a2c7574aafafafa858ff7479c5fa31be3d0b5e7eefbd9717ecf1251e5bf1d1d75f7f
b4e2182b98b368d11c2a24310b09a30ac228529718bf213e8a954354d0bc550e4265787c505b19be2cd9e0f52684c54162a2d7a60b0c81dfee93bee2
3fb6cb4b7843d46bd187630e7b0fc7be16d7106fd8e1aa777de34292981eba6cbbc2ec242b90dd0db242529298c2dad1221a7c317cc33092935e7b66
7cceaf30ed2c43e6e42ff12f876f60b7b4c9523c4909b331d7987b98e3bbaf58b8eece36f1bbe3a4b5ed9224accf472c5ff9183713bf0ce00dd86083
a46e908d0a93c16fd4c8396712eca785f3123a282e0a8a9bb828b879fba51db48326de3ae9f6cb0caf056e052725b68a2c092fc59c66348393f27834
1bc875a962d0e444b3513c200b675823ec9b623252642de2359362260fda102144b3cff9d38dedd64c7753572be3f7ed022bdae312f7f8ac8c4d08f4
733087e430388c0e184b3941312c0393811925154d72388b92c6b0b152ae752a9b263dc21e92e6e36cf961c323c60ab6445a607d4a5a87d57244c800
5200e0c704f44bf5fca294cc4bbf94723e58d272df92338abd250a77377562657c21a9c771d2b5fb097305920356690d2c9499177b81ac0a3b210c72
6346c09461c8352cc005b2cc2684e9defef85bd29f9bef53ce081f2fe97b2ca73d8c94fdf40a44b2351aac312d74696623258f4a94adaf13bc26d94d
3b66360bf7a5eb2cb123607178e23d7d3df779fee851c8ce38dba43e39413826b91373fad92abe7cfdfae5bc277bf30a63bcf50a7f4b496f79f7c98a
f227b79dfbf8d3b32ddb090711556d26994e85df05fad8ac92dd22c5c5c7194d92c12cc5c7c7f5335be2e2650f03cfb3eed5916b9cf21a589dbcd2b9
ac439cd9121f6380c498287b1743943bb183f6494323493b49778eaeb7ba8bbdf8d345ed18411de2a2fd7b6ab655c435471a716d5f7c5a7ada1d6924
f9ba7514c21dff6fdc6f3a6bd7777970c9dbf76d7df9e16d8f9efd33ff945fb8ff8705a58db35facaf585f7af62d16f1e3f4bf2a5b5eefd17dc14393
26c74775fa68df479f67a4bf3760e092dfcd9a1f1fd9e5f0ce63e7295564ad4d44f76f88ee06181ab0ab21e605c84f0514a3769a6cbb4ee7cc0cb287
66610f8dba3d3482b1dd1e8681291e34a649f106cd1430159b36994c13b02d8e55e51f5a2e9e68b948dea7e98cb0861294129dbb509e648664a8a768
2ede1261b2c30b116a9ddde92b8f3fe0adf3d73a9745582102236d26a3251e8dee0129c4f6b74f376666868c75c3f9cbcd64468ee9b6d19923ecc8ac
8cd88cb88cf80c5f46424662dfd4406c202e101ff005120289b9b1b971b9f1b9bedc84dcc4dcd4e2d4c5b1157115f115be8a84c58955a9c1d44ba971
ed4bdb17b52fc88fcb8fcff7e52714c715c717fb8a1316c42d885fe05b901079bd47bb99f570fab385994a21bb9b95707d6c142e1dfa6cd7c2a27575
b5b57deb97ec3ad1728549cfafcddf9737f9d0f87f5e92b2a6944e2cf9686fdaf096853ba6141cd9fcea6157d9d2ae5d77a4a6360b5a1d205a6d51dd
e439bcd0331085755687a92ed2b3cc511bb3360a5cae419156d5183d30562845e6653dc03e2f7cc6b18b19fbf2e316c405e390e06cb795042ad35d2e
e501046b2a853659f8e5f34f3ef9bc282d7fe8f552e9dbd0dafa76e94bbdeaeaa4f413172e9ca022dd5958c0ebf9bfe8aa2f28dc4ed03078b0f51c5e
201e4641df400c94b325b2bddcb6c45ce794eb228879d106970d06bb07446bcd948fb5d97c7ef9a2f6e345a1ae315acc8298aa98608c50d7abc24ed0
f5f0e8b63e140be08511cfe4be7cecd8cbb9cf8cb87deb8416fe21ebc2d4d19be5ec5d9d3a9d3b79f25ca74e3b929208213b73b15e7ea21641258f27
f8b410b5a2ebc0eeae538ccbecb56c2d46c8609406395d9601b17a76949979955a0d3750cb99d5c64c4a9a8018c8aef3a3b8b9b6b6d74bf34fb442eb
89f92fb5bc4174dbbe9d6887fba47b7f69dc5e58c0fa33235dfd0bb8a78d7c6d709511b5dc1043117c127898a9dcb844f1bcc0943a2b3b1859e7aab5
2ef3c67824a3c708c3249763805707b141cf5204f14201d6e590c74ceb1b5b1c1b8c7d2ff652acd217fab2be525f4fdf18a5b321dd986eea6c2e8222
562415798a624c131e14044ed04d874e5bdd8392001874a21be4b2e63dd693fbef7f63e2a4f71ee097f91b2cadf92c33d44a5b97acafb34bf78e3ff4
46b76ebb3b76663d999985b1dbf8a70d6bf7eede08fafb4190fe45b40e837101afa231abf105955550d2acd69ba5300a974c8ad1e6b00c778b0cc42c
025b8b086c87d5d8f5b60878fb3490576d70e90a7d9e6cba769162768a6df6053cb99ea0471842023296854c893f3b4b2897f4af9a49b7b374fe7e5d
4dcdee5755f7badc69935634a7e3fb2b461cdc2968cdc7c8e389d616ca968606fc51d65893ab3c2cbcce817529fedad47a539de3d5e8d89428305a07
a92e976f409a1e6785c4a1e17c4820f81941e91c928a8e0b3a063bfe4a872234e99a45be99b5890ae5d7e111d959b879eb9ad55bb7ae5eb3b596f3a6
825d23476ebcf395bd397be6bfd3dcfccefc3d39b5d2cd6f7ef2c99b6f7cf2c977fc2cff2636eee5ce1d5ffdd3dd9326b25e4cc4a4bd264eda21e84b
21895ca8d3b71be9bd09d0ced40abbb3d6bad6cc24238c109671a0ee0c75b5ef234245a78bfccb9e7c0f1352ec778640a646969e3184cb85b5f3e7af
d95557d7efe5b9478e495b5aee91366eda78684b4b85ea6ed938b9f00761718ed0e1f3e85c110575227f70487e09ea258519651828c2a04641a7f3cd
c2ed0a6b9f6bca278bafe81edd2f92a523b5f493f3af0455f737b45febc77c8cbe9f051cd03fe0b54806b01fb21a2a9457a1defa9266d414f50e1b33
5a61a0a6ef7e3ec7752dd6d59940073903ce5c67beb3d8193ac8adb6e5a0a1039f7b65e04dd387eba72efbf0f08682756a876f8882d73049dc0f6b25
668481b2167a099111b0694a40c955f29562e592a286c027d055f72f8d6dd437c4921425c2f8408aea32453a408d3578ac15b13eac8da98fd20ce074
188d6aaed3e8c8f5469209f6eb7149334526fa9b853e7dce5fd6d351c192405846526e527152555290ae3f257d96d49a64221ee95cf15ccfa96b2cf3
84589636e0f0a23f1eaa9b3d77c5b6bad90f2fdf5657d7b766dea33bb172fe433f9e150c7c768360a0b471f3d37f7aaea542cedf3d75e2fcabf24318
8441f71be5a7fedfcbcff976f9d99bef79d723fd5a823cff8304d1c1428042b66eaeae7f11a47f616a9d0beaacb5e2dd94cb31125d9e01bf7a3715f0
f78d2a8552b5cc50662c339599cb2ca5d6325b99bdcc51a695394b5dc1a84b51ce1bb3c71b5e6195acdeb573cdaa5dbb565d622e7ef1d2dff90fcc89
9f5d387efcc2d76fbef1cd06fe266fe4df9361cb21fbe5663d854f250bb18520145ee296404cbb97a8b52f63af627d2c798841baafb8ceab521adeee
2802a690a7f83c8ec2d8e4aba46973a937b8da92baba6b1e55ead9ee67b7b7ec56cd3baef3a9ecbb765711e21b0e25e89c40a9b46a2139b36085bdd6
546f30ab14740d740973a5eb20f987d36f0b87b037376c5398e058c8935e6357040e8d1fd279c3f304c781c5615dbdb8d7e53c71a8650f316bca2445
fcbfa888fcf81b745a2a5c688b7947b585bca3ae85bce4df2b6577b9a73252f8f7e4da6b31ef9d3146bbc1e84e1cd0414075fa869897ace88fc2e1bb
6e8c79db435e481589ca4cafd96bf15abb92dbea6ce96ced6dea6dee6de96db5f8c0c792a40ee60e968e61e9ee744fc7f00e711de2d37c690949a9e5
e6724bb9b5dc26fe5fc2244935ab16b4a20dede8400da3301a63d02bc79a52d3d3faa6dd975696b620ad2a2d9876292d92729a077f1d5cabfedf06d7
dd8976b874c4f6f195951357f76dd8faf35fc61f9d31e558c1a265937706763ef5f93b53f6ca7d7777e8909717189260efb8ae72c33ebfff5076f6b8
91c372931d496b166ddca5bf81e84106f01fca46d2418a46ec8ad1812f504e586fac305b88c624639acb2e7450778499ba1f6c0c590d327b7bfe18b2
e3c2fbb9c37b0b5f98922dbca0933dcc4af9e26125afbe7a66734585b291bfb6a225583962fda60fa4fc15ec1661c57793168ed5b5df0dbd03de6bfa
bfcccceaddb556d27eb76504d981811ea18e3921893a9f79d50814790e0b231046be23a47657a3a214b65b1881176b6b6f7b69ee9137d9bbec80b4ad
a560d3a6435ba4d22bc15d53265dc2ed02fb9bc90295c9f9a0c29540aafeea4572324911154a2aa8cc09a0f69310fea4a80aa5bc8a0c06f1ce580f16
20142cb8f3c43b5cf18a0cf4579711a1b7b6d772613d0336b6bfb209fc61b074bf542a9549e5d20269a5b445328a834c94f67a28ea8ac66839055258
1aa6c93e633664b35ed84bce300e84816c080e91072a83d580710c8c61e3709c9c6b9c0253d8749c2e4f55a6a9f9c6b930879562a93c5779545d0c8b
592556ca954ab95a0dd56cadb41e9f929f52d6aadb95e7d51ae361e367c656e32de2bd50964924d1371f65f7b27b8ff27b9ae4fce63cdc7525282844
964850c8c196066e33182593131c66a7c50ce0b03b1de0b039ad361095dd66b698ad4e8bc5dccf66316960512af055bba55eb3dbac66938a6074c80e
8b16a2deb01aa34e314b3b09dbe3ad86063d6bd688cd8d6d59d3bf21a35e2bdf47640a7a5e524131aa26b4859b236c9acd6fcbb60d31df611e611b6f
1a6fbedf5c615b605b6573998180b028568bdde288601e49933525c2ecb6b8add1f668472a24912efb649f9266ec604a36275992aca9b68ef68e0e9f
b307f1205bca9033949ee6ee96eed69eb61c7b8e23c3792b0458400a60400e280135600818fb99069807d986d8873802ce3c18c9464aa33157ce55c6
a8a30d638c7799ee328fb68cb68eb38f73e43aa7b029d234f374fb7447beb3d4f888fd1147253c615a6c596cadb455da2b1deb4c6b2c6bacebedeb1d
5b2c5bac3bed3b1d35ce779d9f395b9d9389638a9d85c28bbe8c09e649ab46ac9ebf6ac6f0bcac04de3bc4c6696f3eba7e70799e3ca27935ce10b25e
d3d28cb98627c51719cc9fad85b9c2b20049dc4b8e060f6c0bf2e6f1652dcddfe25a7656ca60d8f24f5ede72b1f9fbd03a5662a814ff37161148cdd1
a386ca9f4af41db99b76acd477f4880dfd126d98cd4aca569d08be571534547edbbc95dfc5dd7c321bc02e4a2998f9edaf21c9d648adb25c612ea747
5a5e369e3707b71d080a40c224179b27d9794bcb7b3cae79fab7644797ebeb2ac57fad85e7cfc05c1d14fd5b80eb61c9f2d0a6e8179b4acb8f06abde
0b9e5855268039d9f23177f183ec69d6c85ec0bbdb7204650745604930381096a2a704d684485b9cd1694dd0dcc393f5b76ac2f669e29f5f0d0d3741
c069b2395f7049d11510b9568d77d55b1ce97dbecacce47d2e66527a909971434a702d2dd0cda3413c10d65bd9d19e2370839e26ecae99949ac27eb9
215f68cf19d675e8306d5228774817ff6552c5173c517067c0193310228ce10eb76c3462b8591d1e7d0d5e2efe67177019c9a06b15f6c843e12fd9d7
9aa05e6102da8b5c7f1b9e49916b2b25b95594ec6a7aaafb5b900962a2a85f1e1a82f4c597eb04e4bfd4d58908ba1dc6fd7f1440b33ddfb4d1b40dc6
1e0147c4408acbcd56a351935df6e1e102be1078023acabe5e30c99492394df5364900c675a8984eba5f2756d25009eaf8fe6bb995088595d25f6557
74bada4ca7a741178a8c23d30746743276d4623cc6e88e2688578d4971a6c494e15daf11aa2153dcf57f71062262e2fd2f2439294bec72a8e34b1aac
0d3724d547c52688603933535825ad3193fe425c6ee3668fee3dae92aa9de7d7a5830a514fa48482bd777953473c4eacbe5dda23c8d9c67b246a12db
435c1eed0d4b15c46c276e3b6a523b6d75ec52604420bcc3409b510b8f741b3593f89753428c29de3f3cf53acc74c4743188f4fa5e48704a15d694b5
1e4342bd233a2e84d2e53ebfc5a77bd6afc8ffabdc3624a7d7f3a20d8fab38ecbc9e2f577923be5891c6ad7bb143f3f9fb1c7d7e8478a3fe71cafbbf
b39f6faf7ffeb079b87d9c497c0967bcfa2d0bad33cce4b10076fef3874d23ede37ef3ed4b77f9a4fecd084894984a4b293e8d871a2a5bd4bdf09112
011554bea0524d650395422a1ba9aca0b29dca522a0bf5b9f34093e7c271f9141c57dc50219f6b6d922f4029b50fc853e041aa1f941be141e97d4817
6dc50507a41c3822e5b47e2c6ac3093820c6e42ff5b9077028b53b4111faa1078def96ebe166bd00d4dc58a4e5a2561ea77dcda162a8276ea7df8067
4f9804cbe15da6b2c5ec17698cf40afa7032be23f794c7c91fcaff527a2a7b95467591fa8e21ca30d9f0ae31c658663c683c67f2996e316d327d6cee
6c2e35efb11458365a3eb73e681b639b6d2bb33d6ddb6f7bddf6853dc7bed4fe9623c7f1944edd1e98079d601a58f5ff19ac13dc903d5238e85f9252
de3d5e7cff209b883919fa3745a2cd209c7aa1b6044636b0ad8dd78dcbd7b515886423dada2ab8d914b80d8aa018e6c16c980e53e9f439fab7519348
977d9009197465516b22cdf0413f9a33074aa8cc86c9500033a1338d0e815934bf2bb56e851974f9e0ceab7b95e8bdc9544fa6350fd1bd90669aff83
53bb5f3d358f4e7a88ce125fbccca2d9028e025af37f77627f6add4febc6c05c9a3189e616e8bb4dd65714e818f9689759742fa6391369dfe934cf47
eb8be8f402fdd9aff719a5ef52421015d1f5008d8a534b686e91be53269d9d05d937ac6a5f238504acf577fa377dbffdf5d0f5507ce3692199b05ffd
8a340222c91345eb5f93fac88ba6d309dd69767f18000361100c268e0c8561301c6e873b201746123d46c1683afb2e180be3e01e26e9ef9414926903
33c26138c24ccc6c983b6b7a66564ebfb6fab6b6ba7f5b3da0ad1e18aa6fcd10f5c07e19ed75565bddad565a1068bdc2b1c98dbf24e3bf32f1e76afc
c98e3f72bcccf19fc9f80f3bfebd1a2f25e30f4fdcaafcc0f162357e5f8d8d4df85d137ecbf19b5ef8753fbcc0f1ab4cfcf2fc28e5cb6a3c4f13cf8f
c27367d395734d78361dbfe0f839c7cf32f16f6efcb41a3fe1f8b10bfffa187e7410ffc2f1439afee16378e6f420e5cc63787a109efa204639c5f183
187c9fe37b1cdfe5f80ec793d5f8f68938e56d8e27e2f0ad4c3ccef1d862a772cc8baf876303c7a31c5fe37884e3618e7fe27888e3ab1ceb391ee478
c08975e5c94a1dc7dafd07955a8efbf74d50f61fc4fd0be47daf242bfb26045a715f407e2519f7727cb91af7707c89630dc73f72dc5d882fda71d7ce
64655721eedce1527626e30e17be4040bfd084db393ecf711bc7ad2edcc2f1b9cd76e5b94cdc6cc7670b31485382d5b889e3c667ac94dfe13356dcf0
7494b2a1109f5eaf294f47e17a0dd799f1298e6bab6dca5a8ed5365c438bd654e3ea55766575075c65c7279b7065d5416525c7aa151394aa8358b540
5ef1876465c5045c1190ff908ccb392e5bda5559c67169577c82d07ce256ac5c62512addb8c4821534505188e544a9f2645cecc4df737c7c9153799c
e322272ee4b8806319c740ebef1e7b4cf91dc7c71ec3f985589ae7514a93f1518ef3383e62c787adf89019e7729cd384254d38bb091f6cc2628e451c
67719c91800f70bcdfd94fb97f144ee738ed319c4a9d291c27732ce43889e3448e05bd30bf09efb5e2048e77731ccf71dc58b332ae09c79af1aef028
e5ae4c1cc371349d3cba1fe6797014d394519178a71b470e0d534672ccb5e01d1c47dcae292338deaee1708ec3e8c9308e438768cad0301c126b5386
6838d88683380eacc601d5d89fe36d5217e5b626ec77106f1d86018e7d39de72b34bb9c58d37f7712837bbb04f6f9bd227d0eac0de36ecc5318763cf
1e6ea56713f6e8ae293ddcd83ddba274d730db82dde230cb86993759944c8e37593023dda264d830dd825dbb9894ae1a763161e74cecd43159e95488
1dd35c4ac7644c736187d464a5c3ad989a8c29c91625c581c9164ce2e8e798e8c004c233c185be428c6fc2384221ae10636de8250a7a39c63461743f
8ca24e14c7c8428c204a45700ca745e151e8e1e8e618c6d145135c1c9d84abb31f6a8fa1a310ed1c6dd670c5c6d14ab3ade168e168d6d0c4d148d38c
1c0d6e540b51a68732498007691439e51a9a227541a6217064b5ac70f172d6e9ff871ffcbf06e0fff88bfd2f181c77f9
>}
\immediate\pdfobj useobjnum \csname EF5O13\endcsname {[ 32 [ 318 ] 40 [ 390 390 ] 43 [ 838 ] 48 [ 636 636 636 636 636 636 636 636 636 636 ] 61 [ 838 ] 65 [ 684 ] 76 [ 557 ] 80 [ 603 ] 83 [ 635 611 ] 97 [ 613 635 550 635 615 352 635 634 278 278 ] 108 [ 278 974 634 612 635 ] 114 [ 411 521 392 634 592 ] 121 [ 592 ] ]}
\immediate\pdfobj useobjnum \csname EF5O14\endcsname stream attr{ /Filter [/ASCIIHexDecode /FlateDecode]}{
789ca58edb0e43601084bfa44ea58a3a5367aa78fff7f347feb810aeecc5ccceec64b270731e07ada00ad476ad63f0c4c4e285cd1b07575e3cc99fd3
5e5f72b0614874f9414c424a46ce9782928a5ab80d2d1d3d83d87f8cfc9998b7fcb202b66703b0
>}
\immediate\pdfobj useobjnum \csname EF5O15\endcsname stream attr{ /Filter [/ASCIIHexDecode /FlateDecode]}{
789c5d53cb6e833010bcf3153eb6878800c14e248454a5170e7da8694f510f602f115231962107febe86b1a9542418cdecce7a1796f85c3d57ba9b58
fc6e0779a189b59d5696c6e16e25b1866e9d8e9294a94e4e9ead4fd9d7268a9df9328f13f5956e87a82858fce182e36467f6f0a486861e23c658fc66
15d94edfd8c3d7f902e97237e6877ad213db4765c914b5aedc4b6d5eeb9e58bc9a779572f16e9a77cef697f1391b62e9ca13b4240745a3a925d95adf
282af6ee2a59d1baab8c48ab7ff12483ad69b7fc74c9075c81dfab7c847cf2f246116d401b1f6d2067a8957953a009200564800320077080001c4319
5455a0ca5755900f28bbc015085942965e969073f492fb31733f668e6ef2839737ba46398ee035a28162108e5c8e413806e1188463108e41f8299441
5574287c2b9e724cc909d0869cd52270a4e0b0048a0e043a1079c88105870aff2dc4f232dd42842fbfecc6b2c8dbe2c9bbb56ee7d66d5f976d59b34e
d3f64398c12caee5fe051eb0d6c5
>}
\immediate\pdfobj useobjnum \csname EF5O16\endcsname {<<  >>}
\immediate\pdfobj useobjnum \csname EF5O17\endcsname {<< /A1 << /Type /ExtGState /CA 0 /ca 1 >> /A2 << /Type /ExtGState /CA 1 /ca 1 >> >>}
\immediate\pdfobj useobjnum \csname EF5O18\endcsname {<<  >>}
\immediate\pdfobj useobjnum \csname EF5O19\endcsname {<<  >>}
\immediate\pdfobj useobjnum \csname EF5O20\endcsname stream attr{/Type /XObject /Subtype /Form /FormType 1 /BBox [0 0 306 273.6] /Resources << /Font \csname EF5O1\endcsname\space 0 R /XObject \csname EF5O16\endcsname\space 0 R /ExtGState \csname EF5O17\endcsname\space 0 R /Pattern \csname EF5O18\endcsname\space 0 R /Shading \csname EF5O19\endcsname\space 0 R /ProcSet [ /PDF /Text /ImageB /ImageC /ImageI ] >> /Filter [/ASCIIHexDecode /FlateDecode]}{
78daed5c4b8fdcc6115e20b7f9153c4a084c75f5bb0f3ed8492cc0c821b217c8c1f2c196e49514f925f981e477e507e62b0ec9aee2b077b8a35d8f16
08ec85668a64f557ef2a92d33f77d419fc47dd4786ff7ff6fdeed15f5ffcf6ead98b2f1e7fdafde54bf9edd9bb1d75aff177850b5ee3ef775cf6187f
573b66f1fdce99887fdf0cffdae4fa88cf66fef472b7fb6ef7e8139cfe0e673ddeb9d2a7dc918fbd4b1157db62fa347f7f337eb7def6e4f8fbfefcfa
7de467f7fcae800398fa0c54e0ce949d77bd8f36792f5691343faeb4fb1442fdbefb79c7baf88895e1a9b72e52a6e22974e4421fa20bac9d4f2f778f
3ea32e7597df0d525f3edf7dd53db8a087ddd7dde5e7bbbf5dee9eec0614bb427da15cac15ab4b5a7bf55cfa58ac0bc5671fb7acee0f57276b7a1b5d
89452caf88edf589726f5cf2068a0a9b00941500d1f6ce276783042089d70008a9cfc55b635cf6760b0077610e21d860fb923c2e976e26896d081617
d8e23dc06eb4c18559f3029b439fbc8b4eba81225e8321c53ebbe4b285423679213058ed0dd2ad1c16edc8ba9e08e69d18e1fabe58ca2080cdc0b60f
ccf8c1c53f1e7697af77be0fd626134ccab45febc1c58fd3919ca375c6c3a3c6236f8623b637993dc83a44da78e4dfe3351e464d314438c178e487f1
480ace532a31e7e3eb7c3f1c49bd33a1986c8d9bb1bd6a22f866e256105bc902c371d4dd7804a08d73729d9f26d4d124bf57e0d175de0e471cafe34b
34c5e5e3d7b4257d315e13c9db00eb548dfe32a22e086276a0bce19a36b6960e1e7d66abb3b48ca83caaadcda70f26c8361a1b5c311b8cd3768fab5b
354edba5da8a7ed974ea535cf7d9c80de11b5c2044cf71077d37ca53a082649238d2e6764a88b45deae9c3e190c88763d547120ec0eb0bb70d69415a
cd85a5cf36b219c917f8484a05d91916ba514d1e574fd4e7980dd9ba7a251d593d96de1225b81b3ce49ad56d7375340114a880d9bc7a251d593de7de
979050c43819dfa82318fb273248c0817d7b5a5d90ae5fdd1a9c3f783295e0ae593d375707b344a518aaab57d2eaea01e5918b154a506fa28b8188ac
bfb608465d0087ee9619521979c17c3d7736470be097cd2cd28eee760c7f371e8908d46c3202ec7816d9927afff4df461a54b23413eccdb3b80ceb9f
3bddc8dbc0daec0a3a9cd8bd7dd1fdb3fba1b3dde71df59e5bf3dec0d248ae2943f8deed3f39f27093e039ae90a962f7c5e34e0f1ea26d97a943f4d3
2a7d7cb97bd26d01b611cedbabf699dd12b83c73d73ef3edd54e49457098512a907334019f05197d22661fa2fdb9f868ca400cf0301306223eba8c9c
b40339b317ecc9fc31f3a9a56a90e0283e2c88a3fe86eb05b964cac17792a9ebd1bc7ad20030e6a15da6fdf53356c7fed64991fc2c92107f263e63f1
4df4455131f5c5ec86d89d38da1ed185c37a6588863455bc02496be2d0bae8b4a626c175d6a7585faa5e00ad66aa12d19af054857fc939709bf30e83
f8e678e2315dc58b70b882fc4231d924c920c6926db15292c2c35309d96aa98bc35c6972724241a0f91a9ba3263571d6ba268f16924c675b4a00d5ec
12ebde3da448b31749f185c31596046565792eb9128260894a8685e36269283ac59248a1346bf29875d9cd9a9e04d759a1627da97b01b4da4988446b
e29fe672fbfb2b5f0dff528ff1f56b5487e7f83af819ca8f65c74b32e5657497991d3d1c66f5d5f4adb2ba6acbb666f5f72f37ab08f4fd9393909db9
de48b170e6249688954a162105624024462ba30fc46c9283dba9f04fec58397a27431d446040b25c104705aaf84f16bc4286e74aaedc347163a41070
2b93d12fe90490701942203a9903929be5123a98892a07a873a7d0125c6b182a043562055ab3269759d78159d397e05a352b10482b08b4d562422eb3
a603537570b7b547058ff03ea6c74869b0f34c270b71bd73ae486998ea2cea66d1a2331ded45a122f5c4d454c37554e9823aeb5fd1676329ceb35d15
0ee1040af5ec314a42572514fa109e481853324e8f5e9d0daacf3e14cd1954cc04ce2c7150c6176b282bd4a0ae4828a9521f9a3e694f72ae9a9638a4
5d24ea6a4529a159d5c7695e79cbe5693dddeb32a026f73fae40ad43d037d84f8276e60a25e5e2c720b35c22842a5d861ba8998cb1410567e1b5c951
582409d05d823706950ef889130071b695d4598f3a49706af52565952338b53a727e0103d14676b879af7204e8be386e2ea5805ccc2601859967aace
11e26c116d82b3884d854344b2402da2be0aa889551b925e752719573d4b18d22a02b4b0a11050da7bcd3beebc708998523e093af20669978c0eb362
1cc664210ea81e03648e0be1418f28c4292a55819aa720ae6a955469044d1f2d261957e34a18d2150468e138553ee16332c1488f8c58c517f41afa6c
008ae8d7a2668c3905d68d6109039161c2303c4bd0664d3ed3508659559d602cd42c6028a354d0d284554069ee35df386bd15aad00aa32e81bbe7f60
d15a85a01ec99e06eddc454bc86590fb26b96a0009b28835347158c4959c656432159e6a21b34a104c8f01ff65990a5861c66544ca923aaa512508a6
f3dd296ed904674220507046e300d597184ad61982e929a59c8bc8104c9c0514569ea92a43a8b3e760538ce7d05cc09803598136ab029a863accaaf2
2a67a9e88a439ba5a29646ac124a83af79c75d172d1953d2272ddfa44a04c092ced4145c30ca292ddf0d03f4b2909ee918ac4d56ba62aaaf513cea75
419dadb0a08f36939cab7d250ee90d1275f51d2961f53395638457321d13b673cb93a9f04b0a9a318fcdc0710003ed9fb304bb28d0b42a2035d441ab
caab9ca5a22b0e6d968a5a1a711650da7bcd3bce59b6d66b80ae0dea49e11f57b6d621a8b7784e8376e6b2a5e4caf0db492e1140952e830d5497108e
a44213544484f5b44811a0178fd44a2a1940630e4342894beaa8479d229ce9819d914bcea0fa8212a671809ab209911629c221c347cac1aa14e1a84a
58f551a92a45e893a760938c6b686a18359025685a15901aeaa055e555ce52d11587364b452d8d380b28edbde61d2245706099f99d4b1d32cb773b97
ef6e22448ebdfe89530cafd0939f1c1a5f3041163ba411fe02bf462fcf4f15dfeedfffa421cceb15bbf61543ac6ce42d5f480850673148e716e34287
d36121d8c25b1855bc50908be117dca6c7ebf2e17ae833268a62170fd7f7a7fc797c37082636869fd1cf6fe69078e65ee194c4838bf7c165c3378943
5f4aa29c282681861fe4a75b456357d14013fc5613e79834dc8508689f888231fc1ee75d6ac7afe329b016cceb0c4070dfe9d1b60522670cb93bd54f
5ec563e13d0eb531232311473d77e6c5a450903eee140f4d2f5f061761a54cf39158917258f3add7cc4fa7d330198057ec7c413ce413cad6901ec210
94b52de8ae6b0b761b4b1d1730be2b1da3eb32f51911583cd248f6e06739bb56ea1b418d851fbbdbe1b1dd74bda4cde9edb81ee617e10f24bcbef119
faf31cfb927d563831311aac8aeed13b055f9025d6cae344093e184b92410542fdf6da98085a8e0f8f4b943a2a590a2e989cc3a2c35bf37cdbd36b19
d01db8384cd44a864a563254266791010929a08550602103aa72ce3184a465a864254365721619c08b6fc62bb0fb9b310045b668192a59c950999c45
86884285b15381850c296188700b2b4c4485bf32380bfeccc1bf48cdc44ff9fd416e166425c1c4e2dee7363441910e52fd3077c783542fc84a1b95c9
39ecc9f73510d6cbfc8c7f31921ce4674196600593fb6e536b2386aa835c6f31d4a0fd5ee67a4156faa84cce62539ec6d241aeb721f4810e72bd202b
192a93b3c81012dfc45ee67acb6f26b8835c2fc84a86cae42c320caf731ce47a8bee2e8445ae9f890a7f6570dfe36aeec8b9932b905475f4952a3bfa
308cf1fca976f492761ff5203a7ba1093906085508b294bbf238511befd70188665c8a205a772942254bb082c9596410cdb89441b4ee52864a563254
26f7dd2b65632ff521c600a98f4a56faa84ceebd3ec49020f521460aa98f4a56faa84ccee2e362509032cc43859460222afc95c159f0cf6382442f46
0a89bf929504138b7bef8f624810da902385d086202b6d5426f75d1f72e090fa10e389d447254bc1059373f8b71c12a40c62a4903254b292a132398b
0c62489032889142ca50c94a86cae4defba51838a43ec47822f551c94a1f95c9bdd7871860a43ee661476a63222a5d54066bba503bbbc43e5914811c
12a0e03aef0a2a02c25e3e9f524fa7f44f63636f89df568d66faf9aefa69ecc7e31398343f6859792244c37635fbe810d852e05f25f3ab42846cb386
4d2cd5fea1fbd607a6cd3377b7f2a0b4506f5330de470fbb6c97466de4d2d0141998b140c140d430e37d521506d9ded1fef7e5272a4a3d844d7d7626
85ccbf07b90d86fcba842f98311da2e63618c6dc678c4a8612e2f03618668f6c7742e86c72b682a26a4ef0e40fd3d97046efed09b1b34557fc1e4bf1
37f5e40f5451a8153ede621c597e3f34dd621c590c64bebc6f1c99a1b7db6f99d6ac9481378b89fcfb3244d91d57ca72c34ab9866d83a3356bd589fc
3e9c72a2abe376719ae955d4a713b9b572abae55ef6fc945adbaf7a6d495f2164ca92be52d5b5357caf7b7e6a256dd776bcaeaf8fea65cd4a7db35e5
a256dd822975ad3ad594ebb56a78edd1649f07a7c68868f9adbf644d905b74212b6e542dffce0c550be6653fdeceaf99954cea0dea72a1c18fb733b4
d764e4e0f79b286ee7e69a134346210fe4dce049db19fa6bb270c82e993078cf7686a13931949e0c5a9b3878cf7686f19a3c18ed092e935afe6d0c6f
bf70739fc9d7e4be184ef099d262680b9c70d89eef061e7d4dfaf1ec8291e5bd91d3503b47446e28831d8c4c3762699b6927b02766f6ec1b390e35a3
052045639c39f5f08f9253e5b8dc514eb4c6ff921b227e2065c9b93e46289f7f2a61e4d67ab9b7422b835a9e3eb8f8e6e2e9c38b0eff7d82bfb70fe1
f9fc1cc1c590f9bedc838b1717cf2e7ee40394188d2b062b17c8fe1ce73fc7a15f2fbebd7883d3dee1fb2f172ff1a9c3e79717af78873c7551376c8d
f7f56ec50ad34f0ad03f66925bf945b991df5ee97f1f37c2cc0555c23a7fb8ebaacd250c1bdcceaf457f3bbd30ed1dff16a9eefc37ed76e94c8cde88
5d6cf71b64f24f009154b8cae7e5469cd1c732eca979b079a7b5a837491eb96aadffaab9cc6fd3ae9ecb6d389b90dbcb4f3b1226f47236cb37c95fce
dbf846defa6d8b66dacb3cbfb1969b7a793b2e92531ade8edfc0ebc6806fe76dfcca65deb9d1668b649267cf145ce67d1b933325f02eae2d30873691
6c1eaefe22611ffb64721fa9d8e68e98b92785c6f5a8c8391ae4d5ac2326f41496dbcd8e6baf5d34693a4001062928a6e3472e47694bc4e4e8a096c5
ae9c8133c878536169ea5246ff28c7d7793d7941482e07de32733af26b739d696fdf4218b15cce95db4f4df5b4114ce9a924de0e3d27221d9fd028aa
908550261ce3c6ee329b71ba4383311059042d86ad0e339fb4bf43937bb86e36994c08fa0e4d607b7245e18772476599f61cf631110c9743396eb5b6
fc3f34afb96aeab9ed03efa630e4fd491c669f0d5ed8b6c0ed7a615b6f3fceda2183a0c2dff175a63d6811861e15a2a474fcc8ab661cb4b1b5af6963
6bdbb4adeb77a29f5a9f0da7ee86df17ccf2375643859c7c7e99dd2c19595c688c97c839ba92ff3c46087b4cdd1d7888b43dc3ff4c25b9142b4fa82b
cea9b18c7be23b714f1583771ff727e6f13cc01e6bd547faccbb1621f02fe7e01af0cc0644bb15e2c369c3f0a4eaf4c7ab503e68abd0785e189ba419
83d97769d4230bc53de81bda6baa24cb6e668a7a64c4e04493f77c3c3fa69042b0c979fd5bcdd09b5c949eae6d090aff24f8fe864bf9f0c3c5ff3f5c
660cee9e85cb93ddeec9ff003f1ef8c7
>}
\expandafter\gdef\csname EFWidth5\endcsname{306}
\expandafter\gdef\csname EFHeight5\endcsname{273.6}
\pdfobj reserveobjnum
\expandafter\xdef\csname EF6O1\endcsname{\the\pdflastobj}
\pdfobj reserveobjnum
\expandafter\xdef\csname EF6O2\endcsname{\the\pdflastobj}
\pdfobj reserveobjnum
\expandafter\xdef\csname EF6O3\endcsname{\the\pdflastobj}
\pdfobj reserveobjnum
\expandafter\xdef\csname EF6O4\endcsname{\the\pdflastobj}
\pdfobj reserveobjnum
\expandafter\xdef\csname EF6O5\endcsname{\the\pdflastobj}
\pdfobj reserveobjnum
\expandafter\xdef\csname EF6O6\endcsname{\the\pdflastobj}
\pdfobj reserveobjnum
\expandafter\xdef\csname EF6O7\endcsname{\the\pdflastobj}
\pdfobj reserveobjnum
\expandafter\xdef\csname EF6O8\endcsname{\the\pdflastobj}
\pdfobj reserveobjnum
\expandafter\xdef\csname EF6O9\endcsname{\the\pdflastobj}
\pdfobj reserveobjnum
\expandafter\xdef\csname EF6O10\endcsname{\the\pdflastobj}
\pdfobj reserveobjnum
\expandafter\xdef\csname EF6O11\endcsname{\the\pdflastobj}
\pdfobj reserveobjnum
\expandafter\xdef\csname EF6O12\endcsname{\the\pdflastobj}
\pdfobj reserveobjnum
\expandafter\xdef\csname EF6O13\endcsname{\the\pdflastobj}
\pdfobj reserveobjnum
\expandafter\xdef\csname EF6O14\endcsname{\the\pdflastobj}
\pdfobj reserveobjnum
\expandafter\xdef\csname EF6O15\endcsname{\the\pdflastobj}
\pdfobj reserveobjnum
\expandafter\xdef\csname EF6O16\endcsname{\the\pdflastobj}
\pdfobj reserveobjnum
\expandafter\xdef\csname EF6O17\endcsname{\the\pdflastobj}
\pdfobj reserveobjnum
\expandafter\xdef\csname EF6O18\endcsname{\the\pdflastobj}
\pdfobj reserveobjnum
\expandafter\xdef\csname EF6O19\endcsname{\the\pdflastobj}
\pdfobj reserveobjnum
\expandafter\xdef\csname EF6O20\endcsname{\the\pdflastobj}
\immediate\pdfobj useobjnum \csname EF6O1\endcsname {<< /F2 \csname EF6O2\endcsname\space 0 R /F1 \csname EF6O9\endcsname\space 0 R >>}
\immediate\pdfobj useobjnum \csname EF6O2\endcsname {<< /Type /Font /Subtype /Type0 /BaseFont /GCWXDV+DejaVuSans-Oblique /Encoding /Identity-H /DescendantFonts [ \csname EF6O3\endcsname\space 0 R ] /ToUnicode \csname EF6O8\endcsname\space 0 R >>}
\immediate\pdfobj useobjnum \csname EF6O3\endcsname {<< /Type /Font /Subtype /CIDFontType2 /BaseFont /GCWXDV+DejaVuSans-Oblique /CIDSystemInfo << /Registry <41646f6265> /Ordering <4964656e74697479> /Supplement 0 >> /FontDescriptor \csname EF6O4\endcsname\space 0 R /W \csname EF6O6\endcsname\space 0 R /CIDToGIDMap \csname EF6O7\endcsname\space 0 R >>}
\immediate\pdfobj useobjnum \csname EF6O4\endcsname {<< /Type /FontDescriptor /FontName /GCWXDV+DejaVuSans-Oblique /Flags 96 /FontBBox [ -1016 -351 1660 1068 ] /Ascent 929 /Descent -236 /CapHeight 0 /XHeight 0 /ItalicAngle 0 /StemV 0 /FontFile2 \csname EF6O5\endcsname\space 0 R /MaxWidth 974 >>}
\immediate\pdfobj useobjnum \csname EF6O5\endcsname stream attr{/Length1 3576 /Filter [/ASCIIHexDecode /FlateDecode]}{
789cb5567b7094d5153ff7fb7df7db4df6c16edec926616159200d0bb891d0440616c803054d348126914836d94d02e4b124e11162004ba8225a1e3a
5ba420d8528ad46a54c6498b7daa834e86e958653a4ec751a7743a9da1337606759a29777bbe4da0c8d499fa87f7e6dc7b7ee79ed7bddfdd73438288
dc3ce8e4aaaea8aca2743288443e4bb3aa6b6bea2ab6add9caf84ec6deeabab52be69c5ef606e37ac6f19aba05c14d053dbf627c95f1bab6ee708c5e
d5761369158cbbdab60d784f9c1c9dcd789475bada631ddd7fdbf4cf373998e9ff4047b83f4616eea49f656cefe81a6caf7a79f551c6e7d9e681ce68
386259f7fb34228b95d74b3b59607f5ebec538c4785667f7c08eec77e9af8c238c73ba7adbc2a2889a183fccd8d91dde11c3662385f11366fe3de1ee
e8fcd2957d8c5f60ffae586fff40e20ff42851caebbc3e3fd6178d1d39accd64ccfbd1ed649e8d9d269bc6082493329352299f0a495454ada9272b9f
1eb744624a7772fd083d4d19c9f59de1be702b8d84fbba7b68a4b52fbc9146dac23dfd3c7646fb781ceceba2918e682ff31d7dd1cd34d219ee619dce
682b4b36877bc234d215eef59ae3007be80e0f74d248cf6653d2db11eea691bead3dac39d0ded3c163a7e9ffb68ca6f2d29de210ef8164897c864a45
a139272af1276ad7f894359b0140b769fad41e6eb6daf6ca084b5651b791a132c4314bb7f8cb97743045f9c99d1355301249acd35c9e0d9e059f217b
48e6949c13af278edef43179ce736f9eb6a933496636dd53c4b90b9f58406334cef43b3a4717b58b748a16737f807e2df66b015e39c3dff403a9d179
3a2eca448628e3d5f78c0c6350ee9367797d3ddbae622f1fd053ecd3f43446bbb521ad96dae9a2bc44c7b8f726e59fd22fc55eba4c47695c5b459fd3
5ed4d341eec7a85f277959a492d28ae9b419893b512b530102f272b27f4abb6988eae9b431666488f792599f116ff0afe5ef9cf37b588f2d6c719cce
8a7dba4f3fabafa28393f9a2850e6a7bc531bd25d987f840b6d371bd459c3332e82d335796d472a6edf43ad376ba24ee12fbb09f331b32339097e992
e51e7dc1645696612ce2fd10d38bf40a051067fbe45e8c763aaeb573accf39934ba8a022f6d64f94c1bf6eca7ccd903a3441f3bcae51cd7f77643474
7f83f7edc6198179b741afcbe21da5da51c7a0772c91a86dd03db27154e68fc26f1dd5fdbe4fbe6af193c0bcd5b50dded1572b2ba6bc56b654b0acae
815913b198e5951581c93b61258d6f8f79b316f355b8c4bbb4d170a8c22a0d69910667cb0c240fa9d246b61449368d49f0a5b1955b9b34fe6455a986
cecbb614ab85776690a1c902bef71ebbebda956b57dcd9656965eeb4b23b68c1922bc17f5f5dc899585cf21f96c93fd714dd9c1a678a12cc482f11be
f419f089a6f1edda87dbc7d587e307aefb0e8ccbcbd78f6b918962ede2f5324afe7ac8b28cbf65b1c80cf567a4a7b95d32739ab43bec4ee970d8ad16
33ff14981546934250aeccf3e465656b9939d29337bd50f316f09c2f3d9ebcf29c6976e8053e82c84c5933dbbbd3b32bd72172f3f273d2a639a5c701
cd675091b0fb0a3c8e393e31c790459ebc59f35cd75e3999279adf3c1ff3b478b4e6f7438e44d129ee878a7a8b6a8a8ce6ab57df3fff92afc5a7355f
7597b9cd73e03fa66c06e6892cb9726dc9d5209f5176d92d07c3786a74fdf7886e81cc36ce7c65614088e6d7428196402c80e6507a22f051e0378153
dc0f057a03b5819a404ab3802f7da95874e76cdf4cc3721b5f62592a4a8259d9c93133c3b02cd83bf688d51a69fade1fa70f8ded65aee13193fbf398
6ff53bfd3b4fb8ea3f1c3a74c2aa9dbebe5e7bb678695653e49de7ae1fd49ef52fcb7d306ab27acb8bad1d7bfab63d7ce6e88c9a1b355b6b7cc6d8fa
fcdc0dd3967c46d3adc942f4ee2ee74737e62f3ebe5ee2a8b21e48dec49b3595bf6ab72a2072ecfbe2e37ffdc25145117e436f6d866efe4e4df7e798
0e98025a6c194fd6b61bad89dfadc91c0cdc47c5d4c9f550231785ccd7114ecd31f5c658f875e3fbaff39b2616266baec90bca6234c96be41455533c
6836bfd493bc7e0b2f2947ec9ce20d9a258ed04aae0d311aa43eda481d1c7d80bc5c8bdbb81e7829480bb99730d7ca1a5e5ac13a035c2506583b4a61
aecbf3587a37f5b0fe7ce6965317772fd7e31bbefa9328ca73946db6f11861cdd4ff236ae9cda8f51c691bc7dac4363dac6de611669baf17b182b94d
6cb78eb6b2461beb8693dea2498b7072475ef6d2c3638c755ad9ef46d6f3b27d2f470f27d76ef75397f4d24f3553fa5b581afd0a1def6d5aeb9219f6
33ee4d460d729e25b4e84bd6376c03b7d9f22b99f88c6997f93fcfff6846f24e69fca58ba996684cdb134abcac30eac74b41bc18c7cf9d78a1df295f
08e2670ae7fc78de89b37efc348e3313f8c9044e2bfcb81c3f52782e885327ebe4a9384edebb5c9eacc3b3419cc8c0f1387e988a630acfa4e1e8307e
70017185a759e3e9613ca570e470b53c328cc3d53874d0230f291cf4e0fb0a4f2a3ca17040e1f1fd85f27185fd85782c88471546b2b057e1bb0a8f28
ec51d8adb04b6178b55f0e47f0b0c2901b3b072fc89d0a833b9ae5e0050ceed1776cf7cb1dcdd811d2b7fbb14d616b1c0311f43bd1b7c52ffb22d812
4b935bfc88a5a197d3ea9d404f28a1d0add0a5b0390b9b3696cb4d116ce4181bcbd1799f4d76e6a0a3dd293b826877221a4184cd2271b429b486edb2
55216c47cb865cd912c186875c72432e1e72a13915eb1f74c8f50a0f3ad0c4164d7134363865e35c3438f19d09ac5b7b41ae53585bdf2cd75ec0da3d
7a7d9d5fd637a33ea4d7f9f180c2fdb5f3e5fd0ab5f351c349d464e03e1beee5acee5d8e353cad51587d8f5baef6e31e37ee565855ed96ab14aadda8
52a854a85058b96258ae5458318ce50aa1092c9bc0d2092c295d219728dcf536ca992baf43990ac5f0ed612c6658aa0764e90a2c52b853a1a41cc109
2cb463814240619e42312f17df816fb95004972cf2616e21e6cc76ca3911cc76c22f52a53f8859f61c396b183e592e7d0a3319cdbc8019ac3fc303ef
749bf44e03bffabf0d1dd3a7db509882c2905ee0423eabe7c7e189232fd72ff322c8cd4993b97ee4a4213bcb2fb39723cb8f4c850c85f409a4b97365
9a829bbdba73e15298a6e0640fce381c1cd0310cbbcd2eed39b0d991aa60e5256b1c06ab1b0a927721cba133d203808b5ffd54a9e540a4428474ca87
1813917d4f8ae26fb6d137ecffebb682ff004f19f87d
>}
\immediate\pdfobj useobjnum \csname EF6O6\endcsname {[ 72 [ 752 ] 109 [ 974 ] ]}
\immediate\pdfobj useobjnum \csname EF6O7\endcsname stream attr{ /Filter [/ASCIIHexDecode /FlateDecode]}{
789c6360186480856a26b10200020d000a
>}
\immediate\pdfobj useobjnum \csname EF6O8\endcsname stream attr{ /Filter [/ASCIIHexDecode /FlateDecode]}{
789c5d50b16ac33010ddf51537264350624ac9600c25593c340d753a850cb27432825a12b23cf8ef73925a177a203deeddbde3f1f8a93db7d644e0d7
e0648711b4b12ae0e4e620117a1c8c65870a9491f1a7cbbf1c85679cc4dd32451c5bab1dab6be09f349c625860f3a65c8f5b0600fc23280cc60eb0f9
3a7585ea66efbf71441b61cf9a06146a3af72efc458c083c8b77ada2b989cb8e647f1bb7c52354b93f144bd2299cbc9018841d90d57baa066a4dd530
b4eadfbc2aaa5eafeb2f475a2f702ff848f4abca74827bc1473af72b4c97530cab6d3987408e7356d96a32692cae717ae7932abd27357778c7
>}
\immediate\pdfobj useobjnum \csname EF6O9\endcsname {<< /Type /Font /Subtype /Type0 /BaseFont /BMQQDV+DejaVuSans /Encoding /Identity-H /DescendantFonts [ \csname EF6O10\endcsname\space 0 R ] /ToUnicode \csname EF6O15\endcsname\space 0 R >>}
\immediate\pdfobj useobjnum \csname EF6O10\endcsname {<< /Type /Font /Subtype /CIDFontType2 /BaseFont /BMQQDV+DejaVuSans /CIDSystemInfo << /Registry <41646f6265> /Ordering <4964656e74697479> /Supplement 0 >> /FontDescriptor \csname EF6O11\endcsname\space 0 R /W \csname EF6O13\endcsname\space 0 R /CIDToGIDMap \csname EF6O14\endcsname\space 0 R >>}
\immediate\pdfobj useobjnum \csname EF6O11\endcsname {<< /Type /FontDescriptor /FontName /BMQQDV+DejaVuSans /Flags 32 /FontBBox [ -1021 -463 1794 1233 ] /Ascent 929 /Descent -236 /CapHeight 0 /XHeight 0 /ItalicAngle 0 /StemV 0 /FontFile2 \csname EF6O12\endcsname\space 0 R /MaxWidth 974 >>}
\immediate\pdfobj useobjnum \csname EF6O12\endcsname stream attr{/Length1 10768 /Filter [/ASCIIHexDecode /FlateDecode]}{
789cd57a0b5c54d5daf7b3f6b3f79e1b33cc0c33dc2f03e38078834054147322c55bc730c9c48e060a7817042f291a680a919a5730cd742a34433332
8f819969526664a752cf399ef2b50b6517324fc7ea84b0789fbd6750ec9cf73beff7fbdeeff7fdbebd597bdd9ff53cff673dcf5a6b16c000c04a1f11
1c23878fc880004800607da8347064e6bd135c97fa71ca0f0710c68c9c707f7a5cedb05300b891ea6b7e7757d6a8a4fb0b45eadc93dad4dc3b212169
f6d2cd8100d209aa9f387d5e6e111a977c012007537dd1f4c50b1d302b221540f318e57941d18c790bfa2f9e0da0a33cec9f915b52041a7a41df9ff2
7e33e62e2d886b7e2591f26e62a7cfccfcdc3ced94372b001c73a97ec04c2a30eed66ca17c2de57bcc9cb7f0e1170687aea5fc3b94af9b5b383db7e2
f3751f02442fa47ccebcdc878bc49df202ca9fa7bc637eeebcfcb8ef871ea3fc75e2e77c5161c9c2d7b5373e00702afd3f2d2ace2f1aa2f99b427a1f
c9341314acfcc0fb089443e8a79629414fe934108667dc9305a6b9b90be74330614a4f6727c0cd94da7a4e7ef17cd0fafa31aa13d4580b025ba7b464
b5ec23a21700ffd79ece4f7e9bba55f23f3dc2ff00ad93b753ee3ce7a57e2ba5e6cefc9fd1be7d14d2871e4ce00f66b0800d42210c22201a62a8dc9f
f28c74866a3bb15b1f0964fa6a7c392de8a89d09ecea4c11a8a552afa172a5464f5f833a9b8c37c7b142d75cda0c5bc1a6cea565b9c5b9d360756ef1
bcf9b07a5a71ee2c583d3d777e097d67e617d37769f15c583d23bf90d2338af3e7c0ea99b9f3a9cdccfc69543227777e2eac9e9b5be850be342757cf
cb5d381356cf9fa39414cec89d07ab8b17cda7960b0be6cfa0ef4c857eb779db1d0f85af4fe16bc884343fd05c520a855504400e097ec1db4849777f
baf26ca62f7dffbfd7083311ddc5ffbe1d4cf6b655e29be96e346e2b9fdc2d5f7c2b2d0c02b8b181d26937bbf6242d44aa08c4c01828868daa852aba
73420f15852e8d23c429ed758be81b4bfd7aead653bd8b52a0abb9ad5daf9bede26fb6ebf94fed04e873b35d6fb576bd3a76576d3fb5b680be7dd5da
47bcfe4334b18dc41d48c9d276ca467a63fc331408349f04838ca8150541a1a3b42fef1234b360441eb8c1e134c8366e633b34f3d817be365d7c7b43
38289a07384039a6e64550fcbf835e518de389df14180e23e11e9800b9900fb3a00816c3c34e833a8b1cc46f6f5a576eb5c883193097d07dd8a9efec
ecfc822cf9e3cebf74bedbf94e6753e7f1ce973b5fea3cd879a0737f67dda5906e1cfd578fd72737f972feea88de80aa4615dcbd38f7f605c54e699d
23aebdf63adc1794b932d217143af7f88299c2045fb050c8f505c562f328e45350fcf50c0ab328d829286b5411055afb409973ca9c8ea0f0b02fc4d0
3a43b83a0d9476b2146880667a4f421dec647b29a7687a0195788443b0061651c929d6ccaa84be54b617aec1396a5909cd5827021b03c9540a705112
e03acb82c3442395d958aa46a6c93b4e3c2cde27368857c4b330502c11cf8a3962094bc667a589d25e0aa9f816cd953310050dec3294c051fc0693f1
98385c34c1653c8b75f0258da2e0d70c1ba0164a89171b2b8432a154b88f4a4e4b676107bd85547f96ed62e788bba3ec51b8004fa2288c825dec02c9
d50c3fc3a3982594915a928502e2ff34d13a4bfd77400919d805a6072ef4a632e29ec69aa67e23b0af74417daf41198d9c05b572836cd338691405b1
bdec146b95b78007cee1ef71017eccd6884e719f380a367811c01cd840b477287de402b6946457de5285bab044cc6175f08d98a39946b4df5224a231
0f0bf7914405708cc212d94c320d616bb08a382d5575785633464ca0fe4441b382a40628c414984da952380887a02fd6c006a2a4ca2b0f947ea69e3b
c5cf48e60d6cbdf0339cc5e134230bc4ab84352d26405ee0558d2c892830e8e330d70baed179f5eef1931cef6447f7edf39bacc3ac71d44366bd71a9
a3a1b3337392182665d74be1f5e8d2d68b2ee767ff55e5677dfb8ccd9ce4a8ef1831dc477544ce702a9b3089924a8e8aa97cc470b54e19b45e72d1df
e89c7ac7f4998ec7cd8f3b073f6ece1fdcd7eb8d68b742b6ad78847afea3502a5b691d1be8f6979f846d26a306d02a4380de64fe646c7d40d6a446d0
779e18943db6de5f4d837b50764b52ab2535f50e30b7b7263259b0dbac41ce5821a5bf75a0505ab1ead1359e9aeaaddb64eb57fcce2b57f8902fbf63
6f7f7a9935b5d278b5345ea13a5e94db5fa38ca76160b08a015aa0f1d2aedfa21b901c68b5db048d738035a5bf504b24ab6b3c6b1e7d54b6b6f2b4cb
9ff2c1df7dc9deba7285bd49542f9207fd5812694d76ba2da0633b352809761182f4b25deb67fea43dad3dadf50e48384f715322b344dba32d4e4b74
4ab4050f0a7d3bceedeb3827f495c48e73754aa28eac94fc5667ba7088d64684be6eda43308109a108982eec8695226d2031a14965f57a6b22a94163
96be5742760c4b664e415bd7f18f3ae9c2aff314b42b3bbf10379005182088f80b903d56f0f86db2ae0bd685fb4762b83d2c98e4bd4efc995baeb79a
af26b218c162b62627592d66212e092c6670c6285f61edcea79fa6bfa79fbec174fc971b37f82f4c2765f2b3fc3d0a6769e064d69f257b7809afe095
bc84ad674bd932b65e59173e23173a99a421f5b9ede9e811058fb452031e9d364a0e47886206f3799f8699a2e1d6a67605b0d6a4ebade7550149b0c3
fee82f0a5306465ba41457b202226763f87696ff2e1bd35e5b27968c6a18d576a18e08904d886348e270d8e58e0b090dc3e0700ba9c7224962baf919
cb56a3c7b64924df0866bdc0f4e141669423cced63ebed5963eb03b31e1c5b6fcb7a90384165de359d6f3d71c2624df571d30d6ee97b561f6eb604a5
126feea4fbc589d244cd327199b438ac3244439e33440c25130a5f088be545a125610bc3574145c8aad05561abc2f7c1be30cb1498e222215206c0c0
3b594aff58678cac49b993252789769bac9181dcf5c9f67b08c6e4dcdf3d5ff1d0b987979d9ff435b38d7830845fafabab5bc2360d9eb76df4929af4
bbdfbb23e9eb377fbfa728827f47d2ef247d9790f43da1c8dd0fec01fa0a5d548523c063377a745be4708f638b7393bccefe5c7c607800a02d243cd6
610e475b944e8e574008ccea925fa7ca4f00905904a993adb5e57a4babf9abab66f5255412995b9717991b95ebc88b16610a8b64769b181d131b9712
49820c20a97ab3146fe236f170d8a6e7f807fceba9a76767bd33eff8e9c63d078f54ef7aeec909c78b4bce647fc5fc9e405754d3c64b3fba5ca7ee48
aad9b0ba7aef92a292d21eb1871d8e0f0f2ddfafcc6b5a3bc55a9a530259f34a770433a211108de980068d4762b852c7fcf4102e6b453fd597184830
a32a989f22d8f9b4a6d6248ba2d796f369ad49248baa58f10c29f78ca2d25e06da7c8d826c5a9497c0e3a00964bd2196f5c6016c1cbbd7ef5ee34456
c016b165b8861949953a168dc99664bb53b56c94b9c0780abf70e14cc754c9d5fe059e6d4fdec73d2ce71469681769288f388f80a96ea718aab15498
23423d1a9bc75c65143cb0d2b84e531b1914cef4180e7ab31c696e67ddf562eee60fcd8ab5908acc4d571503562c98d4c39bbcda09503c8d8239d86d
709b5a146d5cc2900e4f9f497dda580f7e9eff30f5d4ccc927e6bcf8eebb2f8e7f264bba50c737fbfbf3abdffe8dffe47034df917864e7ce233d6209
6d3a2b49f5c4bd967633fddc76a8d6ad64d566ad60d68314624c82709d68553d0af9d254d56ec8680ee50430368559921d8a2b897645ab713c635bae
b31416c53fe3cd3c9ded6687580d9fc93379ae947063090b66fd581f16b4976fe3e5fc115e43d895747e21ada5d91d01a96e074861ac1ac3aab5d667
2c07edd5a64dda759102845bfa8bc9c121063301d7daded2dee4756d040d3f7f95984a74a9be57169d5e6692c4207b3fd69519209ee24704eb22fe95
873fcb17b1b56cea66a6292c6a5fcbaff2ef5900b3ced977816ddadb5136e17eb69dcd63f3d9f651197f7e2887bfcf3fe41ff1f75dc4256124ce268c
2470b9fd846a72d92c1c07832813478a636b694d74eb1235999a722c1745362540f56767de15fedcfe9074c1ebc514bffd2c491a078fb8d38c7e82c9
204446456a7582462f444545a6eb0d9151a29d81fd19dbd6e06a8b580d5b5d9b2ceb7a46ea0d51611a88090b31f5d584d8627a9a3f696a6d6f6da109
a1eac3ac1af1d59faf9adf26b3f6ce7ad3f794f445d93187fc4931538e44c527c4df1b8f5398824e8c6cb70546fd0b034f201da5f4ef919c14288e2a
79efa13daf2cd9bbecf33ff34bfccaec1fca4b5b8b5f3c56b9a3f4f37759d04fb3fe2ad5be357040f9e2e9f95121bd2f1eb9f86962c20723321e7b64
fef2a8e0be27f6bfdd12abd87529c9dd97f6587a3a0f1d23ff1d6508d299e08520b9d1647154441d0d6f743658d605f94110061b755a43146a6d2362
09d8f7ceb72625a90b784253cbf576d2f7db8a974ab5a49229b8e7274624462646253a12a3136386c5b923dc91ee28b7c31ded8ec98cc88ccc8cca74
644667c664c615c5ad89a88cac8caa745446af89d918e789bb1617d9d5b5ab5357879cc89ca81c474e745164515491a328ba3cb23caadc511e1d3c85
66bb8a19a134940db438534c04626c4aff01c9d1ddbd61a070fcf2819585db1b1b1a861d7bec4073c70d263cbf2de74856fef1c97fbf262417944e2b
b97838fe9e8e957505b9279f7dfd84b56c6dbf7e757171ed0a560b08abc9b28ddc40380c7287843682c9d62869d7991ad8360c12412b8cb4580d2322
547b4c4a528cb1e57a93e22d128fe44496477a2251b54a1f7bb4f003b1c454aebd5ce2b30d0d835f5ededc099dcdcb5fee38fdfce6cdfbf66ddefc3c
1e11a6fedaba2f2f970d675a7a87e7727bf3952bcd147c7c95910e6d1046ab500fb0335d85f631c9fe02931afdd86bc18dd606bf75e16176416bd7c2
58c1ea3f225c65b1495d69156b6d51d798ebaa1b73c70f8b288af0447c10712d421a06c3d83061987d5898d44793a04dd0f5d11742212b140aed8561
ba290b481e7bb43a3907da4d3eab2607a851e7b0462c6b3fe477f6d5d9a7a74dff600ebfce4fb3f8f6cf99a641d8f3d88e46933075f2f1d3fdfb1fec
d5870d627ab2f5bbf9a5a66d870fee52b1e613c5c924938156d6316e67885f84ce5a1110d8e88f8db1ce86b863ba46ffd74323624340eb3752b65a1d
2368396dea82bda9c50b3cbfa048944ae8f72aefe5e9a5a0afc2ec9d284166e1966d0d653e95d05e2c3028858e5c7baab7eed9b3b57a4f03e76db907
c68fdf75df1f0ea71e5afe7e7bfbfbcb0fa5360843dff9e493774e7ff2c977fc73fe4d44e42b7d7abdfec683d3a7b1c10c99c8064f9b5ea7ecc58e92
6fcaa3391300fddd21a8033431b9d26469f0dba6678216c629369561237524299bc284b456e2db62254f7128c7aefa70a7c5cb322592d5d52550cc6b
58bebcfa406363fa2b8b4ebe2dd476fc5ed8b57bd7f1da8e4ad9d6b12b3fef0705bf9334f8521a17e9d4dcdb6d928f8b2fc33141625a1132b466da21
2b38b5d0d6db6d30ebdcba4c5d8eae4827a9bed1a92cac271be811736e7864db378a1cb7e8c5bc0adb04a6850cd1ecdd3626ba8d66c92d654a395291
744d92bd4488806cfbb5d58781268274190393ddb1b25517ec0f7284c6ee5719e1c086b06321660d58fcb55a39d3a2f5cf0c0fd68666381540dadbdb
5bbd7bc1b4b496ebea064201c61d90d823b347518f8d3d3cf4bed1e3728fce1e3a424ac5c6de1daf5bc0d9bdc0c58f38b1eaa5e38dc58b36ec6d2c5e
b27e6f63e3b0faa5cbf663d5f2c53f7daec0f8cc4e054661d7b34fbdf15c47a5987370c6b4e537b5481204c080dbb578ec5f6bb1a54b8b8773ec7fb4
0bbfd5a3fddfe8910656d4e8b5ec45aa1504911504c88d5668f46b504e1356fff168b58ff8cd69c2ed1c16520aa57299a64c5ba62bd397194afdca8c
65a632ff327399a5d4ea09b91662e9e62bc9ebdc76e828d97a607ff5960307b65c63567ef5dadff80fcc8297af9c3973e5eb774e7fb393bfc35bf9f7
64c6a964ad363688383c4a765a4b1c2a3ef14e7758974f6c30ad63afe3b108f2872355cf98a178c5a4242faf2d5d6ed1adf3fac54f23695976dd8486
7811c87d77375656d2d838f8e5d2f7a0b3f3bdd2978541e4199f57c2be8e83b2be2e2f971fe3ffa0f7582efbaecb317af58663883b0b24ba6db281e6
99012b4d0dba631abdac056d8655711aaa2590373cff9ee2fe0e6706ec0e5034e65d866fa92b08c7448deeb3f379e2e3e89a807ee178d86a693ede71
889455305d9268b4c2ce2ff0348d1607577c7b8809be2dc4845b5b08a86055a2adc25e15dc68111b5d0db7f610f785694d1aad2d66444f85abf3b7ed
21c897fd74d5fcd355ebed7b88ae2d04c491af70cf0bd7871bc2fdfa9193ee63e8e3374437443fc430c4cfe00007eb21f4d4f734f40a48b025d87b05
f68cec1915ef888fee1157a1af3054f85518d55fb40541d6cb06f443239ad01fcd1882a11886e162842e2e217e58fc43f165f1e5f11be33df1d7e283
691bbee0b79b15d9f9cf9b9501841dae1db76f7255d5b4adc39af6fcf297c9a7e616bc9dbb6a5dfe7ef7fe273f7dbfe0b038ec60cf9e5959eed1d1a6
5edbab761e713a8fa7a4648f1f9be9f2ef51bd6ad78148459703c9adfd28ed221ba4b5d72469fdf105b0b063da4abd8130a63966b69a141b4c6ba2bf
24df5ed87bec206ffa92d79b2a5b745be01066577607b40e275bd81256cad78c2d79fdf50bcf56564abbf89b1b3a3c55e376ecfe48c8d9c0ee547ce9
41b2c249aaf5db60883bfc96fdafd3b363b6063fb27e9b611cf9810cbb628ea9de19d59274d30914da4f284e20803cb8d7ec6eee0162d941c509bcd8
d070f7cb8b4ebec3fec88e0a7b3b7277ef3e5e2b94def01c28987e0df729d20f250f5426e6800c37dc7168112551b0304152221464909905404e1710
de906409054627708d72cad72be71650ce2dca814639752b871a500f9b41de7376eaaddf34b4bea0feb6e17e6294305b2815ca840aa15cd824d40a5a
65201dea6816db5928868ab174468bc778d1a14d81143618078b89da0cc860a371b498218d92ddda8930916563b698a92d8002360b678933a499728e
76112c64a5582a2e9296c96b600dabc22ab14aaa906ba0866d1376e093e293d236799ff4bc5caf3da1bdacedd4de49132e2059a7fce432f4149bcaa6
9ee2bf6f1373dab3f0c00d0f213484105a4a0819d8ddee0cc942877ad182a246892491090c2d82c00c16e5bec8a2d3332532e8355a8dcea2d56ad2f5
1a91895a424ff0a548583f2f80cae1efc1b1f566e56351e193bbf054d2dedf2d9a2c2a9ab4ae06fd4b3cff15be4fea45511f2adaf5b1faa1e21dfafb
c5073493f405fac56c99b858b350bf5e5ca5df2eee16b76936eb37eaf7b217c497c43d9ae7f41e7db81e4549d2e90da16897ecba50433cc64a2e5d2f
83c33898a5e240a9bf66802ed590681c8d19d208dd1883db98ade841c8c607a48972b666a276a22edb90692c343eccca8c4fb1ad9afdac56536ffca3
f1b2b1d398a09cb005a78ed11f012ee6f139acee223fca8f5e64aff0e28b2c9ec58b391d973b4eb2063e4a182304f2056c03f87e619c44bb04094099
ecc9f6fa866bfc474dccaff33650ed69f5f740b5964593e9250b857cd2b506fee306798bdabba31d33359b69b503e64c310758039201699a979cf21c
ddebe1ed93cb3adabfc56dec73219161c7df7945c7d5f6efbdfd5889a64ab94350761ef5a74e69aa7e2ef1f263238a552a45bb42d02910c1145652b6
a5d9f3c1468fa6eadbf63dfc016ee3f96c04bb2ac462d2b7bfe524c54ce6946c0db05aecc2fab2c9bcddb3f7a847612440b0b2a5828977747cc023db
677d4bfe73bddaaf4abd1ba6952311335556d4fba1eebc24db89283a15a2c2fa539e8d1f789ab79429cc9cedf8985bf96bec29d6ca5ec007e1bf259d
82a7ead95809bbd491a3a9fa65dc29e572aa8b9f654a5f467d69ac539a653fafbad97799da375939a1535d470ebb44d53b4ea9373942f6f68ef2e7de
7fc83fed2788d2aa17361f3e626ae98a7ff953fb3da66c9d72ababbd79bf43fd34f378048089fff2a7b6f1a6ec7fba138a15cfaaf72520d0f658584b
eb7314d453a8950fc345e11080140495143ea350436127853c0abbe4783823274389b808ce4836a814af4029c50bc45658a0c492158e0aa97052099a
6638aae4c52fd5baa33886d2bda1109d3090ca0f8ac760288521beb14f93baeabdafb0be5bdcfd1904d3613d1c82ebac37ab1322c92dfe9556c7b9f8
84a8173d62ab34513a2afd2c47ca69f2d3f2454d2fcd524dbdf621ed4eed7fe80af577e8ddfa4cfd5cfd527db5fe65fd1bfa8ff55feb7f35c8864015
9d04cc82de3013fc68a531c376054dd12e04522caa775e9395dfee451d819ba8de932969068194f3a605d0b20c5f1abb958bddd21204b371beb40c36
560077432114c1522886593083465fa8defd4d87788a932091de644a4da3160e48a7360ba1844231e4432ecc833e543a1ae653fb7e94ba0be6d2eb80
fb6ed22a5173f914e7539fc5f4cda396faffc6a8036e8e9a45232da6b194db9af9d45ae12397fafcef8d389c52b3a9df4458442da653db5c955abeda
235795c84154e6d3b788da4c23bab3a89d83fa17d2e8b96add6fe94c50a9941047741a863954aa8c5a426d0b554a493476b27ad37aab57571fc13ba1
3a1f51efb6fff94950ed48f9af05e5bf0e4c6055f71d7608a4334030ed2d7b518b01ea0ded28180363e11eb81732613c493f01eea7911e8049900d0f
2a2b1e9d432526330dd3321d9db00dcc8f19358be6cf4a4a4ebccb17a77be3d4aef86e5f3cdc178ff0c519def8ae445f9ce48b937d718a2f1ed02094
bb3b6f706cb3e1af2efc4712fe52833f9bf0278ed739fedd853f9af06f3578cd853f3c7e97f403c7ab35f87d0db6b6e1776df82dc76f06e3d7e97885
e35749f865cb04e9cb1a6ca1862d13f08bcf13a42fdaf0f304fc8ce3a71c2f27e17fd8f0520d7ec2f1632bfe75055e7c0dffc2f14fd4fc4f2bf0c2f9
91d28515787e249efb284c3ac7f1a330fc90e3071cffc8f17d8e676bf0bde648e93d8ecd91f86e129ee1f8f61a8bf47638be15884d1c4f717c93e349
8e2738bec1f138c7d7391ee3f81ac7a3166cac70498d1c1b5e7d4d6ae0f8ea9129d2abafe1abe5e2913fb8a42353dc9d78c42dfec1858739be528387
38beccb19ee34b1c0fe6e18b263cb0df251dc8c3fd755669bf0bebacf80231fd421beee3f83cc7bd1cf758b196e373cf9aa4e792f059133e93871e6a
e2a9c1dd1c773ded47fb547cda0f773e1522edccc3a77698a5a742708719b7ebf1498edb6a8cd2368e3546aca64ed535b8758b49dada13b79870731b
6edaf89ab489e3c60d53a48dafe1c67271c3132e69c314dce0169f70e17a8eebd6f693d6715cdb0f1f27311fbf0bab1e334855367c8c0e4f54509987
158454850bd7587035c7475759a44739aeb2e04a8ee51ccb38ba3b1f59b1427a84e38a15b83c0f4bb3ec52a90b97715ccaf161132ef1c3c57a5cc471
611b96b461711b2e68c3228e851ce7739c1b8d7338ceb6a44bb327e02c8e3357e00cca1470cce798c7713ac7691c7307634e1b4ef5c3291c1fe43899
63f624bd94dd8693f4f8406088f440124ee4783f8d7c7f3a66d97102334b1382f13e1b8e1f13208de79869c07b398efb9d591ac7f17766bc87e358aa
19cb71cc68b3342600474718a5d1661c65c4911c336a70440d0ee778b7d057babb0dd35fc3bbc6a29be3308e770eb54a77da70689abf34d48a69438c
529abbd31f87187130c7548e8306daa4416d387080591a68c3012906698019530cd83f12938d987487414ae278870113130c52a211130cd8afaf4eea
67c6be3aec9384bd7bb9a4de79d82bde2af57261bc157bc6b9a49e77619c0b635d0629d61f5d06ecc1d1c931c61fa349ce682b3af230aa0d234984c8
3c8c3062382118ce31ac0d43d3318432211c83f33088900ae218489d0243d0ced1c63180a3951a5839ed98fb4a967434af40ff3c347134fa054a468e
7ed4da2f100d1cf566d471d452332d478d0de53c14a952a41960472a454e7b27b324f4456646e0c81a58de9af5acf7ff0f0ffcbf66e07ff944fc275c
3394ef
>}
\immediate\pdfobj useobjnum \csname EF6O13\endcsname {[ 32 [ 318 ] 40 [ 390 390 ] 43 [ 838 ] 47 [ 337 636 636 636 636 636 636 ] 68 [ 770 ] 71 [ 775 ] 76 [ 557 ] 83 [ 635 ] 97 [ 613 ] 100 [ 635 615 ] 103 [ 635 634 278 ] 108 [ 278 974 634 612 635 ] 114 [ 411 521 392 634 592 ] 120 [ 592 ] 8970 [ 390 390 ] ]}
\immediate\pdfobj useobjnum \csname EF6O14\endcsname stream attr{ /Filter [/ASCIIHexDecode /FlateDecode]}{
789cedceb7aa03411004c00679efbd9efcff7fa39643089e824b1428a90a9a69961936f952e3a337d32ad97eb54ebae9a59f418619d5de19573979f7
e9bfd759cde6bcca45962557596753f56d76d9e7906399ff72ca39975ccb7cabfd0500000000000000000000000000000000c06fddf37802a19f0298
>}
\immediate\pdfobj useobjnum \csname EF6O15\endcsname stream attr{ /Filter [/ASCIIHexDecode /FlateDecode]}{
789c5d923f6f833010c5773e85c776880810702a21a42a5d18fa47a59d5006b08f08a91864c8c0b7afed6792aa48c9d3fbdd9def645f782a5f4ad52f
2cfcd0a3a868615dafa4a679bc6a41aca54baf822866b2178b77ee5f0ccd1484a6b85ae7858652756390e72cfc34c179d12b7b7896634b8f01632c7c
d79274af2eece1fb540155d769faa181d4c2f6415130499d39eeb599de9a8158e88a77a534f17e5977a6ec9ef1b54ec462e7238c244649f3d408d28d
ba5090efcd57b0bc335f119092ffe2518ab2b6bbe5c7361f5243cf0e1f819f3cbe59445bd8d6475b8f3b8793d4636fd12189203124811cb60a77c001
d64a0d05e6c0dc63eeb100161e0be014475ba9a10e67e86fa58602a365e667be5b44d139f3d7b0d9e346918439b8bf426f33092148b7e5b8128e5be0
194a368bc93966e0e996831234b552430d8e937d63b011f7147fedd9befff6d07615ecdedef64c5cb5362be696dbed96ddaa5ed16dffa771b255f6f7
0b8c68d42c
>}
\immediate\pdfobj useobjnum \csname EF6O16\endcsname {<<  >>}
\immediate\pdfobj useobjnum \csname EF6O17\endcsname {<< /A1 << /Type /ExtGState /CA 0 /ca 1 >> /A2 << /Type /ExtGState /CA 1 /ca 1 >> >>}
\immediate\pdfobj useobjnum \csname EF6O18\endcsname {<<  >>}
\immediate\pdfobj useobjnum \csname EF6O19\endcsname {<<  >>}
\immediate\pdfobj useobjnum \csname EF6O20\endcsname stream attr{/Type /XObject /Subtype /Form /FormType 1 /BBox [0 0 306 147.6] /Resources << /Font \csname EF6O1\endcsname\space 0 R /XObject \csname EF6O16\endcsname\space 0 R /ExtGState \csname EF6O17\endcsname\space 0 R /Pattern \csname EF6O18\endcsname\space 0 R /Shading \csname EF6O19\endcsname\space 0 R /ProcSet [ /PDF /Text /ImageB /ImageC /ImageI ] >> /Filter [/ASCIIHexDecode /FlateDecode]}{
78dabd5a4db3d53612bd6bcf9f7055365029f4d4dd524b5a0e930953a9d9105ed52c862c524008141002f9faf973e4eb6ba9759f2f2f64007885396a
d9dded3efd21f3f34cb3c76f9aeff9fae7c9ebe9eaab67bfbd78f2ecdb07f7e77f3ceafff5e4fd44f34bfc3cc78697f8f91ddb1ee0e7f9546ff17a12
aff8fbd5f23785e414d77ebbfa719a7e98aefe0ef1f7907a304dc22e69cc1a67cea7ab7a1376a594a41dfaaa4329041797ab57dd1d7a747d101f1ff4
1c0a425997a12e1e5b9149830b21784fe6c91d1a1d9f9e3cdd87c5bf4f3f4fd551f7aaa7941d97185522699a495d88229085efee5f4f575f43c2a5f9
fa87c52bd74fa7ffce770ef1eefcdd7cfdcdf4cfebe9e1b4283315486956313a34f0a20a855d489462501fe5962ad0c19f2b415c9c2fc2311b357af8
a222c4c1a5a292858ae45b6b72833b48c9054a3927ab49075fd62426474c9e20ccb7f509dfe893222e058e6403b3872f6b920b02296af0927cbeb526
37f884599dd784e0369af4f0454d98c9a5e4597389e5b63e919b7c8248705228251bae3d7c59931810502117c46ccab7d6e4269f14ef12c1b996bc3d
7c59938c675389ac241a6ea949b03e3161472e67cec4718ece7bcf61bd13b63bcd29d67bd4bb26c7f5ae770e5fdd9daf5f4ed1c5cc1a483de9f14177
0ecf9615a421f52114a69c4f2b6fd715f5eca34f78dc69e59765855d0e301a26a56de5c7754f2ca1249fbcd069655ef7884798c6c0251d57aebee6a6
e6eb45a8f37ccbb6829729e4032f8e3f79bba137b242b6dd785331c1c41862ec9d95162faeee5a5f02b2c56a0617884785cb4e66f87545946a54211d
cf8cbb94deeb705939dec7efdb13d5e13dc04fc69e865eb627c2f9395154bc30e9ec4941d2a73387f6cd491977f714a231a7a197cd49086314360459
41947db6d7c3fbf690f795b28a30ed0deae0cb16155486a01c4bacda7d368be48245cc2e82e1a1588b1a7cd92222ec493113b1a4cf1672a137488d41
b001c557b92c1ddc39ba98d32510ec47ae4046eac51b3a8ad71cab419315dfd0515c10c13904ab4c4307f180c7164ede2ad3d0511c1915f594d98a6f
e8289e9cc4206aa54fe0288ccea1e422d14a6fe8208eda837e84b291dec0511857825640acf4868ee2a1a6014fd6890d1dc5233896fde0c4860ee2e8
a1333347ab4c434771c432159f06f10d1dc5e1b08cbec42ad3d0413c9153ce21dbf06ae8282e4e329e6a03a0a1a3787059992558f10d1dc5517540d3
419713380aa35aa7140737367410cfc19548d95b551a3a8a677415d0d246634307f102879104b5ca34741447368bd98b8dde868ee26898c8e38558f1
0d1dc515f3224601cbbb868ee2b9e6531a48ddd0511ca1e4130d8e6ce8204e3e399f49a2d5a683c70de45d41754a360c3af86c431d4ec4531c366cf0
d98658bb965806f9137a268e215e4a96f1fe1b7cb661c9fd32bcde0e1e3730aab94736b73eede0b30de4a2a88e4e6df0d906144df2c50f2a3578dc80
9281779478d8d0e0b30d705f6294d861c3069f6dc8c8a798e478d8b0c1671b0a2e316a0f81d1e07143a8e5c0e73cbce8069f6d400f228107d677f0d9
0671aacaa35b1bdcf532cb49536d3f303712e3a56ad5436a1bd6e6af7504e9e6af47875f0f6ff1fba7c3bbc32f87f9f0e2f0e6f0b48e4b883ff42eeb
2f08fe61c6b4f9a62326d474c708c0b286fbf1f2ddb3f93ff39b99e76f66723ed6c3220cac28748a4638a0df91e39550c04a0c887c0dc9ebfced83d9
9e914d010d349a380c82fd9c84b21ce175db58c4043ea38259b41d44f55d05d4468c42c6a27040921cb9ef29b2c0ac9c306877200664a92712c1f413
db81530f9297ea22348b7d37018e4047e47e3228d5314990f84c2f515b59944f413db4f076b0d4b712952111ed2cf2b88111749e056fb66f24708142
8e1a61fb0b74b0a8635c94cce4d61d205938ad73fb8c97e0111b65b11c75594282472d9cd1c113a375ae30e622f4110b7cacb4c5a2eda02866a7b926
17c08c71857d0ae8be2dace8b81983fb1c0ba22cc8c27ae4484792b8c800977a068049735604684c81aa35dd899081a592160fada8a6a44be1606497
82b109961918711b940ad3424a74e20bd761872b787b210df076f20338fbb2243746e866e4b9920d8a564aea1c106a87275c381fe1e4508a63b428da
113c3324db0e76e73ba64b7c343d9cff12cf41ede4234241103b73e5484168c10b8a7f205c57a6ffdf78de08d968de33a4d1bc6742a379ffca1acf71
2ff83a2f25a9111d312e52e448ff13d1d1a87a0c992496e8807328ba4459c7740337aa1bb851ddc21bd5412aef7d3952a9511da44281d0180cd5733d
e52a61a47aa55a407e0a96ea066e54af70d6e3d140477523dda85e61f4564b42eda86e6eb2511d17c2a80362a98eab9263496ca95ef0ecaa6ab654af
671788aaa52beea84e90f0685e966ea7e37aed141198bc94eb8eecb5bf1338b390617b8531c363b3a5fb82c7c2c7c4d3f87e94474b1f2ce1173cd573
5bcbf8da94292ae1e2c78ef3471c392259d6a3d2d7c307564b7b4cfa0eaad70361c3fbda6095083da3257eed724442cc1df32b91fdf681c73274e703
d38d1f8d70bb3ff33d0ae2f5b3d20dadc27ca95598f02e52595b36c430ccadad627d13238aa4c1f0ce198a27fbc5e6fa61ab7dd2ea8e89906c40e8e3
3e448456dd6f3e92de4e881eac47d22c288f3eb397d309d1afa7b3a384e156d19396d3caf7eb0a8922754684c669e5dd7abc0cba50d09011a81fdcf3
66f739a7036e0daa5e4388321e975344f681ceed507c7fe5e9e93919b7ab075d793c14c7a0ade06bcda5ebca8b6505f52cd5c162a9db9f44eb7dbffd
b6ac80152520d316affc618fbedad57adfd27dddf62dddf7e8dbdd959f4ecf6104271225f3a7f1f5aea55f3c7e73da2475aa598e3ed74d8fefec2cd5
cf231b658e9f476a76f7927dae5da53fe3d5978b90d6249a91dc4ad36cff4cf6f1dd5dcdaed615cf8190d0b9b980f76ef7c5dfdaa1ad196e3eaee9e9
f2170b72cb59feda5093bf3af44fe52fccf230f2c3f9ebdf2b3f104c9838a13cff952cf57c376ef7f9f17ed5a0fe9f005478699cfa88e8fc28163cfd
08adffd8cd2bfbbabd5e633a423742f7d456f67dbdf71cc3a87fad45086d325ae158ff3fc13daaf35dfba8c1ddd788eefbe4c3697af83fa6e1e3e3
>}
\expandafter\gdef\csname EFWidth6\endcsname{306}
\expandafter\gdef\csname EFHeight6\endcsname{147.6}
\pdfobj reserveobjnum
\expandafter\xdef\csname EF7O1\endcsname{\the\pdflastobj}
\pdfobj reserveobjnum
\expandafter\xdef\csname EF7O2\endcsname{\the\pdflastobj}
\pdfobj reserveobjnum
\expandafter\xdef\csname EF7O3\endcsname{\the\pdflastobj}
\pdfobj reserveobjnum
\expandafter\xdef\csname EF7O4\endcsname{\the\pdflastobj}
\pdfobj reserveobjnum
\expandafter\xdef\csname EF7O5\endcsname{\the\pdflastobj}
\pdfobj reserveobjnum
\expandafter\xdef\csname EF7O6\endcsname{\the\pdflastobj}
\pdfobj reserveobjnum
\expandafter\xdef\csname EF7O7\endcsname{\the\pdflastobj}
\pdfobj reserveobjnum
\expandafter\xdef\csname EF7O8\endcsname{\the\pdflastobj}
\pdfobj reserveobjnum
\expandafter\xdef\csname EF7O9\endcsname{\the\pdflastobj}
\pdfobj reserveobjnum
\expandafter\xdef\csname EF7O10\endcsname{\the\pdflastobj}
\pdfobj reserveobjnum
\expandafter\xdef\csname EF7O11\endcsname{\the\pdflastobj}
\pdfobj reserveobjnum
\expandafter\xdef\csname EF7O12\endcsname{\the\pdflastobj}
\pdfobj reserveobjnum
\expandafter\xdef\csname EF7O13\endcsname{\the\pdflastobj}
\pdfobj reserveobjnum
\expandafter\xdef\csname EF7O14\endcsname{\the\pdflastobj}
\pdfobj reserveobjnum
\expandafter\xdef\csname EF7O15\endcsname{\the\pdflastobj}
\pdfobj reserveobjnum
\expandafter\xdef\csname EF7O16\endcsname{\the\pdflastobj}
\pdfobj reserveobjnum
\expandafter\xdef\csname EF7O17\endcsname{\the\pdflastobj}
\pdfobj reserveobjnum
\expandafter\xdef\csname EF7O18\endcsname{\the\pdflastobj}
\pdfobj reserveobjnum
\expandafter\xdef\csname EF7O19\endcsname{\the\pdflastobj}
\pdfobj reserveobjnum
\expandafter\xdef\csname EF7O20\endcsname{\the\pdflastobj}
\pdfobj reserveobjnum
\expandafter\xdef\csname EF7O21\endcsname{\the\pdflastobj}
\pdfobj reserveobjnum
\expandafter\xdef\csname EF7O22\endcsname{\the\pdflastobj}
\pdfobj reserveobjnum
\expandafter\xdef\csname EF7O23\endcsname{\the\pdflastobj}
\immediate\pdfobj useobjnum \csname EF7O1\endcsname {<< /F2 \csname EF7O2\endcsname\space 0 R /F1 \csname EF7O9\endcsname\space 0 R >>}
\immediate\pdfobj useobjnum \csname EF7O2\endcsname {<< /Type /Font /Subtype /Type0 /BaseFont /GCWXDV+DejaVuSans-Oblique /Encoding /Identity-H /DescendantFonts [ \csname EF7O3\endcsname\space 0 R ] /ToUnicode \csname EF7O8\endcsname\space 0 R >>}
\immediate\pdfobj useobjnum \csname EF7O3\endcsname {<< /Type /Font /Subtype /CIDFontType2 /BaseFont /GCWXDV+DejaVuSans-Oblique /CIDSystemInfo << /Registry <41646f6265> /Ordering <4964656e74697479> /Supplement 0 >> /FontDescriptor \csname EF7O4\endcsname\space 0 R /W \csname EF7O6\endcsname\space 0 R /CIDToGIDMap \csname EF7O7\endcsname\space 0 R >>}
\immediate\pdfobj useobjnum \csname EF7O4\endcsname {<< /Type /FontDescriptor /FontName /GCWXDV+DejaVuSans-Oblique /Flags 96 /FontBBox [ -1016 -351 1660 1068 ] /Ascent 929 /Descent -236 /CapHeight 0 /XHeight 0 /ItalicAngle 0 /StemV 0 /FontFile2 \csname EF7O5\endcsname\space 0 R /MaxWidth 611 >>}
\immediate\pdfobj useobjnum \csname EF7O5\endcsname stream attr{/Length1 3780 /Filter [/ASCIIHexDecode /FlateDecode]}{
789cb5567b505cd519ffcef9eeb9775996ddbb0bf2300496908d50c2435010ccc40de1a5310185202459b30bbbbcd9254012484a421b188d51a331b3
558c4d6a531be38b5adba1a50fadb63ad5ccb463328e7f389a319dd669ac76c6474bcbd97ef7423066ea4cfda3e7ec39f7fb7de77b9e73ee77171800
38695240afadaaae8144d000583a71936b1bea1b573e95f713c2d7110ed7366eaebce6e44d2f133e49385adf5858dc931b9e05e06b0837b7f70706e0
71fe0ce1a384fbda770dbb1f3b3ebd0a0009b2be8e81cefe3ff7fcfd1572d6478c439d81a101f246fe947f10b675f68d766c9d3ca40208856cfcbe2b
14086acdbf7101687fa0f5d22e62d89e14bf03b05808afecea1f1eb9ea25c824bc9a707a5fa43d006fb31708d711d6fb032303d8a7c6116e21ec0e07
fa4305a5eb07098f907d7d2032341c3b03e701ac77d37ad1c06068e0c8837c05e1198ac906c6ded860a1714208aac93386150a600df0aa9a5b9bc0de
17180e432aed21b5580c608932a57b438361b02cea315ae3e6d342cf2453f26ac820cc4d3fec3209c54497fc1d81a39064fadb13180cb4c14460b03f
0c136d83816e98680f848768ee0a0dd23c3ad807139da108d19d83a15e98e80a8449a62bd4469cde403800137d8188db9829ee89fec070174c847b0d
4ea433d00f13833bc32439dc11eea4b9cbb07f596e5f34a6d8d90320e8bc4ac42350ca328c67ac1adf820e4ea7c6e3554454e2b9b298c3526be8a80e
12a7097ad52499c4a6b47ef6fe97647071a49b9903541162265620879eaab937dcb060c6d4043db1586c3af6e8928d8573cb593a3d43666118d1f4d0
e8a541b1b36c560833f03a8d97e034bcca5f85135046fd76f8153bc8f369e509b80bde161c5e8063ac9c25b1725a7d534d5247c5a43845ebdb48b78e
acbc0d0f914dc3d20cece77b790374c0abe20c4c518f98fc8fe1e7ec009c8387e1755e079fc1016c82c3d4a7604801718e5941f23c386978a20ed046
6339e68b7366ff18f6c35ecae1a43aa326b137cda89f602fb38bf001c5fc266ec31da4710c4eb149255b39a5d4c1e18578d10f87f90136a5f8cdbe97
3664371c53fcecb49a04bf3562254e0345da01bfa0b11bceb01bd9241ea4c8f61a1188737046bb45295c884a1bc3eb291fa0f12c3c0ff918257d3317
b5038ef10ef2f519457206ab2097ac0d015d7388c2553f5585829cc16ab73ecd3d3707a7bdb7b5b85f6bcdca5f7d0574eb9a7b1a1aa61346dd33b158
438bb24cb44e8bf469f458a6154ff6f9af5a3c9fbf7a43438b7bfac7d5558b56abfd55c46b6c21d240c4267e7555fec29da037906e8f71b3a83e882e
ca5283216f9622842a18e728545554a87c0bd428a02902b9ca34950b8b7ef6c2d9579ce5e5ae7267f9b550f8c9c5e222f2a1e9e2c39462636e5df1bc
c619f3793da4c035d494144c513cccc3d7c25aad1eea351ff3713ff3f3e7e039cde54bccd2b2e81a668d2a1fb9e7ff1a83d719cc7f2cce9d9e3b27f2
28ce589a6ce66574527670c11bde754ee172b874e170b82c9ab1a56a1c1776b033baccf60ab4703bb782b3dd611db3ec877d8e6d894ea63bac8971ba
c3121fe7ca7140a27ed69b30007e68002f14811b54dfd90b050505fa8b2f8cdbc7e3b8cf4989b9ca2935fde2854f2eea7f33008d4b491ac399527e89
6c5dc1bcb90e874377381dae4c47a69ee9cc74153a0af54267a1cbeff0eb7ea7df157144f48833e24af279e3f6274612b727d6272a3e56829aaad959
1e2b2b292e2d63255a36cb9d7987798af29fd9559d51d8b7359395cce4cf28fe7fdd2b7e5dbbbeeed896136ace0746a5a00aaffc89f6c40ae7bd75aa
d02c9a108ab05834cc5584c27319271c27c06aa135b0723a3d2b588d3db256580447dc07dbe2d4388b55159c2954143cd695f1fa2717dea0ccccdc0b
d75cbc907259ce960f35fa2d52fa174f3a6c8b621cf65ac19379296fe44da2c9d225ba2de3fc216eb5a145a9502b34afad055bc41dda1d966ed1a7f6
68a338aa8e6a93fc7e3ca8dea34d61920f7c6cc3b4aba9e597101f7b91328ca7eebdc15976432b6d132b89a31b82d96cf7cc34dffdd1bc97f377e58a
1eaa03ffdec9de9affcbfc0ccf9e7f67b186f2d647b20afd95db1d6b3e854c8b5900ffb8cffeeea5e7e7efcd9724d4580e996fc0522da7f7ba5f2e07
4898fcfcbd7ffe2ca10682f45fe0f2a629467d30cc9fa67188f6ff7d68e1e5b134e571e334965a3d7d8343661c1a6e823ce8a25acc41a7bb465f7ab4
f304f26ad4600db618ef9e42df675664d67b8366904c6881e674a96b16698455ac6991562ea305a4b23d8bb40a2bd911584f756900466110baa193bc
0fd30dcf8176aa456e28a6fb5e042544b591841b2a4966982ad43049872000fdb09ab8374398e40b885a077dd4ddf42db8646bc844217a86486717cd
4192b4fe0f5e4b97bc3691a75de4ab8774c2246dc411209dafe7b18aa81ed26b869d24d14eb201d35ac8d4089819b9c94a98e601926923bbdd24e726
fd08790f986b57da6934ad0cd1292ec8ef206ee82b64dc5748359b110e118e985e8b29ce12b8fe4bda9774f3afd0a52f74ec531afbe856fcb7a69977
94d34997c146d834c3c7bdb11f499cf6e073c5f86c149fb1e3d34376f174313e25f1b4079fb4e3290ffe308a4fcce10fe6f0a4c4ef57e0e312bf578c
278e378a13513cbe719d38de88df2dc6c792f058141fb5e294c4475cf8f0187e6716a3128f92c4d1317c48e291076bc591317cb0161f38bc4c3c20f1
f032bc5fe27d12ef957848e23d0733c43d120f66e0ddc57897c489643c20f1db12bf25715ce27e89fb248e6df088b1207e53e25e27ee199d157b248e
8ef8c4e82c8e8e2b23bb3d62c487235e65b7077749dc19c5e1200ed9717087470c0671c7804becf0e0800b231456640ec3de98c47e897d127b93b1a7
bb42f404b19b7c745760d7a678d1958a9d1d76d1598c1d760c0531486ac128b64b6c0bd8449bc4800dfddbd3843f88dbefd4c5f634bc53479f15b76d
4d10db246e4dc02da4b1258aad2d76d19a832d76bc630e9b37cf8a66899b9b7c62f32c6e1e579a1a3da2c9874d5ea5d183b74bbcada140dc26b1a100
eb2988fa24dc148f1b29aa8debf0567adc2a71c32d4eb1c183b738f1668975b54e5127b1d6893512ab2556495c5f3926d64bac1cc37512bd7378d31c
ae9dc335a595628dc41b5fc30aa22a1ab15c7a07f086312c2358aae48bd24abc5ee275124b2ab0780e8b6c5828315fe26a8979b49c772d7e43c75cd4
456e36e664e035abece29a20aeb2a3875985a71857da52c5ca31cc1615225be20a422b66318be4b396a13b335eb81d48ff385ef44e2999f198118719
5e65b98ee9249e1ec56551bc3acd23ae0e625aaa4ba47930d58529c91e91b20e933d7895c424898973e872a6099744275975a6a12ed121d14e16ec51
4c20870963688bb7095b2ac6dbd02ad1424b9628aa24ae4a149485a8408590928fa8d37f29abe0a9c8acc8bc0aa4239b61c1c9fb58deffb7c1ffd9fe
d76dcbff035f8b097c
>}
\immediate\pdfobj useobjnum \csname EF7O6\endcsname {[ 84 [ 611 ] 106 [ 278 579 ] ]}
\immediate\pdfobj useobjnum \csname EF7O7\endcsname stream attr{ /Filter [/ASCIIHexDecode /FlateDecode]}{
789c63601822808504b5ac0c6c0001a90010
>}
\immediate\pdfobj useobjnum \csname EF7O8\endcsname stream attr{ /Filter [/ASCIIHexDecode /FlateDecode]}{
789c5d504d6bc3300cbdfb57e8d81e8adbd0f61402a3bbe4d0ad2cdd69f4e0d872302cb2719c43fefdfcb1a530812df4f49ef410bfb4af2d9900fce6
adec308036a43c4e76f612a1c7c1103b54a08c0cbf55fee5281ce351dc2d53c0b1256d595d03ff88cd29f805362fcaf6b86500c0dfbd426f6880cde7
a52b50373bf78d2352803d6b1a50a8e3b8ab706f6244e059bc6b55ec9bb0eca2ecc9b82f0ea1caf5a1589256e1e484442f684056ef633450eb180d43
52fffa5551f57aa59f8e915ed257c98f049f4586cf7d819fe5234dfdd3a705e91aab7b397b1f8de79365c7c9ab215cafeaac4baaf47e003bd07a48
>}
\immediate\pdfobj useobjnum \csname EF7O9\endcsname {<< /Type /Font /Subtype /Type0 /BaseFont /BMQQDV+DejaVuSans /Encoding /Identity-H /DescendantFonts [ \csname EF7O10\endcsname\space 0 R ] /ToUnicode \csname EF7O15\endcsname\space 0 R >>}
\immediate\pdfobj useobjnum \csname EF7O10\endcsname {<< /Type /Font /Subtype /CIDFontType2 /BaseFont /BMQQDV+DejaVuSans /CIDSystemInfo << /Registry <41646f6265> /Ordering <4964656e74697479> /Supplement 0 >> /FontDescriptor \csname EF7O11\endcsname\space 0 R /W \csname EF7O13\endcsname\space 0 R /CIDToGIDMap \csname EF7O14\endcsname\space 0 R >>}
\immediate\pdfobj useobjnum \csname EF7O11\endcsname {<< /Type /FontDescriptor /FontName /BMQQDV+DejaVuSans /Flags 32 /FontBBox [ -1021 -463 1794 1233 ] /Ascent 929 /Descent -236 /CapHeight 0 /XHeight 0 /ItalicAngle 0 /StemV 0 /FontFile2 \csname EF7O12\endcsname\space 0 R /MaxWidth 863 >>}
\immediate\pdfobj useobjnum \csname EF7O12\endcsname stream attr{/Length1 11240 /Filter [/ASCIIHexDecode /FlateDecode]}{
789cd57a7b5c5465faf8f39ee79c33f7616698e13a0c03c3c53b34888aa24ea4e42dc33453cb0245d4bc80a2b5462e5e56c8d21533b1cc6c72b594cc
d05c052537932c33db5ab5b66fb5dd28bb90b9ad6d85f0f27dde338397faede5f7d7eff33b87e79cf7fe3ef7e779e7000c001cf490c17be3b0e1f990
079301582f6a75df5870f3f894e33de6527d04c1c51bc7df9ab7f11f3bdf01c05ceabf78d3f51346f44a9ff50d4dcea4feda9bc767f8efce599f00a0
bc46fd13a7cf2b2acb7fd3ff2e80eaa5fedf4dbf6791176627e400e8b6539d9794cd9cb7a0ef3d770318a80ecfce2c2a2f031ddd602ca0ba79e6dc25
25c33e4e3650fd0e80e89659338a8af5535fae02480e507fbf59d4607952b785eabfa17acaac798b7e7392c73d46756a8323734ba717adada9fe2d80
8fe6c3b47945bf29939f501750fd25aa7be717cd9b91feede026aa7f4cf89c292b2d5ff43f2d63d601a4eea7fed7cb16ce281ba4fb3b15d396114db3
40f0ca0ca14ba21ba107b5f5840c2a8b3e23f4815c9086e58f9900d6b9458be6430cf195aece4e80cb253192cd99b1703ee84549039956106f3d482c
538c64c3d85caa39e0aaabf3835f97436f7e944adbb5fa41ed795af4741e24f8e0da79fffeea9cf36ffafeeb55febbb19dc3b4e769c2f174d70c51be
86caa2ffb88ce0998df81405d1100f6ef04032b539c0ae71565c424ea14b06859e6ab8a6a3112a245c1e77f52585e789192a8dd48381a46bd2a46f01
2b44d09ef6b0ccb7c26e706a32bfaf6861d134a8295a386f3ed44c5b58341b6aa617cd2fa7e7ac190be9b964e15ca89939a394ca3317ce980335b38a
e6d3985933a651cb9ca2f9455033b7a8d42b9ea43bbf9b57b46816d4cc9f235a4a6716cd839a858be7d3c84525f367d2739658ff5fe897c699b9b367
165da36372985f0cfa696f85e87112cf92a11be9301207e34883c53b167a13fd71d08b9eb1a4e112b531baa3c3347f0c5f4201e49a41f7a1c6ae15b4
7d21d9c8d930a30baf6567579dcd0a976ffd8f720566a575eff9cfe3604a68ac785f2e5fb5c635ed53aeaa2fbc529606005c22bb87dccb53bb910eb9
89178cb8330a88eb540ae9834fe3444488df7421a48bf186c5f44ca379dd0c6ba93f954a60a805b86a5c8fcbe3ba5f1ed7ed57e324da81c609c15189
7af521e975f546885e45686884e85542be87c9565623f45bc952c803324fe88def4289445e4432a9887a5992423a0010e822b4a0647831d5bcc90ed5
c99d6cb36e1efb2c3ca60bef10b841481e60bfc6155197e1517a278357f35f5ea2ca4f9a350006c34da41d45301366c35c28837be0fe6487a6a35ec2
b80ff4d5c68cd1c64cd7c6cc270e6b633a3f231ff0d7ce773bdfe9fc4b6773e7cb9dfb3af776d677eee97cae7377e7337f73fc1fadf59757c847bf15
aed9b47d43809a6405ff43fcee43e087904fe84bd08fc04430200c4267068741f07d0cc14de1350bc2207c7451182209a613cc247012cc0e432cc1dc
30c411cc2728238827107a28f4dc43707f184803921d218014960d0d7092eea350075bd8d3542ba1f605d41294f6c12a584c2dc7d849b65aea4d6d4f
c305384d23abe124d6c9c0464116b502bc479a72914d20f96d6139ccc972742ab984b1f27ef916b9413e279f82fe72b97c4a2e94cb59166e53262a4f
13e4e02ba43f2720111ad847500e87f02bccc22679986c858ff014d6c1e7b48be0e5495807dba1827071b252a8942aa45ba8e555e5146ca6bb94fa4f
b1adec34617788ad84b3f028cad208d8cace125d27e19fb01227489524a22ca984f07f95d63a45f337433919dd5966042ef5a436c29ef69aa63d13b0
b77256bb2f4025ed3c01b6ab0daa53e7a35d04c79e66c758abba0182701aefc005f83e5b25fbe49df2085817e20016c23a5a7bb398a396b02544bbb8
2bc4ead2bd7221ab83afe442dd345afb154191d07ee916a2a8049a08ee556d44d320b60a5713a6a237014ee946c919349f56d02d25aa014a311beea6
5205ec817dd01b6b611dada4d1abf657fe4933b7c89f10cdebd85ae99f700a87917696c8e789d74285c8331cd4a98a8c12835e5e5bbd943ab2b83e30
6e92f7b5c949bd7bfda2eab5e9bcf550506f59e26de8ec2c9824c72b93eb15773da6eaebe554df27ffaaf393debd46174cf2d6770c1f165e7578e130
6a1b3f898aa246cdd43e7c98d62736ad5752e96f6461bd77fa2cef83b6077d031fb4cd18d83be4a1288b216b175ea29e7f2f55a80e8a99fd0311eaa3
b0c96ad1013a5488345a6d1f8cae8f9c30a9118c9d2f0d983cba3e422b4360c0e4167fab3d27e73ab0b5b766325572391dd1be3429bbafa3bf5451b5
62e5aa60edc64736a98e2ff89073e7f8a0cfbf61c73ffe8835b7d27edb69bf526dbfc440844eeca7636072c8917aa0fd722f5e5937322bcae1724a3a
5f3f47765f693b2db9b136b86ae54ad5d1ca733ffa980ffce673f6cab973ec655af57d72877b801345f683b05262b11023d31ad741466b66ff2c97ef
fdd3a739179ea9334fda473110a177c009714c62521c02e6494fc27299dc3466346bdb5f6ccd24d6ea6ccab7022627b32ce693f4751d3fd529677f9e
273858ddf999bc8eb4da4411d7178854830e089ad73bd6c418dc111e74bbe26368ff8b8481ade562abed7c264b96ec364796df61b749e97eb0dbc097
2c9ed2435b9e7882fe9e78e21233f01f2f5de23f328352c04ff137084ed1c659ac2fcb0af2725ec5ab79395bcb96b0fbd85ae1ff3f21173985a82191
045c791894a5a0b25c0741833e517523243293ed4c586a4c48adb539c412ffc5d6331a8144d8fe088c90a5a9fd93ec4a766a963dc995c4d928fe189b
f13a1bd5bebd4e2e1fd130a2ed6c1d2d407a2e8f228addb035901e1b178f316ebb22835d51e43cdb53f6472c41e77a99fc1dd88c1233baa36da826d8
da47d7bb268cae8f9a70fbe87ae784db091314bad47ca6f5a597ec8e9c303657b15bf996d5bb6df6e81cc22de0bf559ea84cd4dd27dfa7dc135f1dab
236f182bc79159b817c13deae2b8f2f845ee155015bb226e45fc0af74ed8196f9f0a53538988ec7ed07f08cbee9be64b5675d94358965f7639559d0a
e4828fb68f21366615ddf44cd55da77f73df99495f32e7f0db63f9c5bababa7bd9fa81f3368dbcb736ef8637aef37ff9f21d3bca12f83724ef62a27e
3bf15a22cd5d1e486016b400a2250fd0a40b2a0c971b98d9086e552f9b35bb3111d5168d60b320f84c6e73abdf2ee86d3993dbea77e48408964f10d1
2704a93d4c947c8ca0a3dd6cb8171e045d14eb0969ac27f66363d9cde69b2d1359095bcceec355cc42241a581266d949b1ed3e7b5236aa5c623c9b9f
3d7ba2e34e25b5fd333cd59eb5930759e13192db56d2d462c23c01ee0cf8e4389dbdca961017d43983b6d5162908cb2d6b74db3dd16e664437186daa
c7d6ce84c0bae465bbcaf66d428b4876b6e6f342b185667f71dec69b6de749969991c477bb6035b89c90949c969eeda14abfb0103ec4d88e60af49bd
da580a3fc3bfbbf3d8ac292fcd79eef5d79f1bf7d404e56c1d7f3822829ffffaeffc07aff7e4759907b66c39909246dc5e47d8d76a76960293022991
2a58aacc108c5283eea81db6a07975f27af79a5473b2c11deb89746352627c2a191e69548b667a2ded2db6f3e21608069c14a3d929e9149e924f2a27
55a27b9f479acaa6b264d5e58c0ae1ca5c7d982f59c22e427c5e61a649fe2869fb034f3ef90001338c797ccc6ba72306ed9bf30953f8854f79073fcf
0a58fc98c771d0a16d4f1d3efcd4b643d292869434fe3dffeeb6a9fcbb6fbee05f6b863b8dedf0881ce721a2ea218d2a1fe5f53704526388a67435e8
e91d74acf7ac49ff43668c39a587db95e28e30902f21871291149f696b6f6ebdd8dcaa91a3b9a956cdb7d8cee7e4645e4d426a1fd2fc942c7f945079
0f7339655f724a76df7e915d03481ed243353b76d4d43cbd83ef58b11e3afff6115fbffce13ff01f7ffc91ffb87dc4fa952b366c58b172bdf4cae6ea
eacd8f57556f9ee8ddb7ec85b7de7a61d93e6ff2f175ef7df9e57beb8eb3a2452b562c222039bdc70a95f7711bd1a60377c0025b24758bac57980c3e
bdad3db7d94fe89e116e476889b829d969e30270dbce0bc2bf9c200bbb9bf45481d48059da480e99b97120c8aa90a7309cd6cc80215357a05b86cb64
994d8dd4bcd589d7a577dbef52ce867c14ada126d21a66f853a02fda757a9d6467925ebc5032180dcc6e341af28c3a09f508cfeb4d8a414f615b31aa
6e798891d2650bedd5de4ade87f6d32c363ae74a20d0874104847d6556c6a60626a22e421f61908c2ec9a98b34a649693aaf2ecde835f6d5651b674b
f74b15ba25c665d20add0a638d1425331346b278f4b15e98aeef66e8cb7271a27eb26186fe6ec33dfa2586e56c2d6e648fa3932c3c32896c9cac9bf9
98dd7782f5664b5925ebfd0aaf3cc92b9b95b3ed7afca9ada792d84e4972db27c47fc1bdb544b99e32dc818118b6d1061b0dcb1d36a39e0e144aac65
a81ddc06d949f4f9bb08d4a240c014e14a740d75dde57adea59025d8c31a929a243ca6dc9376671bf8dacd9bd7f201ecb54b8cf1ce4bfc7525a3e3cf
0f57573dfcf467ef7ff869c74edabf8234ba37e565463a5735517c4834451bacb02b5a6db4dabd558987dc8dbe06fb9a683344638cc5a03725a2de39
3c8dd079e34cabdfafa9734673cbc57652e8e39a7edb7384c5cecf4cc8f46426667a3393329387a60712029e4062c01b480a24172414780a120bbc05
4905c905e965e9ab12aa3dd589d5deeaa455c935e9c1f40be99eaea95d93ba26147a0a130bbd8549659eb2c4326f59d232cfb2c465de65493157fb82
c1acbfdd976d25679046b6939574753089928e7cb47b79e9638d0d0d439b1ed87db2e312939ed9547860c28c2353fe7141ca2aa99856fedefeee633a
96d795141dddf6e24b8eca87faf4a94b4f6f17717b41e767788e78150b4303f150c51e90ad5596078c8d76b9319a9814a77358608473789cad9df2ac
7006c12f9eb7fd705e882bde16bf2cbe263e182fc4a539ab30c2fd5d02d9b0b7c273639f2878e1f8f1170a9e187bd38ea91dfc1dd221f5d66d72f6ee
9e3d3f3b75eab39e3deb5252d81066650e36d0471224ace429aa939cbc1b060462e31ac1ea6c54f46bac0d6c1346cba0976eb43b4cc313b4bcc6ef17
a1bb4578a3e6f399070a3dcb3c410f6afa13661aa53b408c6257f91cdcd6d03070effd273ba1f3e4fd7b3b5e7de6e18777ee7cf8e167f08074e7cfad
3b8b8bd830a6a77b5811779d3c77ee244118af4ae29693ce6265811470314395fe01c5b58b298d667638a6d1d1605ee38e77497a971e464b8e88e16e
0dc5662dbf10cc0b85808ba118d07d6842594230e1ad840b09ca5018ca864a435d43e3955eba0c7d86a197b1144a59a954ea2a8d374c5d20189ca405
318db75a24a0a8a0d398ae932bdbf7994f1dbcfbd569d3df9ac32ff25759f7f64f99ae41daf1c0e646ab74e79423aff6edbba7472f3680195924bb81
7fd8bc69ff9ead4203322885fc89781d0993036ec5c6ccfa5d2aaba664586d324a913ad01914bd25c234c6291228a308bd26117a47d75bb5b208c9b9
cde4559b1d9ae1b4904ddbce53564171e040c055e00aba90502724139808ca942664670925967eaa9f7e13cbe06f37d6d7ef7951753e56306bfabaf6
0c7c7bddd8c3cf0a5ef389f214e2b589cee0a302be587382c1511519d518818d69be86f4264363c48b710969b1a037dfa83a1cdee1ddb598145287e6
969042f0b382d339a4153d96f508f6105aa1893fa4a5d136e94a6e309885558532e3a8e86c3ad4eed8f8c88e1d8f6cdcd1c0795bd1ee71e3b6def2c7
fd39fbee7fb3bdfdcdfbf7e53448835ffbe083d75efde0836ff8a7fcab04cf0bbd7abcf8a7dba74f63031932990d9c36bd4e64e74789c94b88bf221a
f50c58d523f25e689214a697215f84a356816f4bbb707f3643c0506028349419142da6f8445a75b4812eb9f05250757e25e47565bde483b049627ac8
bf7cbe08586c4a4029500a9532e582a28616a10554e7cfad62ee21f2ccc5c4d348e817884503a095a9d5567b83b9c948a109c60a6f98af3966715cc8
c8153cb43b28f0ec2f74fdd925099bf2d9430cd490d3fc8f5cdc70fffd1b773736e6bdb0f8e871697bc71dd2d627b71ed9de512d17ee9951fc5dd86e
166bb28c265946aa8d0e68343788138a23621c3a5cc37f714209f886c65640855aa9abd4571a2a8d95a60a73a5a5d25a195169abb4573882b11762ed
d7e64ad71c64ca1fd9fdecc60dbb776fb8c01cfcfc85bff3ef981d3f3a77e2c4b92f5f7bf5ab2dfc35decabf2523c9215b70b20184e121d2b6ed84a1
f0384302f15d1ea7c1ba86bd884d09e46d6ed4fc4ebef0397e7f08d7962ea7133084bccec71e4a06522fb3867091c8655fad72acbcb171e0de8a37a0
b3f38d8abdd200f23bcf08d8d9b14735d61517f126fe13dd4d45ec9b2eb713921b8e22ecec4089a36ad281dd84d5d6064393cea8ea419fef10aaafe9
11f99a336f08e7b2bf20f2c94821b19057be22ae681c9538b2d79667088f43ab22fbb871bfc37ef248c73e1256c97445a1dd4a2926bc4abba5c3b940
aec52c594de3133d7a83a4338e4f4cf4e4194d9e44d945b162b5ecac72ad8e11b122956245378fd19418af835be2f5569dde993cbc9bc0ea4c6b8bb0
c09c9caee0f183081e8eae64c6fa2d25003aed49990ca48b4c669edbe836b9cd7dc805f632f5320f320c320e320d329bbce0652952376337538fc80c
6786ab4754374fb7c4eedeee4929e955c62a5395b9ca227e036392a41a55139ad182568c401bc6621cc6a35b4e30a467741fdafdaeee95dd9775afe9
1eec7ea17b0c25380baec4ae44269254d577f59121838924b61ff10e1f1abb73caead5d31e19dabce3c7bf4e3936b7e478d18a35339e0d3cfbe8c76f
96ec9787eee9d66dc284c0c8246b8fc7566f39e0f31dc9ce9e3c6e74416a44cac6155b777b842cfb9353f85ed94a364891cdaae8237017d85993beda
68221e938ed91c5661839a53f56b3eb5f562e8c89693b9ef79171356283ca9336a90f0ab69d9c2a3dad9bdac82af1a5dfee28b67b755572b5bf9cbeb
3a82abc76e7ef22f52e13a364478a23d64859334eb77c2a080fb8afdaf31b226678399acdf691a4b7e20df25cc3127a4512dfecb4ea0d4f592700291
145843667739c2a6b13dc2093cd7d070c3dec5475f637f6687a4a73b8a9e7cf2c876a9e2527077c9f40bb853503f983c50a55c082a5c0aa4a35d5664
91152be285920a2ab353da9c2721fc495115910fcba013bf1c6881074281c739419ce4c58110b4037c74e8ecfe2fd26316f8fd08e96ea942aa94aaa4
65d27a69bba4171b19d0405aec62711827a7d1f9b63b7697bdfa6cc8660371a09ca9cf877c361247caf9ca0835a09f0813d9649c2c17e84ba084cdc6
d9f24c65965aa85f0c8b580556c88b95fbd455b08aadc6d5f26aa54aad855ab649da8c8fca8f2a9bd49dca336abdfe25fd47fa4efd1091516719c4cf
38838fb13bd99dc7f81d6d7261fb04dc7d29088cc93c81ed5113c8bf3b0e32d84b19b324d37145fc10f546a6f0bf6c4f9d9af053f85bcd7b526fe57d
e95df16bd3013ae8c80c6870b376c0d18e37dad1467ab74e8cadef68c702ddc3e2376be6cbb6453a22b30089fde5c782879e0ef2f629951ded5fe326
f6a994c9b0e31fbcaae37cfbb7a179ac5cb75afc1e2ee249fdb163bad5ff2cd756e44e5a71b5b6a24b2ce89368c16c565eb9e164f0ad9aa06ef5d7ed
3bf86ddcc967b0e1ecbc9486feaf7f8949b68dc49ce58874d85dd2daca29bc3df8f4a1a040245272b02592957774bcc53dedb3bf26bb5eabcd5bad7d
f3244e646281868af66de46a5cb25cb428fac4a2d2da63c19ab7822737540a644e75bccf1dfc307b9cb5b25d787b38ff51ea289ea6c08840649a96ee
9893622c1ebddd9c64738e49d54e8cc2166db922cb21de06ec068b7d97438aab86984d6aa2a3c9149191fb85dfcf73cffb29f5f1675e93ee5c497934
73d5890ee14d94baaefc87ebb414684ffdf4f434f6f335b950573ef458b76eb3a687f2a20cf1db17e1eba29cfd96803d3e1fa2f551114e59afc728a3
3a26ee0abe3c97b426e0d09383b1555b638e44edb56e324093c204b6e7b9f65b849fd28e4e4ae06b2891b76969fcaf51268c89a33e795408d3e75e68
1498ffdcd828b2922e1c0f3e2f9066fbbe0af3348c63ff4044743ee53a46b35e6f931dd6315102bf107a023bca2c7719644a37ed86268b2410e31a56
da8133e99749a3344a82467ef04ade28121ba5e2179923edaeb6d3eedda1374c09c464e447f7d4f7b0c5bbf4713d0c90a8ea533c86e4b4317dae30aa
d92f9eed1abba2e3137dbb52ec9401f73ed263af0d3645e9529a6213923272735bfc7e7286adb6563ffd85a41c9666ff7efd2fb3aa4be657a5ba0a71
4fa4bb42bcb7b9d3c7ae2451df24ed13ec0ccb1e899b24f690946f7547a60b667631b78b34a98bb71a7569303610d52ddfa2b745c538f53683f8c12f
29de90e81b937e15651a619a1ac4b8bdbb92ec52b5396d934b97d41411e709917431f7d7f4d011f35af6ff226f0fe9e9d5b208d371998667af96cb65
d988ef71d2e4c73eca6ffdfaae88dc1f2051fb6e096fffd6dad2f5fef19df631d6c906f15d5f7ff91b1dcdd3cde3090056fee33b6de3ac937ff5652f
433ea57de7024924dd0f51be9408f504dba9fcbeb48f6c3c1aaa093e21a8252826d84ab08ee021753fbc272f8613ea3938219f2638071572092c509c
b0406e8505d2db9021ca520e1c15a038e090fcb9d67f084751b92794a20ffa53fb1eb909064bfb98acad09507f2d486bc55b5949eb1943a06b228966
5c43cb00980ef7c216f81b73b147d8dfa502e930a6e0709c892d728a7c48519471ca6155a7ce5137a87bd56f7543746b75cdfa727dbdfe7bc34ac366
43bde10dc37b861f8c7663a271907183f173538a6983e9fb101f211327404f980566b2141b3c26b82cbba4287a8bef723a9822bec5c8061a9ca97d03
15650651540b9525d0b3fc7019af6a97af2a2b10c3c686cb2a385909dc00a550064b6021cc8699b4fb22edebee74b2512ff82193ee2c2a4da3115ec8
a3318ba09c6021cc80229807bda87524cca7f17da8743dcca5db0bb75c5eab5cabcda0f70c9a730f3d8b69a4f1bfd8b5dfe55d27d04ef7d05ee2ebdb
7c1a2df028a239ff773b0ea3d2dd346f222ca611d3696c91b6da0c6d469146919756994fcf321a338dd69d4de3bc34bf94762fd2fa7eb9ce786d9572
c2a894ee39d42a762da7b1a5da4a7eda3b0bb2af99d535470a2955e76fb5ff5ff8f595a9e985f8ff15f11f2b767068f9a18b4e6a311461e269f50c4d
3ac3201f6e841124875130066e867144ff78b895f6ba0d26c164b80bbe639276f25598ca74742e7d098e320333ea16cf9fedcfcac90bbf6f08bf8785
dfc3c3effcd0fbfa4cf1cecfcbec7a6785df7d011aa46581ce4b1cdb9cf8732afee4c71f6bf19f56fc81e3458eff48c5efadf8f75abc908adf3d78bd
f21dc7f3b5f86d2db6b6e1376df835c7af06e29779788ee3177efcbc65bcf2792db6d0c096f1f8d9a719ca676df869067ec2f1638e1ff9f16f4efcb0
163fe0f8be03ff6729be7718ffcaf11d1afece523c7be646e5ec523c73239efe4bbc729ae35fe2f16d8e6f71fc33c737399eaac5374e7a9437389ef4
e0eb7e3cc1f1f82abb72dc8daf446133c7631c5fe67894e34b1cffc4f108c7173936713cccf1901d1bab5295468e0d070f2b0d1c0f1e98aa1c3c8c07
97c907fe98aa1c981ae8c40301f98fa9b89fe30bb5b88fe35e8ef51c9fe7b8a7189fb3e2ee675395ddc5f86c9d43793615eb1cb88b90ded5863b393e
c3f1698e3b1cb89de31fb659953ff8719b159f2ac6200d09d6e2931cb73e61a6b3033e61c62d8fc72a5b8af1f1cd36e5f158dc6cc3c78cf828c74db5
166513c75a0b6ea4491b6bf1910d56e5916eb8c18a0fb7e1fa9ac3ca7a8e35eba62a3587b16699bceef7a9cabaa9b82e20ff3e15d7725cf3501f650d
c787fae08344e683d7e3ea074cca6a273e40075a6aa82ec62ae254552aaeb2e3ef38ae5c615756725c61c7e51c9771ace418e8fcedd2a5ca6f392e5d
8af71763c5049752918af7715cc2f13756bcd78cf7187131c7456d58de860bdb70411b96712ce5389fe3dc249cc3f16e7b9e72f7789ccd71d6529c49
95128e333816739cce711ac7a28158d886779a712ac7db394ee13879925199dc86938c785b54ac729b1f2772bc9576be350f27b8703cb329e363f016
278e1b15a98ce35860c29b398ebdc9a68ce578930dc7701c4d3da3398e1a69534645e2c8048b32d286232c7823c7fc5a1c5e8bc338de40c9fc0d6d98
7718af1f8d018e43390e19ec50863871706e8432d881b9832c4a6ea0330207597020c71c8e03fa3b95016dd8bf9f4de9efc47ed926a59f0db34dd8d7
835916f45f6752fc1caf3361668649c9b4608609fbf436287d6cd8db80bdfcd8b347aad2b3187b7477283d52b1bb03bba5a72addaec7f4544c4b3529
6911986ac2148e3e8ec9119844742639d05b8c896de821123cc59860413771d0cd31be0de3f230962ab11c638a319a3815cd318a2645c5a28ba39363
2447070d7070b413adf63cb42dc58862b472b498a3140b47338d3647a189a3d186068e7a1aa6e7a873a25a8c3275caa4012ea456e4746eb029526f64
36048eac8115af5acb7afeff70c1ff6b04feed95f0bfd19d398e
>}
\immediate\pdfobj useobjnum \csname EF7O13\endcsname {[ 32 [ 318 ] 40 [ 390 390 ] 46 [ 318 337 636 636 636 ] 52 [ 636 636 636 ] 56 [ 636 ] 61 [ 838 ] 76 [ 557 863 ] 80 [ 603 ] 97 [ 613 ] 99 [ 550 635 615 352 635 ] 105 [ 278 ] 108 [ 278 ] 110 [ 634 612 635 ] 114 [ 411 521 392 634 592 ] 124 [ 337 ] 8722 [ 838 ] ]}
\immediate\pdfobj useobjnum \csname EF7O14\endcsname stream attr{ /Filter [/ASCIIHexDecode /FlateDecode]}{
789cedce370a43311005c007ce39e71cee7f460b171fe3c22edcb89881e5ad247651f2a3dadbb99e46d537d34a3b9d744bdf4b3f8392c3ea75f471ef
3893674ebffe60566a9e459659655dfacdf3765b6a977d0e258f39e59c4bae2f53b7af7b01000000000000000000000000000000e03fdc1f5b2a0298
>}
\immediate\pdfobj useobjnum \csname EF7O15\endcsname stream attr{ /Filter [/ASCIIHexDecode /FlateDecode]}{
789c5d533d6f833014dcf9151ed3212240308d8490aa7461e8874a3ba10c043f22a4622c4306fe7d6d9f215291e074f7ded9877984e7f2b594fdccc2
4f3db615cdaceba5d0348d77dd12bbd2ad97411433d1b7b367eed90e8d0a4263ae9669a6a194dd18e4390bbf4c719af5c2762f62bcd253c0180b3fb4
20ddcb1bdbfd9c2b48d55da95f1a48ceec10140513d499e5de1af5de0cc44267de97c2d4fb79d91bdba3e37b51c462c723446a4741936a5ad28dbc51
901fcc55b0bc3357119014ffea1187edda6dfdb1ed07d4c08b939f219fbcbc5154c9d124f655d0b8838a159368ed7196e408ca615969baaa68c23616
6a206401597859403eb64e3e7af9415d35458cd4bf58ea5f8c2396851a0839819c79d953a4e448c9f9da03cb09d41f11f767c311c3420d848c33ca7c
1c4f79b7aaae298b41fd21ad14613284c9d2b507166c97f9ed32b75d1c47d60aa881173b0feb87b7a361e7789bbbf6aeb5193937ec6ed6ec94f592b6
ff418dcabaecfd07a090d7ea
>}
\immediate\pdfobj useobjnum \csname EF7O16\endcsname {<< /M0 \csname EF7O17\endcsname\space 0 R /M1 \csname EF7O18\endcsname\space 0 R /M2 \csname EF7O19\endcsname\space 0 R >>}
\immediate\pdfobj useobjnum \csname EF7O17\endcsname stream attr{/Type /XObject /Subtype /Form /BBox [ -6.4 -6.4 6.4 6.4 ] /Filter [/ASCIIHexDecode /FlateDecode]}{
789c6d90310e80200c45f79e820b7c02b506581dbd068b31f1feab80da88b234b4f4bffed69b9d9c59a904782be62067a7e039ca953b1b3888e79af1
cc124329a5989224037d65d25fa8a2caa1b4066b8d15aa4565de2245bed82fa2b677d41a33f5b3f0576104c7df033a8f789ca35f08839531bc0d0647
c4978f6b898d68a113af964c75
>}
\immediate\pdfobj useobjnum \csname EF7O18\endcsname stream attr{/Type /XObject /Subtype /Form /BBox [ -6.4 -6.4 6.4 6.4 ] /Filter [/ASCIIHexDecode /FlateDecode]}{
789c3350c8e23250f0e2d235d433510013b95c70660e980961e92298195c5c4e5c002e230a56
>}
\immediate\pdfobj useobjnum \csname EF7O19\endcsname stream attr{/Type /XObject /Subtype /Form /BBox [ -6.4 -6.4 6.4 6.4 ] /Filter [/ASCIIHexDecode /FlateDecode]}{
789c6d90310e80200c45f79e820b7c02b506581dbd068b31f1feab80da88b234b4f4bffed69b9d9c59a904782be62067a7e039ca953b1b3888e79af1
cc124329a5989224037d65d25fa8a2caa1b4066b8d15aa4565de2245bed82fa2b677d41a33f5b3f0576104c7df033a8f789ca35f08839531bc0d0647
c4978f6b898d68a113af964c75
>}
\immediate\pdfobj useobjnum \csname EF7O20\endcsname {<< /A1 << /Type /ExtGState /CA 0 /ca 1 >> /A2 << /Type /ExtGState /CA 1 /ca 1 >> >>}
\immediate\pdfobj useobjnum \csname EF7O21\endcsname {<<  >>}
\immediate\pdfobj useobjnum \csname EF7O22\endcsname {<<  >>}
\immediate\pdfobj useobjnum \csname EF7O23\endcsname stream attr{/Type /XObject /Subtype /Form /FormType 1 /BBox [0 0 306 183.6] /Resources << /Font \csname EF7O1\endcsname\space 0 R /XObject \csname EF7O16\endcsname\space 0 R /ExtGState \csname EF7O20\endcsname\space 0 R /Pattern \csname EF7O21\endcsname\space 0 R /Shading \csname EF7O22\endcsname\space 0 R /ProcSet [ /PDF /Text /ImageB /ImageC /ImageI ] >> /Filter [/ASCIIHexDecode /FlateDecode]}{
78dac55a4d931cb70d9d731ff30bba72920ecb254000240f394471ac2a555c15595b9583ed436a2d2996f5614b8e7dc98ff70367a69b3d9aed9e9513
47abd5cea2d9e8d704f0f041fd38d218f145e355f4bfb76f86ebcf9efffcddedf32f1f3f1afff2acffedf6c340e32b7cbfc40daff0fd0b6e7b8cef97
83ab7833a468f8f9bafda49282e1739c3efd6b185e0cd77fc6f20f58f57818520ea42675e412b26931850e921ad498ad93beeea45462d05ab36b9c14
f4c2c36378ff98978007a8a1002c1eea92413830144b5a3cb8936afbd41e3c3cc2fbfe32fc38f8365df93e490c4459a25aa1329205d194b0163bf7e8
66b8fe9c460b79bc79d1f6e4e6dbe1abf1c12e3e1cbf196f9e0c7fbd199e0e0dcc902d242e54cb0244275d05912550c6fb46e69c2e04a11f83208a21
d5a8ca4b1b74e25518142dc49a2c67c6d78538e8dc769008eee51af38937cce275246e16b5980856e18b912cf6a4dfdf148ad69475d410636439e8c1
cdc10ac4d0e03a7360d7f960f7f787e3cdab4142ca54a5b030ef1ff360f7ee7045929a950a58c72b3f1caec08a51632e34dd73dbaea4508bc5926a11
d9d6f6efc315ad5273ccd889e395b7775ef9a95de150045b4bb0e384603c5c4951d454b8e6fd95ebcf797ee9576d5167c529209502769f9a0d13c2ff
4478d682e9687629c172ad24c688b01503fef10fbbb80b3bf9d89d26205903fc40adf64066e13a908c9dc939111ca1ca0540780548ad21a758b3f440
66e1592035c08f8874ac39c42c9985525e271aa0882b2888e1545cf03e3d8c4eba8e839802229db8c0f56c13c9da7e10ac4cacb92c36a4936e20110d
62b5c538626d0bc99a8b10cc9c18c451164866e90612abc1cc6a2a126bd944624bc66999d7b5969088e108a906b24cce31b3aa43c41d54355d5fecfe
b9fb80af71f7edeebbdd0b7c3ddfbd7f38c24039522afea762d9f3dddbdd2dfe1d775f8370da9ae174cd6d5bf51314be6c2b7fd8bd834a977cd87dfd
10788709ef7826671371403acce07f447849a9541ddf3f1fff31be1d797c325288ead91799225235109d0c9e5fdaa744a0b6a082946a9251377cf978
5c961c5d62ae25d4caca6e281404b8df6fa74881844a65584a3514648b0258112c898744a71b63bc7262b526c6b709b9185505694b5d306801394537
f79c8529a6c0acf80131ea18a41998dae3168fae6d71a9016ee8389034ac16268fa26abe3d062f4086d414d585ae0dbe1d0b8d49e134a822bc62e9d2
6d8df06a68ce4d0cb3c714192f8edc57395b13b384482a75ffe6aa2593470b25f835e7ea16e9770494114ca2247f7757efd58bfbf433f8d3938559c6
bbcdf2de0ba98b8df8fee521565ac174b49fa7e6c98660546ce4ed9bf1fa8b387ef66e5a6d0179104a6aac34325e03e585c070796b6d0ac91130d236
6dadbd72e8c877096bd3d662529034925156c0ded48cca033b1cdd9ab4f97e57000d08553972d95aec3589c13990405b3db28123bad995505c88966d
20d891824c9fe088b6b55a25e4aa19fe5fca261281576231f257ddb48bb9a7c03170835c60432143b8c1e817acb512c1aa66da3bd2d3e1e9f8db382d
472d85134a540f072480982a13da20a4b17a0888754eab01b493954e380d0b6ae6944e388d4136a88ce884d3501f94e2efb6e4342724a0b0134ecba1
c22ff6ec35739a383754a505a7e18150638d79664e4b2122281bb3769ce6d541662ef584d49cc8405794cb92d5720cd52c0a2f59cd1f93538e27a486
1aa844b27d3b38731af8c428256e4ae66e61ded68ee966638d771bebfdbebba5cbcdbbc6764718fe96cdef6885ec2a180c0ecdc9b6d6c29f612bf8c4
14adb41280c00a97d04c5b4b9dbb6256f0916daebd022d9624de9f6c2e85238a7acb8c967553314cab94c95216ddd48c72298a3386c9360e85af20bc
2b79e5b3a95a90f9915bc18af182c516509750521864d3226865c14666cc88c3adc5a01a0f806828bcb7dda2787a30848f6cafcd82beaf8a2c007f32
251e662c20da162f428c90c9fb8a0041d49e77779577b64f9aa73e8b46e95983d89e39cd76eed0d6cf12ce4e8ba0ec1e6328acf671d2998a765cab68
1b49e6dcc65799bcdfc7a7fd98454ea47b82a48fa4ffe36a6d812582a06784b841453a3111aacb5a7d58e0c22435fb0c0a528421c7daa40616f3bc3c
b81c8d2bbad7262fee74a5ad2ed34bb29bd64e84fb37df2b98a535efd1f45abd51dbc3e931306cdc3eee554c88d1d3ebfec6fee590d10e2f37efc351
76dbf601355d03392f8d704ad9bf4ea7155290b8bffc1203e4a808db567588a3a7fcd3975b08a77d58488f5bd66b9d77b7c7d0dba247dc59ae7bb9d9
c81ffbc3adcf561f0db1cd7ff7a3d4be2366bc6d6a77a004410fcf767668c6c7c6f8c1eefb363f426fc3a9d68a92408e43269a17fda92d326c4e2ce8
c4ea3cf1a2c35c0b64a3aae62ddde18acc83a93baab5f11ed55a17bd869dd78fa277922ea2b793fe0e1548ef1f8af674423945cf247ddd4b5145da51
3a6be8a493cddb207dd5eec8afa849fe8f76d7d9ee878386feb0221e0e2b06544e084c3c7b991fbcb86e69a1a7fc59daa78d59c3fd8e1ef01a5eb4a7
93b3875ebc3ee6aea8b65125dfebf4e1cc180e9c11b820fe97387af12a0e7637c77393a2b6cd9f7e00c170192dd99696e8a4eb28d49d00ad87c6a4bf
e5f481bdb9a8514e4e1f7af13a10370baaf7ffd2e903da3e34d82d7474a9e9cef387bf1d6202ec4045d06e4fd1f2fc7862804e4b401d5e8aedaffc3c
c55154ce70beb27dcfebc3b94024e1c2887abbd789c1f7a727067320a18266b162adf02bd9cf9516d2b3e160c5ab349485e20688e80451956e0fa6f5
ec11d404c5d00725f7ab059459ba0ec5fcdc864b42859ff80228bc06c59f9ec0f465016596ae43297e20a5518024960ba0c81a149f1b0832112d2dd4
89d7c1f8e435a3e831b854910bd0d82a1a466c92e6b4345227de408388464a94e2d30dba004d59452305455286ba259a59bc814614bc4e3e3e4e171c
f1d8baffba5e8b717ffcd6a199c51b684c43453f8d462fb6f9e8161abee350c3eb88bd4a10249811ede782c9221a2ffe98ccfe73a012114663ef93a8
33547273e4282ea8908a57d3573ef646493f971c1aa2f4cc0314292674bc19f474b2fa02000f76d7775ee1e3692e98dbc34d91bdda2fe56e386dbfbc
993d57915c36831c2eee424f67905dd1d1732ffb40cf52297979a2c1de5cb1483a154f25432ee8236a6e67258c3884537b71001fa0ec84e8e2e4e7e9
e6ff3f04f15791adcc0395857db66db855d16d2077ec75e397cc9950e1a222005ed8ba898fe581c062e2c7852e053f8b773d692cc84f95559b388303
63c9d846c45846e9147d9689e415106e20e8fe2406993044d0aa614baabf8cb6496b5707e089a98260caef7982325baa0d158fd63ae6baf1fa0beec7
3c006b51c54fb8ad11a22825548dad875b5f1db7165c010ac24d7dca743cbf59594d3032de8daa702e5bab9138b39ff4e44cc7539915d51976cbe2ff
d3a2306fadcea11a8a22f855ba00740c55c12a8861dbdcb12401de5f337a84e359c88aea64e8ba618b360abd0408ca0db0649acebf560c9d9c94c813
f1f67e80297d930b5a28ebb7fad3d8e8d05f7f35a650fc940225eaf80d62e3dbc53870bcc738f07c5938f7628bb2707b1e78bee13bdfc4f98cef1efd
a13f7d189efe0a344dd3a9
>}
\expandafter\gdef\csname EFWidth7\endcsname{306}
\expandafter\gdef\csname EFHeight7\endcsname{183.6}
\pdfobj reserveobjnum
\expandafter\xdef\csname EF8O1\endcsname{\the\pdflastobj}
\pdfobj reserveobjnum
\expandafter\xdef\csname EF8O2\endcsname{\the\pdflastobj}
\pdfobj reserveobjnum
\expandafter\xdef\csname EF8O3\endcsname{\the\pdflastobj}
\pdfobj reserveobjnum
\expandafter\xdef\csname EF8O4\endcsname{\the\pdflastobj}
\pdfobj reserveobjnum
\expandafter\xdef\csname EF8O5\endcsname{\the\pdflastobj}
\pdfobj reserveobjnum
\expandafter\xdef\csname EF8O6\endcsname{\the\pdflastobj}
\pdfobj reserveobjnum
\expandafter\xdef\csname EF8O7\endcsname{\the\pdflastobj}
\pdfobj reserveobjnum
\expandafter\xdef\csname EF8O8\endcsname{\the\pdflastobj}
\pdfobj reserveobjnum
\expandafter\xdef\csname EF8O9\endcsname{\the\pdflastobj}
\pdfobj reserveobjnum
\expandafter\xdef\csname EF8O10\endcsname{\the\pdflastobj}
\pdfobj reserveobjnum
\expandafter\xdef\csname EF8O11\endcsname{\the\pdflastobj}
\pdfobj reserveobjnum
\expandafter\xdef\csname EF8O12\endcsname{\the\pdflastobj}
\pdfobj reserveobjnum
\expandafter\xdef\csname EF8O13\endcsname{\the\pdflastobj}
\pdfobj reserveobjnum
\expandafter\xdef\csname EF8O14\endcsname{\the\pdflastobj}
\pdfobj reserveobjnum
\expandafter\xdef\csname EF8O15\endcsname{\the\pdflastobj}
\pdfobj reserveobjnum
\expandafter\xdef\csname EF8O16\endcsname{\the\pdflastobj}
\pdfobj reserveobjnum
\expandafter\xdef\csname EF8O17\endcsname{\the\pdflastobj}
\pdfobj reserveobjnum
\expandafter\xdef\csname EF8O18\endcsname{\the\pdflastobj}
\pdfobj reserveobjnum
\expandafter\xdef\csname EF8O19\endcsname{\the\pdflastobj}
\pdfobj reserveobjnum
\expandafter\xdef\csname EF8O20\endcsname{\the\pdflastobj}
\pdfobj reserveobjnum
\expandafter\xdef\csname EF8O21\endcsname{\the\pdflastobj}
\pdfobj reserveobjnum
\expandafter\xdef\csname EF8O22\endcsname{\the\pdflastobj}
\immediate\pdfobj useobjnum \csname EF8O1\endcsname {<< /F2 \csname EF8O2\endcsname\space 0 R /F1 \csname EF8O9\endcsname\space 0 R >>}
\immediate\pdfobj useobjnum \csname EF8O2\endcsname {<< /Type /Font /Subtype /Type0 /BaseFont /GCWXDV+DejaVuSans-Oblique /Encoding /Identity-H /DescendantFonts [ \csname EF8O3\endcsname\space 0 R ] /ToUnicode \csname EF8O8\endcsname\space 0 R >>}
\immediate\pdfobj useobjnum \csname EF8O3\endcsname {<< /Type /Font /Subtype /CIDFontType2 /BaseFont /GCWXDV+DejaVuSans-Oblique /CIDSystemInfo << /Registry <41646f6265> /Ordering <4964656e74697479> /Supplement 0 >> /FontDescriptor \csname EF8O4\endcsname\space 0 R /W \csname EF8O6\endcsname\space 0 R /CIDToGIDMap \csname EF8O7\endcsname\space 0 R >>}
\immediate\pdfobj useobjnum \csname EF8O4\endcsname {<< /Type /FontDescriptor /FontName /GCWXDV+DejaVuSans-Oblique /Flags 96 /FontBBox [ -1016 -351 1660 1068 ] /Ascent 929 /Descent -236 /CapHeight 0 /XHeight 0 /ItalicAngle 0 /StemV 0 /FontFile2 \csname EF8O5\endcsname\space 0 R /MaxWidth 974 >>}
\immediate\pdfobj useobjnum \csname EF8O5\endcsname stream attr{/Length1 3672 /Filter [/ASCIIHexDecode /FlateDecode]}{
789cb5567b5094d7153ff73bdfdd5d96dd6517160416706577791456cda25888a3ab82a012c1801689ab2cecf250f6e182c8231853b189da0435ce56
a906d25a6bcc8b5a2743629ba64dd22493719ac6984ca6ed249938d3c90cc9d8a949a6b4b9dbf32dc8244e33d3fc917bf7de7b7ee79ed77d7ce72e30
0030512783b1aaa2722da4800a80651137adaaaeb6bea2b7660fe12584ad55f59b57e79f5df912e106c2d1dafa45ae9de5c1df129e22bca535e00d43
54fa27805441b8abb5b7c77a666c228ff004c974b585db037fdff98f97c99962ff48bbb73b0c6aaa209f27ac6befea6f33b53fbf8ef025d2b9d4e1f7
fad45bfe900ca0eea0f9d20e62e81ee7af101e276cef08f4f4a53dc77484c926a477855abdac098e10fe2b6143c0db17c65daa04c24a7cd6a037e05f
58ba2602a0211dc9180e75f7c4fe046701b44a3c0bc3117ff8f8312997700dc5a003656f48325e2442083cce539a16b2200758c5da9a06d0d0ee5189
c5666567e68fc30930c7e707bc116f0b0c7b2381200cb744bc9d30dcea0d7653dfe18f50df1fe982e1767f88e8f6887f170c77788324d3e16f21ce2e
6fd00bc35dde9055e97bc842c0dbd301c3c15d0a27d4ee0dc070644f90247bda82edd47728f66f8b68362ed9c08ed21a8097f05350ca7294315689ef
429b44bb2c25aa10514e94e4d935cc95bab64a1f71fc10509985998daa03eca3afc9e06ccb8aaf1ca082108b63190a6854d1c8680fc9423ca6f8181b
8b9d9cb331b3cf0573bbadc8cc34259ac06ca3d8998d2d82497883daefe102bc2abd0ae3b08ceaddf0023b243969e61c3c00ef71092ec16956c6ccac
8c66afaaccaa7e7e909fa7f96da45b4d56de8347c8a6626912ee9306a53a688357f91518a51a8af36fc0f3ec005c8393f086540d9fc3016c8011aaa3
d02d03bfc6b420a422ba41e4892a400bb56c74f26bf17a03ee83416880b3aa4995995d8d477d8ebd445fcbc714f355dc86bb49e3349c6707659b7c5e
ae86919978b11946a4036c546e8ed741da90bd705a6e6617546678458995387514691bfc86da5eb8c2ee6407f1104536a844c0afc115f57a79d14c54
ea215c4aeb016a4fc345706294f4e36b51b5c169a98d7c7d4e915cc10a28246bdd0066faba21f55915975162506c354e488e75be09f7a646eb6b5b17
388b6f8356a3da3a017513fa7eeb642c56d7285bf8d6099e35810ecd84ecb07df84d931f3a8b37d4355a277e5d59316bb5b2b98278f58d442a88d8c4
afac70cedc090d48747b90a8edb18fe463b42f09900767dc760869533798379a6ad2b5f7aa6a160ca46fcad7da1313000de654b4d90d364b6ebef1e6
c5310bf34cbded4ebc617fd33e66df615f69e79ea9b232d3bc3253d91d60bc7e73eaba51dcfcd4f8e9e28be102c63ccc5d1c82100b4921ecd2841242
da5062c81e7284f20ea71c369f4a39654e1a4919318fa48ea48da68c9a551ef0b08ba51ae67127dc2878b360bc606581ec612bd8d22579b65c959a2f
30b054735a89ab7499dac06cb9794b97242f5bc14a5c69a96615be87383278e8ad525cfa9ff32af989d37def6cb3b55e0d7cfcefa1fe8deff6bff0a1
4e3a61cc350d3cfcd4d8485eb1f1c8c9d225cf1614bc7fe519e6ec5d5f77fdc5859df42555d359afa45d2962a9ee6e734ab2c9c85393b84eaf3370bd
5ea751d379f204547298c419830c9e69c94c9b27a5a6734be6fc1cc99a4d6316b75832cbd393742867db00596a424d9e75c0b22f43cf3232b3d29393
0cdca247c9a68242a6b3655bf4f93696afe285964c7bb1b2c599ccf3f2a5b0a5d92279de76eb6385e3548f16860a6b0b559ea9a9b72f3d636bb6499e
29136db92959f9519bd9ff45cb975fbfb97cca45c731af6c315d02b5917fa2e69f109eed153cdbbe0289dc9a7b71b1930eec59b7b3d91976a2c79d12
73beeffc9d739cea5167c859e7ac75267818da52e68ee336ba441d3f8a79b70e44bde8c0e4fd1a8dafe9476fcd1f9c3c4054e3830af59749db86d7bb
07ce181bfe3678f48c463afbe536e9d1a215694dbed71ffb72447ad4b132e31ebf42cacd4fb7b4ef8ff4de7beee482da995c296d3d1579e78fc77724
2dff0ce66be289eecffb0cefdf1abff8e0cb12fd5acd91f84d9fcbd974a601910da03ff8c507ff7a4ebf167cf4467fb5a864250f28e62f50a3d74f9e
82edea3794db30578ed0bb3813830a3742117450be95c0086ee5f54583a49f7dc3d4d0a4646e99de4cb6389ed3159a411aa1195a02035b3b4b23e4d1
3f81195afe0acd219d0dccd22ab0b3e3b086724f18fa21029dd04ede7bc04ab9be95f28d155cb0986a09512d246185d524d34359a887a4fde0a5bc5f
4cdc751024f98544ad822eaa56caf7b76c75c7919f463fe9f452ef2349edffe1b574ce6b0379ea255f3b492748d24a1c5ed2f9761e2b88da497a5b60
0f49b492ac376ecd1fd7f0c65764252b41eac324d342763b49ce4afa21f2ee8dcfdd6ea73e6ea51b6a67e57713d7ff0d32d6dba4b6c423ec261c8a7b
75519c25b0f46bdab7749db7e9d22b1cfb8cda3e6882ff5554f13b25d14957431dc0a4b4df1dfb95c009073ee3c2a7a3f894019fec36f0275df884c0
0b0e7cdc80e71df8cb289e9bc65f4ce359813f2fc79f097ccc85e363f57c3c8a6377ade263f5f8a80bcf98f174147faac55181a792f1e410fee43246
059e20891343f888c0e3c7aaf8f1213c568547472cfca8c0110b3e2cf021813f167844e0e14339fcb0c04339f8a00b1f10389c860704fe50e0fd02f7
0bbc4fe03e81431b1c7cc887f70a1c34e140ff653e20b0bfcfc3fb2f63ff7eb96faf83f779b0cf2def7560afc03d51ecf161b70123bb1d3ce2c3dde1
64bedb81e1640c5158a1690cba63020302bb04ee4ac39d9de57ca70f3bc9476739766c4ce41de9d8de66e0ed2e6c33a0df873e52f345b155608b57c7
5b047a75d8bc238337fb70c77623df9181db8de8d1e2b67bf47c9bc07bf4d8441a4d51dcda68e05b0bb0d1803f98c62d9b2ff32d02373778f8e6cbb8
79bfdc50efe00d1e6c70cbf50ebc5be0a6ba857c93c0ba85584b41d49a716322de4551ddb50a6b68a811b861bd896f70e07a13ae13585d65e2d502ab
4cb85660a5c00a816b560ff13502570fe12a81ee695c398d2ba67179e96abe5ce09daf613951e5f55826dc61fcfe102e23582a3b79e96a5c2a7089c0
9272744de3621d2e12e814582cb088a68beec0ef19b1108dbcd0860539989f67e0f93ecc33a08369b9c385765d3ab70fa18d97739bc05c42b9977101
c92fb0a0757e22b72621fdab78d13d2acf4fc49c04cc71cbd946cc22f1ac285aa29899e1e0993ecc484fe6190e4c4fc679690e3e6f15a6393055a059
60ca34269b3278b24013593565a05160924003593044514f0ef543a84bd4715d3a26ea502b5043539a28aa485c2590d32a7839ca846427a291fe2f69
b9948e4c8bcc2d4316b249e63bf8102bfa6e0b7cc7f6bf6dc9fe2f01ab2568
>}
\immediate\pdfobj useobjnum \csname EF8O6\endcsname {[ 101 [ 615 ] 109 [ 974 ] ]}
\immediate\pdfobj useobjnum \csname EF8O7\endcsname stream attr{ /Filter [/ASCIIHexDecode /FlateDecode]}{
789c6360183680058dcf0a000125000a
>}
\immediate\pdfobj useobjnum \csname EF8O8\endcsname stream attr{ /Filter [/ASCIIHexDecode /FlateDecode]}{
789c5d503d6bc4300cddfd2b345e87c37781760a8172b764e8074d3b950e8e250743631bc719f2ef2bdbd7142ab0859ede7b3c242ffdb57736817c8d
5e0f94c058879116bf464d30d2649d38378056a7db547e3dab20248b876d4934f7ce78d1b620df78b9a4b8c1e111fd48770200e44b448ad64d70f8b8
0c151ad610be692697e024ba0e900cdb3da9f0ac660259c4c71e796fd37664d91fe37d0b044d99cf3592f6484b509aa2721389f6c4d5416bb83a410e
ffed9baa1acd4e7fb8677a6d9fb57f15182b8c37181966bb5f6176ce67d863eb35464e5c6e55a2e690d6d17ecee04356e5f703343d78c4
>}
\immediate\pdfobj useobjnum \csname EF8O9\endcsname {<< /Type /Font /Subtype /Type0 /BaseFont /BMQQDV+DejaVuSans /Encoding /Identity-H /DescendantFonts [ \csname EF8O10\endcsname\space 0 R ] /ToUnicode \csname EF8O15\endcsname\space 0 R >>}
\immediate\pdfobj useobjnum \csname EF8O10\endcsname {<< /Type /Font /Subtype /CIDFontType2 /BaseFont /BMQQDV+DejaVuSans /CIDSystemInfo << /Registry <41646f6265> /Ordering <4964656e74697479> /Supplement 0 >> /FontDescriptor \csname EF8O11\endcsname\space 0 R /W \csname EF8O13\endcsname\space 0 R /CIDToGIDMap \csname EF8O14\endcsname\space 0 R >>}
\immediate\pdfobj useobjnum \csname EF8O11\endcsname {<< /Type /FontDescriptor /FontName /BMQQDV+DejaVuSans /Flags 32 /FontBBox [ -1021 -463 1794 1233 ] /Ascent 929 /Descent -236 /CapHeight 0 /XHeight 0 /ItalicAngle 0 /StemV 0 /FontFile2 \csname EF8O12\endcsname\space 0 R /MaxWidth 974 >>}
\immediate\pdfobj useobjnum \csname EF8O12\endcsname stream attr{/Length1 9704 /Filter [/ASCIIHexDecode /FlateDecode]}{
789cd5390978146596efaf57d547f5599dee1ca493748e2609674c08100ed3861b14830424289a8624843321e10e1844484441402008226400110362
64584830224846406474069859565d46c541772232b338ce84e4ef7d55dd6070677766bf6fbf6fbfadea57f5dfef7efffbab810180831e2278460c1d
361cc2201e80f5a0d6f011b90f8ff756f6984bf5a100827dc4f80939c97bb34f03601ef5d73ef440de486b43ff06aab7d218f7c3e37ba7cfccda1803
2019a87fe2b439fe320c5b40eb4905d45f366de17c0fcc88c902d07d4c755e5c367dcebc3e0b670218a80e07a6fb2bca404f37189751dd3c7df692e2
47decbd843f56701c2969514f90b0d53deab068811a9bf6f09355876e95fa0fa20aa2795cc99bfb8b6c9514575c207ab66974ef33f9035dc4af56b54
1f3ac7bfb84c7c41370f203685ea9eb9fe3945c9df0e6ea6fa48a2e7525969c57ce16b711175b9a9bfb0acbca86ca0fe8fead015c44309a8b23243f0
12a886e0d1da5490a1170c0261e8f007f3c03adb3f7f2e44824a25040200774bdae85945e573c1109ac7a84fd0de0610589e3a92cd64b5600407fcaf
5f81bd04c7082e063e0d1c23f894ca9ffe4333ffa151ffc03aa7ee5d8d68b818a422f8d6dacefded593fbeefb6cffabb081949d20c16b0916547904e
ba402cc451ab85eaaafc85d0380cbd4590e8a90bd5f4a41546ba0d0f8dc54e23d43e23810c26000d8795b0d84109d9c38bb0199c9a3d2cf597fba7c2
2a7ff99cb9b06a6ab97f06ac9ae69f5b41cf92a2727a2e299f0daba6179552797a79d12c5855e29f4b634a8aa652cb2cff5c3fd9b2bfd4a33ec9ae56
cdf1cf2f81557367a92da5d3fd736055f982b934727ef1dce9f42c51d7ef647b9d65a1d2f53bf81a72619019f49fa98dc24a62a980ecfb72480005f7
0af04e9d9584ca13feaecc8191cf090bfffe38981c1cabbeef963bad714ffbe44ef5f21fcb427f80dbeba93ce8eed414d2865de3373ce4638c742486
7a054d5f2992aa79138d0549f565265ad90655b35286b48da6c406dff85b2816c81305930ed1200a821892a3ef0eb2dce2618554f32438744eee64db
f573d897a131c10b43a0c61499e020d59856176103bd3d748bda3b1306c060c881a1308a3434158a600694c1025892e0d034e981b47b46f8a1104a60
0e94c3627544e0cbc0a5c0c7815f063e089c0bfc22f06ee0cdc0a1c0c1c0eb81fd817dff6aed44d17f7505635b4ba866d3300641a53f8d201382b16c
4008547e068740956a4e08d4958686c042302a04aa1e7343a0aeef27984aa06aab90a088408d7b646bc43b8093600e4119413881aaf705045d081613
2c218805487004013c2c131ae13cdda7a01e76b07d542ba6f679d452271c86d534bb114eb3f36c8dd093daf6c14db848236be03cd68bc0464306b502
5c219bb845d1f808ad91c59c2c4baf13411c2b1e111f111bc5ebe205e827568817c402b18265e06e69a2b48f200b7f41b6728ee24b23bb0a15701cbf
c10c6c16878a56b88a17b01ebe222caafccec37ad80b95448b9395429550293c422d67a40bb09dee52eabfc076b28b44dd71f60c5c8697501446c24e
7699f83a0f7f8667304fa07d0e338462a2ff0cad7581e66f870a0a4e97990c5ce84e6d44bd265ff519833da5cbda7d13aa08731eecd535ea9cfa44c2
a24a6c1f3bcd5a759ba00e2ee2e3380f3f61abc54471bf3812d607258005b09ed6deaeced115b325c4bb7a57aaab0b8bc402560fdf8805faa9b4f62f
548e08e711e111e2a8189a0916e9ecc4d340b61ad710a56a6f0c5cd08f167bd37c5a41bf9cb80628c54c9849a54a380487a127d6c27a5a49e357d74f
fa33cddc217e4e3caf67eb843fc3051c0aa9502cde2059aba6520b704caf93441418f4f0d81b04efa8c206dfb8499eb3f9f13d7bfca4eab1eb3d0d90
db6059e2690c04722789d1527e83e46e40afa141f4267efe5f757edeb3c798dc499e868e614343ab0e2b184a6de3275151ad5133b50f1baaf5a9481b
242ffd46153478a695789eb33f9738e0397bd1809ec15844bb3ef93652a926f0a5b89eb463a27d2ad117a6ab73409d79a3636da4d16d8b45b72b3ad2
dede7aabf53eb05fbbd56abf91c61204c5eec84877287621391d143b2426a84fe1f91dafbc42bf575eb9cd8cfc87dbb7f90fcc28e5f20bfc43820b2c
83ee3e2ca38e57f06a5ec32bd83ab6842d65ebd498f539b9f764da0964f0f95c3958270a75d2d37aa8331ae2746e843866b25f1ad360cb9bd444837d
fdf35b5bda89a0deade9b75a2fb5a6910cf213d8111bda44614abf7845caf46628f1ae78ce46f36dace80336ba7d6fbd5831b27164dbe57a5a80f425
8e268eddb0d3971cd5251a23dd8a24822249628efd67ca664b9d73a3487e0b765960b23bc28eba187bfb980657de9886f0bcc7c63438f31e234a3070
b27f7ecba5d69327154756889a5b1a357abbf4ad5efa9635b8ed4a4416d1e64b9f204e9426ea978a4ba585d135517af2ea28b10ba9d73d1f16ea1674
a9889eef5e09d5512bbbac8c5ee9de0ffba3952930c54b4c64f6857ef7b3cc3e5d131374faccfb5946bae872eaf43aa05072aafd41126386ffa1d7aa
9fbcb878e9a5495f33e7b0c7a2f8adfafafa456ce380395b472daacd19f2e17de95fbff7f8ab6531fc0fc4fd0ed27705719f0265be5ee00a93ab8d71
d59eb03a97a5ceb849e7aef36c4adca85bebda931aee0e037446b9bb7aec6e74c61975a9aa10c2f3eef06fd4f82701dc6a252e4902f6d66bb7aeb5da
7f7fc3aedd249534e63316c6fae3fc9ec27811a6b058e6728af1095d9333638991bec4557796192cdcc31e666fdcc33fe65f3f716666ded93927ce34
bd7ae8e8969d7b5e1a7fa2bce25cfeef99f905f4c6b56cf8ec4f5eefe9fbd26bd7afdab26f5159456552d7231ecfaf0e2f3ba05a38c575712fd99440
3bc1d3be1866410b205a72004dfa3a89e1d3466696c1ad338866abfdd3310d2662cca231665619bb34a8a5355d51f57aedd2a0d674e24553ac788e94
7b4e55693713748391904f1bc622780ef4e1ac3b7465ddb12f1bcb1e363f6c99c88ad902b61457330ba9d2c8e23143c970252a894a7c26eab8c07826
bf7cf95cc71392b7fd4bbcd09eb19fd7b182d3a4a19da4a142a23c069ef0258a5df44ab53da64b9dde59675f6311eae069cb5afdded8083793d10db2
5d176b6f679df56257c90f798b5df5165291bde586eac0aa07937a784b503b61645f8a2a737039e11eb5a8daf80ca33aea7a4cead1c692f825fedd13
a74b269f9cf5c6071fbc31ee6779d2e57afea2cdc66ffcdb1ff9f71ecff9fbd28eeed87134a92b49fb2639f4693191628b1edc3eabee19711f3c2390
c0458834d8db4996aaaf90a3f453a54170f3225d9c8b899c665f6105d227b83b34db023b04dd0ed120311112d5c92d34d97e49757a9576f5a62db38d
ab80bbf7dfaca77994c58bad243d0365c41e18e24b8c822db2718be369b6457e234e311984b0a83809acee7029caddcb086e87184f112ebd9d5655ed
5873e55635d465a51db62530368529f141a1dc2d78e3d3c3552fd4ec39319e6d6243f7bcf2ca1edeccba6fdeb871333709e2f5b615cbb6bcca6fdeee
f85a38d7f159cdf36b570bc5fcfed2f27965fb4ebeb566b7d373fea5b3ff421c13bd524388de5e3e176c3112a576836097418ab2a483db283ab408ac
64858823d91d2e08d308cbf0a8a137de1bafbd5319db748b65b238fe393fcf73d82e7698d5f2129ecbfd52efdb8b5824ebc57ab0887d7c2b5fc19fe2
b5646b849d6d22ec08bd8ec19b02a3ec5054dddcae5a0f68b15645498e6c49937c52ae5420ad977649ba2961194ae2b90f3e902eb775872017e23a8d
8b3018e08b240e54461c76d940a92bf191ada88c3883822631d3a25aecf6996cae3857b6eb49d79b2e49e3e8ae84291088dd9992c836f175dbb7afe3
fdd9d9db8cf1c06dfe81d4bbe3a3176baa5fdcf7e5279f7dd1b19ff05792cff4a4ac40062f3453548f334518adf07a84aec9aa78aae38ebb9b121b95
b5116688c0488bd1608a438373585722e7c34bade9e4e72a412dd76eb593dedf5763579692450ee29b9b1693169b1697e6498b4f4bc84ef6c5f8627d
713e8f2fde97901b931b9b1b97ebc98dcf4dc84d2e4b5e1d53135b1357e3a9895f9db021b92ef96672ec9da97726dd9950105b1057e029882f8b2d8b
2bf394c5af885d11b7c2b3223e720a492041e77286931406b37e4a62a695252674cdecd33723be738c0c174e5c3df874e9b6a6c6c6ece6670f9eefb8
cd84d7b6161ccd2b3a31f9df6f0a19c595532bae1c497db0e3e9fa62ffa9ddef9c74543ddfab577d7272bbaaabe324abbd3a27edf86ee8ef8bc226b3
cdd814e95a6b6b8cde1a050ec78848b3ced065788caaacf45b9afeafdd6a51059376b42076456c5d2c129d1a394152991643281f205a9355a7c6af5e
7bf1c5d754e87861c05b951fd289f0c3cab706343509bdcf5fbf7e9e4078a4d0cf9bf95fe86ef617ee276a049847744da6923d48579726b03a9b24c3
5a6b23db8a11221884118ac3342c46f386f4f4bb74b5dc43979211121ba52940a2621a8941e9e1eec6c6016f2d3b1f80c0f9656f759c210af7ef272a
f1a8f0c45f5bf717fad95066a07ba89fbb428486e8aa22db724234ed9949e062c66ac3b392eb75263599d9db914d8e46f35a77b44b30b80c304670d8
86b935125bb4bc404d9bae693be22d2de8fa52b363ca62ea623e8eb91923654336cb16b25dd9d1520f7d6f436f630fb9144a59a950ea2a8d364e9947
fcb8e2b5a0dccf453c695e4ea2d6f762aa2d8855ed87cd178ecd3c3375dac7b3f82d7e86a5b67fc1f48dc2abcf6e6fb20a4f4c3e71a64f9f43dd7ab0
fe4c66616c08ffac65eb91433bd58ceb38f96b21c93a0cfa90051801ad4c5763551acd5b65261860acea23c33577d50c60502b095a714450482c7069
912751091a001532b43d245c2c6c5cb66ccbc1a6a69c9f2f38f5beb0b7e37161e7ae9d27f676d4e89c1d3b8b0abf536def14215f4278d5f8de9d7687
13e25bd04cbb834184e17777876bed6a60b01b7dc65c6381b1cc4881410d37ea6671aa912eb1e0769dcef94d880f7d0ce9260126fbbaea1cc6481be8
62f42e734d8c071ba39ba3ec7a506c06832e5731d872dd9164d6895a0c6aa728a4656d83065dbba56df52a73beb0b4a4dca4b2a40d497574bf9b7435
299064246e35fe5c9d79fe91795790f9d4612757be79a2a97cc1fa7d4de58bd6ed6b6aca6e58b2f400ae59b6f0fb2f5451fc6c872a0a61e7ee97dfdd
d35123161c9a3e75d95d4d100761d0f75e4d34ff6d4d5cbba3892305ae8f5cc24f75e1fa3bba20c4aa2a8256bd80f0aa79ff68cafb9b1cd0646e54f3
7e876d1c3a5cc37e92f7fb12b3a32aa15257a5af325419abe42a53a5b9ca5265adb255d9ab944a475dd4cd28a553fc228fbbe77850b1f9e0812d9b0e
1edc749339f88d9b7fe4df3105af5e3f77eefad767cf7cb3839fe5adfc5b32e12cb25427ebafc6293e51dc4b14aaf1e07e5ff49d78d0685dcbdec1e6
188a0523b4a8d02952518a732724f88cc198f0bb58914df1de154d284cdd13be2a9a9a7e8c5242ff3bb16b7fc7219d5cdf294eb13fdc090a41bde168
a24e81349f5367223b33618db5d1d8ac977506300ca74dbb2568cd14092e7da8bafe91dcb05d61aac634f7eda4ae081c1d37aac78ed7888ee3abc37a
b9f18843397fa2e33029ab789a2411b6d2c0977886b025c375df208b59b09ac6c7c51a8c825e1e1f17179b239b62e3441754b335a2b3dab526b24911
9bbcb4e1a5c4caa6b8683d3c126db0ea0dce8461292a55975aaf11395a36a129977f7fc3fefd0dd5a6b41cd7fa2d6dd07aed999f701892c9df7d73dc
b2dbe436f7a200d5c3d4c33cd038501e681a683679c0c392841439c5d42dacb7b3b7ab5b784a6c4a5caa27353e29b95aae36559bab2dea171226083a
596742335ad08a36b4631476c168748b31c6e4dea9d9a94fa656a5ae48dd905a977a33359212e679cc159492333c4e3b37e8123b27a8bd4986eab698
1e8ecf8fdd3f79cd9aa99bb35b5efde19f279f9e5dfcbe7fe5daa203be032ffdee97c547c4ec4329297979be51f1d66edbd6ec389a98782233337fdc
985caf2d69cbca9d0763555df6a3d0f4276927f920ed3b56c960c3d74161cd861ad94432261bb33bacaa0f0e6a5153d05016163c2050447c331811d5
64da193e90b9d41d9bf6a00c852d62957cf5988a77deb9bcbba646dac9df5bdf51b766ecf65dbf160ad6b3fbd5787888bc7092e6fd4e18e873ffe8ff
6b65d6ec6c3493f73b4d63290e0c77a9ee9815b4a86be9778340a9eba41a04c2280a07ddeeeefed7951d5283c01b8d8d43de5a70ea2cfb881d17f675
f877ed3ab157a8bc5d77b078da4ddcaf723f9022d012b1004c6c886fb8a4d0e1525450d4ab2f4964024345109849fdbe2c2b4699a92f93ac37e88d8a
c1a0cf91f522130df0ae24844a824167568feb729e7a08794ccd25e9a16867139d9a51aa0715b51c3c3fb704b3418afc77ec8f0ecf764308a46f7f5a
57cf5e2fc9a22877115d725779b0789f3c417c543f492e9617b2a5e242fd7c799db852de26ee12b7ea5f9437c8fbd8ebe29be2abfa3d729dec965194
24a36cea822ec965ec624ac5ae92d7d8cde4b10c6059d84feaa3ef6bcc32a55946e170699871b4c967c987892c5fc8c747a589ba7cfd44c34463be29
d7526a59ccaa2c2fb3cdfa036cafbec1f291e5aa2560e9ad9ef4844423a35f869189857c16abbfc28ff3e357d8cf79f91596ca52c5828eab1da75823
1f298c16c2f93cb65ed5c1608a7555a4031b7bde37446f108c0ad8543103d8ac8a0d6c16c56c01f565b5c826d9ac984c728ec564b48349aac177aca6
66bbd562968d3a04834db499ec771460d0c46eea247653f0d8ae49dd4e86d41aca75ffa6e80dea778c887455e6377520197446b484cb1116bb25d192
6919253f2c8fb54c364e9667ca359615964d16870c448449329bac265b04730976d12e45c84e93d3dcc5dac5960c49142d3ca2474a35a418bd729229
c99c6ce966ed66f328fd2093650a69629ad45fee6bea6bee6fc9b266d9d29407c0c77c820f7da24ff2e97c7a9f21c7384c1e6119651d65f32979308e
8d132660ae984bfa9940fa79d4f8a83cc134c19c6fcdb7e52ac5ac582891675867d80a944ac362eb62db1a78ceb8dab4dabcc6b2c6bac6b6cdb8c5b4
c5bcddbaddb6d7b4d77cc07ac0d6a07ca45c55024a11e952b2b2e0078a6ca6ea3343d83476f3b24db31fcccb88e7034fb327181d8acf2edd3eb23a4f
1cdbbe1967ab9abc22f4943e117e4b198e72948eaf2203d17ea9453bb66a8756edc02afcb63ef8ad5ec8dfe63cfbcba54fda067d0f7106ed03f8af9e
b25ebbf3fee137ed0f5af38dea3f5486bbdfcb699e7e0e8f01b0f21f7ed336ce9aff9fbeb1c78b17b4efcf20101ae1793a6146400dc1e704b5043b08
0a09760a87e1a6ee085c9196c3395d2a9c131e8773e24582eb502939e138c13cb115e6490e382e64c129fd796aa3b2f815b5513f8ea6727728c544e8
47ed87c466184830585db3133d39b004fec266b373428a304228114e0a014cc469f8577198b851bc2df5939e944e4aedba87743b7417f52bf5ef5212
7edab8df7842e32c1ef3a03b948099e2a51db6a992105d4238bdd56fda7a98acfeb3211a493069da7f066a994138d58265010c6c78a88c9ddac54e65
0922d9d85059074e560c43a014ca88ee729801d309fb7cf0400a4c83547aa7431add19549a4a233cc4df0ceaaf20288722f0c31ce841ada3602e8def
45a5076036dd1e78e4ee5a155aad88de453467213d0b69a4fc0f60ed7b176b1e615a48b8d42fd77369b44a879fe6fccf300ea5d24c9a371116d08869
34d6afad56a4cdf06b1c796895b9f454ff159a4aebcea0711e9a5f4ad8fd5adf4fd719afad524114d179066651ab8ab582c6966a2ba513ee0cc8bc67
d69d39a1ff46034f69ffb5fde72b5ef301f55f5175d77441b8f60f6b0cadeaa5fca807ad3b1486c108180963e041c88571c4f7789840381e8549f018
3c0edf350a2b7c81db1cdb9cf8572ffe251d7fa8c53f5bf17b8eb738febb17ff64c53fd6e24d2f7ef7dc03d2771c6fd4e2b7b5d8da867f68c37fe3f8
cd00fc3a07af73fc7d3a7e756dbcf4552d5ea381d7c6e3975ff496be6cc32f7ae3e71c7fc7f16a3afeab133fabc54f397ee2c07f598e57dec67fe6f8
1b1afe9be578f9d208e9f272bc34022ffe3a5abac8f1d7d1f82b8e1f73fc88e32f395ea8c50fcfc74a1f723c1f8b1fa4e3398eefaf56a4f7ddf88b70
6ce1789ae37b1c4f713cc9f15d8e2738bec3b199e3db1c8f2bd854ed959a38361e7b5b6ae478ece814e9d8db786c8578f49fbcd2d129be001ef589ff
e4c5231c7f5e8b8739bec5b181e39b1c0f15e21b563c78c02b1d2cc403f50ee98017eb1df83a11fd7a1beee7f81ac77d1c5f75e05e8e7b765ba53de9
b8db8a3f2bc43a1a52578bbb38ee7cc54c190fbe62c61d2f47493b0af1e5ed76e9e528dc6ec76d32bec4716bad45dacab1d6825b68d2965adcbcc92a
6d4ec14d567cb10d376e785bdac871c3fa29d286b771c30a71fd0b5e69fd145cef135ff0e23a8e6b9fef25ade5f87c2f7c8ed87cee015cf3ac495ae3
c467290da7869a42ac2649557b71b582ab383eb352919ee1b852c1a739aee058c5d117786af972e9298ecb97e3b242accc7349955e5cca7109c7c556
5c64c685322ee038bf0d2bdab0bc0de7b56119c7528e7339ce8ec7591c672a39d2ccf1388363c9729c4e95628e451c0b394ee33895a37f0016b4e113
669cc2f1318e9339e64f92a5fc369c24e3a3e151d2a3e93891e304c23c2107f35c389ed9a5f191f88813c78d0e93c671cc35e1c31cc73e6497c6727c
c88e0f721c433d63388e1e65974687e1a8188b34ca8e232d3882e3f05a1c568b43390ea14d68481be6bc8d0f8c411fc76c8ef70f7648f73b71f0209b
34d88183065aa441be800d075a7000c72c8efdfb39a5fe6dd8afaf5deae7c4be9926a9af1d334dd82716332c987e9f494ae7789f09d37a9ba4340bf6
3661af9e46a9971d7b1ab1473a76efe695ba1762b75487d4cd8ba90e4c49f64a290f60b217bb7a4d52571b7a4d98c4319163820de389cf78077a0a31
ae0d638985d8428cb1a09b24e8e618dd865d72308a2a511c230b31822415c1319c268547a18ba393631847070d7070ca4e7b4a4a0eda97a3ad10ad1c
2de670c9c2d14ca3cde168e228dbd1c8d140c30c1cf54ed415a2489d2259800ba915390a54177a22b32370648dac70f53ad6fdffc305ffd704fcb757
cc7f009ea56dc5
>}
\immediate\pdfobj useobjnum \csname EF8O13\endcsname {[ 32 [ 318 ] 48 [ 636 636 636 636 636 636 ] 58 [ 337 ] 61 [ 838 ] 66 [ 686 ] 68 [ 770 ] 73 [ 295 ] 80 [ 603 ] 97 [ 613 635 ] 100 [ 635 615 ] 104 [ 634 278 ] 109 [ 974 634 612 635 ] 114 [ 411 521 392 634 ] 120 [ 592 592 ] 8722 [ 838 ] ]}
\immediate\pdfobj useobjnum \csname EF8O14\endcsname stream attr{ /Filter [/ASCIIHexDecode /FlateDecode]}{
789cedcec90a02311404c006f7dd71dff5ffbfd21064c0d31cbc78a8823cd21d1292fca8d771decf20c38c32cea4eda675ceda3c2f6bd1a6e5d7fd55
e70fd669cadc645bd32efb4f7fc831a79ccbee926b6eb9d7f69167e78b00000000000000000000000000000000f04f5e6f090f020e
>}
\immediate\pdfobj useobjnum \csname EF8O15\endcsname stream attr{ /Filter [/ASCIIHexDecode /FlateDecode]}{
789c5d52cb6e833010bcf3157b4c0f11810468248454a5170e7da8b4a72807072f115231962107febe366382542418cdcc3ebc78c353f95aaa76a4f0
d3f475c52335ad928687fe6e6aa62bdf5a154431c9b61e3d9bbf75277410dae46a1a46ee4ad5f4419e53f865cd6134136d5e647fe5a78088c20f23d9
b4ea469b9f5305a9ba6bfdcb1dab9176415190e4c6967b13fa5d744ce19cbc2da5f5db71dadab435e27bd24cf1cc231ca9ee250f5ad46c84ba7190ef
ec5350ded8a70858c97f7e9420edda3ce263170f38032f4ede43de2790171a0162c01e70584291294085cf145e9690a59725e4034a39380321a3a683
3310f211f2d1cb47c8094e97f819123f438ad3a6bef64ae1a245ea275c29dc6750df69a570314ce61b7a9a32a059cc3936c38499efb350fcbbecb098
88459bcc775da975e33872a98033f0e2ee78b94c77dd6e371fbb54df8db16b342ff0bc3f6e735ac58f1dd7bd7659eefd032190cd3f
>}
\immediate\pdfobj useobjnum \csname EF8O16\endcsname {<< /I1 \csname EF8O17\endcsname\space 0 R /M0 \csname EF8O18\endcsname\space 0 R >>}
\immediate\pdfobj useobjnum \csname EF8O17\endcsname stream attr{/Type /XObject /Subtype /Image /Width 379 /Height 203 /ColorSpace [ /Indexed /DeviceRGB 1 <ffffff23313d> ] /BitsPerComponent 1 /Filter [/ASCIIHexDecode /FlateDecode] /DecodeParms [null << /Predictor 10 /Colors 1 /Columns 379 /BitsPerComponent 1 >>]}{
789cedd8319283300c0550312e2839828fe2a3c1d17c141f81928259ad24444276427653adc27c158c47e8d10805c7441e037fc9b55f892da63227b9
1225be57308d8d0062018d911772202176cab69ae92400fe1d94993a696c6f054bc75ced094d3b5d66af908735802b00294d06f42d915743afec099b
69808060ffed95a4f6bb0eab3536e972b959806080e4fe1ed91a7bf87efae402c403a37f296dfbd37ad909155e3d214b06b80ca894a42a4b8a4a231f
61bd93f8210062016fac4ff3d67a9de9cc0c101a48745bea90f9590d100eec23ba6c9f4efdef313c2907b804d873d5124fcb01e2013bb6cb7ae4530f
fb5a80e0c00f6cb6e57939402cf070e8f63a0002019b48eb74fea51ce0e3c1e4e7ad009f042c86bf940344027939dfa6020406afb6a9001702e58d72
8038e0ad728030e01b317e861b
>}
\immediate\pdfobj useobjnum \csname EF8O18\endcsname stream attr{/Type /XObject /Subtype /Form /BBox [ -4.98205081 -4.98205081 4.98205081 4.98205081 ] /Filter [/ASCIIHexDecode /FlateDecode]}{
789c6d90390ec53008447b4ee10b8ce50d25b429738d3451a4dcbfcd66f1f9320d420cbc017238288595ee801ca75a12e770528a8da53636c5146711
a9f294b809dfa51c4b69d3d7d5b38d54854e2803cafd61df11f550b9430ca33b182b236afbc8d77423c71f23049e17c69530ee8dbfbbe01c0ee73f70
1f09e7e3709d604edc8916ba0078975993
>}
\immediate\pdfobj useobjnum \csname EF8O19\endcsname {<< /A1 << /Type /ExtGState /CA 0 /ca 1 >> /A2 << /Type /ExtGState /CA 1 /ca 1 >> >>}
\immediate\pdfobj useobjnum \csname EF8O20\endcsname {<<  >>}
\immediate\pdfobj useobjnum \csname EF8O21\endcsname {<<  >>}
\immediate\pdfobj useobjnum \csname EF8O22\endcsname stream attr{/Type /XObject /Subtype /Form /FormType 1 /BBox [0 0 306 190.8] /Resources << /Font \csname EF8O1\endcsname\space 0 R /XObject \csname EF8O16\endcsname\space 0 R /ExtGState \csname EF8O19\endcsname\space 0 R /Pattern \csname EF8O20\endcsname\space 0 R /Shading \csname EF8O21\endcsname\space 0 R /ProcSet [ /PDF /Text /ImageB /ImageC /ImageI ] >> /Filter [/ASCIIHexDecode /FlateDecode]}{
78daed58cb6e5c370c9db596fd02a1ab781199a428892ad045dc34460214a813035d1459a4cee485d8a993b46efebe479a3bf7e18e274eb30dec6b8f
295ef288220f295f7af6842ff677a97d9f9dbbc3fbebbf5f9fad1f1f1ff99f9eccff3afbe0d8bfc1f3122fbcc17385d78ef1bc74cdc4b98b94f1fb6d
ffcd9582e1338d9f5e39f7c21dde83fa07681d3b2735a41a6bf562a1e46439751b12b8722d3613bf9d8bb96828b5960cf168622eec8e2efd0ef35204
b252cda2987ad61c2846abc9bf5ffbdffc853fbc270d1eb689e7aac35c06e0b2db30db040def73dee5e8ecdc1f3e647fff9d3f7127feabb150c00b57
8e42a164263149852c52358a550081024b55e692b5f8c7c7eed26d8f35c640061f9948c547b81778662d257694bf10508eda1c03204553e26c3e871c
4b258d6cf24df99bf237e5776e53cf436d12a870539fb6618b267163c92d896d92a6205b5e7347707f35d4ebdd5eb014aaa2ac2db582e51c34c5d829
c51d9dbac3070c64c59fbee89c7bfadcfdeeeface4c03ff5a78fdccfa7c0d7c138ab81083b90058899742f08039b52d55aa8805d6e078257f45f189c
2c548a529738e6e2bd403869c894950ba764b70dc72e2402ff4624bc443217ef4522a22181a1113cbd258eb81347b5906be1780dc74cbc1f47d5a000
9c4a945b27882e91cc03ac35588577e50c7309f941d8e0c618b5ce63253533cd7009d20cdf59dd3ff0a76f5c0ac9046793a9b5a0a67467b5ee2b8098
4915bda9b7cabef2e7b09249285131447c58f9d857d05615fbe698cbb8f26a78276d7291226f57fcf04e244d39a9d4b259397c2013ccf3aeb408fed0
87a3a01ad5b8b4d0a3f35e13ee4c450bf0c29c5a3757b194919395f7d7c4aef31f20a00472cc1a750e6112ee87905240f367d3c64f7b21a43d102a07
9354539e4398843b21e89021a01246c2128e53ca9713c334b921dd09658ec96e066226dd8b820529902a206424ce6760ec0b45632421adf838873149
f7c3c07924b5828c473efe0f729ae6580e6a84af058c49ba1f46069989a654d0c4ea6760c89211fa1da09934eca4e79861ef26a8d3f4393a78381423
4763c371c4b1e82f6e2ce0e737d2c1cd14f2cfb022959294a8939f5bd1c17aa283af1ccac53ff2fd2e14102774fedb0ee67e799b704a015b4f62d364
8e7347cfa3542a22a2296072cfd26a11674b683e35fbc4a1623e88101a288b33c69539715490436d5dd367b492a231b6f465aa986f92aa79a413ac8b
360b2c187b90bb683c05479f087783268e25d4a86d0f9602ce00116fe2e6bb55da9c25186e363dc4a38b698ed6a5059da548043826982004a1b534ae
3158e2027bcc1c4aac39377c422558e4669c4174b114b4b82617b8e2e67e410a0876285631a16170635c90704e0da1a0a1955c0559837444aab07438
9263286aadb9b06a005cdc279bbc14441edb835b6c0eac4e35f5ab25b05151bd46024fe60360cb8369f85b1eeeceebedce0b2b6c7ec95d18ead4f3ef
e57851dddec317bd2207cb027ee696bf637e5c6b18d779e1d703df02cbadc1e1678468bd7a7f8020840c22404f5d3d5b5d4076b1fab8f2f86b8de78f
a6b3fab07ab57adddafc5337b25bff5fc1662c9e636bac3d0d1d6c34c462e29956f713d3d8a6808f86b1436de0a5a1fadf0dbc6015b3d04cfed7a0df
8967125fec163f1fc573dbcf06db9504845346f9fb2e8fa0c84a7952ff34a89785f20f033de5bac0b765ad46589d5847ceb2f90053d1e2388d4e6601
f9b1af233d335a844dfe64bb15438ce30ec31b2accadabe8dcf3d6f2f7dfdd6099671ceadcc9bf4660aa17
>}
\expandafter\gdef\csname EFWidth8\endcsname{306}
\expandafter\gdef\csname EFHeight8\endcsname{190.8}
\endgroup
\edef\EFSetResources{\noexpand\pdfpageresources{ /XObject << /EF1 \csname EF1O23\endcsname\space 0 R /EF2 \csname EF2O20\endcsname\space 0 R /EF3 \csname EF3O20\endcsname\space 0 R /EF4 \csname EF4O22\endcsname\space 0 R /EF5 \csname EF5O20\endcsname\space 0 R /EF6 \csname EF6O20\endcsname\space 0 R /EF7 \csname EF7O23\endcsname\space 0 R /EF8 \csname EF8O22\endcsname\space 0 R >> }}\EFSetResources

\begin{document}
\title{Diagonal Periods and Newton Supports of Rules 30, 86 and 135}
\author{\authname{Tigran Nersissian}\\[2pt]
\authadd{Independent researcher}\\
\authadd{Nice, France}\\
\authadd{ORCID: \texttt{0009-0001-8755-7412}}\\
\authadd{Email: \texttt{tigran.nersissian\symbol{64}hotmail.com}}
}
\markboth{Complex Systems}{Diagonal Periods and Newton Supports}
\maketitle
\begin{abstract}
The single-seed Rule 30 cone has unbounded least diagonal periods in
both directions. On the finite-support side, first contact fixes the
least Newton index and an integer polynomial lift gives a Fibonacci
ceiling. For reflected Rule 86, a backward map on periodic tail
profiles retains the zero boundary and bounds periods independently
of transients. At depth $m$, the Rule 30 period is at least
$\lfloor m/2\rfloor+1$; for natural Rule 86 cuts through depth $m$,
the largest period is at least $\tfrac12\log_2(m+2)$. This proves
infinitely many reset/integration period doublings. The physical
Rule 135 orbit is a translated complement of Rule 30 after its first
step, yielding finite or cofinite supports. A separately initialized
auxiliary recurrence has exact index-by-index stabilization. Residue
tables and block formulas describe the remaining structure and permit
evaluation from supplied representations. Construction and minimization
remain separate costs.
\end{abstract}

\keepsection
\section{Introduction}

The diagonals on either side of a single-seed Rule 30 cone become
periodic, but periodicity alone leaves a basic question unanswered:
can all their least periods remain bounded as the diagonal moves
deeper into the cone? We prove that they cannot. The two directions
require different information from the orbit. One argument starts
with the first active cell of each diagonal; the other uses its
eventual periodic tail and the zero boundary outside the cone.
Following both arguments reveals how the initial geometry constrains
periods even when the full pattern has no simple description.

\begingroup\widowpenalty=10000\clubpenalty=10000\brokenpenalty=10000
Newton supports provide the common representation used here. A binary
sequence is encoded by the indices of its odd binomial coefficients;
the support calculus translates the rule into operations on those
indices. Finite supports describe periodic sequences, and masked dyadic
blocks collect groups of coefficients into direct digit tests. This
representation separates three quantities that can otherwise be
confused: the period of a sequence, the cost of evaluating a supplied
formula, and the cost of constructing that formula. A large period is
compatible with a short digit description, while the existence of a
description gives no bound on the work needed to find it.
\par\endgroup

Our first period argument concerns the finite supports on the Rule 30
side of the cone. The first active cell determines their minimum index,
which increases with diagonal depth. An integer polynomial lift controls
the maximum index through a Fibonacci degree bound. The exact relation
between finite Newton support and least period then proves unbounded
period growth. This argument uses the earliest nonzero value and the
interpolation order of each diagonal; it does not require an explicit
formula for every intervening support coefficient.

Reflection gives Rule 86 and exposes the less direct mechanism.
Its diagonals have periodic tails driven by adjacent diagonals:
a one in the driver resets the recurrence, whereas an eventually
zero driver leaves a discrete integration problem. The known
reset/integration criterion describes a record increase when it
occurs. The question addressed by Theorem~\ref{thm:profile} is why
such increases must continue. Pairs of periodic profiles satisfy a
backward recurrence whose zero boundary pair is fixed. The prescribed
first nonzero pair forces all subsequent predecessors to be distinct.
Counting those pairs gives a logarithmic lower bound on the largest
period reached by a specified depth. The proof compares absolute
phases on tails, so it needs no uniform estimate for their transient
lengths.

Rule 135 requires an additional distinction about initial data. Its
physical single-seed orbit becomes a translated complement of Rule 30
after the first step. The corresponding support transformation therefore
has to include that step and the changed background. A related
recurrence obtained by prescribing two finite diagonal seeds defines a
different object. We analyze it separately: an interval becomes
permanently included in its supports, and each index has an exact
stabilization time. The residues that remain at powers of two give a
concrete setting for questions about density, popcount and compression.

The period results are Theorems~\ref{thm:unbounded} and
\ref{thm:profile}. They provide the main orbit conclusions and
explain why the two directions are treated separately. The Rule 135
sections then determine how much of this structure survives a
translated complement or a change of diagonal seeds. In particular,
the exact auxiliary stabilization law isolates a permanent part of
the support, leaving residues on which compression and density can
be studied. The finite tables record those residues at stated depths
and windows. Thus the later questions arise from explicitly
identified parts of the recurrence that the period and stabilization
proofs leave undetermined.

\keepsection
\subsection{Relation to Earlier Work}
Rule 30 was studied by Wolfram in the statistical mechanics of cellular
automata \cite{wolfram1983}. Rowland proves local nested structure and,
in Proposition 2, the white-stripe and odd-parity criterion for period
doubling on the regular side \cite{rowland}. Our reset/integration proof
restates that criterion in the diagonal coordinates used here. We then
retain the zero boundary profiles in the backward recurrence to obtain
a quantitative bound on the largest period. The record indices are
tabulated in OEIS A364239 \cite{oeis}.

Kopra's rapidly left expansive framework gives aperiodicity results for
spatial traces, including the relevant distinction between traces of
width one and two \cite{kopra}. Those results concern a different
direction through the space--time diagram. The fixed-diagonal theorems
below do not determine the eventual behavior of Rule 30's center column.
Pivato's defect formalism treats interfaces in one-dimensional automata
\cite{pivato}; the driver used here is the value of a specified adjacent
diagonal, rather than a particle variable.

The support operations come from the foundational paper \cite{U1}.
Classical algebraic and fractal analyses of additive cellular automata
provide comparison cases \cite{mow,willson}. The counting argument for
block decompositions follows the circuit-counting method associated
with Lupanov and the Boolean-complexity framework of Wegener
\cite{lupanov2,wegener}. Its conclusion concerns almost all sets of a
given density, so it cannot establish optimal compression of the
particular supports in this paper.

\keepsection
\section{Support Representations}\label{sec:bg}

We first establish the common language needed for all three rules.
The same support may be viewed as a set of Newton indices, as a binary
function on a finite digit window, or as a symmetric difference of
blocks. Each view serves a different step: the set recurrence follows
the local rule, the digit view explains periods, and the block view
gives an evaluation formula. The coordinate conventions are fixed
before using any of these interpretations.

We write $\NN=\{0,1,\ldots\}$ and work over $\FF$. For nonnegative
integers, $r\subseteq t$ means that every set bit of $r$ is set in $t$;
$\pc(r)$ is the number of set bits. Symmetric difference is $\symd$,
and Boolean exclusive-or is $\xor$.

\keepsection
\subsection{Coordinates and the Newton Transform}
\label{bg:I-1-3}
Let $A_R(t,i)$ be the infinite-lattice orbit of Rule $R$ with
$A_R(0,i)=[i=0]$, and set
\[
 b_R(m,n)=A_R(n,n-m+1),\qquad m\ge1,\ n\ge0.
\]
For a local rule $g(p,q,r)$ this rotation gives
\[
 b_R(m,n+1)=g\bigl(b_R(m,n),b_R(m-1,n),b_R(m-2,n)\bigr).
\]
The boundary values come from the infinite-lattice orbit. They are zero
for $m\le0$ for Rules 30 and 86, whose all-zero state is fixed. Rule 135
requires the nonzero background treated in Section~\ref{sec:physical}.
For radius $r$, the same calculation with $A(n,rn-m+1)$ gives offsets
$0,1,\ldots,2r$.

Every binary sequence $d$ has a unique Newton support $S\subseteq\NN$:
\begin{equation}\label{eq:newton}
 d(n)=E(S)(n):=\Xor_{r\in S}\binom nr\pmod2,
 \qquad [r\in S]=\Xor_{t=0}^r\binom rt d(t).
\end{equation}
The sum is finite for each $n$, even when $S$ is infinite. Lucas' theorem
identifies the matrix on $[0,2^k)$ with
$Z[t,r]=[r\subseteq t]$. Its square is the identity over $\FF$: the
number of intermediate submasks between $r$ and $t$ is a power of two,
odd only when $r=t$. This proves both formulas and uniqueness.

\label{bg:I-1-6}
The product of two sequences corresponds to OR--convolution,
\[
 A\conv B=\{c:\#\{(a,b)\in A\times B:a\vee b=c\}\text{ is odd}\}.
\]
It satisfies $E(A\conv B)=E(A)E(B)$. The operation is well defined
for arbitrary subsets of $\NN$: a pair
contributing to $c$ lies in $[0,c]^2$. Addition corresponds to symmetric
difference. Increment is $\Inc A=\{a+1:a\in A\}$ and satisfies
$\Delta E(\Inc A)=E(A)$, where $\Delta d(n)=d(n+1)\xor d(n)$.

\label{bg:I-1-13}\label{bg:I-1-14}
Write the source $h=g\xor p$ in algebraic normal form. Replacing its
constant monomial by $\{0\}$, addition by $\symd$, and products by
$\conv$ gives a support expression $H$. The initialized recurrence is
\[
 S_m=(b_R(m,0)\{0\})\symd\Inc H(S_m,S_{m-1},S_{m-2}).
\]
When $h$ is independent of $p$, this is explicit in previous diagonals.
Otherwise the increment makes it triangular in the Newton index:
the coefficient at index $r>0$ uses only indices below $r$ of the
current support. Thus each finite prefix is computed uniquely. This
fact alone implies neither a finite support nor automaticity.

\keepsection
\subsection{Finite and Cofinite Masks}
\label{bg:II-2-1}\label{bg:II-2-4}

Finite masks describe freely varying digits in a bounded window.
Cofinite masks allow all sufficiently high digits to vary and arise
naturally when a support is transformed by complementation or a finite
initial correction. Both can be stored with finite data under the
conventions below. Their evaluation properties differ, however, so
the distinction is retained in the notation and in the period
arguments that follow.

For disjoint finite bitmasks $v,M$, let
\[
 B(v,M)=\{v\vee s:s\subseteq M\}.
\]
A valid block has $2^{\pc(M)}$ elements. We interpret a displayed pair
with $v\wedge M\ne0$ as the empty set. On $k$ bits there are $3^k$
valid blocks, and $\binom kj2^{k-j}$ have $2^j$ elements. For valid
inputs the product rule is
\[
 B(v_1,M_1)\conv B(v_2,M_2)=B(v_1\vee v_2,M_1\vee M_2),
\]
with the empty-output convention just stated \cite{U1}.

An infinite mask is a set of bit positions. The same notation defines
a block using its finite submasks. A cofinite mask contains every bit
above some position $j$. We store its lower bits and the explicit
threshold $2^j$. Two useful infinite supports are
\[
 R(v,2^j)=\{v+q2^j:q\ge0\},\quad 0\le v<2^j,
 \qquad \Up t=\{r:t\subseteq r\}.
\]
The former has all bits $\ge j$ free; the latter has every bit not in
$t$ free. Formula~\eqref{eq:newton} gives
\begin{equation}\label{eq:transientatoms}
 E(R(v,2^j))(n)=[v\subseteq n][n<2^j],
 \qquad E(\Up t)(n)=[n=t].
\end{equation}
Indeed the allowed free bits contribute a factor $2^q$; its parity is
one precisely when none of them occurs in $n$.

\label{bg:II-2-19}\label{bg:II-2-6}
Splitting a free bit reverses the merge identity
\[
 B(v,M)\symd B(v\vee 2^j,M)=B(v,M\vee2^j),
 \qquad 2^j\wedge(v\vee M)=0.
\]
Together with cancellation and the absorption identities of \cite{U1},
these moves rewrite a finite XOR list without changing its support.
Convolution of lists requires all supplied pairs; increment fragments
each block at its binary carry positions. Neither operation gives a
constant-cost construction or guarantees a minimal list.

\label{bg:I-1-19}\label{bg:III-3-24}
For product-free sources with polynomial boundary data, the integer
coefficient recurrence is linear. The initialized Rules 90 and 150
have finite formulas; in particular Rule 90 has empty even diagonals.
For the polynomial lift of Rule 30 used below, the dimension bound is
$F_{m+1}$, with $F_1=F_2=1$. This bound is specific to that lift and
boundary data, rather than a universal dimension for all rules.

\label{bg:V-4-7}\label{bg:XI-5-7}
On a dyadic window the zeta transform uses $O(N\log N)$ XOR operations
and $O(N)$ stored bits. It makes OR--convolution pointwise and increment
a prefix XOR, as shown in Section~\ref{sec:zeta}. For comparison,
the diagonal basis of a single-seed rule is
$B_R[t,k]=A_R(t,t-k)$ for $t,k\ge0$. Thus $B_{204}=I$,
$B_{220}[t,k]=[k\le t]$ and
$B_{60}[t,k]=\binom tk\bmod2$; Rule 102 gives the mirrored Pascal
system. The parity-staircase basis associated with Rule 28 has column
zero identically one and column $k\ge1$ equal to
$[t\ge k][t-k\text{ even}]$. These basis matrices are
distinct from the set-valued blocks $B(v,M)$.

\keepsection
\subsection{Exact Periods}\label{sec:period}

A period statement concerns the values of the represented sequence,
rather than the number of indices in its support. Lucas' digit
criterion provides the connection: a finite support reads only a
bounded set of binary positions, and the highest active Newton index
determines the least dyadic period. We begin with a single block to
make the digit dependence explicit, then pass to the support as a
whole, including the effect of finite initial corrections.

\keepstatement
\begin{lemma}[Evaluation of a Block]\label{lem:blockeval}\label{thm:zero}
For a finite anchor $v$ and a finite or cofinite mask $M$, with disjoint
set bits, put $g_{v,M}=E(B(v,M))$. Then
\[
 g_{v,M}(n)=[v\subseteq n][M\wedge n=0].
\]
\end{lemma}
\begin{proof}
Lucas' theorem counts the contributing submasks as
$[v\subseteq n]2^{\pc(M\wedge n)}$. This is odd exactly under the two
displayed conditions. The count is finite even for an infinite mask,
since only bits present in $n$ can contribute.
\end{proof}

\keepstatement
\begin{corollary}[Finite and Cofinite Masks]\label{cor:dicho}
For a finite mask let $k$ be the bit length of $v\vee M$, with $k=0$
when $v=M=0$. Then $g_{v,M}$ has period dividing $2^k$. If instead
$M$ contains every bit at positions $j,j+1,\ldots$, then
$g_{v,M}(n)=0$ for $n\ge2^j$.
\end{corollary}
\begin{proof}
In the finite case the tests read only the lowest $k$ bits. In the
cofinite case disjointness forces every bit of $n$ at position at least
$j$ to vanish.
\end{proof}

\keepstatement
\begin{theorem}[Exact Least Period of a Finite Support]\label{thm:exactper}
Let $S\subseteq\NN$ be finite and non-empty, $D(t)=\Xor_{r\in S}\bin tr$, and let $k$ be
minimal with $S\subseteq[0,2^k)$. Then $D$ is purely periodic with minimal period
\emph{exactly} $2^k$.
\end{theorem}

\begin{proof}
For $k=0$, $S=\{0\}$ and $D=1$, of least period one. Suppose $k\ge1$.
$2^k$ is a period by Corollary~\ref{cor:dicho}. If $d$ is the least positive
period, divide $2^k=qd+r$, $0\le r<d$. Periodicity gives $D(t+r)=D(t+2^k)=D(t)$,
so $r=0$. Thus the minimal period divides $2^k$ and is
some $2^j$; it suffices to exclude $2^{k-1}$. Let
$A=\{r\in S:\ \text{bit }k{-}1\text{ of }r\text{ is set}\}$, non-empty by minimality of $k$.
Take $t<2^{k-1}$ and $t'=t+2^{k-1}$. For $r\notin A$ the bit $k-1$ is absent from $r$, so
$[r\subseteq t]=[r\subseteq t']$; for $r\in A$ we have $[r\subseteq t]=0$. Hence
\[
\begin{aligned}
 D(t')\xor D(t)&=\Xor_{r\in A}[\,r\subseteq t'\,]
 =\Xor_{r'\in A'}[\,r'\subseteq t\,],\\
 A'&=A-2^{k-1}\subseteq[0,2^{k-1}).
\end{aligned}
\]
If $2^{k-1}$ were a period this would vanish for every $t<2^{k-1}$. The matrix
$\big([\,r'\subseteq t\,]\big)_{t,r'}$ on $[0,2^{k-1})$ is the zeta matrix, unitriangular
hence invertible over $\FF$, so its columns are linearly independent and $A'=\varnothing$ ---
a contradiction.
\end{proof}

The empty support also gives a sequence of least period one. Thus the
largest surviving Newton index determines the period after all XOR
cancellations, rather than the largest anchor in an unreduced list.

\keepstatement
\begin{theorem}[The Finite Part Carries the Period]\label{thm:split}
Let $F,C$ be finite index sets and $D=\Xor_{i\in F}g_{v_i,M_i}\xor\Xor_{i\in C}g_{v_i,M_i}$ with $M_i$ finite for $i\in F$
and cofinite for $i\in C$, and set $S_F=\symd_{i\in F}\Bl{v_i}{M_i}$ (a finite set) and
$T=\max(\{0\}\cup\{2^{j_i}:i\in C\})$, where $M_i$ contains every bit at position at least $j_i$. Then $D$ is ultimately periodic with transient at most $T$ and
ultimate period exactly
\[
 \pi=2^{\lceil\log_2(\max S_F+1)\rceil}\qquad(\pi=1\text{ if }S_F=\varnothing).
\]
\end{theorem}

\begin{proof}
By Corollary~\ref{cor:dicho} every term of $C$ vanishes for $t\ge T$, so on $[T,\infty)$
the sum reduces to $\Xor_{r\in S_F}\bin tr$, whose minimal period is $\pi$ by
Theorem~\ref{thm:exactper}. If $q$ is a period of that tail, then for any $t$
choose a multiple $a\pi$ with $t+a\pi\ge T$. Periodicity of the finite part gives
$E(S_F)(t)=E(S_F)(t+a\pi)=E(S_F)(t+q+a\pi)=E(S_F)(t+q)$.
Thus $q$ is a period of the whole finite-part sequence and is a multiple of $\pi$.
\end{proof}

The finite masks determine the periodic tail; the cofinite masks
describe a finite transient in time. A cofinite block itself has
infinitely many indices. Its truncation to an increasing dyadic
window may have exponentially many elements even when its description
uses only a few stored masks. A mask that is neither finite nor
cofinite need not have either of these two behaviors.

\keepsection
\subsection{Direct Evaluation}\label{sec:evaluation}
\label{bg:II-2-7}\label{bg:I-1-15}
For a query $n\ge0$, write
$\ell=\max(1,\lceil\log_2(n+1)\rceil)$. Let $b\ge1$ bound the bit
lengths of stored anchors, finite masks and explicit modulus thresholds
in a supplied representation. Binary addition, comparison and a Lucas
containment test on these data cost $O(b+\ell)$ bit operations.
Complemented masks are stored by their finite exceptional bits and
their threshold. Constructing or minimizing a representation is a
separate task.

\keepstatement
\begin{theorem}[Evaluation from Supplied Blocks]\label{thm:logeval}
With a finite-mask block decomposition $S_m=\symd_{i=1}^{t}\Bl{v_i}{M_i}$ ($t=K(m)$ for the supplied decomposition),
\[
 b(m,n)\;\equiv\;\Xor_{i=1}^{t}\big[\,(v_i+M_i)\subseteq(n+M_i)\,\big]\pmod2,
\]
each bracket a bitwise test on integers of length $O(b+\ell)$, so $b(m,n)$ is evaluated in
$O(t(b+\ell))=O(K(m)(b+\ell))$ time --- without simulating prior generations.
\end{theorem}
\begin{proof}
The binomial transform is linear over $\FF$. For each block,
Vandermonde--Lucas compression gives
$\Xor_{r\in B(v_i,M_i)}\binom nr=\binom{n+M_i}{v_i+M_i}$, and Lucas' theorem evaluates
its parity by the displayed containment test. XORing the $t$ block contributions proves
the formula. Each test scans $O(b+\ell)$ bits, including the block parameters.
\end{proof}

The empty list returns zero in constant time. For a nonempty list,
sequential evaluation uses $O(b+\ell)$ working bits beyond its stored
terms. The same bound applies to a mixture of finite and cofinite
blocks by Lemma~\ref{lem:blockeval}. In particular $\Up t$ contributes
only the equality test $n=t$.

\begin{example}[The Eighth Rule 30 Support]\label{ex:eighth}
The sixth and seventh supports are
$S_6=B(5,0)\symd B(3,4)$ and
$S_7=B(3,0)\symd B(5,0)\symd B(8,0)$. Their linear and product
terms give the source
\[
 \{3,5,7,8,11,13,15\}.
\]
Incrementing gives
$S_8=\{4,6,8,9,12,14,16\}$. Its decomposition
$B(4,10)\symd B(8,1)\symd B(16,0)$ yields
\[
 b_{30}(8,n)=\binom{n+10}{14}\xor\binom{n+1}{9}
                  \xor\binom n{16}\pmod2.
\]
Seven singleton terms have become three Lucas tests. This is a property
of this supplied decomposition, rather than a general compression ratio.
\end{example}

\keepstatement
\begin{theorem}[Companion, Support and Periodic Forms]\label{thm:three}
Let $S$ be finite and nonempty, put $N=\max S+1$, and let
$d=E(S)$. For every $n\ge0$,
\[
 d(n)=\sum_{q=1}^N d(N-q)
       \binom n{N-q}\binom{N-1-n}{q-1}\pmod2.
\]
This is the companion evolution for $(E-1)^N$, where $E$ shifts $n$.
It equals both the block formula of Theorem~\ref{thm:logeval} and
$d(n\bmod P)$, where $P=2^{\lceil\log_2N\rceil}$ is the least period.
The same interpolation formula over the rationals reconstructs any
integer-valued polynomial lift of degree below $N$ from its first $N$
values; reducing those values modulo two suffices for its parity.
\end{theorem}
\begin{proof}
The rational polynomial $f(n)=\sum_{r\in S}\binom nr$ has degree
$N-1$. The $q$th displayed interpolation polynomial has value one at
$n=N-q$ and zero at every other integer in $[0,N)$, by its two
binomial factors. Thus the sum reconstructs $f$ and may be reduced
modulo two. Its $N$th forward difference vanishes, giving the
companion recurrence. The other formulas follow from linearity and
Theorem~\ref{thm:exactper}. Evaluate each generalized binomial product
as an integer before reducing it. For $n<N$ use the initial values;
for $n\ge N$, upper negation turns the second factor into
$(-1)^{q-1}\binom{n-N+q-1}{q-1}$, so two ordinary Lucas tests suffice.
\end{proof}

An interpolation coefficient has Newton support
$\{r<N:N-q\subseteq r\}$. It is a single block for every $q$ when
$N$ is a power of two; that assertion is false at general orders.
For the preceding example $N=17$ and $P=32$, so the order-17
companion and the three block terms describe the same period-32 sequence.

\keepsection
\subsection{Compression and Its Limits}\label{sec:compression}
\label{rem:costmeasure}

Once evaluation is expressed as a sum of block tests, the block count
becomes a natural measure of a supplied formula. Its minimum value is
a property of the particular support and of the allowed blocks.
Counting all possible descriptions gives a useful comparison with
arbitrary sets, but the cellular-automaton recurrence selects a very
special collection of supports. The counting result below is therefore
used to frame the compression question, with its quantifiers made
explicit.

We write $K(S)$ for the minimum number of valid blocks in a finite XOR
decomposition of $S$. The notation $K(m)$ in an evaluation bound counts
the supplied terms and need not equal $K(S_m)$.

\keepstatement
For a target in $[0,2^k)$, intersecting every supplied block with
that window either deletes it or gives a valid $k$-bit block. Thus
allowing higher bits cannot lower the minimum by evading the finite
block census.

\begin{theorem}[Counting Bound for Almost All Supports]\label{thm:lupanov}
Fix $\delta\in(0,1)$ and let $H$ be the binary entropy. For all but a vanishing fraction of
the sets $S\subseteq[0,2^k)$ of cardinality $\lfloor\delta2^k\rfloor$,
\[
 K(S)\;\ge\;(1-o(1))\,\frac{H(\delta)\,2^{k}}{\log_2\!\big(3^{k}\big)}
\]
blocks in any XOR decomposition. The lower-bound scale is $2^k/k$;
this counting argument does not supply a matching upper bound.
\end{theorem}

\begin{proof}
There are $3^k$ valid blocks on $k$ bits. For each $j\le K$, at most
$3^{kj}$ ordered lists of $j$ blocks occur, so at most $(K+1)3^{kK}$ sets
can be represented with at most $K$ terms. Stirling's formula gives
$\log_2\binom{2^k}{\lfloor\delta2^k\rfloor}=H(\delta)2^k+O(k)$.
Choose $\varepsilon_k=k^{-1/2}$ and
$K=\lfloor(1-\varepsilon_k)H(\delta)2^k/(k\log_2 3)\rfloor$.
The fraction representable with at most $K$ blocks is at most
$2^{-\varepsilon_kH(\delta)2^k+O(k)}$, which tends to zero.
As $\varepsilon_k\to0$, this proves the stated lower bound.
\end{proof}

The binary entropy in the theorem is
$H(\delta)=-\delta\log_2\delta-(1-\delta)\log_2(1-\delta)$.
Density alone gives no corresponding lower bound for a particular set:
a half-cube has density $1/2$ and is one block. Consequently a measured
factor near two from greedy compression does not establish optimality
for Rule 30 or for the auxiliary recurrence below.

\keepstatement
\begin{algo}[Finite Block Search]\label{alg:search}
Given a support bitmap on $[0,N)$, $N=2^k$, one exact greedy cover
chooses its least remaining element as an anchor and repeatedly adds
the first unused free bit whose whole block is still present. Emit the
block, remove its elements and repeat. Each step removes at least one
element, so the procedure terminates and the emitted disjoint blocks
have union equal to the input support.
\end{algo}
With explicit membership tests this implementation takes at most
$O(N^2(k+1)^2)$ indexed Boolean operations: there are at most $N$
emissions, at most $k+1$ sweeps over $k$ candidates per emission, and
at most $N$ tested members per candidate. The bitmap and stored output
use $O(N(k+1))$ bits. These are conservative construction bounds,
distinct from the bit cost of a later query.

Merge--split local search can reduce the resulting list while preserving
the XOR invariant. A learned search can use the same exact invariant
and the reduction in list length as its reward. With a fixed window,
list-length cap and move horizon, this is a finite decision process;
no learning performance or minimum-size guarantee is assumed here.
Exhaustive search through lists of at most $N$ of the $3^k$ blocks
terminates and can find a minimum, but its exponential cost supplies
no practical evaluator for large supports.

\keepsection
\section{Rule 30}\label{sec:r30}

The local map is $g(p,q,r)=p\xor q\xor r\xor qr$. Its source
$q\xor r\xor qr$ contains only the two preceding diagonals, so
\[
 S_m=\Inc\bigl(S_{m-1}\symd S_{m-2}\symd(S_{m-1}\conv S_{m-2})\bigr),
 \qquad S_1=\{0\},\quad S_2=\{1\}.
\]
In this section $S_m$ always denotes the Rule 30 support and
$b(m,n)=b_{30}(m,n)$. These sets
are finite by induction. Their first and last indices have different
origins: the first is fixed by propagation from the seed; the last is
bounded by the degree of an integer polynomial lift.

\keepsection
\subsection{First Contact and the Degree Bound}

The two ends of a finite support can be controlled by different pieces
of information. Before its first nonzero term, a sequence has only
zero Newton coefficients, so first contact identifies the lower end.
At the upper end, a polynomial lift bounds the order of nonvanishing
differences. Combining these observations gives useful support bounds
without determining all coefficients between them, which is precisely
what is needed for the period argument.

\keepstatement
\begin{theorem}[First-Contact Floor]\label{thm:sol}
For Rule~30, $\min S_m=\lfloor m/2\rfloor$.
\end{theorem}
\begin{proof}
The first nonzero time identifies the minimum support index. $b(m,n)\equiv\Xor_{r\in S_m}\binom nr$ vanishes for
$n<\min S_m$, since every $\binom nr$ with $r>n$ is zero; conversely at $n=\min S_m$ the
unique term with $r=n$ contributes $\binom nn=1$ and all others vanish, so
$b(m,\min S_m)=1$ and $\min S_m$ is exactly the first index at which the diagonal is
nonzero. In the rotated frame of Section~\ref{sec:bg} the diagonal $m$ occupies the cells
$(t,t-m+1)$. The condition $|t-m+1|\le t$ first holds at
$t=\lceil(m-1)/2\rceil=\lfloor m/2\rfloor$; below it the automaton has not been reached
from the seed. At first contact the cell is on the left boundary (odd $m$) or one site
inside it (even $m$). Write $\ell(t)$ and $\ell'(t)$ for the values of the leftmost and
second--leftmost cells of the cone at time $t$, i.e.\ at positions $-t$ and $-t+1$. From the
single seed, $\ell(0)=1$ and $\ell'(0)=0$. Rule~30 is $g(x,y,z)=x\xor(y\vee z)$, so the cell
at position $-t-1$ at time $t+1$ reads the neighbourhood $(0,0,\ell(t))$ and equals
$\ell(t)$, while the cell at position $-t$ reads $(0,\ell(t),\ell'(t))$ and equals
$\ell(t)\vee\ell'(t)$. Hence $\ell(t+1)=\ell(t)$ and
$\ell'(t+1)=\ell(t)\vee\ell'(t)$, so by induction $\ell(t)=1$ for every $t$ and
$\ell'(t)=1$ for every $t\ge1$. Both leftmost cells of the cone are therefore one from
$t=1$ on, and the cell reached at first contact is one. Thus the first nonzero diagonal
value occurs exactly at $\lfloor m/2\rfloor$.
Consequently $\max S_m\ge\min S_m=\lfloor m/2\rfloor\to\infty$.
\end{proof}

Choose the integer lift
$a_m(n+1)-a_m(n)=a_{m-1}(n)+a_{m-2}(n)+a_{m-1}(n)a_{m-2}(n)$,
with $a_1=1$, $a_2(n)=n$ and $a_m(0)=0$ for $m>1$. Indefinite
summation produces an integer-valued polynomial at each fixed depth,
whose reduction modulo two is $b_{30}(m,n)$. The causal zero region
is illustrated in Figure~\ref{fig:spectrum}.

\keepstatement
\begin{theorem}[Fibonacci Ceiling]\label{thm:fibceil}
$\max S_m\le F_{m+1}-1$.
\end{theorem}
\begin{proof}
Let $d_m$ be the degree of the integer lift. Antidifferencing raises degree by at most one,
and Rule~30's source contains only the two previous lifts and their product, so
$d_m+1\le(d_{m-1}+1)+(d_{m-2}+1)$. Since $d_1=0$ and $d_2=1$, induction gives
$d_m+1\le F_{m+1}$. Reduction modulo two can delete binomial coefficients but cannot
create an index above $d_m$.
\end{proof}

\begin{figure}[htbp]
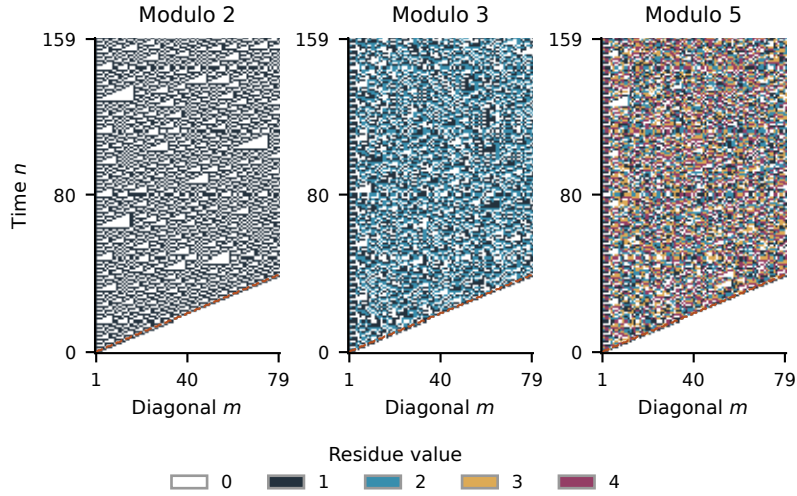
\centering
\EmbeddedFigure{\textwidth}{1}
\caption{Fixed diagonals $1\le m\le79$, $0\le n<160$.
Left: the binary Rule 30 orbit. Middle and right: the same chosen
integer lift reduced modulo three and five. The key identifies each
residue value; only the appropriate residues occur in each panel.
The dashed line is $n=\lfloor m/2\rfloor$, below which the seed has
not reached the diagonal. Equality with the first nonzero time is
asserted for the binary orbit; the other panels illustrate the causal
zero region.}
\label{fig:spectrum}
\end{figure}

\keepsection
\subsection{Unbounded Least Periods}
Write $D^{30}_{m-1}=b_{30}(m,\cdot)$.
\keepstatement
\begin{theorem}[Unbounded Diagonal Periods]\label{thm:unbounded}
For Rule~30 from $\delta_0$, the diagonal $D^{30}_{m-1}$ is purely periodic with minimal period
$2^{\lceil\log_2(\max S_m+1)\rceil}\ge\lfloor m/2\rfloor+1$. Hence the period grows at least
linearly in $m$. Writing $P_m$ for that least period and
$Q_m=\max_{1\le j\le m}P_j$, every strict increase of $Q_m$ is a doubling,
and their number by depth $m$ is at least
$\lceil\log_2(\lfloor m/2\rfloor+1)\rceil$; in particular there are infinitely many.
\end{theorem}

\begin{proof}
The exact minimal period of a finite support is Theorem~\ref{thm:exactper}: it equals
$2^{k}$ with $k$ minimal such that $S_m\subseteq[0,2^k)$, so the period is at least
$\max S_m+1$. Theorem~\ref{thm:sol} gives $\min S_m=\lfloor m/2\rfloor$, and
$\max S_m\ge\min S_m$.
The first increase is $P_2=2$. For $m\ge3$, if $Q_{m-1}=2^e$, both preceding supports lie below $2^e$. Their OR-products and
symmetric differences also lie below $2^e$, and increment puts the next support
in $[1,2^e]$. Hence $P_m\le2^{e+1}$, so each new record doubles.
Since $Q_1=1$, the count is $\log_2 Q_m$, giving the lower bound.
No monotonicity of consecutive $P_m$ is needed.
\end{proof}

The contact floor and Fibonacci ceiling give the explicit enclosure
\[
 2^{\lceil\log_2(\lfloor m/2\rfloor+1)\rceil}
 \le P_m\le2^{\lceil\log_2 F_{m+1}\rceil}.
\]
For example, the support in Example~\ref{ex:eighth} has largest
index $16$, hence least period $32$, while its seven coefficients
are evaluated by three block terms. Period, support cardinality
and block count therefore give three different measurements even
at this small depth. The upper endpoint above comes from the
integer lift; cancellation modulo two can make it strict.

This theorem concerns fixed diagonals. In the center column the depth
varies with time, $c(t)=b_{30}(t+1,t)$, so the result cannot be applied
there by fixing $m$.

\keepsection
\subsection{Finite Growth Data}\label{sec:growthdata}
The initial maxima, for $1\le m\le11$, are
$0,1,1,2,4,7,8,16,25,32,58$. For example,
\[
 S_5=\{2,3,4\},\quad S_6=\{3,5,7\},\quad
 S_7=\{3,5,8\}.
\]
The supports through $m=44$ give the following least-squares slopes
against $m$. Each fit uses all integer depths in its indicated window.
The finite sets are constructed directly from the displayed support
recurrence. Given two predecessor bitmaps, choose a power $P$ of two
strictly larger than both maxima. On $[0,P)$, apply the zeta transform
to each bitmap, multiply pointwise, and transform back to obtain their
OR convolution. XOR the two linear terms and shift the resulting
bitmap upward by one, retaining the possible endpoint $P$. Since the
OR of two indices below $P$ remains below $P$, this procedure includes
the entire next support. Choosing the next window from its actual
maximum makes the computation independent of an extrapolated growth
law. One step uses $O(P\log(P+1))$ Boolean operations and $O(P)$
working bits; Section~\ref{sec:zeta} proves the transform identities.
\begin{table}[t]\centering\fontsize{9}{10.5}\selectfont
\caption{Dependence of the fitted Rule 30 growth slopes on the depth window.}
\label{tab:growth}
\begin{tabular}{|rrr|}\hline
Depth window & $\log_2|S_m|$ slope & $\log_2\max S_m$ slope\\\hline
8--20 & 0.285842 & 0.299437\\
10--26 & 0.228614 & 0.216147\\
12--28 & 0.253626 & 0.238782\\
20--44 & 0.408292 & 0.402749\\
30--44 & 0.438717 & 0.437158\\\hline
\end{tabular}
\end{table}
The two slopes differ by as much as $0.014845$, and neither is stable
under changes of the window. In particular these data do not determine
an asymptotic exponential rate. The previously proposed scale
$2^{0.294m}$ remains an unestablished estimate. On $28\le m\le44$,
the smallest value of $\log_2|S_m|/m$ is
$\log_2(2035)/35\approx0.314023$, attained at $m=35$.
Thus even on this finite range a lower bound of $0.32$ fails.
The degree ceiling holds throughout the
computed range. Figure~\ref{fig:growth} places these measurements
between the two proved bounds.

\begin{figure}[t]
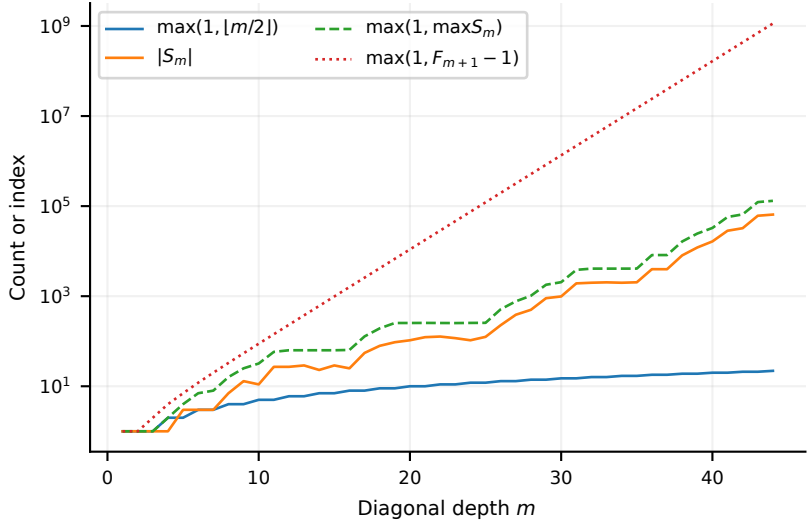
\centering
\EmbeddedFigure{\textwidth}{2}
\caption{Rule 30 support sizes and extreme indices through depth 44.
The contact floor and Fibonacci ceiling are proved bounds. The
intermediate curves are values of the finite supports; Table~\ref{tab:growth}
records the dependence of fitted slopes
on the chosen window.}
\label{fig:growth}
\end{figure}

\keepsection
\subsection{The Zeta--Floor Factorization}\label{sec:zeta}

The support recurrence can also be studied after applying the zeta
transform. In that form, OR convolution becomes pointwise
multiplication, while the increment operation contributes the
remaining index dependence. The factorization below separates these
effects and identifies which part of the recurrence is simplified by
the transform. It provides another exact description of the same
diagonals, whose usefulness still depends on the size of the data
being transformed.

Define $\Zet(A)(n)=\Xor_{r\subseteq n}[r\in A]$ and write
$\widetilde S_m=S_{m-1}\symd S_{m-2}\symd(S_{m-1}\conv S_{m-2})$.
Thus $S_m=\Inc\widetilde S_m$ and $\Zet(A)=E(A)$.

\keepstatement
\begin{theorem}[Zeta--Floor Factorization]\label{thm:zf}
Put $Y_m=1\xor\Zet(S_m)$ and
$\widetilde Y_m=1\xor\Zet(\widetilde S_m)$. Then
\[
 \widetilde Y_m(n)=Y_{m-1}(n)Y_{m-2}(n).
\]
\end{theorem}
\begin{proof}
The zeta transform preserves addition and takes OR--convolution to
pointwise multiplication. If $u=\Zet(S_{m-1})$ and
$v=\Zet(S_{m-2})$, then
$1\xor\Zet(\widetilde S_m)=1\xor u\xor v\xor uv=(1\xor u)(1\xor v)$.
\end{proof}

The quadratic source has become a pointwise product. Increment is
linear but is not pointwise in these coordinates; it becomes a prefix
XOR. This distinction explains the algorithm without attributing
nonlinearity to the carry itself.

\par
\begin{proposition}[Failure of Multiplicativity]\label{prop:comm}
With $A=\{1\}$, $B=\{2\}$:
\[
\begin{aligned}
 \Inc(A\conv B)&=\Inc(\{1\vee2\})=\{4\},\\
 \Inc(A)\conv\Inc(B)&=\{2\}\conv\{3\}=\{3\}.
\end{aligned}
\]
Hence $\Inc$ is not an endomorphism of $(\mathscr P_{\mathrm{fin}}(\NN),\symd,\conv)$ and
the multiplicativity identity fails.
\end{proposition}
\begin{proof}
The two displayed evaluations give the distinct singletons $\{4\}$ and $\{3\}$.
Thus $\Inc(A\conv B)\ne\Inc(A)\conv\Inc(B)$ for this pair, which is enough to disprove
the homomorphism identity. This is not a computational lower bound; in the zeta frame the carry is a linear-time prefix XOR.
\end{proof}

\keepstatement
\begin{proposition}[Triangular Carry Matrices]\label{prop:carrymat}
On a window of length $N$, with increment truncated at its last
coordinate,
\[
 \Zet\Inc\Zet^{-1}=L,\quad L[i,j]=[i>j],
 \qquad \Zet(I+\Inc)\Zet^{-1}=I+L.
\]
Here $L^N=0$ and $(I+L)^{-1}=I+J$, where $J$ is the one-step
lower shift. Both triangular maps can be applied in $O(N)$ XOR
operations and $O(1)$ auxiliary bits in an in-place scan.
\end{proposition}
\begin{proof}
Let $y=E(A)$ and $z=E(\Inc A)$. Pascal's identity gives
$z(0)=0$ and $z(t+1)\xor z(t)=y(t)$, hence
$z(t)=\Xor_{s<t}y(s)$. This proves the strict triangle. Adding
$y(t)$ gives the inclusive triangle, whose inverse takes adjacent
differences and has matrix $I+J$. Strict triangularity proves
nilpotence. A running XOR applies either triangle in one pass.
\end{proof}

\begin{remark}[Shift and Cumulative Sum]\label{rem:whichcarry}
The support recurrence uses the nilpotent $L$. The inclusive matrix
$I+L=B_{220}$ is unipotent and invertible. For $N\ge2$, the nonzero
nilpotent $L$ cannot be diagonalized: a diagonal nilpotent matrix would
be zero. This is a specific obstruction to diagonalizing the truncated
shift, not a lower bound on applying it. Failure of multiplicativity
in Proposition~\ref{prop:comm} is likewise not a commutator of two
linear operators and supplies no general simultaneous-diagonalization
theorem.
\end{remark}

Figure~\ref{fig:comm} retains the explicit counterexample and groups
the full finite block census by the size of the disagreement.
\begin{figure}[htbp]
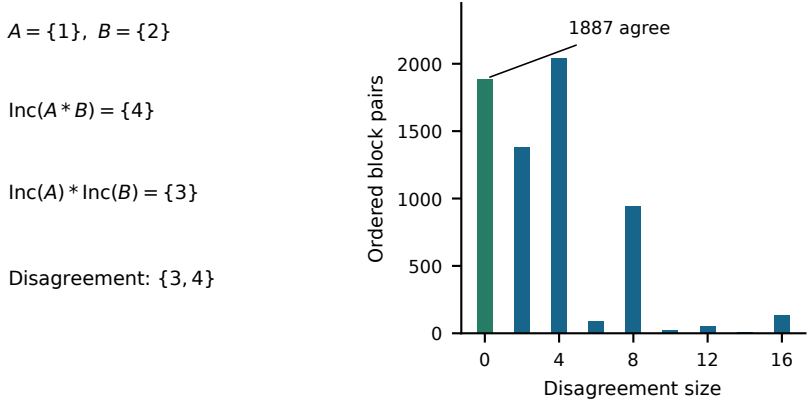
\centering
\EmbeddedFigure{\textwidth}{3}
\caption{Failure of multiplicativity of increment.
The witness $A=\{1\}$, $B=\{2\}$ gives the outputs $\{4\}$ and
$\{3\}$. The histogram counts all $81^2=6561$ ordered pairs of
valid four-bit blocks by the size of
$\Inc(A*B)\symd(\Inc A*\Inc B)$; $1887$ pairs agree and $4674$
disagree. The count describes this finite family, not the cost of a
prefix XOR or a commutator of linear operators.}
\label{fig:comm}
\end{figure}

\begin{remark}[Dyadic Probe Rows]\label{rem:probe}
If $S_m\subseteq[0,2^k)$, then
$b_{30}(m,2^k-1)=|S_m|\bmod2$. In addition,
$|S_m|=|\widetilde S_m|$ because increment is a bijection.
The factorization therefore computes this parity from the two
predecessor transforms. When a support extends beyond the window,
the row gives only the parity of its intersection with that window.
The floor alone does not remove this restriction, and no center-column
formula follows by dropping the unaccounted indices.
\end{remark}

\par
\subsection{The Center Column and Finite Automata}
Wolfram's prize questions concern eventual periodicity, frequency and
evaluation cost of $c(t)=A_{30}(t,0)$ \cite{prize}. The representations
above do not answer those questions. They do permit finite comparisons
with automatic sequences, provided the observation being counted is
specified.

For a prefix $u(0),\ldots,u(L-1)$ and a depth $s$, take the words
\[
 \bigl(u(2^a j+r)\bigr)_{0\le j<q},\qquad
 0\le a\le s,\quad0\le r<2^a,\quad q=\lfloor L/2^s\rfloor.
\]
Their number of distinct values is at most $2^{s+1}-1$. Comparing
all words at this same length avoids counting distinctions caused
only by different truncation lengths. Figure~\ref{fig:observedkernel}
compares the resulting counts for the center column and fixed diagonals.

\begin{figure}[htbp]
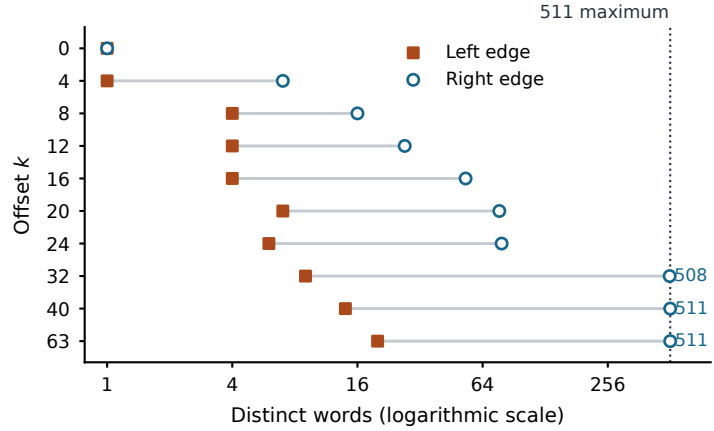
\centering
\EmbeddedFigure{\textwidth}{4}
\caption{Finite decimation counts under a common-length
comparison. In (a), each prefix has length $L=2^{15}$ and every word
at depth at most $s$ is truncated to $q=\lfloor L/2^s\rfloor$ terms.
The highlighted rows attain the finite maximum $2^{s+1}-1$:
the Rule 30 center prefix and the seeded Bernoulli sample have
identical counts on these windows, although their words differ.
In (b), each horizontal segment joins the counts for
$(x_t(k-t))_{k\le t<8192}$ and $(x_t(t-k))_{k\le t<8192}$,
on the left and right sides of the same orbit, at depth eight.
Here $q=32$ for $k=0$ and $q=31$ for the other displayed offsets.
The logarithmic count axis makes both sides visible; the dotted
line is the finite maximum $511$. These finite counts do not
determine the size of the infinite kernel.}
\label{fig:observedkernel}
\end{figure}

The Rule 30 center-column prefix of length $131072$ gives a stronger
finite certificate: for all $0\le a\le12$, $0\le r<2^a$, the
$8191$ words $(c(2^a j+r))_{1\le j\le31}$ are pairwise distinct.
Consequently any deterministic finite automaton with output that
generates the full sequence while reading binary digits from least
to most significant needs at least $8191$ states. To see this,
read the $a$ low digits of $r$; every positive $j$ is then a valid
continuation. Two equal states would give equal words for all these
continuations. The certificate forces distinct states. It does not
prove that the full sequence has an infinite automaton kernel, nor
does similarity to one Bernoulli sample establish randomness.

\keepsection
\section{Rule 86 and Periodic Tails}\label{sec:r86}

The finite-support argument for Rule 30 used the first nonzero term
of each diagonal. On the reflected side, the useful structure lies
in the eventual tail. We therefore change the object being followed:
each diagonal is now a binary recurrence driven by its neighbors,
and its tail is represented by a periodic profile. This makes the
role of a reset visible and permits a separate analysis of period
growth and transient behavior.

\keepstatement
\begin{proposition}[Two Diagonal Cuts]\label{prop:cuts}
Rule 86 has local map $g(p,q,r)=p\xor q\xor r\xor pq$.
Reflection $i\mapsto-i$ conjugates it to Rule 30. Its natural
diagonal has source $q\xor r\xor pq$, which depends on the current
diagonal; its reflected diagonal has exactly the finite Rule 30
supports of Section~\ref{sec:r30}.
\end{proposition}
\begin{proof}
Reflecting interchanges the left and right arguments $p,r$, changing
the nonlinear term $pq$ into $qr$. The seed is reflection invariant,
so the identity holds for every generation and every cell. Applying
the Newton transform to equal diagonal sequences gives equal supports.
\end{proof}

In this section $b(m,n)=b_{86}(m,n)$ on the natural cut.
On the finite cut, a singleton decomposition gives
$K(S_m)\le|S_m|\le\max S_m+1\le P_m$. This is an available
decomposition bound, not an upper bound for every redundant supplied
list. On the natural cut, we retain the initial transient explicitly.

\keepsection
\subsection{Reset and Integration}
Set $D_m(t)=A_{86}(t,t-m)=b_{86}(m+1,t)$ for $m\ge0$, and
$D_{-2}=D_{-1}=0$. In particular $D_0=1$ and $D_m(0)=0$ for $m>0$.

\keepstatement
\begin{proposition}[Diagonal Recurrence on the Regular Side]\label{prop:regrec}
Write $D_m(t)=A_{86}(t,t-m)$ for the diagonals of Rule~86 from $\delta_0$. Then
\begin{equation}\label{eq:regrec}
 D_m(t+1)\;=\;D_{m-2}(t)\;\xor\;\big(D_{m-1}(t)\vee D_m(t)\big).
\end{equation}
\end{proposition}

\begin{proof}
Rule~86 is the reflection of Rule~30, so its local map is
$A_{86}(t+1,x)=A_{86}(t,x+1)\xor\big(A_{86}(t,x)\vee A_{86}(t,x-1)\big)$. Substituting $x=(t+1)-m$ gives
$A_{86}(t,t+2-m)=D_{m-2}(t)$, $A_{86}(t,t+1-m)=D_{m-1}(t)$ and $A_{86}(t,t-m)=D_m(t)$.
\end{proof}

For $u=D_{m-2}$, $v=D_{m-1}$ and $w=D_m$, the recurrence becomes
\[
 w(t+1)=u(t)\xor(v(t)\vee w(t)).
\]
When $v(t)=1$, the output is $1\xor u(t)$ and forgets $w(t)$.
When $v(t)=0$, it is the cumulative XOR of $u$. These two cases
determine the period behavior.

\keepstatement
\begin{theorem}[An Active Driver Forbids Growth]\label{thm:driver}
Let $u,v$ be ultimately periodic with common period $P$ and transient at most $\tau$, and let
$w$ satisfy $w(t+1)=u(t)\xor(v(t)\vee w(t))$. If $v(t_0)=1$ for some $t_0\ge\tau$, then $w$ is
ultimately periodic with period dividing $P$.
\end{theorem}

\begin{proof}
Since $v$ is $P$-periodic beyond $\tau$, $v(t_0+kP)=1$ for every $k\ge0$, and at those points
the reset gives $w(t_0+kP+1)=u(t_0+kP)\xor1=u(t_0)\xor1$, a value independent of $k$; in
particular $w(t_0+P+1)=w(t_0+1)$. Now induct on $t\ge t_0+1$: if $w(t+P)=w(t)$ then
\[
\begin{aligned}
w(t+P+1)&=u(t+P)\xor\big(v(t+P)\vee w(t+P)\big)\\
&=u(t)\xor\big(v(t)\vee w(t)\big)=w(t+1).
\end{aligned}
\]
Hence $P$ is a period of $w$ from $t_0+1$ on.
\end{proof}

\keepstatement
\begin{theorem}[Inactive Driver]\label{thm:parity}
Suppose $u$ is periodic from $\tau$ onward with least ultimate period
$\pi$, and $v(t)=0$ for $t\ge\tau$. Put
$\sigma=\Xor_{s=\tau}^{\tau+\pi-1}u(s)$. Then the least ultimate
period of $w$ is $\pi$ if $\sigma=0$, and $2\pi$ if $\sigma=1$.
\end{theorem}
\begin{proof}
Summing the recurrence over one period gives
$w(t+\pi)=w(t)\xor\sigma$ for $t\ge\tau$. Thus $\pi$ is a
period in the first case and $2\pi$ is a period in either case.
If $\rho$ is the least ultimate period of $w$, then
$u(t)=w(t+1)\xor w(t)$ makes $\rho$ a period of $u$. Hence
$\pi\mid\rho$ and $\rho\mid2\pi$. In the odd case $\pi$ is
excluded by the displayed identity; in the even case it is a period.
\end{proof}

\keepstatement
\begin{theorem}[Construction of Every Periodic Tail]\label{thm:tails}
Every $D_m$ is ultimately periodic with a power-of-two least period
$\pi_m$. A finite initial segment and one repeating word can be
constructed from those of its two predecessors.
\end{theorem}
\begin{proof}
The assertion holds for the two zero boundaries and $D_0=1$.
Suppose the predecessors are periodic beyond a common $\tau$, with
common period $P$, a power of two. If the driver has a one in its
periodic word, Theorem~\ref{thm:driver} gives a period dividing $P$.
Otherwise Theorem~\ref{thm:parity} gives a period dividing $2P$.
Induction proves the assertion for every depth.

For the construction, first evaluate $w$ through time $\tau$. Thereafter
the state $(t\bmod P,w(t))$ has $2P$ possibilities and evolves
deterministically. Record states until the first repetition. The
intervening word repeats forever by induction on the update. Reduce
that word to its least period, then extend it backward while it agrees
with the computed prefix. This gives the least transient and least
ultimate period, rather than a guess from repeated samples.
\end{proof}

\keepstatement
\begin{algo}[Tail Construction]\label{alg:tail}
Store each predecessor by its transient and one periodic word. Apply
the construction in Theorem~\ref{thm:tails}, starting with $w(0)=0$.
The state search uses $O(\tau+P)$ indexed updates and
$O((\tau+P)\log(\tau+P+2))$ stored bits when visited states carry
binary time indices. A direct divisor test for the repeating
word costs at most $O(P^2)$ indexed comparisons; a prefix-function
period test reduces this to $O(P)$. Binary address operations add
at most a factor $O(\log(\tau+P+2))$ in a bit model. These costs
depend on the actual predecessor words, not only on the depth $m$.
\end{algo}

\keepstatement
\begin{lemma}[Two Profiles Determine All Earlier Periods]\label{lem:pairperiod}
For $m\ge1$, the running largest period through $m$ is
\[
 Q_m:=\max_{0\le j\le m}\pi_j=\max(\pi_{m-1},\pi_m).
\]
\end{lemma}
\begin{proof}
Solving the diagonal recurrence backward gives
$D_{j-2}(t)=D_j(t+1)\xor(D_{j-1}(t)\vee D_j(t))$.
If $P$ is a common eventual period of $D_{j-1},D_j$, it is also
an eventual period of $D_{j-2}$. Iterate backward. Since all least
periods are powers of two, their least common multiple is their
maximum. The opposite inequality is part of the definition of $Q_m$.
\end{proof}

\keepstatement
\begin{corollary}[Active Driver and the Running Period]\label{cor:nogrowth}
If $D_{m-1}$ is not ultimately zero, then
$\pi_m\mid\operatorname{lcm}(\pi_{m-1},\pi_{m-2})$ and
$Q_m=Q_{m-1}$.
\end{corollary}
\begin{proof}
Theorem~\ref{thm:driver} gives divisibility, and
Lemma~\ref{lem:pairperiod} identifies the preceding common period
with the running maximum.
\end{proof}

\keepstatement
\begin{corollary}[Record-Doubling Criterion]\label{cor:criterion}
For $m\ge1$, $Q_m>Q_{m-1}$ if and only if $D_{m-1}$ is ultimately
zero and the parity of one least period of $D_{m-2}$ is odd. In
that case $\pi_m=2\pi_{m-2}=2Q_{m-1}$.
\end{corollary}
\begin{proof}
An active driver prevents a new record. For an ultimately zero driver,
Lemma~\ref{lem:pairperiod} gives $Q_{m-1}=\pi_{m-2}$, since the
zero sequence has least period one. Theorem~\ref{thm:parity} then
gives exactly the stated alternatives.
\end{proof}

The same argument controls the additional transient introduced by
one step. If the predecessors have a common tail starting at $\tau$
with period $P$ and the driver is active, its first reset occurs at
some $t_0\in[\tau,\tau+P)$. The output is then periodic from
$t_0+1\le\tau+P$. If the driver is zero throughout the tail,
integration is periodic from $\tau$ itself, with period dividing
$2P$. Thus each step has an explicit transient bound from its input
words, even though no uniform bound in depth is used below.

This recovers the period-doubling criterion proved by Rowland from
Wolfram's observation \cite{rowland}. A white stripe is necessary;
the parity condition is also necessary. No conclusion about the
spacing of record depths follows from this criterion alone.

\keepsection
\subsection{The Zero Boundary Forces Unbounded Periods}\label{sec:resist}

The reset criterion explains how a period changes at one step, but
does not by itself force infinitely many changes. For that purpose
we use the boundary condition. If all periods were bounded, only
finitely many pairs of profiles could occur. The backward map would
then have to reconcile this finite collection with an arbitrarily
long chain ending at the zero boundary pair. The proof turns that
constraint into a bound depending explicitly on the allowed period.

\keepstatement
\begin{proposition}[Lower Bound on the Transient]\label{prop:transient}
$D_m(t)=0$ for $t<m/2$. Consequently, if the ultimate periods are bounded by $P$ and the
ultimate regime is not identically zero, the transient satisfies $\tau_m\ge m/2-P$.
\end{proposition}

\begin{proof}
The light cone gives $a(t,x)=0$ for $|x|>t$, and $|t-m|>t$ whenever $t<m/2$; hence
$D_m(t)=0$ for every $t<m/2$. Suppose the ultimate periods are bounded by $P$ and the
ultimate regime of $D_m$ is not identically zero, and suppose for contradiction that
$\tau_m<m/2-P$. Beyond $\tau_m$ the sequence is periodic with some period $\pi\le P$, and
a non-zero periodic sequence takes the value $1$ at least once in every window of length
$\pi$. The interval $[\tau_m,\,m/2)$ lies beyond the transient and has length greater than
$P\ge\pi$, so it contains such a window and therefore a $1$. But $D_m$ vanishes identically
on $[0,m/2)$. Hence $\tau_m\ge m/2-P$.
\end{proof}

The lower bound applies to diagonals with a nonzero periodic tail.
Their transients need not be uniformly bounded. The following
argument avoids that issue by comparing absolute time phases on the
periodic tails only.

\keepstatement
\begin{theorem}[Backward Profile Bound]\label{thm:profile}
For every $m\ge0$,
\begin{equation}\label{eq:profilebound}
 m+2\le4^{Q_m},\qquad Q_m\ge\tfrac12\log_2(m+2).
\end{equation}
Consequently the least ultimate periods of the natural Rule 86
diagonals are unbounded, and the record-doubling criterion occurs
infinitely often.
\end{theorem}
\begin{proof}
Fix $m$ and put $P=Q_m$. Extend each periodic tail $D_j$,
$-2\le j\le m$, to its $P$-periodic profile $p_j$ on $\ZZ/P$,
using absolute time residues. For any fixed triple of adjacent
diagonals the recurrence holds beyond a common transient. Taking a
sufficiently large time in each residue class therefore gives
\[
 p_{j-2}(t)=p_j(t+1)\xor(p_{j-1}(t)\vee p_j(t)).
\]
On the $4^P$ pairs of binary profiles define the total map
\[
 B(u,v)=\bigl(t\mapsto v(t+1)\xor(u(t)\vee v(t)),\ u\bigr).
\]
For $\Sigma_j=(p_{j-1},p_j)$ we have $B\Sigma_j=\Sigma_{j-1}$.
The boundary supplies $\Sigma_{-1}=(0,0)$, a fixed point of $B$,
and $\Sigma_0=(0,1)$, which is different from it.

We claim that $\Sigma_{-1},\Sigma_0,\ldots,\Sigma_m$ are distinct.
Suppose $\Sigma_a=\Sigma_b$ with $-1\le a<b\le m$. If $a=-1$,
applying $B^b$ makes $\Sigma_0=(0,0)$, a contradiction. If $a\ge0$,
applying $B^{a+1}$ gives
$\Sigma_{b-a-1}=\Sigma_{-1}$. Since $b-a-1\ge0$, applying
$B^{b-a-1}$ gives the same contradiction. There are thus $m+2$
distinct pairs in a set of size $4^P$, proving the bound. As
$m\to\infty$, $Q_m$ is unbounded; each strict increase is a
doubling by Corollary~\ref{cor:criterion}.
\end{proof}

The equivalent bounded-period argument is useful in other settings.
An infinite sequence of backward predecessors in a finite set lies
in the eventual image $X_\infty=\bigcap_{j\ge0}B^j(X)$. On that
set $B$ is a permutation, so the sequence would be purely periodic
in depth. Here this alternative is excluded by the two boundary
profiles: a nonzero pair maps to the zero fixed point after finitely
many steps and cannot lie on a cycle. The zero boundary pair is
essential; retaining only the constant profile $D_0=1$ misses the
contradiction.

\begin{remark}[Record Indices]\label{rem:trigopen}
The first record depths in the $D_m$ convention are
$0,3,8,29,400$. In the $b_{86}(m,n)$ convention they are
$1,4,9,30,401$, agreeing with the first five entries of A364239.
The next listed entry is $87868$ \cite{oeis}; our construction was
run through $D_{1024}$ and does not independently reproduce that
entry. The infinitude of actual record depths follows from
Theorem~\ref{thm:profile}. A formula for their spacing remains open.
\end{remark}

\keepsection
\subsection{Finite Descriptions of the Supports}

\keepstatement
\begin{theorem}[Tail and Transient Support]\label{thm:resolved}
Suppose a binary sequence $d$ is periodic from $\tau$ onward with
period $P=2^k$. Let $p$ be its $P$-periodic extension to all
nonnegative times, let $F\subseteq[0,P)$ be the Newton support of
$p$, and put $J=\{t<\tau:d(t)\ne p(t)\}$. Then the full support is
\begin{equation}\label{eq:tailcertificate}
 S=F\symd\mathop{\symd}_{t\in J}\Up t.
\end{equation}
Every natural Rule 86 diagonal consequently has a finite
finite-mask/cofinite-mask block representation, and a finite
residue-class representation with finite correction.
\end{theorem}
\begin{proof}
Invert the first $P$ values of $p$ by the Pascal transform. The
result $F$ lies in $[0,P)$, so $E(F)$ is $P$-periodic and agrees
with $p$ everywhere. Equation~\eqref{eq:transientatoms} adds
exactly the discrepancies in $J$, proving the support formula by
uniqueness. For $2^j>t$,
$\Up t=\bigcup_{v<2^j,\ t\subseteq v}R(v,2^j)$, a finite
disjoint union of residue classes. Theorem~\ref{thm:tails} applies
this construction at each depth. Computing $F$ costs
$O(P\log(P+1))$ XOR operations and $O(P)$ bits; finding $J$
costs $O(\tau)$ comparisons, with indices stored in binary.
\end{proof}

\keepstatement
\begin{theorem}[Evaluation from a Residue Representation]\label{thm:47}
Suppose a diagonal support has a certified finite decomposition
$S_m=\symd_iR(v_i,2^{j_i})\symd F$, with $F$ finite. Put
$P_{\max}=\max(\{0\}\cup\{2^{j_i}:i\})$ and
$E_m=\max F$ when $F\ne\varnothing$, with $E_m=-1$ otherwise. Then:
\begin{enumerate}
\item \emph{Recurrence:} an infinite support has no finite Newton expansion.
For $n\ge P_{\max}$ the residue contribution vanishes, and the remaining
tail is $\Xor_{x\in F}\binom nx\bmod2$. If $F=\varnothing$ this tail is zero;
otherwise it is annihilated by $(E-1)^{E_m+1}$, where $E$ shifts $n$.
Its least ultimate period is $1$ for $E_m\le0$, and
$2^{\lceil\log_2(E_m+1)\rceil}$ for $E_m\ge1$.
\item \emph{Residue / block:}
$b(m,n)=\Xor_i[v_i\subseteq n][n<2^{j_i}]
\xor\Xor_{x\in F}\binom nx\pmod2$.
\item \emph{Evaluation:} each supplied residue or singleton term costs
$O(b+\ell)$ bit operations. For $K(m)\ge1$ terms the total is
$O(K(m)(b+\ell))$; an empty representation returns zero in constant time.
\end{enumerate}
\end{theorem}
\begin{proof}
Residue evaluation gives $[v_i\subseteq n][n<2^{j_i}]$, so every residue
term is zero past $P_{\max}$. The finite-part period and its minimality
follow from Theorem~\ref{thm:exactper} and the tail argument of
Theorem~\ref{thm:split}, including the empty and constant cases. A zero
tail need not mean a zero sequence before the transient: for example
$S_m=R(0,1)=\NN$ gives $b(m,0)=1$ and $b(m,n)=0$ for $n\ge1$.
The evaluation formula is linearity of the Newton transform; its cost follows
by scanning each term once.
\end{proof}

The certificate in Theorem~\ref{thm:resolved} proves the existence
required by this residue formula. In its unexpanded form it uses
at most $|F|+|J|\le P+\tau$ blocks: one singleton for each finite
coefficient and one cofinite block $\Up t$ for each discrepancy.
It is therefore a constructive bound on description size in terms
of the actual tail and transient. It does not imply a short
representation uniformly in depth. Expanding every cofinite block
into residue classes can greatly increase the term count.

\begin{example}[The First Twenty Natural Supports]\label{ex:res}
Table~\ref{tab:regular} gives $F,J$ in
Equation~\eqref{eq:tailcertificate} for $S_m^{86}$, using
$d=b_{86}(m,\cdot)=D_{m-1}$. Each row specifies the entire support,
including rows with the same tail but different transients.
\end{example}
\begin{table}[t]\centering\fontsize{9}{10.5}\selectfont\setlength{\tabcolsep}{3pt}
\caption{Exact tail and transient descriptions of $S_m^{86}=F\symd\bigtriangleup_{t\in J}U(t)$.}
\label{tab:regular}
\begin{tabular}{|rrrlp{155pt}|}\hline
$m$ & $\tau$ & $\pi$ & $F$ & $J$\\\hline
1 & 0 & 1 & $\{0\}$ & $\varnothing$\\
2 & 1 & 1 & $\{0\}$ & $\{0\}$\\
3 & 2 & 1 & $\varnothing$ & $\{1\}$\\
4 & 2 & 2 & $\{1\}$ & $\{1\}$\\
5 & 2 & 1 & $\{0\}$ & $\{0,1\}$\\
6 & 2 & 2 & $\{1\}$ & $\{1\}$\\
7 & 2 & 2 & $\{1\}$ & $\{1\}$\\
8 & 0 & 1 & $\varnothing$ & $\varnothing$\\
9 & 2 & 4 & $\{0,2\}$ & $\{0,1\}$\\
10 & 5 & 1 & $\{0\}$ & $\{0,1,2,3,4\}$\\
11 & 6 & 4 & $\{0,1,2\}$ & $\{0,3,4,5\}$\\
12 & 8 & 4 & $\{2,3\}$ & $\{2,6,7\}$\\
13 & 6 & 4 & $\{1,2,3\}$ & $\{1,2,3,5\}$\\
14 & 10 & 4 & $\{0,3\}$ & $\{0,1,2,4,5,6,7,8,9\}$\\
15 & 11 & 4 & $\{0,2\}$ & $\{0,1,4,5,7,8,9,10\}$\\
16 & 13 & 4 & $\{0,1,2\}$ & $\{0,3,4,7,8,11,12\}$\\
17 & 16 & 4 & $\{0,1,3\}$ & $\{0,2,3,4,6,7,9,11,12,13,14,15\}$\\
18 & 15 & 4 & $\{3\}$ & $\{3,7,9,10,12,13,14\}$\\
19 & 20 & 4 & $\{1,2\}$ & $\{1,2,5,6,10,12,14,15,16,17,18,19\}$\\
20 & 18 & 4 & $\{2\}$ & $\{2,3,6,7,10,14,17\}$\\
\hline\end{tabular}
\end{table}

In residue notation the initial examples become
\[
\begin{aligned}
S_2^{86}&=R(0,1)\symd\{0\},&S_3^{86}&=R(1,2),\\
S_4^{86}&=R(1,2)\symd\{1\},&S_8^{86}&=\varnothing,\\
S_9^{86}&=R(0,2)\symd\{0,2\}.&&
\end{aligned}
\]
For $m=10$,
\[
 S_{10}^{86}=R(0,8)\symd R(5,8)\symd R(6,8)
              \symd R(7,8)\symd\{0\}.
\]
Its evaluation is
$[n<8](1\xor[5\subseteq n]\xor[6\subseteq n]\xor[7\subseteq n])\xor1$.
In fact it is zero for $n<5$ and one thereafter. The same diagonal
has the simpler tail description $F=\{0\}$ and
$J=\{0,1,2,3,4\}$ from Table~\ref{tab:regular}: its constant tail
is corrected at exactly the first five times. These two descriptions
show explicitly how an infinite support can encode a short transient
followed by a constant sequence. Also
$b_{86}(3,n)=[n=1]$, while the eighth diagonal is identically zero.
For a dyadic window $[0,2^a)$ with $2^a>t$, the exact truncation
of $\Up t$ is $B(t,(2^a-1)\mathbin{\xor}t)$; the truncation of
$R(v,2^j)$ is $B(v,2^a-2^j)$ for $a\ge j$.

\keepsection
\subsection{The Milestone Shift Model}\label{sec:shift}
\label{rem:nocontinuum}\label{rem:correction}
The companion theory also describes a formal polynomial built from
the record milestones. Its algebraic identity should be separated
from an annihilator of the physical diagonal. For instance
$b_{86}(3,n)=[n=1]$ has infinite Newton support; no finite power of
$E-1$ annihilates that full sequence. An eventual period does give
a recurrence on its tail. It does not give a bounded-degree
polynomial lift of the initial transient or justify differentiating
such a lift as in the conditional transport model \cite{U6}.

\begin{definition}[Closest-Down Map and Milestone Polynomial]\label{def:cdv}
Let $a_0=1<a_1<a_2<\cdots$ be the actual record indices in the
$b_{86}$ convention. Theorem~\ref{thm:profile} makes this sequence
infinite. For $n\ge1$, let $\mathrm{CDV}(n)=a_j$ be its greatest
member not exceeding $n$, put $\mathrm{MF}(a_j)=2^j-1$ and
$\delta(n)=\mathrm{MF}(\mathrm{CDV}(n))+1$. Define
\[
 T_{lm}(n,k)=(-1)^{k+1}\binom nk+(-1)^k\binom n{k-\delta(n)},
 \quad 1\le k\le n+\delta(n),
\]
where a binomial coefficient outside $[0,n]$ is zero.
\end{definition}

\keepstatement
\begin{theorem}[Pure-Power Form of the Model]\label{thm:pp}
For every $n\ge1$, $\delta(n)$ is a power of two and
\[
 T_{lm}(n,k)\equiv\binom{n+\delta(n)}k\pmod2.
\]
The coefficient polynomial, with its constant term included, is
\[
 1-\sum_{k=1}^{n+\delta(n)}T_{lm}(n,k)z^k
   =(1-z)^n(1-z^{\delta(n)})
   \equiv(1-z)^{n+\delta(n)}\pmod2.
\]
\end{theorem}
\begin{proof}
The definition gives $\delta(n)=2^j$. The integer polynomial
identity follows by the binomial theorem and shifting the second
sum. In characteristic two,
$(1+z)^{2^j}=1+z^{2^j}$, without any condition on overlaps in the
binary digits of $n$. Multiplication by $(1+z)^n$ proves the
coefficient identity. The argument defines a pure-power companion
for this model; it does not identify its initial data with a
physical Rule 86 diagonal.
\end{proof}

The degree is piecewise affine: starting at $n=4$ it is
$6,7,8,9,10,13,14,15,16,17,18,19,\ldots$. It jumps when a record
changes the value of $\delta$. Modulo two, its coefficients are
constructed with $n+\delta+1$ Lucas tests, using
$O((n+\delta+1)\log(n+\delta+2))$ bit operations.
Figure~\ref{fig:r86deg} links each shift plateau to its degree
and displays the coefficient change at the record $n=9$.

\begin{figure}[htbp]
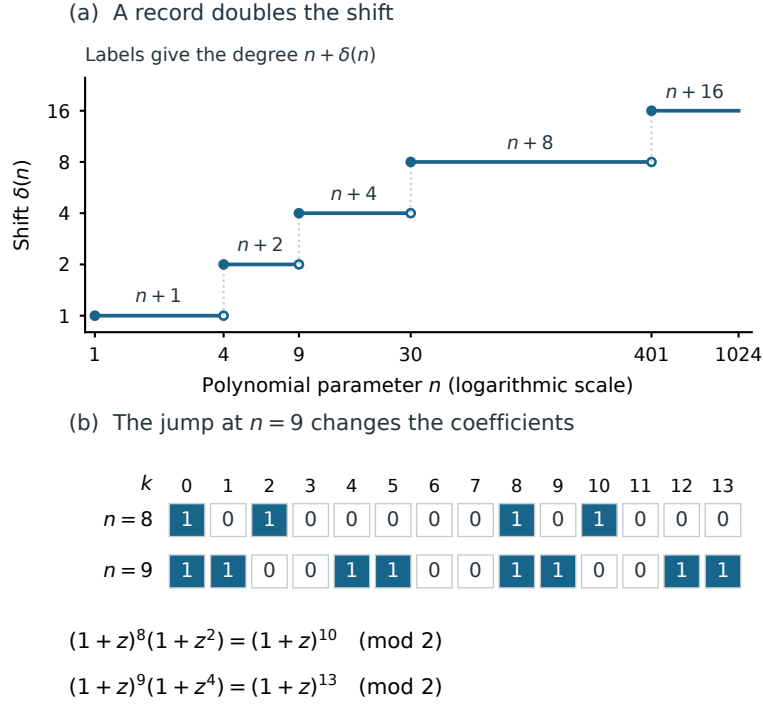
\centering
\EmbeddedFigure{\textwidth}{5}
\caption{The milestone polynomial of
Definition~\ref{def:cdv}. In (a), the shift doubles at record
indices $4,9,30,401$ on the displayed range $1\le n\le1024$.
Filled endpoints belong to the new plateau; open endpoints do not
belong to the old one. The labels give the corresponding degree
$n+\delta(n)$. In (b), the change from $n=8$ to $n=9$ replaces
$\delta=2$ by $\delta=4$ and the degree $10$ by $13$.
The numbered cells give every coefficient for $0\le k\le13$,
including the constant term, in the two products displayed below.
Theorem~\ref{thm:pp} identifies them with the indicated Pascal rows
modulo two. This is the defined polynomial model, not a measurement
of minimal recurrence orders for physical diagonals.}
\label{fig:r86deg}
\end{figure}

\begin{remark}[Arithmetic Period Lifting]\label{rem:wall}
There is a related, but distinct, linear term in Fibonacci period
lifting. Let $Q=\left(\begin{smallmatrix}1&1\\1&0\end{smallmatrix}\right)$,
let $p$ be prime and let $h$ be the order of $Q$ modulo $p$, the
Pisano period. Write $Q^h=I+pA\pmod{p^2}$. Then
\[
 (Q^h)^s=I+spA\pmod{p^2}.
\]
Thus the period modulo $p^2$ is $h$ if $A=0\pmod p$, and $ph$
otherwise: any such period is a multiple of $h$, and the displayed
identity tests the multiplier. The off-diagonal entry of $pA$ is
$F_h\pmod{p^2}$. For odd $p$, $h$ is even and $\det Q^h=1$;
if $F_h=0\pmod{p^2}$, the two diagonal entries coincide with a
number $a\equiv1\pmod p$ satisfying $a^2=1\pmod{p^2}$, hence
$a=1\pmod{p^2}$. Therefore $p^2\mid F_h$ is equivalent to
failure of this period lift. For $p\ne2,5$ this is the exceptional
condition studied in connection with Wall--Sun--Sun primes
\cite{wall,jahnel}. The diagonal parity $\sigma$ and the matrix
$A$ both measure a first-order change over one period, but they
belong to different recurrences. No coefficient independent of the
initial Fibonacci index $r$ is asserted for every progression
$F_{r+hs}$.
\end{remark}

\keepsection
\section{The Physical Rule 135 Orbit}\label{sec:physical}

A relation between local rules must be applied together with the
initial configuration. Complementation changes the background, and
a spatial shift changes the diagonal labels. For Rule 135 the first
update is the point at which these data match the stated Rule 30
relation. We derive the orbit identity with that indexing in place
before translating it into supports. This fixes the physical object
to which the subsequent formulas apply.

Rule 135 has local map $g_{135}(a,b,c)=1\xor a\xor bc$, the
complement of Rule 120. The all-zero configuration becomes all ones,
so imposing a permanent zero exterior changes the problem. The
single-seed infinite-lattice orbit has the following exact relation
to Rule 30.

\keepstatement
\begin{theorem}[Single-Seed Conjugacy]\label{thm:seedconjugacy}
For every $t\ge1$ and $i\in\ZZ$,
\begin{equation}\label{eq:physical-orbit}
 A_{135}(t,i)=1\xor A_{30}(t-1,i-1).
\end{equation}
Hence $b_{135}(m,n)=1\xor b_{30}(m,n-1)$ for $n\ge1$, while
$b_{135}(m,0)=[m=1]$.
\end{theorem}
\begin{proof}
Put $g_{135}(a,b,c)=1\xor a\xor bc$. Direct expansion gives
$g_{135}(1\xor a,1\xor b,1\xor c)=1\xor g_{30}(a,b,c)$.
At time one, the product of adjacent cells in the single seed is zero,
so $A_{135}(1,i)=1\xor[i=1]$. This is the translated complement of the
Rule 30 seed. Translation commutes with either local update. Induction using
the complement identity proves \eqref{eq:physical-orbit} at every time. Substituting
$i=n-m+1$ gives the diagonal formula, and time zero is the specified seed.
\end{proof}

In the next theorem $S_m$ is the finite Rule 30 support, and $V_m$
denotes a set of tail parities, not a physical space--time region.
\keepstatement
\begin{theorem}[Finite or Cofinite Support]\label{thm:physical135}
For a fixed $m$, let $R=\max(S_m\cup\{0\})$ and set
\[
 V_m=\{j:0\le j\le R,\ |S_m\cap[j,R]|\text{ is odd}\},\qquad
 \eta_m=|S_m|\bmod2.
\]
Put $\varepsilon_m=[m=1]\xor1\xor\eta_m$. The unique binomial support
of the physical Rule 135 diagonal is
\begin{equation}\label{eq:physical-support}
 S_m^{135}=\{0\}\symd V_m\symd(\varepsilon_m\NN),
\end{equation}
where $0\NN$ means the empty set and $1\NN$ means $\NN$.
\end{theorem}
\begin{proof}
As a polynomial identity over the rationals,
$\binom{n-1}r=\sum_{j=0}^r(-1)^{r-j}\binom nj$.
To prove it, write $\binom nr=\binom{n-1}r+\binom{n-1}{r-1}$ and induct on $r$.
Reduce the finite sum modulo two: its support coefficients are the tail parities
defining $V_m$. Thus $E(\{0\}\symd V_m)(n)=1\xor E(S_m)(n-1)$ for $n\ge1$.
At $n=0$ this expression has value $1\xor\eta_m$ because
$\binom{-1}r=(-1)^r$.
Finally $E(\NN)(0)=1$ and $E(\NN)(n)=2^n\bmod2=0$ for $n\ge1$.
The last term of \eqref{eq:physical-support} corrects exactly the initial value.
Unit triangularity on every finite window gives uniqueness for infinite supports too.
\end{proof}

The orbit identity also gives the least ultimate period of every
physical Rule 135 diagonal: it is exactly the Rule 30 period $P_m$.
Indeed a one-step translation in time and complementation both
preserve the periods of a tail, and the value at time zero is a
single finite correction. Thus unbounded diagonal periods transfer
to the physical Rule 135 orbit without using the auxiliary recurrence.

A backward parity scan constructs $V_m$ and $\varepsilon_m$ from
a support bitmap on $[0,R]$ in $O(R+1)$ Boolean operations and
output bits. For a single query $n\ge1$, it suffices to complement
the Rule 30 answer at $n-1$; time zero is the specified seed. The
first physical supports are $S_1^{135}=\NN$ and
$S_2^{135}=S_3^{135}=\{1\}$.

\keepsection
\section{An Auxiliary Recurrence}\label{sec:r135}

The physical conjugacy determines the orbit with its prescribed
background. Prescribing two finite diagonal seeds instead produces
a related set recurrence with a different advantage: each fixed
index has an exact time after which it remains present. We study
this auxiliary recurrence through that stabilization question.
Its solution separates a permanent interval from a moving residue,
giving the later density and compression questions a precise
object. The seeds are part of the definition throughout.

The source $1\xor bc$ also suggests a recurrence with two imposed
finite initial supports. Those initial supports do not describe the
physical Rule 135 boundary. We retain the resulting sequence as a
separate algebraic object and write it as $H_m$ throughout.

\keepstatement
\begin{proposition}[The Auxiliary Recurrence]\label{prop:spec}
Define the boundary-seeded sequence $H_1=\{0\}$, $H_2=\{1\}$ and, for $m\ge3$,
\begin{equation}\label{eq:135}
 H_m=\Inc\bigl(\{0\}\symd(H_{m-1}\conv H_{m-2})\bigr).
\end{equation}
This uniquely defines finite sets. These are auxiliary supports, not the
physical single-seed supports $S_m^{135}$.
\end{proposition}
\begin{proof}
The two initial finite sets are given. A finite Cartesian product yields finitely
many OR values, and symmetric difference and increment preserve finiteness.
Induction therefore proves existence, uniqueness and finiteness at every depth.
The physical supports are instead given by Theorem~\ref{thm:physical135}.
\end{proof}

\keepsection
\subsection{Initial Supports}
Iterating \eqref{eq:135} by hand gives
\[
\begin{aligned}
H_1&=\{0\},\quad H_2=\{1\},\quad H_3=\{1,2\},\quad H_4=\{1,2,4\},\\
H_5&=\{1,2,3,6,7\},\quad H_6=\{1,2,3,6,7,8\},\\
H_7&=\{1,2,3,4,7,8,10,11,12,15,16\},\\
H_8&=\{1,2,3,4,6,7,9,14,16,18,19,20,23,24,25\},\\
H_9&=\{1,2,3,4,5,10,14,17,22,23,24,25,26,27,28,29,31,32\}.
\end{aligned}
\]
\begin{table}[t]\centering\fontsize{9}{10.5}\selectfont\setlength{\tabcolsep}{3pt}
\caption{The auxiliary supports at depths 2 through 30. The minimum is one throughout; End is $\lfloor(m+1)/2\rfloor$.}
\label{tab:auxiliary}
\begin{tabular}{|rrrrrrrr|}\hline
$m$ & $|H_m|$ & $\max H_m$ & End & $m$ & $|H_m|$ & $\max H_m$ & End\\\hline
2 & 1 & 1 & 1 & 17 & 109 & 193 & 9\\
3 & 2 & 2 & 2 & 18 & 134 & 254 & 9\\
4 & 3 & 4 & 2 & 19 & 129 & 254 & 10\\
5 & 5 & 7 & 3 & 20 & 135 & 255 & 10\\
6 & 6 & 8 & 3 & 21 & 146 & 254 & 11\\
7 & 11 & 16 & 4 & 22 & 137 & 253 & 11\\
8 & 15 & 25 & 4 & 23 & 133 & 255 & 12\\
9 & 18 & 32 & 5 & 24 & 130 & 256 & 12\\
10 & 31 & 58 & 5 & 25 & 263 & 512 & 13\\
11 & 35 & 63 & 6 & 26 & 391 & 769 & 13\\
12 & 34 & 63 & 6 & 27 & 514 & 1024 & 14\\
13 & 41 & 63 & 7 & 28 & 919 & 1794 & 14\\
14 & 33 & 63 & 7 & 29 & 1097 & 2048 & 15\\
15 & 34 & 64 & 8 & 30 & 1936 & 3843 & 15\\
16 & 75 & 128 & 8 &  &  &  & \\
\hline\end{tabular}
\end{table}

Table~\ref{tab:auxiliary} also shows that sizes are not monotone: they decrease from $41$ to
$33$ between depths 13 and 14, and from $146$ to $137$ between
depths 21 and 22. The minimum is always one for $m\ge2$. By
contrast the Rule 30 minimum is $\min S_m=\lfloor m/2\rfloor$.
The constant source term in Equation~\eqref{eq:135} supplies the
index one and propagates the interval described next.

\keepsection
\subsection{The Closed Kernel and Exact Stabilization}
\keepstatement
\begin{theorem}[Closed Kernel]\label{thm:kernel}
For all $m\ge2$, $\;C_m:=\{1,2,\dots,\lfloor(m+1)/2\rfloor\}\subseteq H_m$.
\end{theorem}

\begin{proof}
Strong induction. $m=2$: $H_2=\{1\}$ and $\lfloor3/2\rfloor=1$. $m=3$: $H_3=\{1,2\}$ and
$\lfloor4/2\rfloor=2$. Let $m\ge4$ and let $0\le c\le\lfloor(m-1)/2\rfloor$; by
\eqref{eq:135} it suffices to show $c\in\{0\}\symd(H_{m-1}\conv H_{m-2})$, since $\Inc$ then
places $c+1$ in $H_m$ and $c+1$ ranges over $C_m$.

\emph{Case $c=0$.} In the OR--convolution, $r\vee s=0$ forces $r=s=0$. Now $0\notin H_{m-1}$
for $m\ge4$: by \eqref{eq:135} every element of $H_{m-1}$ is of the form $x+1$ with $x\ge0$,
so $H_{m-1}\subseteq\{1,2,\dots\}$, whether or not $H_{m-1}$ is empty. Hence no pair
contributes to $0$, so $0\notin H_{m-1}\conv H_{m-2}$ and the $0$ of the constant set
$\{0\}$ survives the symmetric difference. Thus $1\in H_m$. (Emptiness does not in fact
occur: the induction hypothesis gives $1\in H_{m-1}$ for $m\ge3$.)

\emph{Case $1\le c\le\lfloor(m-1)/2\rfloor$.} Any pair $(r,s)$ with $r\vee s=c$ satisfies
$r,s\subseteq c$ bitwise, hence $r,s\le c$. By the induction hypothesis
$H_{m-1}\supseteq\{1,\dots,\lfloor m/2\rfloor\}$ and
$H_{m-2}\supseteq\{1,\dots,\lfloor(m-1)/2\rfloor\}$, and since
$c\le\lfloor(m-1)/2\rfloor\le\lfloor m/2\rfloor$, \emph{every} nonzero $r\le c$ lies in
$H_{m-1}$ and every nonzero $s\le c$ lies in $H_{m-2}$. The contributing pairs are therefore
exactly
\[
 \{(r,s):\ r,s\subseteq c,\ \ r,s\ge1,\ \ r\vee s=c\}.
\]
Count them bitwise. For each bit set in $c$ the pair $(r_b,s_b)$ must lie in
$\{(0,1),(1,0),(1,1)\}$ --- three choices --- and for each bit not set in $c$ it must be
$(0,0)$. That gives $3^{\pc(c)}$ pairs with $r\vee s=c$, from which we remove the pair with
$r=0$ and the pair with $s=0$; these two are distinct because $c\neq0$ forbids $r=s=0$. Hence
\[
 \mathcal N(c)=3^{\pc(c)}-2,
\]
which is odd because $3^{\pc(c)}$ is odd. An odd number of contributing pairs means
$c\in H_{m-1}\conv H_{m-2}$, and since $c\ge1$ it is untouched by $\symd\{0\}$. Thus
$c+1\in H_m$.

The two cases together give $C_m\subseteq H_m$.
\end{proof}

The induction uses both predecessors only below the proposed next
index. Its interval is sufficient for an explicit upper bound on the
time after which each coefficient stays one.

\keepstatement
\begin{corollary}[Stabilization Bound]\label{cor:stab}
With $M_0(e):=\min\{m\ge2:\ \forall m'\ge m,\ e\in H_{m'}\}$, one has $M_0(e)\le\max(2,2e-1)$ for
every $e\ge1$.
\end{corollary}

\begin{proof}
For $m\ge\max(2,2e-1)$, $e\in C_m$ since $\lfloor(m+1)/2\rfloor\ge e$, and
Theorem~\ref{thm:kernel} then puts $e$ in $H_m$ for every such $m$.
\end{proof}

\keepstatement
\begin{theorem}[Exact Stabilization]\label{conj:stab}\label{prob:opt}
For every $e\ge2$, $e\notin H_{2e-2}$ and $M_0(e)=2e-1$.
Also $M_0(1)=2$.
\end{theorem}
\begin{proof}
The upper bound is Corollary~\ref{cor:stab}. We prove the last
absence by induction on $e$. For $e=2$, $H_2=\{1\}$.
Let $e\ge3$ and put $c=e-1$. The kernel theorem says that
$H_{2e-3}$ contains all of $1,\ldots,c$ and $H_{2e-4}$ contains
$1,\ldots,c-1$. The induction hypothesis says
$c\notin H_{2e-4}$. Only submasks of $c$ can contribute to the
OR coefficient at $c$.

Write $p=\pc(c)\ge1$. The number of pairs of nonzero submasks
with OR equal to $c$ is $3^p-2$. The second predecessor omits
exactly the submask $c$ among the nonzero submasks relevant here.
Removing the pairs whose second component is $c$ removes $2^p-1$
pairs. The remaining number is
\[
 3^p-2-(2^p-1)=3^p-2^p-1,
\]
which is even. Hence $c$ is absent from the OR--convolution.
Since $c>0$, the constant source does not change it. Increment
therefore leaves $e=c+1$ absent from $H_{2e-2}$. This proves the
induction and makes the upper bound optimal. Index one is present
at every $m\ge2$ by the constant source.
\end{proof}

This determines stabilization, not the first time an index appears:
sporadic earlier memberships can occur. For example $4$ is absent
from $H_6$ and present at every depth from 7 onward; $5$ is absent
from $H_8$ and present at every depth from 9 onward.
Figure~\ref{fig:kernel} compares these times with the growth of
the complete support.

\begin{figure}[htbp]
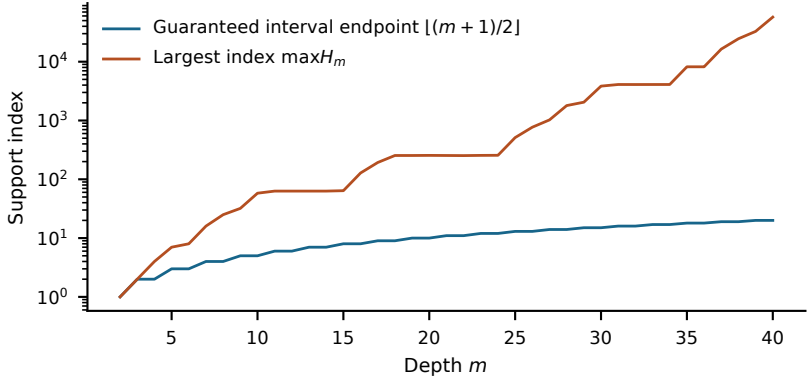
\centering
\EmbeddedFigure{\textwidth}{6}
\caption{Extent of the auxiliary support through
depth $40$. The guaranteed interval ends at $\lfloor(m+1)/2\rfloor$,
while the largest index of $H_m$ can be much farther away. The
logarithmic axis displays both on a common scale. Permanent membership
has exact stabilization time $M_0(e)=2e-1$ for every $e\ge2$;
Figure~\ref{fig:r135mem} displays the last absent cells and the permanent
region. The interval theorem does not describe membership at every
index between the two curves.}
\label{fig:kernel}
\end{figure}

\keepsection
\subsection{Renormalization Points}

Index-by-index stabilization controls any fixed portion of the support
once the diagonal depth is sufficiently large. To compare the part
that continues to change, it is useful to choose depths associated
with powers of two and remove the interval already understood.
The residues defined below retain the remaining digit structure.
Their existence follows from the recurrence, while a concise formula
uniform in the scale would require additional information about
correlations between their indices.

\begin{definition}[Renormalization Point]\label{def:renorm}
For $k\ge0$, let $m(k)=\min\{m\ge2:\max H_m=2^k\}$ and
write $H_{m(k)}=T_k\cup\{2^k\}$ with $T_k\subseteq[0,2^k)$.
We call $T_k$ the residue. It is unrelated to the trinomial
coefficients denoted $T^{(q)}_k$ in other units.
\end{definition}

\keepstatement
\begin{proposition}[Existence of Every Renormalization Point]\label{prop:renormexists}
Each $m(k)$ exists. For $k\ge1$, $m(k)\le2^{k+1}-1$.
\end{proposition}
\begin{proof}
The kernel theorem makes $\max H_m$ unbounded and ensures
$\max H_{2^{k+1}-1}\ge2^k$. At the first depth with maximum
at least $2^k$, both predecessors lie below $2^k$. Their OR values
also lie below $2^k$, and increment produces no index above
$2^k$. The maximum at first crossing is therefore exactly $2^k$.
The case $k=0$ is $H_2=\{1\}$.
\end{proof}

With predecessor supports stored as bitmaps, compute until the first
crossing of $2^k$. Before that crossing, no support index is larger
than $2^k$. Zeta convolution therefore uses windows of size at most
$2^{k+1}$. The construction takes
$O(m(k)(k+1)2^k)$ Boolean operations and $O(2^k)$ working bits,
including the residue bitmap. The proved depth bound gives the
coarser total $O((k+1)4^k)$. This is a terminating construction,
not an assertion of efficient compression of $T_k$.

The term renormalization refers here to successive dyadic thresholds.
It does not assert that the residues at consecutive levels are
rescaled copies. The explicit finite table is
Table~\ref{tab:renorm}.
\begin{table}[t]\centering\fontsize{9}{10.5}\selectfont\setlength{\tabcolsep}{3pt}
\caption{Renormalization depths and residues of the auxiliary recurrence.}
\label{tab:renorm}
\begin{tabular}{|rrrrrrr|}\hline
$k$ & $m(k)$ & $m(k)/k$ & $|T_k|$ & $\max T_k$ & $2^k-1$ & $|T_k|/2^k$\\\hline
2 & 4 & 2.00 & 2 & 2 & 3 & 0.5000\\
3 & 6 & 2.00 & 5 & 7 & 7 & 0.6250\\
4 & 7 & 1.75 & 10 & 15 & 15 & 0.6250\\
5 & 9 & 1.80 & 17 & 31 & 31 & 0.5312\\
6 & 15 & 2.50 & 33 & 62 & 63 & 0.5156\\
7 & 16 & 2.29 & 74 & 127 & 127 & 0.5781\\
8 & 24 & 3.00 & 129 & 253 & 255 & 0.5039\\
9 & 25 & 2.78 & 262 & 511 & 511 & 0.5117\\
10 & 27 & 2.70 & 513 & 1019 & 1023 & 0.5010\\
11 & 29 & 2.64 & 1096 & 2046 & 2047 & 0.5352\\
12 & 34 & 2.83 & 2046 & 4094 & 4095 & 0.4995\\
13 & 36 & 2.77 & 4011 & 8190 & 8191 & 0.4896\\
14 & 37 & 2.64 & 8238 & 16383 & 16383 & 0.5028\\
15 & 39 & 2.60 & 16421 & 32766 & 32767 & 0.5011\\
\hline\end{tabular}
\end{table}

\begin{example}[Near-Mersenne Maxima]\label{ver:mers}
For $2\le k\le15$, the computed gap $2^k-\max T_k$ lies in
$\{1,2,3,5\}$. The Mersenne value occurs at
$k=3,4,5,7,9,14$. Of the eight other levels, the gap is two at
$k=2,6,11,12,13,15$, three at $k=8$, and five at $k=10$.
These are finite values; Proposition~\ref{prop:renormexists}
proves existence of all later levels but does not bound their gaps.
\end{example}

\keepsection
\subsection{Popcount Frequencies and Density}
\label{ver:pop}\label{rem:half}
Let $N_j(k)=|\{s\in T_k:\pc(s)=j\}|$. A uniform subset of a
$k$-bit window at density $|T_k|/2^k$ has expected count
$|T_k|\binom kj/2^k$ in popcount class $j$. We use this as a
reference profile for the deterministic residues, without assuming
independent bits. Table~\ref{tab:popcount} gives the actual errors
for two levels; Figure~\ref{fig:popcount} displays the deviations of
the full normalized profiles from the binomial reference.
\begin{table}[t]\centering\fontsize{9}{10.5}\selectfont\setlength{\tabcolsep}{3pt}
\caption{Popcount counts, binomial reference values and signed relative errors.}
\label{tab:popcount}
\begin{tabular}{|rrrrrrr|}\hline
$j$ & $N_j(15)$ & Ref. & Error \% & $N_j(14)$ & Ref. & Error \%\\\hline
5 & 1516 & 1504.9 & +0.74 & 1009 & 1006.6 & +0.24\\
6 & 2501 & 2508.2 & -0.29 & 1510 & 1509.9 & +0.00\\
7 & 3228 & 3224.8 & +0.10 & 1776 & 1725.6 & +2.92\\
8 & 3175 & 3224.8 & -1.54 & 1487 & 1509.9 & -1.52\\
9 & 2538 & 2508.2 & +1.19 & 963 & 1006.6 & -4.33\\
\hline\end{tabular}
\end{table}

In these rows the largest absolute relative error is about $4.33\%$.
Across all $8\le k\le15$, even the central three popcount classes
need not be within $2\%$ of the reference profile. Density also varies
nonmonotonically: $|T_{12}|/2^{12}\approx0.4995$ and
$|T_{13}|/2^{13}\approx0.4896$, followed by values approximately
$0.5028$ and $0.5011$ at levels 14 and 15. This finite agreement
with one-half does not establish a limit, independence, or a lower
bound for the block complexity of these particular sets.

\begin{figure}[htbp]
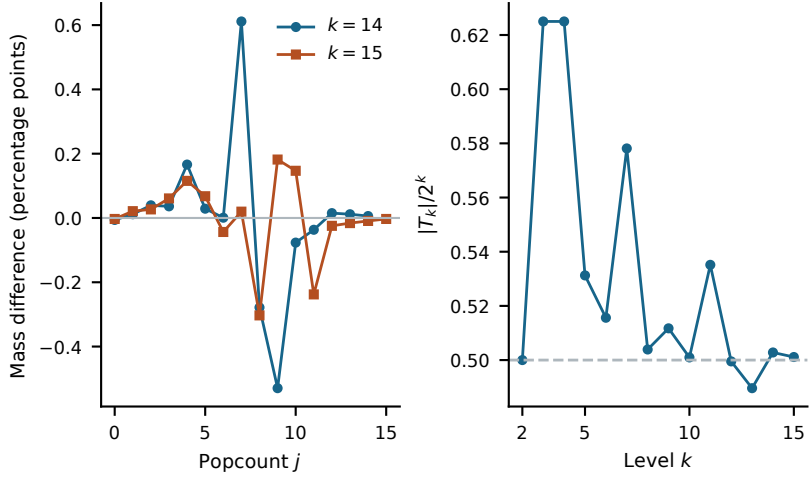
\centering
\EmbeddedFigure{\textwidth}{7}
\caption{Residues of the auxiliary recurrence. Left:
$100\bigl(N_j(k)/|T_k|-\binom kj/2^k\bigr)$ for $k=14,15$, in
percentage points. Subtracting the reference mass makes the
differences visible; this normalization differs from the relative
errors in Table~\ref{tab:popcount}. Right: densities for
$2\le k\le15$, with one-half marked. These finite deterministic
data do not assert independence or a limiting density.}
\label{fig:popcount}
\end{figure}

\keepsection
\subsection{Evaluation and the Cost of the Kernel}
For any supplied decomposition of $H_m$, its auxiliary sequence
$E(H_m)(n)$ is the XOR of the corresponding Lucas tests. With
$K(m)$ supplied terms the query cost is $O(K(m)(b+\ell))$.
The interval in Theorem~\ref{thm:kernel} gives a part of this
representation whose size is bounded independently of the residue.

\keepstatement
\begin{proposition}[Block Cost of the Kernel]\label{prop:cheap}
$C_m=\{1,\dots,\lfloor(m+1)/2\rfloor\}$ is a contiguous interval, and every interval
$[1,N]$ decomposes into at most $2\lceil\log_2(N+1)\rceil$ masked dyadic blocks. Hence the
kernel contributes $O(\log m)$ terms to the evaluation, independently of $|H_m|$.
\end{proposition}

\begin{proof}
Write $N+1=\sum_{i\in E}2^i$. In decreasing order of $i\in E$, emit
$B(v_i,2^i-1)$, where $v_i=\sum_{j\in E,\ j>i}2^j$.
These intervals $[v_i,v_i+2^i)$ partition $[0,N]$. Each block is
valid because $v_i\wedge(2^i-1)=0$. Their number is
$|E|\le\lceil\log_2(N+1)\rceil$. XOR with the singleton $\{0\}$
removes zero and gives $[1,N]$, using at most
$\lceil\log_2(N+1)\rceil+1\le2\lceil\log_2(N+1)\rceil$ blocks
for $N\ge1$.
\end{proof}

Thus $H_m=C_m\symd(H_m\setminus C_m)$ separates a part costing
$O(\log m)$ block terms from the remaining finite support. The
interval construction takes $O((\log(m+2))^2)$ bit operations,
including its binary output masks. The cost of compressing the
remaining support is unresolved. A large finite value of $|H_m|$
does not itself prove exponential growth, and the almost-all bound
of Theorem~\ref{thm:lupanov} does not apply to $T_k$ merely because
its measured density is near one-half. Figure~\ref{fig:r135mem}
shows the permanent region and the earlier sporadic memberships.

\begin{figure}[htbp]
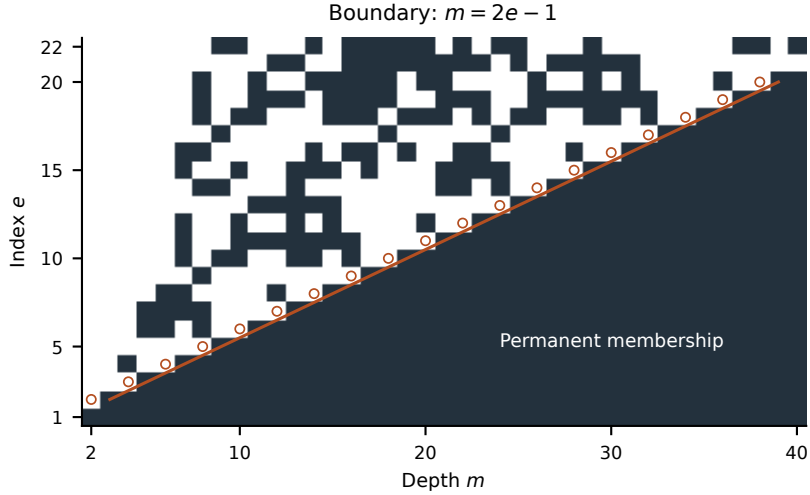
\centering
\EmbeddedFigure{\textwidth}{8}
\caption{Membership of $e$ in $H_m$ for
$2\le m\le40$, $1\le e\le22$; dark cells indicate membership.
The orange line marks $m=2e-1$. The outlined white cells at
$m=2e-2$ are the last absences for $2\le e\le20$ in this window.
Every subsequent cell in the same row is present, as proved in
Theorem~\ref{conj:stab}. Earlier dark cells need not be permanent.
Both axes use integer indices.}
\label{fig:r135mem}
\end{figure}

\keepsection
\subsection{The auxiliary integer lift}\label{sec:auxinteger}

The finite auxiliary supports also have an integer counterpart.
It is useful to retain it explicitly, since its initial diagonal
polynomials specify a different boundary problem from the physical
single-seed Rule 135 orbit.

\begin{proposition}[Integer realization of the auxiliary recurrence]\label{prop:auxinteger}
Define $A_1(n)=1$, $A_2(n)=n$ and, for $m\ge3$,
\[
 A_m(0)=0,\qquad \Delta A_m(n)=1+A_{m-1}(n)A_{m-2}(n).
\]
These are integer-valued polynomials with nonnegative Newton
coefficients $D_{m,r}$ and degree $F_{m+1}-1$. Their odd
coefficients have support $H_m$. Here $D_{m,0}=0$ for $m\ge3$, and for $r\ge1$,
\[
 \begin{aligned}
 &D_{m,r}-[r=1]\\
 &\quad=
 \sum_{\substack{0\le j,k\le r-1\\j+k\ge r-1}}
 D_{m-1,j}D_{m-2,k}
 \frac{(r-1)!}{(j+k-r+1)!(r-1-j)!(r-1-k)!}.
 \end{aligned}
\]
\end{proposition}
\begin{proof}
The binomial product identity counts two subsets by their union,
with the displayed factorial coefficient. Discrete summation
raises each binomial index by one, including the constant source.
This proves the recurrence and nonnegativity. The product has a
positive leading coefficient, so its degree cannot cancel;
$d_m=d_{m-1}+d_{m-2}+1$ with $d_1=0,d_2=1$ gives the stated degree.
Reducing the initialized difference equation modulo two gives
exactly $H_m=\operatorname{Inc}(\{0\}\symd(H_{m-1}\ast H_{m-2}))$.
Uniqueness of the Newton transform identifies the odd coefficients.
\end{proof}

For example, $A_3=\binom n1+\binom n2$ and
$A_4=\binom n1+\binom n2+4\binom n3+3\binom n4$.
The first has degree two, already distinguishing the correct
$F_{m+1}-1$ law from an index-shifted degree formula.

\subsection{Exact block correlations}\label{sec:blockdensity}

A supplied block representation also determines the density of its
binary evaluations in a dyadic window. An individual block fixes
some input bits to one and others to zero. Several blocks can
constrain the same bit, so their evaluations must be counted jointly.

\begin{proposition}[Density from compatible bit constraints]\label{prop:blockdensity}
Let $S=\symd_{i=1}^K B(v_i,M_i)$ with valid masks $v_i,M_i<2^b$,
and let $f(n)=\Xor_{r\in S}\binom nr\pmod2$. For a nonempty
$J\subseteq\{1,\ldots,K\}$ put
$V_J=\bigvee_{i\in J}v_i$ and $Z_J=\bigvee_{i\in J}M_i$.
The exact proportion of ones in $0\le n<2^b$ is
\begin{equation}\label{eq:blockdensity}
 \begin{aligned}
 &\frac{\#\{n<2^b:f(n)=1\}}{2^b}\\
 &\quad=\sum_{\varnothing\ne J\subseteq\{1,\ldots,K\}}
 (-2)^{|J|-1}[V_J\wedge Z_J=0]\,
             2^{-\operatorname{pc}(V_J\vee Z_J)}.
 \end{aligned}
\end{equation}
\end{proposition}
\begin{proof}
Write $X_i(n)=[v_i\subseteq n][n\wedge M_i=0]$, the block
evaluation. As an equality of real-valued indicators,
$f=(1-\prod_i(1-2X_i))/2$. Expanding the product gives the
coefficient $(-2)^{|J|-1}$ for $\prod_{i\in J}X_i$.
That joint event is impossible if $V_J\wedge Z_J\ne0$.
Otherwise it fixes exactly $\operatorname{pc}(V_J\vee Z_J)$
independent binary positions and has the displayed probability.
Averaging the expansion proves the formula.
\end{proof}

If the footprints $v_i\vee M_i$ are pairwise disjoint, the
$X_i$ are independent and the formula reduces to
\[
 \frac12-\frac12\prod_{i=1}^K
       \left(1-2^{1-\operatorname{pc}(v_i\vee M_i)}\right).
\]
Disjointness of the support blocks themselves does not suffice.
For example, $B(1,0)=\{1\}$ and $B(2,1)=\{2,3\}$ are disjoint,
but their evaluation events are respectively $n_0=1$ and
$(n_1,n_0)=(1,0)$. On $0\le n<4$ their XOR has density $3/4$;
the independence formula would give $1/2$. Equation~\eqref{eq:blockdensity}
retains the missing correlations. Even a proved density law for
fixed-depth diagonals would require an additional argument before
being applied to the moving center column.

\section{Further work and conclusion}

The backward profile argument proves that regular-side period records
continue indefinitely. The remaining issue is their spacing. For the
indices $a_j$ of Definition~\ref{def:cdv}, the bound
$a_j\le4^{2^{j-1}}$ for $j\ge1$ leaves room for much sharper estimates.
A useful next result would control the time between resets in the
profile dynamics. The proposed next record after $87868$, namely
$2107985256$, lies beyond the computed range and is not an established
term of the orbit's record sequence.\label{prob:records}

For the auxiliary recurrence, Theorem~\ref{conj:stab} determines the
permanent interval exactly. Its removal isolates the residues $T_k$.
Proposition~\ref{prop:renormexists} already constructs these sets in
finite time; a more informative description would reveal their bit
correlations and control the length of a compressed formula directly
from $k$. The associated renormalization depths $m(k)$ may then be
compared with a law of the form $m(k)=\alpha k+o(k)$.
The measured ratios for $8\le k\le15$ lie in $[2.60,3.00]$, with
$m(7)/7=16/7$; the last four are approximately $2.83,2.77,2.64,2.60$.
These values do not determine a limiting constant near $2.6$ or $2.7$.
\label{prob:closed}\label{prob:asym}

A statistical study of $T_k$ should begin with joint bit constraints.
Equation~\eqref{eq:blockdensity} gives an exact model of the correlations
that an independence approximation can miss. For the residues, natural
quantities are $|T_k|/2^k$ and the distribution of
$(\operatorname{pc}(t)-k/2)/(\sqrt{k}/2)$ for uniform $t\in T_k$.
Convergence to $1/2$ for the former or to the standard normal law for
the latter would require estimates uniform in $k$. Finite histograms
alone provide neither estimate.\label{prob:stat}

The upper endpoint poses a related arithmetic question: controlling
$2^k-\max T_k$ and the frequency of $2^k-1\in T_k$ would distinguish
persistent near-endpoint structure from the observed occurrences
$k\in\{3,4,5,7,9,14\}$ through $15$. An exact recurrence for these
boundary bits would be more useful than extrapolation from their
finite frequency.\label{prob:mers}

\subsection{Consequences for the Diagonal Sequences}

Unbounded diagonal periods follow here from two features fixed by
the single-seed orbit: the position of first contact and the eventual
boundary profiles. On the finite-support side, the earliest nonzero
value gives the least Newton index, while the polynomial lift gives
a Fibonacci upper bound for the largest. The exact support-period
relation converts these into period information. On the reflected
side, the backward profile map turns a hypothetical bound on all
periods into a finite state space with too many distinct predecessors
of the zero boundary. The second argument explains why the
reset/integration mechanism must produce further record increases,
even when no transient bound is available.

The boundary analysis also determines how Rule 135 enters this picture.
Its physical single-seed orbit follows from a translated-complement
identity after the first step. The auxiliary recurrence with two
finite diagonal seeds has its own exact stabilization law,
$M_0(e)=2e-1$ for the indices covered by the theorem. Keeping the
two constructions distinct allows the physical conjugacy and the
permanent-interval result to be used with their correct hypotheses.
It also explains why similar-looking recurrences can have different
support descriptions when their boundary values differ.

The block formulas address evaluation after those descriptions have
been obtained. They replace a sum over individual support indices
by a collection of digit tests and state the cost in terms of the
number and width of the supplied blocks. The remaining difficulty
is to control those quantities as depth increases. A support-period
formula alone does not determine block count, and a finite density
measurement does not determine the correlations on which compression
depends.

The sharpest next period question is the spacing of the record
indices. The profile count proves that records continue, but leaves
a large gap between its bound and the finite tables. For the
auxiliary recurrence, exact stabilization shifts attention to the
renormalized residues and their minimum block descriptions.
Information about those sets would improve the corresponding
evaluation bounds; density or popcount alone does not determine
their correlations. These are separate ways to strengthen the
results proved here: refine the period growth, describe the moving
residue, or control the cost of obtaining its representation. The
eventual behavior of Rule 30's moving center column requires an
argument that also follows the changing diagonal index.


\begin{thebibliography}{99}\footnotesize
\setlength{\itemsep}{1pt}
\bibitem{wolfram1983}
S. Wolfram, ``Statistical Mechanics of Cellular Automata,'' \textit{Reviews of
Modern Physics}, \textbf{55}(3), 1983 pp. 601--644.
doi:10.1103/RevModPhys.55.601.

\bibitem{rowland}
E. S. Rowland, ``Local Nested Structure in Rule 30,'' \textit{Complex Systems},
\textbf{16}(3), 2006 pp. 239--258. doi:10.25088/ComplexSystems.16.3.239.

\bibitem{oeis}
OEIS Foundation. ``The On-Line Encyclopedia of Integer Sequences, A364239.''
\url{https://oeis.org/A364239}.

\bibitem{kopra}
J. Kopra, ``Rapid Left Expansivity, a Commonality between Wolfram's Rule 30 and
Powers of $p/q$,'' \textit{Theoretical Computer Science}, \textbf{946}, 2023
article 113668. doi:10.1016/j.tcs.2022.12.018.

\bibitem{pivato}
M. Pivato, ``Defect Particle Kinematics in One-Dimensional Cellular Automata,''
\textit{Theoretical Computer Science}, \textbf{377}(1--3), 2007 pp. 205--228.
doi:10.1016/j.tcs.2007.03.014.

\bibitem{U1}
T. Nersissian, ``The Support-Set Calculus and the Algebra of Dyadic Blocks,''
companion manuscript, revised September 2026.

\bibitem{mow}
O. Martin, A. M. Odlyzko, and S. Wolfram, ``Algebraic Properties of Cellular
Automata,'' \textit{Communications in Mathematical Physics}, \textbf{93}(2),
1984 pp. 219--258. doi:10.1007/BF01223745.

\bibitem{willson}
S. J. Willson, ``Cellular Automata Can Generate Fractals,'' \textit{Discrete
Applied Mathematics}, \textbf{8}(1), 1984 pp. 91--99.
doi:10.1016/0166-218X(84)90082-9.

\bibitem{lupanov2}
O. B. Lupanov, ``A Method of Circuit Synthesis,'' \textit{Izvestiya Vysshikh Uchebnykh Zavedenii, Radiofizika}, \textbf{1}, 1958 pp. 120--140.

\bibitem{wegener}
I. Wegener, \textit{The Complexity of Boolean Functions}, Chichester: Wiley--Teubner, 1987.

\bibitem{prize}
S. Wolfram. ``Announcing the Rule 30 Prizes.'' (Oct 1, 2019)
\url{https://writings.stephenwolfram.com/2019/10/announcing-the-rule-30-prizes/}.

\bibitem{U6}
T. Nersissian, ``Monotone Branch Rigidity in Conditional Cellular Automata,''
companion manuscript, revised September 2026.

\bibitem{wall}
D. D. Wall, ``Fibonacci Series Modulo m,'' \textit{The American Mathematical Monthly},
\textbf{67}(6), 1960 pp. 525--532. doi:10.1080/00029890.1960.11989541.

\bibitem{jahnel}
A.-S. Elsenhans and J. Jahnel, ``The Fibonacci Sequence Modulo $p^2$---An
Investigation by Computer for $p<10^{14}$,'' arXiv:1006.0824, 2010.
\end{thebibliography}
\end{document}